\documentclass[preprint, 12pt,1p, proof, authoryear]
{elsarticle}
\makeatletter
\def\ps@pprintTitle{%
 \let\@oddhead\@empty
 \let\@evenhead\@empty
 \def\@oddfoot{}%
 \let\@evenfoot\@oddfoot}
\makeatother

\usepackage{tikz}
\usetikzlibrary{patterns}
\usepackage{longtable, booktabs,multirow,lipsum,array}
\usepackage[authoryear, round]{natbib}
\usepackage{amssymb}
\usepackage{amsmath}
\usepackage{enumerate}
\usepackage{mwe}
\usepackage[toc, header]{appendix}
\usepackage{mathtools}
\usepackage{graphicx}
\usepackage[dvipsnames]{xcolor}
\usepackage{soul}
\usepackage{hyperref}
\usepackage{amsfonts}
\usepackage{graphicx}
\usepackage[utf8]{inputenc}
\usepackage[all]{xy}
\usepackage{csquotes}
\usepackage{ulem}
\usepackage{amsmath}
\usepackage{amsthm}
\usepackage{import}
\usepackage{adjustbox}
\usepackage{caption}
\usepackage{pdflscape}
\usepackage{tikzpeople}
\usepackage{wasysym}
\usepackage{setspace}
\usepackage{booktabs}
\usepackage{arydshln}
\usepackage{comment}
\usepackage{mathbbol}

\usepackage{subcaption} 
\usepackage{hyperref} 
\usepackage{rotating} 
\usepackage{prodint}

\usepackage{ mathrsfs }
\usepackage[flushleft]{threeparttable}
\usepackage{algorithm,algpseudocode}

\usepackage[sectionbib]{bibunits}
\defaultbibliography{references}
\defaultbibliographystyle{plainnat}

\usepackage{titlesec}

\usepackage{enumitem}

\usepackage{titletoc}

\newcommand\DoToC{%
  \startcontents
  \printcontents{}{1}{\textbf{ }\vskip3pt\hrule\vskip5pt}
  \vskip3pt\hrule\vskip5pt
}

\usepackage{stmaryrd}
\usepackage{arydshln}

\usepackage{tikz}
\usetikzlibrary{positioning}

\usepackage{algorithm}
\usepackage{algpseudocode}
\usepackage{booktabs}

\usepackage{geometry}
\hypersetup{
    colorlinks=true,  
    linkcolor=green, 
    citecolor=red,
    filecolor=magenta, 
    urlcolor=cyan, 
}

\usepackage{tocloft}

\newcommand{\expectation}{\mathbb{E}}

\newcommand{\prstudyp}{
  \scalebox{0.5}{
      \begin{tikzpicture}
        \draw (0,0) -- 
              ++(0.4,0) --
              ++(0,0.3) -- cycle;
      \end{tikzpicture}
    }
}

\newcommand{\prgapp}{
  \scalebox{0.5}{
      \begin{tikzpicture}
        \draw (0,0) --
              ++(0,0.3) -- 
              ++(0.4,0) -- 
              cycle; 
      \end{tikzpicture}
    }
}

\newtheorem{lemma}{Lemma}
\newtheorem{proposition}{Proposition}
\newtheorem*{proposition*}{Proposition}

\newtheorem{corollary}{Corollary}
\newtheorem{definition}{Definition}

\newcommand\independent{\protect\mathpalette{\protect\independenT}{\perp}}
\def\independenT#1#2{\mathrel{\rlap{$#1#2$}\mkern2mu{#1#2}}}

\newcommand{\norminf}[1]{{\sup\limits_{s \in \left[0, t\right]} \left| #1 \right|}}

\newcommand{\norminfdiscrete}[1]{{\sup\limits_{m \in \{0, \dots, k \}} \left| #1 \right|}}

\newcommand{\CTC}{(CTC)}

\newcommand{\RExch}{{(\prstudyp - \text{e})}}
\newcommand{\RCons}{{(\prstudyp - \text{c})}}
\newcommand{\RPos}{{(\prstudyp - \text{p})}}

\newcommand{\RAID}{(\prstudyp-{\text{ecp}})}
\newcommand{\FAID}{ 
(\prstudyp\prgapp-\text{ecp})
}

\newcommand{\hi}{{HI}}

\newcommand{\mb}{
  \mathcal{M}\left(\mathbb{B}\right)
}
\newcommand{\mbs}{
  \mathcal{M}\left(\mathbb{B}_{\prstudyp}\right)
}
\newcommand{\mbf}{
  \mathcal{M}\left(\mathbb{B}_{\prstudyp\prgapp}\right)
}

\newcommand{\mbps}{
\mathcal{M}(\mathbb{B}^*) 
}

\newcommand{\mbcirc}{
\mathcal{M}(\mathbb{B}^{\circ}) 
}

\newcommand{\RNPcont}{{(\prstudyp-\text{np-cont})}}

\newcommand{\ctFull}{{\prstudyp\prgapp - \text{ctc}}}

\newcommand{\ctPost}{{(\neg\prstudyp)\prgapp - \text{ctc}}}

\newcommand{\ass}[1]{\textcolor{blue}{#1}}

\begin{document}

\begin{frontmatter}
    \title{Causal inference with staggered entries and effects that change over calendar time} 
\author{
\vspace{3ex}
Lorenzo Gasparollo\footnotemark[1] 
\textit{and} Mats J. Stensrud
}

\address{\vspace{4ex}
Institute of Mathematics, Ecole Polytechnique Fédérale de Lausanne, Switzerland
\vspace{-4ex}}

\footnotetext[1]{\textbf{First and corresponding author:}
Lorenzo Gasparollo, Institute of Mathematics, Ecole Polytechnique Fédérale de Lausanne, Switzerland. \url{lorenzo.gasparollo@epfl.ch}}

    \begin{abstract}
Studies with staggered entry, in which individuals enroll at different calendar times, are ubiquitous in medicine and related disciplines. Because these studies usually have a fixed administrative end of follow-up, identification of the estimand of interest relies on assumptions about the right-censoring mechanism. The assumptions are often considered plausible only conditional on covariates, including the time an individual entered the study ($E$). 
Yet, censoring assumptions formulated conditional on $E$ are ill-posed due to positivity violations. Here, we study the consequences of this issue for common procedures in causal survival analysis, such as those based on marginal structural models. We further give conditions for valid identification and introduce sensitivity analyses to assess practical implications. We illustrate our methodology through two case studies. The first builds on the seminal article by Hernán et al. (2000) on the effect of zidovudine treatment. The second reanalyzes a recent study on the effect of the mRNA vaccine in patients receiving immune checkpoint inhibitors.
\end{abstract}

\begin{keyword}
Causal inference, survival analysis, censoring, staggered entries, sensitivity analysis.
\end{keyword}

\end{frontmatter}

\clearpage
\setstretch{1}

\begin{bibunit}
    \section{Introduction}
\label{sec: introduction}
A characteristic feature of studies with staggered entries, which are common in medicine and related fields, is that individuals are right-censored when the study ends. 

This type of censoring is often referred to as administrative, generalized type I, or end-of-follow-up censoring. Investigators who analyze staggered-entry data routinely rely on independent censoring assumptions \citep{andersen_statistical_1993, klein_survival_2003} or, alternatively, sequential exchangeability assumptions, i.e., independence assumptions involving potential outcomes under interventions on censoring \citep{robins_new_1986, hernan_causal_2025}. These \textit{censoring assumptions} permit us to draw inference about outcomes over time, even when complete outcome data are unavailable for some individuals.

Censoring assumptions are often considered plausible conditional on covariates. In particular, many studies adjust for the calendar time at which an individual entered the study, $E$, thereby operating under censoring assumptions formulated conditional on $E$. The reason is that adjustment for $E$ is needed when $E$ depends on the outcome of interest, as administrative censoring itself depends on $E$, see, for example, \citet[p. 230]{hernan_causal_2025} and \citet{robins_recovery_1992, robins_g-estimation_1992, robins_estimation_1992, robins_analytic_1993, robins_adjusting_1994}. 
Such a dependence reflects calendar-time changes, that is, variation in risk over calendar time, which has historically been modeled as a period effect \citep{clayton_models_1987} and referred to as ``calendar-time trends'' \citep{hansen_estimating_2017} or ``secular trends'' \citep{hernan_causal_2025}. Adjustment for $E$ is common in studies based on classical survival analysis methods (e.g., \citet{van_not_improving_2024, harrysson_temporal_2025, lian_survival_2025}) as well as in target trial emulations (e.g., \citet{garcia-albeniz_effectiveness_2017, dickerman_avoidable_2019}). The seminal works of \citet{hernan_marginal_2000, hernan_observational_2008}, which have received thousands of citations and inspired a large literature on causal inference from observational data in which adjustment for $E$ is common, also adjusted for $E$.

We will show, however, that conventional adjustment for $E$ is often problematic under calendar-time changes: conditioning on $E$ renders conventional censoring assumptions ill-posed, whereas excluding $E$ leads to violations of the conventional assumptions. 

Our aim is therefore to assess the consequences of conditioning on $E$ and to develop procedures permitting valid identification of the parameter corresponding to the mathematical formulation of the research question of interest, i.e., the estimand of interest. Methods introduced by Robins, such as g-estimation of structural nested failure time models \citep{robins_analytic_1993,robins_adjusting_1994} and parametric g-computation \citep{robins_new_1986}, can be used to adequately adjust for $E$.

In this work, however, we study two widely used classes of procedures for causal inference, for example those routinely used in target trial emulations. In practice, adjustment for $E$ is carried out for both classes, although its justification is often unclear.

The first class, hereafter referred to as the weighted survival curve (WSC) procedures, includes methods based on the Kaplan–Meier estimator \citep{kaplan_nonparametric_1958} and its stratified \citep{cupples_ageadjusted_1995} and weighted \citep{cole_adjusted_2004,xie_adjusted_2005} formulations, which are often applied in medical and epidemiologic studies \citep{westreich_time_2010, weintraub_comparative_2012, bleicher_time_2016,kishan_radical_2018, pouwels_estimating_2020}. We show that, without relying on untestable conditions that are often scientifically implausible, WSC procedures generally target parameters that differ from the estimand of interest, and we give an explicit expression for this difference. The second class, hereafter referred to as MSM-hazard procedures, concerns the parameters of a Marginal Structural Model (MSM) \citep{hernan_marginal_2001, hernan_observational_2008}. In practice, these procedures are widely used in target trial emulations and cover (an approximation of) the proportional hazards model \citep{cox_regression_1972} as a special case. Previous work has suggested that MSM-hazard procedures cannot address positivity violations \citep{moore_causal_2012} or that they can but under unspecified conditions \citep{cole_constructing_2008}. Here, we articulate explicit assumptions under which identification of MSM-hazard parameters is possible.

For both classes of procedures, we elaborate on the assumptions required for valid identification, which we term \textit{extrapolation assumptions}, as they rely on extrapolations beyond the support of the observed data. 

Because extrapolation assumptions are generally not expected to hold under calendar-time changes, we propose sensitivity analyses that quantify the minimum post-study change in outcome risk sufficient to nullify the inferred effect, thereby allowing investigators to assess the robustness of their conclusions to the extrapolation assumptions.

We establish most of our results within a discrete-time potential-outcome framework for causal inference, with further results derived in a continuous-time counting-process framework. These results are illustrated through two examples.
The first is a simulation study inspired by the seminal work of \citet{hernan_marginal_2000} on the effect of zidovudine treatment on survival among HIV-positive men. It illustrates our theoretical results and provides additional support for their findings. The second is a target trial emulation analysis of a recent study published in Nature on the effect of SARS-CoV-2 mRNA vaccination in patients receiving immune checkpoint inhibitors \citep{grippin_sars-cov-2_2025}. It demonstrates the interpretation and implications of extrapolation assumptions in a target trial emulation analysis.

    \section{Censoring assumptions articulated conditional on $E$ are ill-posed}
\label{sec: censoring_assumptions_are_illposed}

\subsection{Notation and observed data structure}
Consider a population of i.i.d. individuals for whom
we measure variables at discrete time points $k \in \{0, \dots, K\}$ \textit{from} the calendar time they entered the study, $E$, with $E \in \{0, \dots, \mathcal{T}\}$, where $\mathcal{T}$ denotes the (translated) calendar time of the administrative end of study. Let $K, \mathcal{T} > 0$. 

At each $k \in \{0, \dots, K\}$, let $L_k$ denote baseline (for $k=0$) and time-varying (for $k>0$) covariates; $C_{k-1}$ the right-censoring indicator (1 right-censored, 0 uncensored, with $C_{-1}=0$); $A_k$ the treatment indicator (1 treatment, 0 control); and $Y_k$ the event indicator (1 death, 0 survival). We assume that $Y_k = 1 \Rightarrow Y_{k+1}=1$ and $C_{k-1} = 1 \Rightarrow C_{k}=1$ a.s. for $k \in \{0, \dots, K-1\}$, and adopt a temporal order $(L_k, C_{k-1}, A_k, Y_k)$ for each $k \in \{0, \dots, K\}$; any variable in $(L_m, C_{m-1}, A_m, Y_m)$ for $m \notin \{0, \dots, K\}$ is considered to be $\emptyset$.

We let overbars denote treatment, covariate, or survival histories, for example, $\overline{A}_{k} \equiv \{A_m : m \in \{0, \dots, k\}\}$, and let $\overline{0}$ and $\overline{1}$ denote vectors of $(k+1)$ zeros and ones, respectively, for $k \in {0,\dots,K}$. To lighten notation, for events $\overline{Y}_k =  \overline{0}$ and $\overline{C}_{k-1} = \overline{0}$ we write $Y_k=0$ and $C_{k-1}=0$, respectively.

Finally, we will also give results for the special case of a point treatment $A$ in place of $\overline{A}_K$, with no time-varying confounders ($L_{k} = \emptyset, k > 0$).

The remaining notation is introduced as needed throughout the article, and Table \ref{tab: nomenclature} in Appendix \ref{sec_app: guide_appendices} describes all the symbols succinctly.

\subsection{Illustrative example}
\label{sec: illustrative_example_censoring_assumptions_ill_posed}
\subsubsection{Exchangeability between counterfactual outcomes as a censoring assumption}
Consider the effect of zidovudine treatment on survival in human immunodeficiency virus-positive men, as studied in the seminal paper by \citet{hernan_marginal_2000}.\footnote{Our arguments also extend to studies adopting more recent formulations of the target trial emulation framework. Since individuals can be eligible in multiple (non-randomized) nested trials \citep{hernan_observational_2008}, the individual-trial-specific calendar time of eligibility can be conceptualized as the calendar time of study-entry $E$ in a population generated from the original target population by cloning individuals every time they are eligible in one of the nested trials.} 
One way to conceptualize this effect is through an intervention on two variables at each $k \in \{0, \dots, K\}$: the treatment under scientific investigation, $A_k$, and the right-censoring indicator, $C_{k-1}$.
\footnote{
Introducing interventions on censoring may be unfamiliar to some readers outside the causal inference literature. However, this is a standard representation of right censoring in the potential-outcomes framework (see, e.g., \citealp{robins_correcting_2000, hernan_causal_2025}) and is adopted here to make the identifying assumptions explicit.
} 
The intervention on $A_k$ evaluates the effect of assigning the treatment under scientific investigation according to a treatment strategy, $a_k$, that differs from that observed in the data, while the intervention that sets $C_{k-1}$ to 0 prevents loss to follow-up and administrative censoring.

For $\overline{a}_k\in\{\overline{0},\overline{1}\}$, let $Y_k^{\overline{a}_k, {c}_{k-1}=0}$ denote the outcome that would occur at follow-up $k$ under such an intervention, and consider an estimand traditionally of interest,
\begin{align}\label{eq: survival_contrast}
S_k^{\overline{a}_k} \coloneqq 
P \{Y_k^{\overline{a}_k, {c}_{k-1}=0} = 0\}, \quad
\text{\footnotesize $k\in\{0,\dots,K\}.$}
\end{align}
In words, $S_k^{\overline{a}_k}$ is the survival function, hereafter \textit{survival}, resulting from an intervention where, possibly contrary to fact, individuals received treatments $\overline{a}_k$ (e.g., zidovudine if $\overline{a}_k=\overline{1}$) and remained uncensored through follow-up $k$.

When adjustment for $E$ is deemed necessary for the identification of $S_k^{\overline{a}_k}$ from the observed data, as in \citep{hernan_marginal_2000}, consider the classical
\textbf{ex}changeability (\ass{ex}), \textbf{con}sistency (\ass{con}), and \textbf{pos}itivity (\ass{pos}) assumptions \citep{ hernan_causal_2025}, 

which can be articulated by requiring
\begin{itemize}[leftmargin=2cm, labelsep=0.1cm]
    \item[\ass{ex}$(e,k)$]:  {\footnotesize$\left\{Y_K^{\overline{a}_K, {c}_{K-1}=0}, \dots, Y_k^{\overline{a}_k, c_{k-1}=0} \right\}
\independent
\mathbb{1}_{\{A_k = a_k, C_{k-1}=0\}}
\mid \overline{L}_k, \overline{A}_{k-1}=\overline{a}_{k-1}, {Y}_{k-1} = {C}_{k-2} = 0, E=e, \, \overline{a}_K\in\{\overline{0},\overline{1}\}$,}
\item[\ass{con}$(e,k)$]:   {\footnotesize $E = e, \overline{A}_{k} = \overline{a}_k, {C}_{k-1}={0}  \Longrightarrow Y_k^{\overline{a}_k, c_{k-1}=0} = Y_k, \, \overline{a}_k\in\{\overline{0},\overline{1}\}$}, and
\item[\ass{pos}$(e,k)$]: {\footnotesize$P\left( {A}_{k}=a_k, C_{k-1}=0  \mid \overline{L}_k, {Y}_{k-1} = 0, \overline{A}_{k-1}=\overline{a}_{k-1}, {C}_{k-2}=0, E=e\right) > 0, \, \overline{a}_k\in\{\overline{0},\overline{1}\},$} 
\end{itemize} 
to hold a.s. for every $(e,k) \in \{0, \dots, K\}\times \{0, \dots, \mathcal{T}\}$. To simplify the presentation, we wrote $E$ and $L_0$ separately, even if $E \in L_0$. We use this convention throughout the remainder of Section \ref{sec: censoring_assumptions_are_illposed}.

\subsubsection{The problem: conditioning on $E$}
\label{sec: the_problem_conditioning_on_e}
The articulation of the positivity assumption \ass{pos} depends on that of the exchangeability assumption \ass{ex}: under \ass{con}, \ass{pos} is formulated to characterize the well-definedness of \ass{ex}. As discussed by \citet{robins_correcting_2000}, \ass{ex} serves the same purpose as independent censoring assumptions \citep{andersen_statistical_1993, klein_survival_2003} in that it enables identification of the parameter of interest (e.g., survival) from incomplete observations. But a formulation of \ass{ex} conditional on $E$ is problematic because of a violation of positivity (\ass{pos}): for individuals who enter the study late in calendar time, we cannot observe outcomes far into their follow-up if they are event-free at the administrative end of the study, i.e., if they experience neither failure nor right-censoring before administrative censoring. We will subsequently make this issue precise, study its consequences, and consider potential solutions. 
 
To fix ideas, consider the case where $\mathcal{T}=K=2$ and take any $l_0 \in \textnormal{supp}(P(L_0))$, such that, say, $e=1$. Because entries are staggered, individuals who entered at calendar time $\tau=(e=1)$, and who remained alive and uncensored by follow-up $k=1$, $\{E=1, Y_1=C_{0}=0\}$, will necessarily be censored at follow-up $k=2$, that is, $C_1=1$. Figure \ref{fig: secc} illustrates this phenomenon, which is a violation of positivity (\ass{pos}). More elaborately, we have that
{\footnotesize
\begin{multline}
    0 \leq P(A_{\textcolor{teal}{k\equiv 2}} =1, C_{\textcolor{teal}{k-1 \equiv 1}} = 0 \mid {L}_2={l}_2, Y_1 = 0, \overline{A}_1 = \overline{1},   {C}_0 = 0, L_1 = l_1, L_0 =l_0, \textcolor{teal}{E=1}) \\
    \leq P(C_{1} = 0 \mid {L}_2={l}_2, Y_1 = 0,  \overline{A}_1 = \overline{1}, {C}_0 = 0, L_1 = l_1, L_0=l_0, \textcolor{teal}{E=1}) =0 , \\  \text{ for $P-$ almost all $l_2, l_1$},
\end{multline}
}
i.e., $P(A_{\textcolor{teal}{k\equiv 2}} =1, C_{\textcolor{teal}{k-1 \equiv 1}} = 0 \mid {L}_2={l}_2, Y_1 = 0, \overline{A}_1 = \overline{1}, {C}_0 = 0, L_1 = l_1, L_0=l_0, \textcolor{teal}{E=1}) =0$ for \text{$P-$ almost all $l_2, l_1$}, i.e., \ass{pos}($e=1,k=2$), and thus \ass{pos}, is violated. This implies that the conditioning sets defining \ass{ex} lie outside the support of the observed data distribution, and therefore \ass{ex} is ill-posed. See Appendix \ref{sec_app: censoring_assumptions_meaning_violation} for a formal justification.
\begin{figure}[H]
\includegraphics[width=13.5cm]{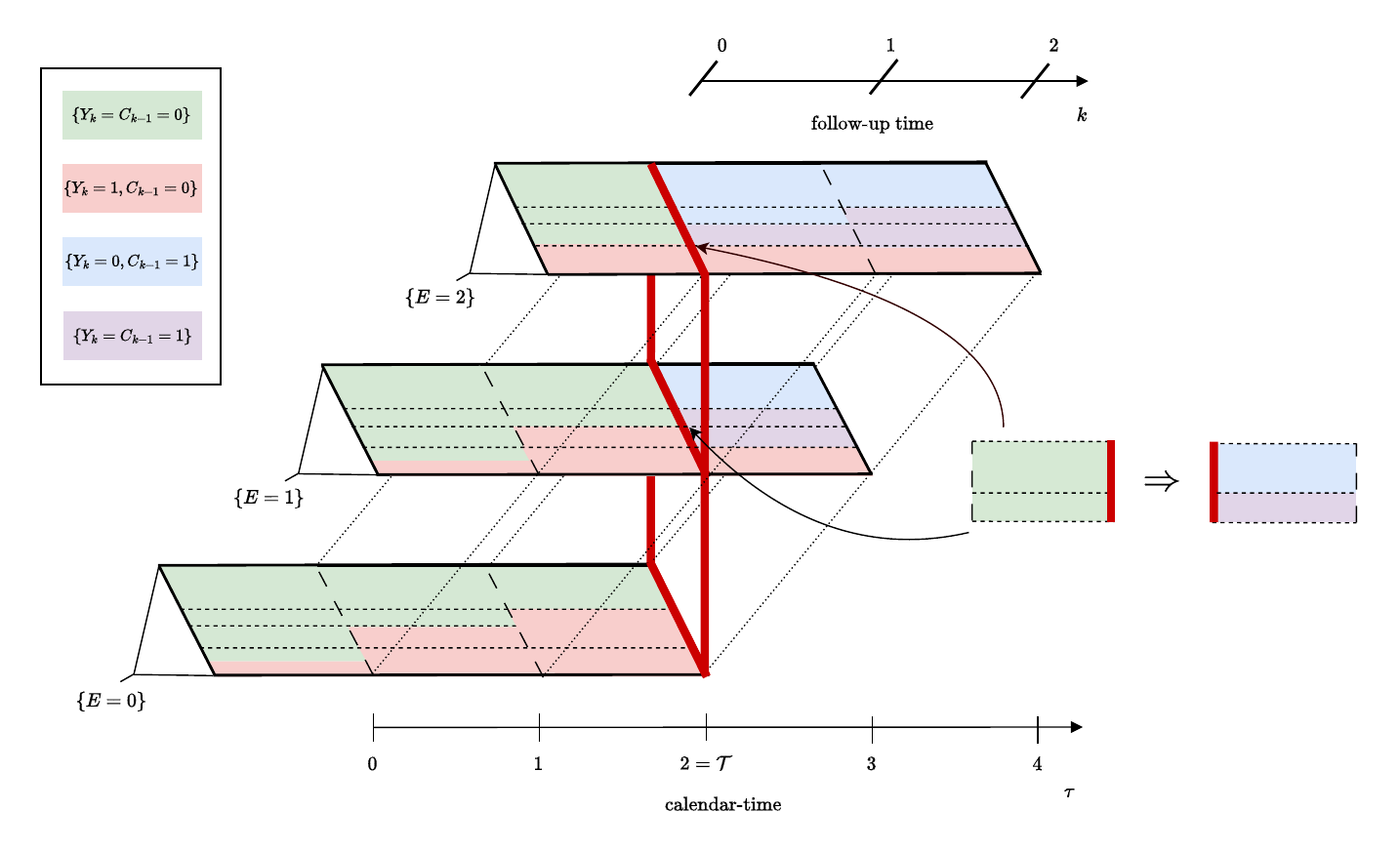}
\caption{The red vertical lines are the time of administrative end of the study $2=\mathcal{T}$, which, depending on $E$, corresponds to different individual-specific follow-up times $k$. Thus, the event of interest is censored if it happens after $k=1$ for $\{E=1\}$, or after $k=0$ for $\{E=2\}$.}
\label{fig: secc}
\end{figure}

It is the dependence between $E$ and the outcome of interest, here $Y_k^{\overline{a}_k,c_{k-1}=0}$, that usually prompts researchers to formulate censoring assumptions, here \ass{ex}, conditional on $E$. This follows because event-free individuals are censored at the calendar time of the administrative end of study, $\mathcal{T}$, i.e.,  
\begin{multline}\label{eq: secc}
    \mathbb{1}_{\{E + k \geq \mathcal{T} + 1\}} = 1, Y_{k-1}=0, C_{k-2}=0 \Longrightarrow  C_{k-1} = 1, \text{ $P-$\text{a.s.}} \\
    \text{ for every } k \in \{0, \dots, K\},
\end{multline} 
and therefore $C_{k-1}$ itself depends on $E$. That is, $E$ can be thought of as a confounder for the effect of the \textit{composite} treatment, $(C_{k-1}, A_{k}=a_k)$, on the outcome of interest, $Y_k^{\overline{a}_k,c_{k-1}=0}$ \citep{hernan_causal_2025}. In particular, what invariably leads to the ill-posedness of \ass{ex} is a dependence between $E$ and $Y_k^{\overline{a}_k, c_{k-1}=0}$ that arises at a calendar time $\tau$ \textit{greater} than the administrative end of the study $\mathcal{T}$: while the condition \ass{pos}$(e,k)$ need not be a priori violated for $\tau = (e+k) \leq \mathcal{T}$, it is necessarily violated if $\tau = (e+k) > \mathcal{T}$ and it is this violation that renders \ass{ex} ill-posed.

Such a line of reasoning applies to any outcome defined under an intervention that prevents censoring. For example, consider a two-arm randomized controlled trial (RCT) where only a single treatment, $A \in \{0,1\}$, is randomized at baseline, and let $Y_k^{c_{k-1}=0}$ denote the outcome under an intervention that abolishes censoring. Suppose that $E$, in the arm $A=1$, is believed to be associated with $Y_k^{c_{k-1}=0}$ so that the researcher operates under the censoring assumption $\left(Y_K^{{c}_{K-1}=0}, \dots, Y_k^{c_{k-1}=0}\right)\independent C_{k-1} \mid {Y}_{k-1} = 0, {C}_{k-2} = 0, E, A=1$ for every $k\in\{0,1,2\}$. Then, for $k=2$ and $E=1$, from \eqref{eq: secc}, it follows that
\begin{align}
    P(C_{\textcolor{teal}{k-1 \equiv 1}} = 0 \mid Y_1 = 0,  {C}_0 = 0, \textcolor{teal}{E=1}, A=1) = 0,
\end{align}
i.e., positivity is violated and, correspondingly, the censoring assumption is ill-posed.

We say that calendar-time changes \ass{\CTC} are present when dependencies between the outcome of interest and the time of study entry arising during the post-study period result in the following recurring identification problem:
{\footnotesize
\begin{multline*}
\text{``censoring assumptions are only deemed plausible conditional on time of study-entry} \\
\Rightarrow 
\text{the censoring mechanism
\eqref{eq: secc} implies positivity is violated} \\
\Rightarrow \text{censoring assumptions are ill-posed''}. 
\end{multline*}
}
We elaborate on the notion of \ass{\CTC} in Appendix \ref{sec_app: notion_double_bind}, clarifying how it differs from previously known formulations. Appendix \ref{sec_app: continuous_time_secc} discusses how this problem manifests in continuous-time settings and elaborates on the sense in which censoring assumptions are ill-posed when formulated within a classical continuous-time framework, e.g., as in \citep{andersen_statistical_1993}. Insodoing, we clarify that the issues presented here do not arise because of our specific adoption of the potential-outcome framework. Unless otherwise stated, we henceforth focus on discrete-time settings.
    \section{Ill-posed censoring assumptions in WSC and MSM-hazard procedures: consequences and solutions}
\label{sec: consequences_and_solutions_in_wsc_and_msm_hazard}
Here, we aim to describe the consequences of \ass{\CTC} in WSC and MSM-hazard procedures and to establish conditions that characterize when survival $S_k^{\overline{a}_k}$, as in eq. \eqref{eq: survival_contrast}, is identifiable from the observed data. 

Denote by $\ass{\FAID}$ the intersection of \ass{ex}, \ass{con}, and \ass{pos}. For assumption $\ass{\FAID}$ to be satisfied, \ass{ex}$(e,k)$ and  \ass{con}$(e,k)$ and \ass{pos}$(e,k)$ must hold for every $(e,k) \in \{0, \dots, \mathcal{T}\} \times \{0, \dots, K\}$, that is, they must all hold during the study period, i.e., at every calendar time $\tau = (e+k) \in  \{0,\dots, \mathcal{T}\}$ \textit{and} during the post-study period, i.e., at every calendar time $\tau = (e+k) \in\{\mathcal{T}+1,\dots, \mathcal{T}+K\}$. As elaborated in Section \ref{sec: censoring_assumptions_are_illposed}, however, it is because \ass{pos}$(e,k)$ is violated during the post-study period that \ass{ex}, and thus $\ass{\FAID}$, is ill-posed. We therefore articulate a restricted formulation of $\ass{\FAID}$, denoted by $\ass{\RAID}$, in which conditions are required to hold solely during the \textit{study period}, i.e., $\ass{\RAID}$ is satisfied when \ass{ex}$(e,k)$, \ass{con}$(e,k)$, and \ass{pos}$(e,k)$
hold for every $(e,k)$ such that $\tau=(e+k) \in \{0, \dots, \mathcal{T}\}$. 

Figure \ref{fig: faid_vs_raid} illustrates the distinction between $\ass{\RAID}$, required to hold only in the green area, and $\ass{\FAID}$, which instead must hold for every $k \in \{0, \dots, K\}$, i.e., in both the green and lavender areas. We refer the reader to Appendix \ref{sec_app: censoring_assumptions_discrete_time} for alternative definitions of $\ass{\RAID}$ and $\ass{\FAID}$, and to Appendix \ref{sec_app: continuous_time_secc} for analogous continuous-time conditions.
 
\begin{figure}[ht]
\centering\includegraphics[width=10cm]{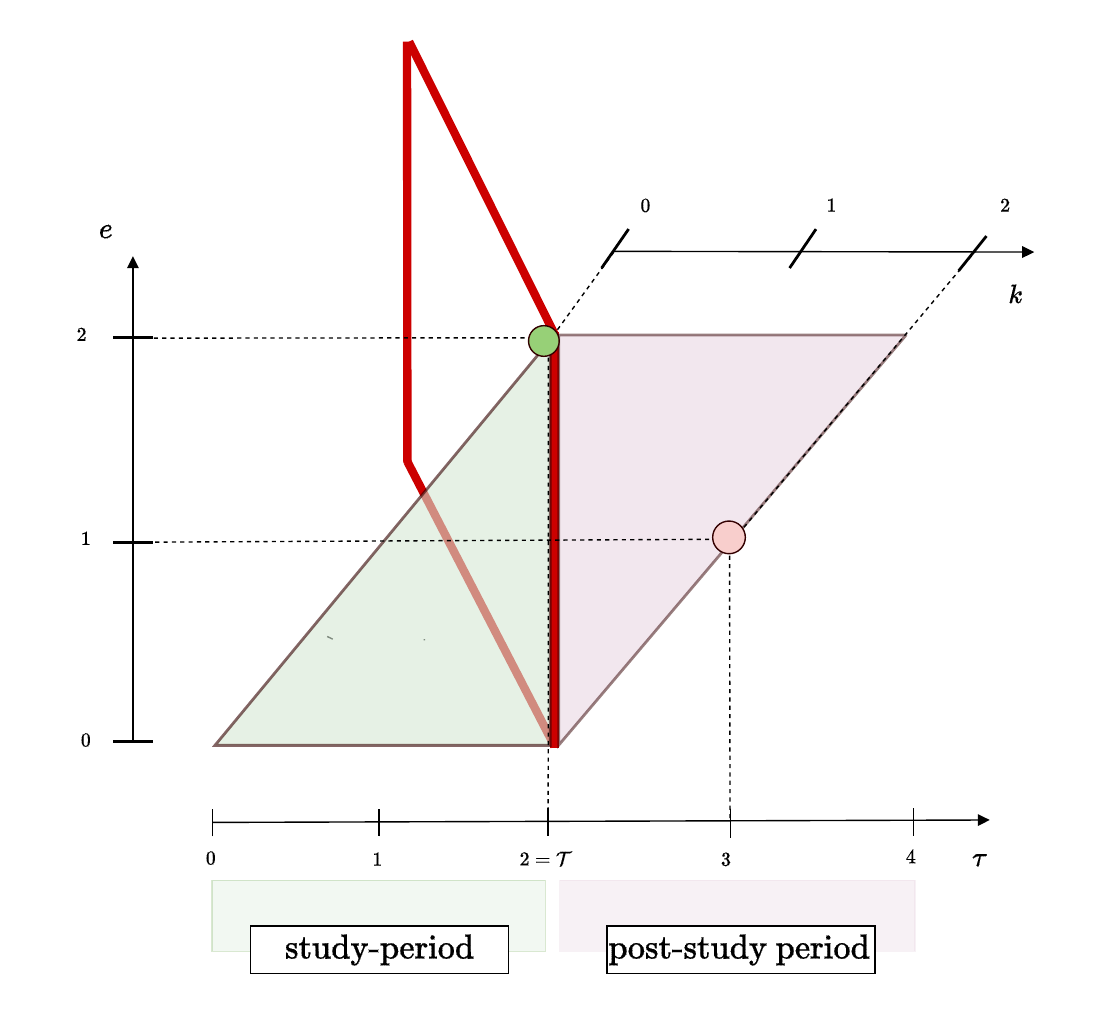}
\caption{Consider \ass{pos}$(e,k)$ in the context of our running example with $K=\mathcal{T}=2$. Condition \ass{pos}$(e=2,k=0)$ is not automatically violated: $(e=2,k=0)$ corresponds to calendar time $\tau=(e+k)=2$ located in the study period; more generally, this holds for every $(e,k)$ such that $e+k \leq 2$, so $\ass{\protect\RAID}$ is not automatically ill-posed. Condition \ass{pos}$(e=2,k=1)$, however, is violated: it corresponds to calendar time $\tau=(e+k)=3$ located in the post-study period; thus, $\ass{\protect\FAID}$ is ill-posed.}
\label{fig: faid_vs_raid}
\end{figure}

We discuss existing methods proposed to address \ass{\CTC} and related issues in Appendix \ref{sec_app: literature_review}. Therein, we also elaborate on why our contributions differ from previous works. Unless otherwise stated, we hereafter assume that the maximum follow-up time, $K$, equals the (translated) calendar time of the administrative end of the study, $\mathcal{T}$, i.e., $K=\mathcal{T}$.

\newpage
\subsection{Weighted Survival Curves procedures}
\label{sec: wsc_procedures}
Let $k \in \{0, \dots, K\}$. For each $m \in \{0, \dots, k\}$, let $h_m^{*, \overline{a}_m} \coloneqq P^*(Y_m=1 \mid \overline{A}_m=\overline{a}_m, C_{m-1}=0, Y_{m-1}=0)$ denote the hazard at follow-up $m$ for individuals who received treatments $\overline{a}_m$ and remained uncensored through follow-up $m$ in the pseudopopulation $P^{*} \equiv P^{*}(P)$ generated from $P$ by means of the weights $\{W_m\}_{m=0}^k$. Consider
\begin{align}\label{eq: wscp_y1_km}
    S_k^{*,
    \textnormal{km}, \overline{a}_k} \coloneqq \prod_{m=0}^k (1 - h_m^{*, \overline{a}_m}),
\end{align}
which represents the discrete-time analogue of the statistical parameter targeted by the Kaplan–Meier estimator \citep{kaplan_nonparametric_1958} for the special case in which $W_m=1$ for every $m \in \{0, \dots, k\}$. Expression \eqref{eq: wscp_y1_km} also represents the statistical parameter targeted by the IPTW-adjusted Kaplan–Meier estimator \citep{cole_adjusted_2004, xie_adjusted_2005} when $\{W_m\}_{m=0}^k$ depends on the inverse probability of receiving treatment. See Appendix \ref{sec_app: iptw_technical} for a formal definition and account of the role of weights in IPTW-theory.

Under $\ass{\FAID}$, it is well-known that $S_k^{\overline{a}_k}$ can be expressed through hazards in the pseudopopulation $P^*$ \citep{hernan_causal_2025}; $h_m^{*, \overline{a}_m}$ equals \begin{align}
    h_m^{\overline{a}_m} \coloneqq P(Y_m^{a_m,  \overline{a}_{m-1}, {c}_{m-1}=0}=1 \mid Y_{m-1}^{\overline{a}_{m-1}, {c}_{m-2}=0} = 0),
\end{align}
for every $m \in \{0, \dots, k\}$ and, thus, $S_k^{*, \textnormal{km}, \overline{a}_k}$ equals $\prod_{m=0}^k (1 - h_m^{\overline{a}_m}) = S_k^{\overline{a}_k}$. 
Yet, under \eqref{eq: secc}, we show that if $\ass{\RAID}$ holds, then $S_k^{*, \textnormal{km}, \overline{a}_k}$ equals 
\begin{align}\label{eq: wscp_y1_km_biased}
    \prod_{m=0}^k(1 - h_{\prstudyp,m}^{\overline{a}_m} ), 
\end{align}
where we use the symbol $\prstudyp$ in 
\begin{align*}
    h_{\prstudyp,m}^{\overline{a}_m} \coloneqq P(Y_m^{a_m,\overline{a}_{m-1}, {c}_{m-1}=0}=1 \mid Y_{m-1}^{ \overline{a}_{m-1}, {c}_{m-2}=0} = 0, E \in \{0, \dots, \mathcal{T}-m\})
\end{align*}
to denote that the counterfactual hazard at follow-up time $m$ pertains only to individuals who entered the study at a calendar time $e$ such that $\tau = (e + m) \in \{0, \dots, \mathcal{T}\}$; that is, individuals who entered during the calendar-time period $\{0, \dots, \mathcal{T}-m\}$, rather than the entire population.
See Corollary \ref{corollary: kiss_when_e_not_in_v} in Appendix \ref{sec_app: iptw_preliminaries} for a formal proof.

By definition of $h_{\prstudyp,k}^{\overline{a}_k}$ we have $h_{\prstudyp,0}^{\overline{a}_0} = h_{0}^{\overline{a}_0}$, and thus $S_k^{*, \textnormal{km}, \overline{a}_k}$ coincides with $S_k^{\overline{a}_k}$ if 
\begin{align}\label{eq: condition_no_bias_km}
h_{\prstudyp,m}^{\overline{a}_m} = h_{m}^{\overline{a}_m}, \quad \textnormal{for every $m \in \{1, \dots, k\}$},
\end{align}
in the absence of perfect cancellations.

Returning to our running example, where $\mathcal{T}=K=2$, this implies that for a given treatment strategy $\overline{a}_2$, $S_2^{*, \textnormal{km}, \overline{a}_2}$ will differ from $S_2^{\overline{a}_2}$ whenever the proportion of individuals experiencing failure at follow-up time $k=1$ differs between individuals with $E \in \{0,1\}$ and the entire population ($E \in \{0,1,2\}$), i.e., whenever $h_{\prstudyp,1}^{\overline{a}_1} \neq h_{1}^{\overline{a}_1}$.\footnote{Analogous arguments apply to survival estimands analogous to $S_k^{\overline{a}_k}$; see Appendix \ref{sec_app: supplementary_wsc_deferred_main_text} for these and complementary results on a companion procedure.} 

Letting $h_m^{\overline{a}_m}(e) \coloneqq P(Y_m^{a_m,  \overline{a}_{m-1}, {c}_{m-1}=0}=1 \mid Y_{m-1}^{\overline{a}_{m-1}, {c}_{m-2}=0} = 0, E=e)$, assumption \eqref{eq: condition_no_bias_km} can be equivalently assessed by verifying that
\begin{align}\label{eq: hazard_variation_dependence_wsc_main}
\frac{\expectation\left(\mathbb{1}_{\{E+m \textcolor{teal}{\leq} \mathcal{T}\}} \cdot \prod_{j=0}^m (1 - h_j^{\overline{a}_j}(E))\right)}{\expectation\left(\mathbb{1}_{\{E+m \leq \mathcal{T}\}} \cdot \prod_{j=0}^{m\textcolor{teal}{-1}} (1 - h_j^{\overline{a}_j}(E))\right)}= 
\frac{\expectation\left(\mathbb{1}_{\{E+m \textcolor{teal}{>} \mathcal{T}\}} \cdot \prod_{j=0}^m (1 - h_j^{\overline{a}_j}(E))\right)}{\expectation\left(\mathbb{1}_{\{E+m > \mathcal{T}\}} \cdot \prod_{j=0}^{m\textcolor{teal}{-1}} (1 - h_j^{\overline{a}_j}(E))\right)}
\end{align}
holds for every $m \in \{1, \dots, k\}$ and every $k \in \{1, \dots, K\}$; see Corollary \ref{corollary: equivalence_hazard_under_ct_special} of Appendix \ref{sec_app: notion_ctc_discrete_time} for details. Equation \eqref{eq: hazard_variation_dependence_wsc_main} clarifies that assumption \eqref{eq: condition_no_bias_km} imposes a variation dependence between hazards in the study and post-study periods, a restriction that is difficult to justify on substantive grounds under \ass{\CTC}. Hence, whether $S_k^{\overline{a}_k}$ coincides with $S_k^{*, \textnormal{km}, \overline{a}_k}$ for all $k \in \{0, \dots, K\}$ reduces to determining whether \ass{\CTC} are present. Addressing this question based on the data, however, requires extrapolation (in calendar time) beyond the observed law $P^{\textnormal{o}} \equiv P^{\textnormal{o}}(P)$. For example, consider a researcher who finds that $E \independent_{P^*} Y_m \mid \overline{A}_{m}, C_{m-1}=Y_{m-1}=0$ for every $m \in \{0, \dots, k\}$ and therefore assumes that \ass{\CTC} are not present. Any dependence between $E$ and $Y_k^{\overline{a}_k, c_{k-1}=0}$ that arises during the post-study period is not identifiable from $P^{\textnormal{o}}$ without additional assumptions, as no observed data are available to assess it. We illustrate the practical significance of this limitation in Section \ref{sec: case_study_mrna_vaccine}.
\subsection{MSM-hazard procedures}
\label{sec: msm_hazard_procedures} 
MSM-hazard procedures often rely on two assumptions in addition to the nonparametric identification assumptions \ass{$\FAID$} assumed by WSC procedures. To be explicit, let $k \in \{0, \dots, K\}$ and consider
\begin{align}
         h_k^{\overline{a}_{k}}(v) \coloneqq P(Y_k^{a_k, \overline{a}_{k-1}, {c}_{k-1}=0}=1 \mid Y_{k-1}^{\overline{a}_{k-1}, {c}_{k-2}=0} = 0, V=v), \quad V \subseteq L_0, 
         \label{eq: counterfactual_hazard}
\end{align}
that is, the hazard at time $k$ for the subset of individuals $\{V=v\}$ under an intervention that sets $\overline{A}_k$ to the treatment strategy $\overline{a}_k$ and abolishes censoring. A common assumption is
\begin{itemize}[leftmargin=1cm, labelsep=0.5cm]
    \item[\ass{HI}] : for every $k \in \{0, \dots, K\}$ and every $\overline{a}_{k}$, $h_k^{\overline{a}_k}(v)$ does not depend on $\overline{a}_{k-1}$,     \label{cond: hi}
\end{itemize}
which, although not strictly needed to establish our results, we assume throughout, thereby allowing us to write $h_k^{\overline{a}_k}(v)$ as $h_k^{a_k}(v)$. Second, $h_k^{a_k}(v)$ must be correctly specified, an assumption that we articulate, coherent with the one assumed by \citet{hernan_marginal_2000, hernan_observational_2008, dickerman_avoidable_2019}, as
\begin{itemize}[leftmargin=2cm, labelsep=0.5cm]
    \item[\ass{${\mb}_{\mid V}^{\textnormal{cs}}$}] : there exists a unique $\beta_{\prstudyp\prgapp} \in \mathbb{B}$ such that for every $a_k \in \{0,1\}$, $k \in \{0, \dots, K\}$ and for all $v \in \textnormal{supp}(P(V))$, $h_k^{a_k}(v) = p(v,a_k,k; {\beta}_{\prstudyp\prgapp}) \in \mb$, with $\mb$ denoting a regression model defined as a collection of functions  $p(\cdot;\beta)$ indexed by a set $\mathbb{B} \ni \beta$.
    \label{ass: cs}
\end{itemize}
For example, in the work by \citet{hernan_marginal_2000},
{\footnotesize
\begin{align}\label{eq: model_hernan_00}
    \mathcal{M}(\mathbb{B}) = \Big\{(v,a_k,k) \mapsto \textnormal{expit}\Big( 
    &\langle \textnormal{ordinal}(e), \beta_{E} \rangle \nonumber\\
    +\ &\langle \textnormal{categorical}(v\setminus e), \beta_{V\setminus E} \rangle \nonumber\\
    +\ &\langle a_k, \beta_{A} \rangle \nonumber\\
    +\ &\langle \textnormal{natural cubic splines}(k), \beta_{K_{S}} \rangle 
    \Big):\nonumber \\
    &\textnormal{ for }\left(\beta_E, \beta_{V\setminus E}, \beta_{A_k}, \beta_{K_{S}}\right)^{\textnormal{T}} \equiv \beta \textnormal{ in some } \mathbb{B} \subseteq \mathbb{R}^{d},  d>0
    \Big\}. 
\end{align}
}
The subscript $V$ emphasizes that the correct specification is formulated for hazards defined conditional on a vector of baseline covariates $V$. For a fixed $V$, we will interpret condition \ass{${\mb}_{\mid V}^{\textnormal{cs}}$} as a function of $\mb$ that is either true or false. Throughout, we will drop the subscript, unless emphasis is needed.

Together, assumption  \ass{\hi} and condition \ass{${\mb}^{\textnormal{cs}}$} for a given $\mb$ impose a regularity structure on $h_k^{a_k}(v)$: \ass{\hi} establishes that, minimally, the non-response variables determining \eqref{eq: counterfactual_hazard} are $(v,a_k,k)$ while \ass{${\mb}^{\textnormal{cs}}$} determines its functional form via $p(v,a_k,k; \beta_{\prstudyp\prgapp})$ as a function of $(v,a_k,k)$.
Thus, when the estimand of interest is survival, in our case $S_k^{\overline{a}_k}$, MSM-hazard procedures encode additional assumptions on $S_k^{\overline{a}_k}$ beyond those imposed by WSC procedures, since, under \ass{\hi} and \ass{${\mb}^{\textnormal{cs}}$},
\begin{multline} \label{eq: survival_via_hazards}
   S_k^{\overline{a}_k} =\expectation\left(\prod_{m=0}^{k}\left(1-h_m^{a_m}(V)\right)\right)=\expectation\left(\prod_{m=0}^{k}\left(1-p(V,a_m,m;\beta_{\prstudyp\prgapp})\right)\right).
\end{multline}
The MSM-hazard procedure usually targets a statistical parameter $\beta^{*} \equiv \beta^{*}(P^*)$ of a pseudopopulation $P^{*}(P)$ generated from $P$ through weighting by $\{ W_m\}_{m=0}^K$. For the trivial choice of the weights $W_m=1$, which yields $P^*=P^{\textnormal{o}}$, procedures adopting the Cox Proportional Hazard model \citep{cox_regression_1972}
fall within the class of MSM-hazard procedures, since pooled logistic models (e.g., \eqref{eq: model_hernan_00}) approximate the Cox model when the hazards within discrete-time intervals are sufficiently small \citep{dagostino_relation_1990, hernan_marginal_2000}.
 
While not always stated explicitly, MSM-hazard procedures must be consistent with the assumption that $h_k^{*, a_k}(v) \coloneqq P^*(Y_k = 1 \mid A_k = a_k, Y_{k-1} = 0, C_{k-1} = 0, V=v)$, i.e., the hazard of failure at time $k$ in the pseudopopulation $P^*$ for individuals in $\{A_k = a_k, {C}_{k-1}=0, V=v\}$, is also correctly specified:
\begin{itemize}[leftmargin=2cm, labelsep=0.5cm]
    \item[\ass{${\mb}_{\mid V}^{*\textnormal{-cs}}$}] : there exists a unique $\beta^* \in \mathbb{B}$ such that for every $a_k \in \{0,1\}$, $k \in \{0, \dots, K\}$ and for all $v \in \textnormal{supp}(P^*(V \mid A_k = a_k, Y_{k-1} = 0, C_{k-1} = 0))$, $h_k^{*,a_k}(v)=p(v,a_k,k; {\beta}^*) \in \mb$, with $\mb$ denoting a regression model, e.g., \eqref{eq: model_hernan_00}.
    \label{ass: mcp_pseudo}
    \end{itemize}
Just like \ass{${\mb}_{\mid V}^{\textnormal{cs}}$}, for a fixed $V$, we interpret \ass{${\mb}_{\mid V}^{*\textnormal{-cs}}$} as a function of $\mb$ that is either true or false and we drop the subscript unless emphasis is needed.
Under \ass{$\FAID$} and \ass{${\mbps}^{*\textnormal{-cs}}$} for some working model $\mbps$, the statistical parameter targeted by the procedure, $\beta^*$, can be shown \citep{robins_marginal_2000} to coincide with  $\beta_{\prstudyp\prgapp}$, and, from \eqref{eq: survival_via_hazards}, thus 
\begin{align}\label{eq: survival_star}
   S_k^{*, \textnormal{msm}, \overline{a}_k} \coloneqq \expectation\left(\prod_{m=0}^{k}\left(1-p(V,a_m,m;\beta^{*})\right)\right) 
\end{align}
to coincide with $S_k^{\overline{a}_k}$.

However, we have no guarantee that  $S_k^{*, \textnormal{msm}, \overline{a}_k}$ is equal to $S_k^{\overline{a}_k}$ because positivity (\ass{pos}) is violated  (and thus \ass{$\FAID$} is ill-posed). Relatedly, $\beta^*$ is not guaranteed to be a parameter characterizing $S_k^{\overline{a}_k}$, without imposing additional assumptions.

\subsubsection{Valid identification of $S_k^{\overline{a}_k}$ in MSM-hazard procedures}
Here we articulate assumptions that allow identification of $S_k^{\overline{a}_k}$ when positivity fails.

\begin{proposition}\label{proposition: msm_hazard_characterization}
Under \ass{$\RAID$}, if \ass{${\mbps}^{*\textnormal{-cs}}$} holds and if $\mbps$ is identifiable from the observed data, then, under mild regularity conditions,
{\small
\begin{align*}
{\mbps}_{\mid V}^{\textnormal{cs}} \textnormal{ with }
    \textnormal{$E \in V$ } \Leftrightarrow  S_k^{\overline{a}_k} = S_k^{*, \textnormal{msm}, \overline{a}_k}& \\
    \textnormal{ \small for every $k \in \{0, \dots, K\}$ and $\overline{a}_k \in \{\overline{0}, \overline{1}\}$}.
\end{align*}
}
\end{proposition}
In words, under a relaxation of the canonical assumption, invoking \ass{$\RAID$} in place of \ass{$\FAID$}, and under a correctly specified and identifiable working model $\mbps$, restrictions on the functional form of the hazards are necessary and sufficient for $S_k^{*, \textnormal{msm}, \overline{a}_k}$ (eq. \eqref{eq: survival_star}) to coincide with $S_k^{\overline{a}_k}$. These restrictions require $h_k^{a_k}(v)$ to be correctly specifiable when $E$ is included in $V$ for the identifiable working model $\mbps$. Proof of Proposition \ref{proposition: msm_hazard_characterization}, which, along with a definition of the mild regularity assumptions, is given in Appendix \ref{sec_app: staggered_entry_characterization}, makes it clear that analogous results apply under relaxations of \ass{\hi}, as well as in settings
with a point treatment $A$.

To quantify departures of $S_k^{*, \textnormal{msm}, \overline{a}_k}$ from $ S_k^{\overline{a}_k}$ for $k \in \{0, \dots, K\}$ when ${\mbps}_{\mid V}^{\textnormal{cs}}$ with $E \in V$ is not met, we derived the tightest possible bounds under the assumptions of Proposition \ref{proposition: msm_hazard_characterization} and under conditions that bound the otherwise unknown $h_k^{a_k}(v)$ as a function of $p(v,a_k,k;{\beta}^*)$. See Corollary \ref{corollary: bounds_ass_characterization_not_met} of Appendix \ref{sec_app: staggered_entry_characterization}. 

\subsubsection{Remarks on the interpretation of the conditions involved in Proposition \ref{proposition: msm_hazard_characterization}}
\label{sec: extrapolation_assumptions_in_msm_hazard}
Throughout this section,  all equalities are understood to hold for any $k \in \{1, \dots, K\}$, for every $a_k \in \{0,1\}$, and for all $v \in \textnormal{supp}(P(V))$. Suppose further that \ass{$\RAID$} holds and that \ass{${\mbps}_{\mid V}^{*\textnormal{-cs}}$} is satisfied for some $V \not\ni E$. Then, $p(v,a_k,k;\beta^*) \in \mbps$ is given by
{\small
\begin{align*}
    h_{\prstudyp,k}^{a_k}(v) \coloneqq P(Y_k^{a_k,  \overline{a}_{k-1}, {c}_{k-1}=0}=1 \mid Y_{k-1}^{\overline{a}_{k-1}, {c}_{k-2}=0} = 0, V=v, E \in \{0, \dots, \mathcal{T}-k\})
\end{align*}
}
for every $k \in \{1, \dots, K\}$; see Lemma \ref{lemma: omission_e_is_problematic} of Appendix \ref{sec_app: staggered_entry_characterization} for a proof. Thus, $p(v,a_k,k;\beta^*)$ targets the hazards for the subpopulation of individuals entering before or at calendar time $\mathcal{T}-k$, $\{E \leq \mathcal{T}-k, V=v \}$, and not for the target subpopulation, $\{ E \leq \mathcal{T}, V=v \} = \{V=v\}$. This distinction matters because, under \ass{\CTC}, $E$ depends on the outcome of interest $Y_k^{\overline{a}_k, c_{k-1}=0}$ given $V$, and hence $h_{\prstudyp,k}^{a_k}(v)$ is unlikely to equal $h_k^{a_k}(v)$ unless $E \in V$ (see Appendix \ref{sec_app: notion_ctc_discrete_time}, e.g., Proposition \ref{proposition: equivalence_hazard_under_ct}). Accordingly, \ass{${\mbps}_{\mid V}^{*\textnormal{-cs}}$} must be formulated with $E \in V$. In particular, if $E \notin V$, then \ass{${\mbps}_{\mid V}^{\textnormal{cs}}$} and \ass{${\mbps}_{\mid V}^{*\textnormal{-cs}}$} cannot hold simultaneously for the same parameters. That is, for $\beta^*, \beta_{\prstudyp\prgapp} \in \mathbb{B}^*$, $\beta^* \neq \beta_{\prstudyp\prgapp}$, regardless of the flexibility of $\mbps$. See Lemma \ref{lemma: omission_e_is_problematic} in Appendix \ref{sec_app: staggered_entry_characterization} for a proof, and Section \ref{sec: case_study_hernan_00_sim} for an illustration of the practical consequences of this incompatibility result.
    
\section{C-values and extrapolation curves}
\label{sec: c_values_and_extrapolation_curves}
For WSC and MSM-hazard procedures, identification of $S_k^{\overline{a}_k}$ therefore relies on extrapolations from $P^{\textnormal{o}}$ that constrain the laws concerning the \textit{post-study period}. Yet, in many settings these \textit{extrapolation assumptions} will not be met under \ass{\CTC}. In such cases, researchers often seek pragmatic solutions to address their substantive questions using the observed data; our aim is to assist them in this task.

Suppose, as in \citet{hernan_marginal_2000}, that we want to compare the effect of zidovudine ($a_k=1$) vs. no zidovudine ($a_k=0$) by means of a Proportional Hazard MSM as specified by a working model $\mbps$ defined as in, say, \eqref{eq: model_hernan_00}. The protective effect of zidovudine found by \citet{hernan_marginal_2000} implicitly relied on the extrapolation assumption \ass{${\mbps}^{\textnormal{cs}}$}. Suppose, however, that, based on subject-matter knowledge, \ass{${\mbps}^{\textnormal{cs}}$} is considered implausible, as the post-study hazards implied by the procedure are deemed to deviate from the true counterfactual hazards, $h_k^{a_k}(v)$ (we will offer a substantive justification for this hypothesis in Section \ref{sec: c_values_case_studies}). A pragmatic question is whether the protective effect of zidovudine would have changed if we collected observations during the post-study period. Here we will describe one such approach, and we also give some alternative ideas in Appendix \ref{sec_app: first_sensitivity_analysis}. 

\subsection{Definition of c-values}
One way to address our pragmatic question is to ask what minimal post-study variation would suffice to nullify the protective effect of treatment. C-values provide a principled way to formalize this idea; they capture the smallest $\textnormal{c}$ such that a $\textnormal{c}$-fold \textit{increase} in post-study risk under the most beneficial treatment strategy (here, receipt of zidovudine) and a corresponding $\textnormal{c}$-fold \textit{decrease} under the least beneficial treatment strategy (here, no treatment) would suffice to nullify the protective effect of zidovudine. To fix ideas, let $k \in \{1, \dots, K\}$, define $R_k^{\rightsquigarrow,a} \coloneqq 1- S_k^{*, \textnormal{msm}, \overline{a}_k}$ ($a \in \{0,1\}$, $\overline{a}_k = \overline{1} \cdot a)$, and consider
\begin{align*}
   RR_k^{\rightsquigarrow}\coloneqq \frac{R_k^{\rightsquigarrow,1}}{R_k^{\rightsquigarrow,0}}.
\end{align*}
$RR_k^{\rightsquigarrow}$ is the ratio, as implied by a given procedure (here, MSM-hazard), of the cumulative risk (or incidence) accrued by follow-up $k$ under receipt of zidovudine to that under no treatment; that is, the ratio between the proportions of individuals who experienced failure \textit{by} follow-up time $k$ under the treatment strategies $\overline{a}_k=\overline{1}$ ($R_k^{\rightsquigarrow,1}$) and $\overline{a}_k=\overline{0}$ ($R_k^{\rightsquigarrow,0}$). For every $a \in \{0,1\}$, $R_k^{\rightsquigarrow,a}$ can be decomposed into the sum of two proportions: one relative to individuals who \textit{first} experienced failure during the study period, $R_{\prstudyp,k}^{a}$, and one relative to individuals who \textnormal{first} experienced failure during post-study period, $R_{\prgapp,k}^{\rightsquigarrow,a}$, i.e.
$R_k^{\rightsquigarrow,a} = R_{\prstudyp,k}^{a} + R_{\prgapp,k}^{\rightsquigarrow,a}$ for $a \in \{0,1\}$.
We say that post-study variations, $\textnormal{c}_{k}^{a} > 0$ for $a \in \{0,1\}$, nullify a given effect inferred under extrapolation assumptions if the induced
 \textit{shifted risks}, $R_k^{\textnormal{c}, a}(\textnormal{c}_{k}^{a}) \coloneqq R_{\prstudyp,k}^{a} + \textnormal{c}_{k}^{a} \cdot R_{\prgapp,k}^{\rightsquigarrow,a}$, $a\in\{0,1\}$, are equal, i.e., 
\begin{align}\label{eq: shifted_risks_equality}
    \frac{R_k^{\textnormal{c}, 1}(\textnormal{c}_{k}^{1} )}{R_k^{\textnormal{c}, 0}(\textnormal{c}_{k}^{0} )} = 1.
\end{align}
We define $\textnormal{c}_k^{0}$ and $\textnormal{c}_k^{1}$ as the c-values for the treatment strategies $\overline{0}$ and $\overline{1}$, respectively, if $(\textnormal{c}_k^{0}, \textnormal{c}_k^{1})$ is the unique solution to \eqref{eq: shifted_risks_equality} that minimizes $(\ln(\textnormal{c}_{k}^{0} \cdot \textnormal{c}_{k}^{1}))^2$, i.e., if it penalizes departures of $\textnormal{c}_k^{0} \cdot \textnormal{c}_k^{1}$ from 1. 

C-values admit closed-form expressions. In particular, when a perfectly symmetric solution exists, i.e., when $(\textnormal{c}_k^{0}, \textnormal{c}_k^{1})$ satisfies both $\textnormal{c}_k^{0} \cdot \textnormal{c}_k^{1} = 1$ and \eqref{eq: shifted_risks_equality}, as in the case studies we will present in Section \ref{sec: c_values_case_studies}, we have $\textnormal{c}_{k}^{0} = \left(-d_k + \sqrt{d_k^2 + 4m_k}\right) \cdot {(2m_k) }^{-1} , \textnormal{c}_{k}^{1} = (\textnormal{c}_{k}^{0})^{-1} $ and $d_k= (R_{\prstudyp,k}^{0} - R_{\prstudyp,k}^{1})\cdot {(R_{\prgapp,k}^{\rightsquigarrow,1})}^{-1}, m_k=R_{\prgapp,k}^{\rightsquigarrow,0} \cdot {(R_{\prgapp,k}^{\rightsquigarrow,1})}^{-1}$. See Appendix \ref{sec_app: c_value} for a formal description. 

As $k$ decreases, the proportion of individuals for whom extrapolation is required and who are subject to variations that could nullify a given effect increases. Accordingly, $\textnormal{c}_{k}^{0} + \textnormal{c}_{k}^{1}$ will often increase as $k$ decreases, reflecting that larger post-study variations are needed when a greater proportion of individuals rely on extrapolation. For some $k \in \{1, \dots, K\}$ no pair $(\textnormal{c}_k^{0}, \textnormal{c}_k^{1})$ satisfying \eqref{eq: shifted_risks_equality} might give shifted risks compatible with $R_{\prstudyp,k}^{a}$ and $R_{\prgapp,k}^{\rightsquigarrow,a}$ for $a \in \{0,1\}$. This occurs when the proportion of individuals who can first experience failure during the post-study period under the most beneficial strategy is too small to nullify the effect. In such cases, we say that c-values at follow-up $k$ do not exist. See Appendix \ref{sec_app: existence_c_values} for a formal mathematical treatment.

\subsection{Sensitivity analysis with c-values and extrapolation curves}
C-values play a role analogous to e-values introduced by \citet{ding_sensitivity_2016}. E-values quantify the minimum strength of association between a treatment–outcome pair attributable to an unmeasured variable required to explain away the effect under no unmeasured confounding; c-values quantify variations induced by the censoring mechanism due to staggered entries, formalized in \eqref{eq: secc}, which precludes observation of event-free individuals during the post-study period.\footnote{Like e-values, c-values accommodate arbitrary estimation procedures, as they are defined via a nonparametric decomposition of the risks $R^{\rightsquigarrow,a}$ for $a \in \{0,1\}$. Thus, although introduced here in the context of MSM-hazard procedures, they are not tied to any specific estimation procedure.}

To offer a complementary interpretation of the variations captured by the c-values, $\textnormal{c}_k^0$ and $\textnormal{c}_k^1$, we derive additive shifts, $\delta_k^0$ and $\delta_k^1$, in the post-study hazards that induce equivalent shifted risks. In the case studies of Section \ref{sec: case_studies}, for every $a \in \{0,1\}$ ($\overline{a}_k = a \cdot \overline{1}$), $\delta_k^a$ is the unique solution to
\begin{align*}
    1-\expectation\left(\prod_{m=0}^{k}\left(1-p_{\delta_k^a}(V,a_m,m;\beta^{*})\right)\right) = R_k^{\textnormal{c}, a}(\textnormal{c}_k^a),
\end{align*}
where $p_{\delta_k^a}(v,a_m,m;\beta^{*}) = p(v,a_m,m;\beta^{*})$ if $e+m \leq \mathcal{T}$ and $p(v,a_m,m;\beta^{*}) + \delta_{k}^a$ if $e+m > \mathcal{T}$ (with $e \in v$). See Appendix \ref{sec_app: equivalence_c_values_hazards} for a formal mathematical treatment.

Finally, letting $K_{\textnormal{ex}}$ denote the minimum follow-up for which c-values exist, we refer to the collections
\begin{align*}
    \{{\textnormal{c}}_k^0, {\textnormal{c}}_k^1, \delta_k^{0}, \delta_k^{1} : k \in \{K_{\textnormal{ex}}, \dots, K\}\}
\end{align*}
as \textit{extrapolation curves}. For each $k \in \{K_{\textnormal{ex}}, \dots, K\}$, these curves characterize the magnitude and direction of post-study variations, relative to those implied by the extrapolation assumptions, required to nullify a given effect. Rather than reasoning about assumptions that would permit identification of $S_k^{\overline{a}_k}$ but are likely to be deemed implausible, investigators \textit{first} reason about the magnitude and direction of post-study variations on the risk or additive-hazard scale and \textit{then} compare these with those implied by the extrapolation curves. We illustrate these ideas in the two case studies in Section \ref{sec: c_values_case_studies}.

    \section{Case studies}
\label{sec: case_studies}
The first case study is a more stylized, didactic example in which we analyze the effect of zidovudine on survival in HIV-positive men \citep{hernan_marginal_2000} using WSC and MSM-hazard procedures. We describe six scenarios, each based on a synthetic dataset generated using the technique proposed by \citet{seaman_simulating_2024} and constructed to be compatible with user-specified MSMs. This allows direct comparison between a ground truth and procedure-specific estimates.
We then reanalyze a recent study published in Nature that investigated the effect of SARS-CoV-2 mRNA vaccination on clinical outcomes in patients treated with immune checkpoint inhibitors \citep{grippin_sars-cov-2_2025}. Specifically, we emulate a target trial building on the specification proposed by \citet{dumas_re-evaluating_2025}; the emphasis here, however, is on the tenability of the additional extrapolation assumptions required for sound interpretation of the target trial. We also relate these considerations to target trial practice more broadly and, in Appendix \ref{sec_app: additional_complications_grippin}, discuss an additional issue in the original analysis of \citet{grippin_sars-cov-2_2025}.
We conclude by illustrating how c-values and extrapolation curves can supplement the findings of \citet{hernan_marginal_2000} and complement those of \cite{dumas_re-evaluating_2025}. 

\subsection{Case study 1: the effect of zidovudine on survival among HIV-positive men}
\label{sec: case_study_hernan_00_sim}
Following the original analysis of \citet{hernan_marginal_2000}, we considered $16 = (K=\mathcal{T}=15) + 1$ equally-spaced time intervals of 6 months, each representing an individual's visit in the MACS cohort (details of the MACS cohort can be found in \citep{hernan_marginal_2000}). To focus on our main points, we assumed administrative censoring as the sole source of censoring; the specific data structure can be found in Appendix \ref{sec_app: case_study_1_data_structure}.
We study six scenarios; scenarios A), B), and C) are distinguished from their respective counterparts, D), E), and F), only in that individuals are followed even during the post-study period, i.e., until failure or the end of the follow-up time $K$, whichever comes first, i.e., individuals are not administratively censored. \ass{$\RAID$} holds for scenarios A), B), and C), while \ass{$\FAID$} holds for scenarios D), E), and F).

In each of the six scenarios, the estimand of interest is survival under receipt of zidovudine $S_k \equiv S_k^{\overline{a}_k = \overline{1}}$ for every $k \in \{0,\dots, K\}$. In Scenario A), \ass{$\mathcal{M}_{\circ}^{\textnormal{cs}}$} holds for a stylized version of \citet{hernan_marginal_2000} in which the outcome is negatively dependent on the time since-study entry, i.e., $h_{k}^{a_k}(v) \in \mathcal{M}_{\circ}$ decreases as $e\in v$ increases (for any fixed $k$, $a_k,$ and $v\setminus e$). 

Scenario B) differs from scenario A) only in that \ass{$\mathcal{M}_{\circ\circ}^{\textnormal{cs}}$} holds for a model 
$\mathcal{M}_{\circ\circ}$ that presents an additional term relative to $\mathcal{M}_{\circ}$, $\prgapp_{\textnormal{int}}(e,k)$, on the logit scale $-0.5 \cdot \mathbb{1}_{\{a_k = 1, e + k > \mathcal{T}\}}$, to which a substantive interpretation is given in Section \ref{sec: c_values_case_studies}. Scenario C) differs from Scenario B) only in that the study-entry distribution, $P(E)$, rather than being uniformly distributed, is skewed toward the origin to capture the enrollment dynamics of individuals in the MACS cohort, as specified by a discrete-time analogue of an exponential distribution. The exact data-generating mechanisms are reported in Appendix \ref{sec_app: scenarios_dgms}.

For each scenario, we considered survival curve estimates derived from three procedures:
\begin{enumerate}
    \item $\hat{S}_k^{*, \textnormal{msm}} \equiv \hat{S}_k^{*, \textnormal{msm}, \overline{a}_k=\overline{1}}$ (MSM-hazard; we assume \ass{${\mathcal{M}_{\circ}}_{\mid V}^{*\textnormal{-cs}}$} in scenarios A), B), and C) and 
    \ass{${\mathcal{M}_{\circ\circ}}_{\mid V}^{*\textnormal{-cs}}$} in scenarios D), E), and F); $E \in V$ in all scenarios),
    \item $\hat{S}_k^{*, \textnormal{msm}-\textnormal{we}} \equiv \hat{S}_k^{*, \textnormal{msm-we}, \overline{a}_k=\overline{1}}$ (MSM-hazard; we assume \ass{${\mathcal{M}_{\textnormal{sat}}}_{\mid V}^{*\textnormal{-cs}}$}, $E \notin V$, for a saturated model $\mathcal{M}_{\textnormal{sat}}$),
    \item $\hat{S}_k^{*, \textnormal{km}} \equiv \hat{S}_k^{*, \textnormal{km}, \overline{a}_k=\overline{1}}$ (WSC-Kaplan-Meier; targeting $S_k^{*,
    \textnormal{km}}$ as in \eqref{eq: wscp_y1_km}).
\end{enumerate}
Explicit forms and point estimates for each estimator are reported and tabulated in Appendices \ref{sec_app: procedures_estimators_case_study_1} and \ref{sec_app: numerical_results_case_study_1}, respectively. In Appendix \ref{sec_app: plots_case_study_1} we further compare these curves with the estimand of interest, $S_k$, implied by the scenario-specific data-generating mechanism, thereby validating the theoretical results derived in Section \ref{sec: consequences_and_solutions_in_wsc_and_msm_hazard}.

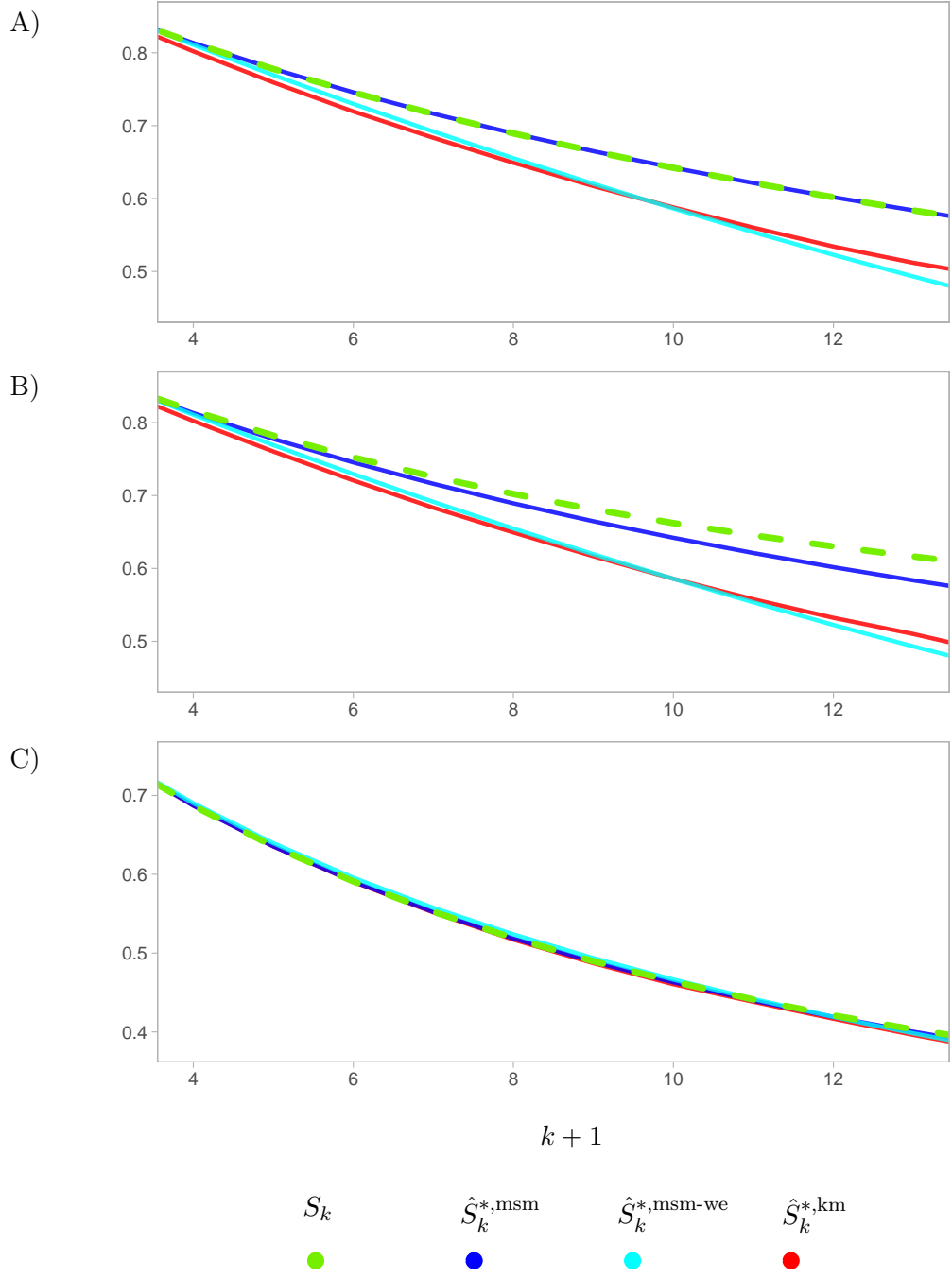
\begin{figure}[H]
    \centering
    \adjustbox{trim={.0 \width} {.01\height} {0\width} {.02\height},clip}%
{
    \begin{tikzpicture}[scale=1]
        \node[anchor=north, scale=0.8] at (0,0){%
            \pgfimage{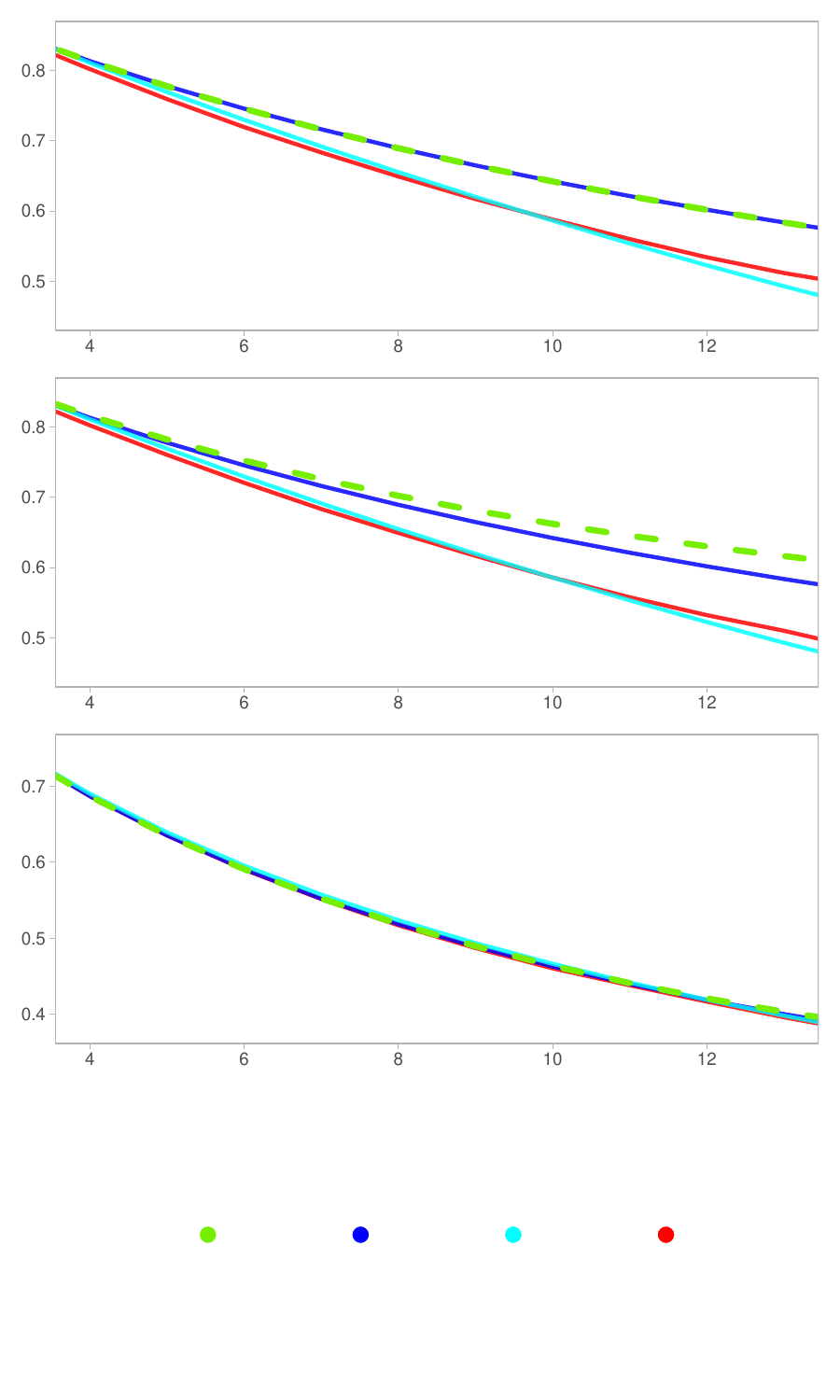}};

        \node[anchor=north west] at (-7.5,-0.4){\small A)}; 
        \node[anchor=north west] at (-7.5,-5.5){\small B)}; 
        \node[anchor=north west] at (-7.5,-10.7){\small C)}; 
        
        \node[anchor=north] at (0.5,-16){\small $k+1$};
        \node[anchor=north west] at (-3.4,-17){\small ${S}_k$};
        \node[anchor=north west] at (-1.2,-17){\small $\hat{S}_k^{*,\text{msm}}$};
        \node[anchor=north west] at (1,-17){\small $\hat{S}_k^{*,\text{msm-we}}$};
        \node[anchor=north west] at (3.3,-17){\small $\hat{S}_k^{*,\text{km}}$}; 
\end{tikzpicture}
}
\setlength{\abovecaptionskip}{-40pt}
\caption{Simulation results for Scenarios A), B), and C) with sample size $n = 5 \cdot 10^6$. For clarity, the x-axis displays time $k^{\prime} \coloneqq k+1 \in \{4, \dots, 13\}$ (all survival curves equal 1 at $k^{\prime}=0$).}
\label{fig: sim_hernan_00_reduced}
\end{figure}

The top panel of Figure \ref{fig: sim_hernan_00_reduced} presents the results for Scenario A). The estimate $\hat{S}_k^{*, \textnormal{msm}}$ coincides, up to a finite-sample error, with $S_k$ as the conditions of Proposition \ref{proposition: msm_hazard_characterization} are met: \ass{$\RAID$} and \ass{${\mathcal{M}_{\circ}}_{\mid V}^{\textnormal{cs}}$} with $E \in V$ hold. For $\hat{S}_k^{*, \textnormal{msm}-\textnormal{we}}$, conditions of Proposition \ref{proposition: msm_hazard_characterization} are not met as $E \notin V$. Hence, no matter how flexible the model $\mathcal{M}(\mathbb{B}_{\textnormal{sat}})$ is, it will never satisfy
\ass{$\mathcal{M}_V(\mathbb{B}_{\textnormal{sat}})^{\textnormal{cs}}$}, as \ass{$\mathcal{M}(\mathbb{B}_{\textnormal{sat}})_V^{\textnormal{cs}}$} and \ass{$\mathcal{M}(\mathbb{B}_{\textnormal{sat}})_V^{*-\textnormal{cs}}$} cannot hold simultaneously if $E \notin V$. Specifically, the working model $p(v,a_k,k;\beta^*)$ targets $h_{\prstudyp,k}^{a_k}(v)$, not $h_k^{a_k}(v)$, and thus $\hat{S}_k^{*, \textnormal{msm}}$ targets $\expectation\left(\prod_{m=0}^{k}\left(1-h_{\prstudyp,m}^{a_m=1}(V)\right)\right) \neq \expectation\left(\prod_{m=0}^{k}\left(1-h_{m}^{a_m=1}(V)\right)\right)= S_k$. Similarly, as elaborated in Section \ref{sec: wsc_procedures},
$\hat{S}_k^{*, \textnormal{km}}$ targets $\prod_{m=0}^{k}(1-h_{\prstudyp,m}^{\overline{a}_m=\overline{1}})$, see eq. \eqref{eq: wscp_y1_km_biased}. Thus, $\hat{S}_k^{*, \textnormal{msm}-\textnormal{we}}$ and $\hat{S}_k^{*, \textnormal{km}}$ underestimate $S_k$ because the outcome of interest is negatively dependent on time since study-entry (both marginally and conditional on $V$), implying that $h_k^{a_k=1}(v) < h_{\prstudyp, k}^{a_k=1}(v)$ (for all $v \in \textnormal{supp}(P(V))$) and $h_k^{a_k=1} < h_{\prstudyp, k}^{a_k=1}$ at each $k \in \{0, \dots, K\}$. 

The middle panel of Figure \ref{fig: sim_hernan_00_reduced} presents the results for Scenario B). None of the survival estimates approaches $S_k$. Departures of $\hat{S}_k^{*, \textnormal{msm}-\textnormal{we}}$ and $\hat{S}_k^{*, \textnormal{km}}$ from $S_k$ can be explained by arguments analogous to those presented for scenario A); $\hat{S}_k^{*, \textnormal{msm}}$ deviates from $S_k$ as the interaction term $\prgapp_{\textnormal{int}}(e,k)$ in  $h_k^{a_k}(v)$ prevents conditions of Proposition \ref{proposition: msm_hazard_characterization} from being satisfied. 

We emphasize that in scenarios A) and B) the discrepancies from the procedure-specific survival curves and $S_k$ are attributable to administrative censoring. Had we observed individuals \textnormal{even} during the post-study period, as in scenarios D) and E), the same procedures could have led to markedly different estimates; $\hat{S}_k^{*, \textnormal{km}}$ under \ass{$\FAID$}, and $\hat{S}_k^{*, \textnormal{msm}}$ and $\hat{S}_k^{*, \textnormal{msm}-\textnormal{we}}$ additionally under \ass{${\mathcal{M}_{\circ\circ}}_{\mid V}^{*-\textnormal{cs}}$}, for any choice of $V$, including $E \notin V$, would have coincided, up to a finite-sample error, with $S_k$. See Figure \ref{fig: sim_hernan_00_df} in Appendix \ref{sec_app: plots_case_study_1} for the corresponding plots.

The bottom panel of Figure \ref{fig: sim_hernan_00_reduced} presents the survival curves for scenario C). As argued in Appendix \ref{sec_app: comments_cf_study_case_1}, these curves are nearly indistinguishable because survival is defined marginally with respect to $P(E)$ and is therefore dominated by individuals entering early, who constitute nearly the entire target population ($P(E \in \{0,\dots, 5\}) \simeq 0.95$). Thus, the survival curves in scenarios C) and F) are essentially indistinguishable, indicating that the use of post-study observations is inconsequential for identification of $S_k$ (see Figure \ref{fig: sim_hernan_00_df} in Appendix \ref{sec_app: case_study_1}). This phenomenon does not extend to estimands defined within the subpopulation $\{E=e\}$. In particular, if the estimands of interest were $S_{k}(e) \coloneqq P(Y_k^{\overline{a}_k=\overline{1}, {c}_{k-1}= 0}=0\mid E=e)$ for $(e,k) \in \{0,\dots,15\}^2$, then a modified version of procedure 1), producing estimates $\hat{S}_{k}^{*,\textnormal{msm}}(e)$, would underestimate $S_{k}(e)$ whenever $\tau=(e+k)\in \{16,\dots,30\}$, as the interaction term $\prgapp_{\textnormal{int}}(e,k)$ is not identifiable from the observed data; see Figure \ref{fig: s_e_diff} in Appendix \ref{sec_app: case_study_1}.

\subsection{Case study 2: effect of SARS-CoV-2 mRNA vaccination on survival in patients treated with immune checkpoint inhibitors}
\label{sec: case_study_mrna_vaccine}
We now illustrate our ideas in a real-data setting. Consider one-month-spaced discrete follow-up times $k \in \{0, \dots, K=44\}$ and the data structure 
with a point treatment $A$ and no time-varying confounders. When $A$ is a point treatment, the estimands of interest simplify to
\begin{align}\label{estimand_mrna}
    S_k^{a} \coloneqq P \{Y_k^{a, c_{k-1}=0} = 0\} \\ 
    a \in \{0,1\}, \, \textnormal{ for } k \in\{0, \dots, K=44\}. \notag
\end{align}
In words, $S_k^{a}$ is the survival under SARS-CoV-2 mRNA vaccination $(a=1)$ and no vaccination $(a=0)$, had censoring been abolished.
Throughout, we consider two target trials (hereafter, trials) to estimate \eqref{estimand_mrna} among patients with non-small cell lung cancer $(n=277 \textnormal{ per arm } a \in \{0,1\})$. Trial A) is specified as in \citet{dumas_re-evaluating_2025}. Trial B) differs from trial A) only in that administrative censoring is treated as informative; accordingly, we assume \ass{$\RAID$} to allow for potential dependence between the outcome and the time of study-entry. 

Figure \ref{fig: real_data_analysis} illustrates the weighted Kaplan–Meier survival curves under vaccination, $\hat{S}_k^{*,\textnormal{km},a=1}$, and no vaccination, $\hat{S}_k^{*,\textnormal{km},a=0}$, estimated following trial B). The precise definition of \ass{$\RAID$} is provided, and point estimates are tabulated in Appendix \ref{sec_app: specifics_numerical_case_study_2}.

\begin{figure}[H]
    \centering%
{
    \begin{tikzpicture}[scale=1]
        \node[anchor=north, scale=0.65] at (0,0){%
        \pgfimage{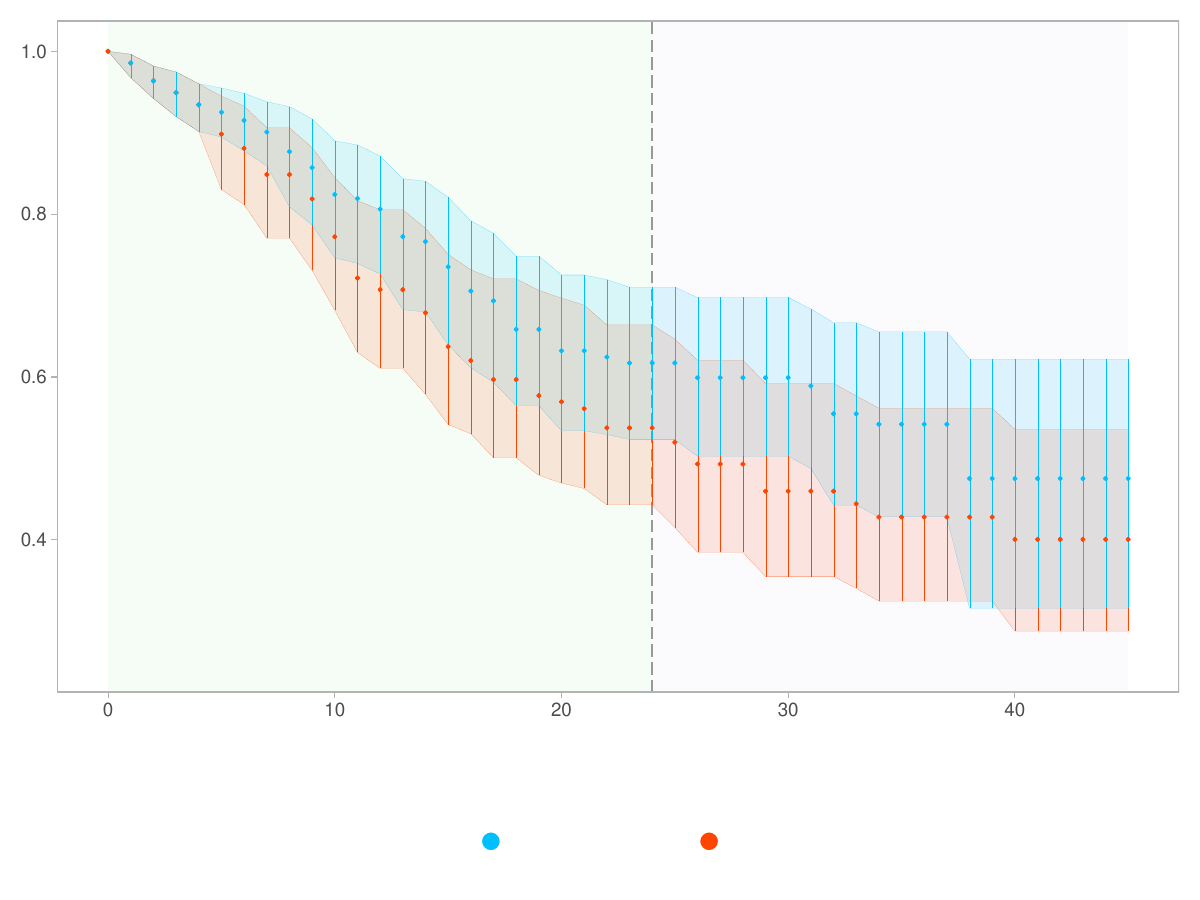}
        };    
        \node[anchor=north] at (0.5,-8){\scriptsize $k+1$};
        \node[anchor=north west] at (-1.5,-8.6){\scriptsize $\hat{S}_k^{*,\text{km},a=1}$};
        \node[anchor=north west] at (0.9,-8.6){\scriptsize $\hat{S}_k^{*,\text{km},a=0}$}; 
\end{tikzpicture}
}
\caption{$\hat{S}_k^{*,\text{km},a=1}$ and $\hat{S}_k^{*,\text{km},a=0}$ estimates (95\% point-wise  confidence intervals obtained by non parametric bootstrap). The x-axis reports time $k^{\prime} = k+1$ (at $k^{\prime}=0$, time of the first ICI initiation, all survival curves, by default, equal 1). $\hat{S}_{k^{\prime}}^{*,\text{km},a=1}=\hat{S}_{k^{\prime}}^{*,\text{km},a=0}$ during the grace period $\{0,1,2,3\}$ ($\simeq$ 100 days). The dashed-grey vertical line marks the follow-up time $k+1=23+1$ after which extrapolation assumptions are needed for the parameter targeted by $\hat{S}_k^{*,\text{km},a}$ to coincide with $S_k^{a}$.}
\label{fig: real_data_analysis}
\end{figure}
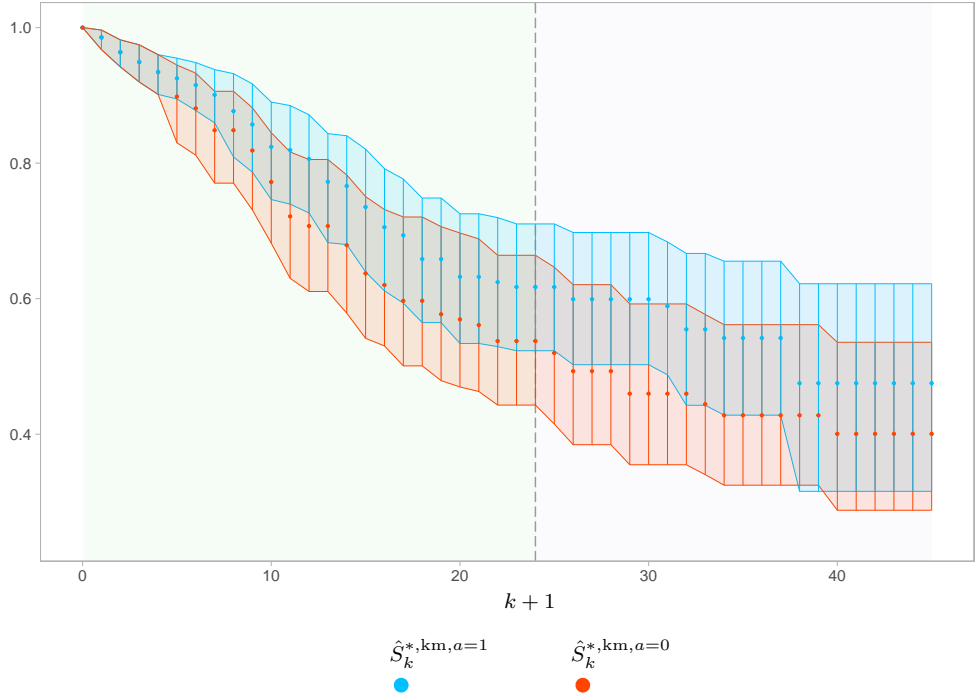
The confidence intervals presented in Figure \ref{fig: real_data_analysis}, which resemble those of trial A), intersect at every $k \in \{0, \dots, K\}$, suggesting an insufficient sample size to detect a positive effect of the vaccine. 
However, even with a (near infinite) sample, assumptions about post-study-period features are needed to interpret $S_k^{a}$ as the parameter targeted by $\hat{S}_k^{*,\textnormal{km}, a}$ for every $k \in \{0, \dots, K\}$, and to ensure the validity of the designs of both trials A) and B), that is, that the assumptions underlying the trial designs are compatible with the underlying data-generating mechanism.

To make this concrete, recall that eligible individuals in both trials entered the study from December 2020 ($\tau=0$) until September 2022 ($\tau=21$) and were followed until August 2024 (included), the administrative end of the study ($\tau=\mathcal{T}=44$). In particular, the eligible target population for which $S_k^{a}$ is evaluated is $\{E \leq 21\}$, i.e., $P\{E \leq 21\} = P\{E \leq \mathcal{T}\}=1$. We now distinguish two cases.

Suppose \ass{\CTC} are absent. Then, the design of trial A) is trivially valid, and for trial B), we have
\begin{multline*}
    \frac{{S}_k^{*,\textnormal{km},a}}{S_k^a}= \frac{\prod_{m=0}^{k}\left(1 - \textcolor{teal}{h_{m}^{a}}\right)}{\prod_{m=0}^{k}\left(1 - h_{m}^{a}\right)} = 1,
\end{multline*}
where $h_{m}^{a} \coloneqq P(Y_m^{a, c_{m-1}=0}=1 \mid Y_{m-1}^{a, c_{m-2}=0} = 0)$.
That is, the design of trial B) is valid because all hazards in $S_k^{*,\textnormal{km},a}$ refer to the target eligible population, $\{E \leq 21\}$.

Suppose instead that \ass{\CTC} are present, even if $Y_k$ is independent of $E$ for every arm $a \in \{0,1\}$ in $P^*$. Then, the design of trial A) is invalid, as it assumes the absence of \ass{\CTC} when they are in fact present. For trial B), define
\begin{align*}
    h_{\prstudyp,m}^{a} \coloneqq P(Y_m^{ a, c_{m-1}=0}=1 \mid Y_{m-1}^{a, c_{m-2}=0} = 0, \textcolor{teal}{E \in \{0, \dots, \mathcal{T}-m\}}).
\end{align*}
For $k > 23$, we have 
\begin{multline*}
    \frac{{S}_k^{*,\textnormal{km},a}}{S_k^a}= \frac{\prod_{m=0}^{23}\left(1 - h_{m}^{a}\right) \cdot \prod_{m=24}^{k}\left(1 - h_{\prstudyp, m}^{a}\right)}{\prod_{m=0}^{23}\left(1 - h_{m}^{a}\right) \cdot \prod_{m=24}^{k}\left(1 - h_{m}^{a}\right)} = \frac{\prod_{m=24}^{k}\left(1 - \textcolor{teal}{h_{\prstudyp, m}^{a}}\right)}{\prod_{m=24}^{k}\left(1 - h_{m}^{a}\right)}.
\end{multline*}
Thus, ${S}_k^{*,\textnormal{km},a}$ coincides with $S_k^{a}$ only if, for every $m\in\{24, \dots, k\}$,
\begin{align}\label{eq: mrna_condition_no_bias_km}
h_{\prstudyp, m}^{a}=h_{m}^{a}
\end{align}
which is isomorphic to assumption \eqref{eq: condition_no_bias_km}. Hence, using arguments analogous to those in Section \ref{sec: wsc_procedures}, \eqref{eq: mrna_condition_no_bias_km} is unlikely to hold under \ass{\CTC} and therefore
\begin{multline*}
\frac{{S}_k^{*,\textnormal{km},a}}{S_k^a}=  \frac{\prod_{m=24}^{k}\left(1 - \textcolor{teal}{h_{\prstudyp, m}^{a}}\right)}{\prod_{m=24}^{k}\left(1 - h_{m}^{a}\right)} \neq 1.
\end{multline*}
That is, the design of trial B) is invalid because it no longer targets $S_k^a$: among the hazards defining $S_k^{*,\textnormal{km}, a}$, only those evaluated at $k \leq 23$ (light-green area) correspond to the target eligible population, $\{E \leq 21\}$; for $k > 23$ (lavender area) the hazards instead correspond to different populations depending on the follow-up time $k$, i.e.  $\{E \leq \mathcal{T} - k\} \subset \{E \leq 21\}$.\footnote{For example, if $k=24$, only individuals $\{E \leq \mathcal{T}-k = 44-24 = 20\}$ are included; if $k=25$, only individuals $\{E \leq \mathcal{T}-k = 44-25 = 19\}$ are included, and so on.}

Thus, the designs of both trials A) and B) are invalid, irrespective of whether assumptions concerning the study period are compatible with \ass{\CTC}. In particular, in trial B), the absence of evidence in $P^*$ of dependence between $Y_k$ and $E$ for both vaccinated and unvaccinated individuals does not rule out dependence between $E$ and $Y_k^{a, c_{k-1}=0}$ arising during the post-study period, and hence the presence of \ass{\CTC}. This point, hitherto not captured by existing notions of calendar-time changes, illustrates the limits of what can be learned about trial design from the observed data.

\subsubsection{Remarks on calendar time and target trial procedures}
Assumptions isomorphic to \eqref{eq: mrna_condition_no_bias_km} could, in principle, be falsified had individuals been observed even during the post-study period.
However, had we conducted the same analysis at the time the original study \citet{grippin_sars-cov-2_2025} was received (November 2024, $\tau=47$), assumption \eqref{eq: mrna_condition_no_bias_km} could not have been falsified: the only observations that could have been collected for the target population, $\{E \leq 21\}$, in the post-study period are limited to follow-up $k=26$, as $\tau^{\prime}=47 < (21 + k)=\tau^{\prime\prime}$ for $k \in \{27, \dots, K=44\}$. Thus, when asserting that $h_{k}^a=h_{\prstudyp,k}^a$ for $k \in \{27, \dots, K=44\}$, and hence that \eqref{eq: mrna_condition_no_bias_km} holds, we must reason about events that had not yet occurred at the time the analysis was conducted.

While not directly applicable to our analysis, these considerations are critical to the validity of target trial–based analyses. For instance, consider the important work by \citet{dickerman_avoidable_2019} investigating the effect of statin initiation on cancer-free survival. Let time be discretized into equally spaced 1-year intervals. \citet{dickerman_avoidable_2019} designed a target trial using MSM hazard–based procedures. Individuals were eligible if they entered the study between 1999 ($\tau = 0$) and 2015 ($\tau = 16$) and were followed until 2016 ($\mathcal{T} = 17$). The reported outcomes were cancer-free survival curves under statin initiation vs. no initiation up to $K = 10$ years. Because the analysis was conducted in 2019 ($\tau = 20$), at any follow-up time $k \in \{1, \dots, K\}$, survival must be derived by evaluating MSM-hazard models during the post-study period. 

Establishing the absence of an effect of statin initiation would not have been possible at the time the analysis was conducted without extrapolation: $\tau^{\prime}=20 < (17 + k) = \tau^{\prime\prime}$ for $k \in \{4, \dots, 10\}$. While likely inconsequential in the compelling analysis of \citet{dickerman_avoidable_2019}, this step is formally required to interpret $S_k^{*,\textnormal{km},a}$ as $S_k^{a}$ for every $k \in \{0, \dots, K\}$ and $a \in \{0,1\}$. We return to this issue in the next section.

\subsection{Sensitivity analyses with extrapolation curves}
\label{sec: c_values_case_studies}
In Sections \ref{sec: case_study_hernan_00_sim} and \ref{sec: case_study_mrna_vaccine}, we elaborated on the complications stemming from not observing individuals after the end of the study. We now illustrate the use of c-values and extrapolation curves as a sensitivity analysis strategy to corroborate the findings of \citet{hernan_marginal_2000}, as reproduced in scenario C) of case study 1, and to assess the reliability of the effect of SARS-CoV-2 mRNA vaccination inferred from a target trial emulation reanalysis of \citep{grippin_sars-cov-2_2025}.

\subsubsection{Case study 1}
\label{sec: c_values_case_study_1}
Consider procedure 1). Let $\hat{RR}_k^{\rightsquigarrow} = (1 - \hat{S}_k^{*, \textnormal{msm},\overline{a}_k=\overline{1}}) \cdot {(1 - \hat{S}_k^{*, \textnormal{msm},\overline{a}_k=\overline{0}})}^{-1}$. 
Because the estimated effect of zidovudine is protective (i.e., $\hat{\beta}_A^* < 0$), the proportional hazards MSM implies that, for sufficiently large $n$, $\hat{RR}_k^{\rightsquigarrow} < 1$ for every $k \in \{0, \dots, K\}$.\footnote{We have $\textnormal{sign}(\hat{\beta}_A^*) = \textnormal{sign}(\hat{RR}_k^{\rightsquigarrow} - 1)$ with high probability.} 

C-values exist for $k \in \{K_{\textnormal{ex}} = 13, 14, 15=K\}$, i.e., for $k \in \{1,\dots, 12\}$, the proportion of individuals who can first experience failure during the post-study period is insufficient to nullify the protective effect of zidovudine. Numerical values of the extrapolation curves are reported in Table \ref{table: ex_curves_case_study_1}. 

These results, together with subject-matter considerations, provide additional support for the protective effect of zidovudine reported by \citet{hernan_marginal_2000}.

In the MACS cohort, individuals (HIV-positive men) received zidovudine in two main forms: either as monotherapy or in combination with other therapies. In \citet{hernan_marginal_2000}, it is assumed that once an individual initiates zidovudine therapy, he remains on zidovudine; this implies that if zidovudine is later used in combination with other therapies, then both monotherapy and combined therapy (with the adjuvant therapy complementing the monotherapy starting during the post-study period) are consistent versions of zidovudine defined as a nominally-fixed treatment. Hence, the effect of zidovudine investigated in \citet{hernan_marginal_2000} pertains to zidovudine administered as both a monotherapy and combined therapy. 

Suppose that, consistent with later findings \citep{darbyshire_delta_1996}, combined therapy is more effective than monotherapy, as in the true data-generating mechanism of scenario C), where the interaction term $\prgapp_{\textnormal{int}}(e,k)$, given on the logit scale by $-0.5 \cdot \mathbb{1}_{\{a_k = 1, e + k > \mathcal{T}\}}$, is present. Because close to the end of the study conducted by \citet{hernan_marginal_2000}, the most prevalent type of zidovudine administration shifted from monotherapy to combined therapy \citep{detels_effectiveness_1998}, the effect of zidovudine as a single, nominally-fixed treatment must have improved during the post-study period relative to the study period. Such an improvement, however, could not be captured by the estimate $\hat{\beta}_A^{*}$ because it manifested entirely during the post-study period. An implication is that extrapolation assumptions are not met in \citet{hernan_marginal_2000}, but the failure of the extrapolation assumptions likely underestimated the true effect. Specifically, nullifying the protective effect of zidovudine would require a substantial increase in the (extrapolated) post-study hazards under receipt of zidovudine (see $\delta_k^{1}$ in Table \ref{table: ex_curves_case_study_1}). Yet, the deviations from the post-study hazards under receipt of treatment are expected to be strictly negative (in the true data-generating mechanism, they vary within $[-0.06, -0.01]$). Therefore, for any $k \in \{K_{\textnormal{ex}} = 13, 14, 15=K\}$, the variations required to nullify the protective effect of zidovudine seem to be clinically implausible.
\begin{table}[ht]
\centering
\scriptsize
\begin{tabular}{ccccccc}
  \toprule
$k$ & $\text{c}_k^0+\text{c}_k^1$ & $\text{c}_k^0$ & $\text{c}_k^1$ & $\delta_k^0$ & $\delta_k^1$ \\
  \midrule
   13 & 7.54 & 0.13 & 7.41 & -0.10 & 0.43 \\
   14 & 5.11 & 0.20 & 4.90 & -0.09 & 0.21 \\
   15 & 3.68 & 0.30 & 3.38 & -0.05 & 0.11 \\
   \bottomrule
\end{tabular}
\caption{Extrapolation curves for case study 1, scenario C).}
\label{table: ex_curves_case_study_1}
\end{table}

\subsubsection{Case study 2}
\label{sec: c_values_case_study_2}
Consider now case study 2. Owing to the limited sample size, we estimated the (extrapolated) effect of the vaccine by fitting a proportional hazards MSM, which yields a protective effect, i.e., $\hat{RR}_k^{\rightsquigarrow} < 1$ for all $k \in \{0, \dots, K\}$. C-values exist for $k \in \{27, \dots, 44\}$; for $k \leq 26$, the proportion of individuals who can first experience failure during the post-study period is insufficient to nullify the protective effect of the vaccine. Numerical values of extrapolation curves are reported in Table \ref{table: ex_curves_case_study_2} and depicted in Figure \ref{fig: extrapolation_curves_mrna_vaccine}. To interpret these results, we contrast two clinically motivated scenarios that differ in the direction of post-study variations under vaccination. Suppose first that mRNA vaccination improves survival by gradually increasing the tumour sensitivity to immune checkpoint inhibitors. Such immune remodeling develops over many months, so the hazards under vaccination would only differ from the hazards under no vaccination at larger follow-up times $k$. This could lead to an underestimation of the protective effect of vaccination, i.e., for $k \in \{K_{\textnormal{ex}} = 27, \dots, 44=K\}$, the deviations from the (extrapolated) hazards under vaccination are expected to be strictly negative. Thus, analogous arguments to those presented in Section \ref{sec: c_values_case_study_1} can be used to corroborate the protective effect of the vaccine. Conversely, suppose that vaccinated individuals are more likely to complete early therapy and survive long enough to receive later-line regimens, which are highly toxic. Then, because post-study events are associated with higher follow-up times $k$,\footnote{For any given population $\{E=e\}$, the follow-up $k_{\prstudyp}$ for event $\textnormal{Ev}_{\prstudyp}$ occurring during the study period is smaller than that ($k_{\prgapp}$) of an event $\textnormal{Ev}_{\prgapp}$ occurring in the same population, but during the post-study period: $e+k_{\prstudyp} \leq \mathcal{T}$, but $e+k_{\prgapp} > \mathcal{T}$ $\Rightarrow$ $ k_{\prstudyp} < k_{\prgapp}$.} the inability to measure these during the study period would lead to an overestimation of the protective effect of the vaccine. Such knowledge can guide the choice of the maximum follow-up $K_{\textnormal{exe}}$ beyond which estimates should not be considered reliable. If quantitative knowledge is available, $K_{\textnormal{exe}}$ can be selected by specifying a threshold for post-study variations; for example, setting a threshold of $0.05$ on the additive-hazard scale and reporting estimates up to the largest $K_{\textnormal{exe}}$ such that $\delta_{K_{\textnormal{exe}}}^{1} > 0.05$ yields $K_{\textnormal{exe}} = 35$. In the absence of quantitative knowledge, $K_{\textnormal{exe}}$ may be selected using an elbow criterion based on $\textnormal{c}_k^{0} + \textnormal{c}_k^{1}$, reflecting a conservative worst-case profile (in our case, this gives $K_{\textnormal{exe}} = 32$, see Figure \ref{fig: extrapolation_curves_mrna_vaccine}). 
\begin{table}[H]
\centering
\scriptsize
\begin{tabular}{ccccccccccc}
  \toprule
$k$ & $\textnormal{c}_k^0 + \textnormal{c}_k^1$ & $\textnormal{c}_k^0$ & $\textnormal{c}_k^1$ & $\delta_k^0$ & $\delta_k^1$ \\ 
  \midrule
    27 & 21.702 & 0.046 & 21.656 & -0.042 & 0.945 \\
    28 & 15.210 & 0.066 & 15.144 & -0.039 & 0.463 \\
    29 & 11.311 & 0.089 & 11.222 & -0.036 & 0.280 \\
    30 & 8.801 & 0.115 & 8.686 & -0.031 & 0.188  \\
    31 & 7.154 & 0.143 & 7.011 & -0.029 & 0.139 \\
    32 & 6.007 & 0.171 & 5.836 & -0.026 & 0.107 \\
    33 & 5.100 & 0.204 & 4.896 & -0.023 & 0.083 \\
    34 & 4.432 & 0.238 & 4.194 & -0.021 & 0.066 \\
    35 & 3.920 & 0.274 & 3.646 & -0.019 & 0.053 \\
    36 & 3.518 & 0.312 & 3.206 & -0.017 & 0.044 \\
    37 & 3.220 & 0.348 & 2.872 & -0.016 & 0.037 \\
    38 & 2.987 & 0.384 & 2.603 & -0.014 & 0.031 \\
    39 & 2.805 & 0.419 & 2.385 & -0.013 & 0.027 \\
    40 & 2.661 & 0.453 & 2.209 & -0.012 & 0.023 \\
    41 & 2.548 & 0.485 & 2.063 & -0.012 & 0.020 \\
    42 & 2.457 & 0.515 & 1.942 & -0.011 & 0.018 \\
    43 & 2.386 & 0.543 & 1.843 & -0.010 & 0.016 \\
    44 & 2.327 & 0.569 & 1.758 & -0.009 & 0.015 \\
   \bottomrule
\end{tabular}
\caption{Numerical results for extrapolation curves for case study 2.}
\label{table: ex_curves_case_study_2}
\end{table}
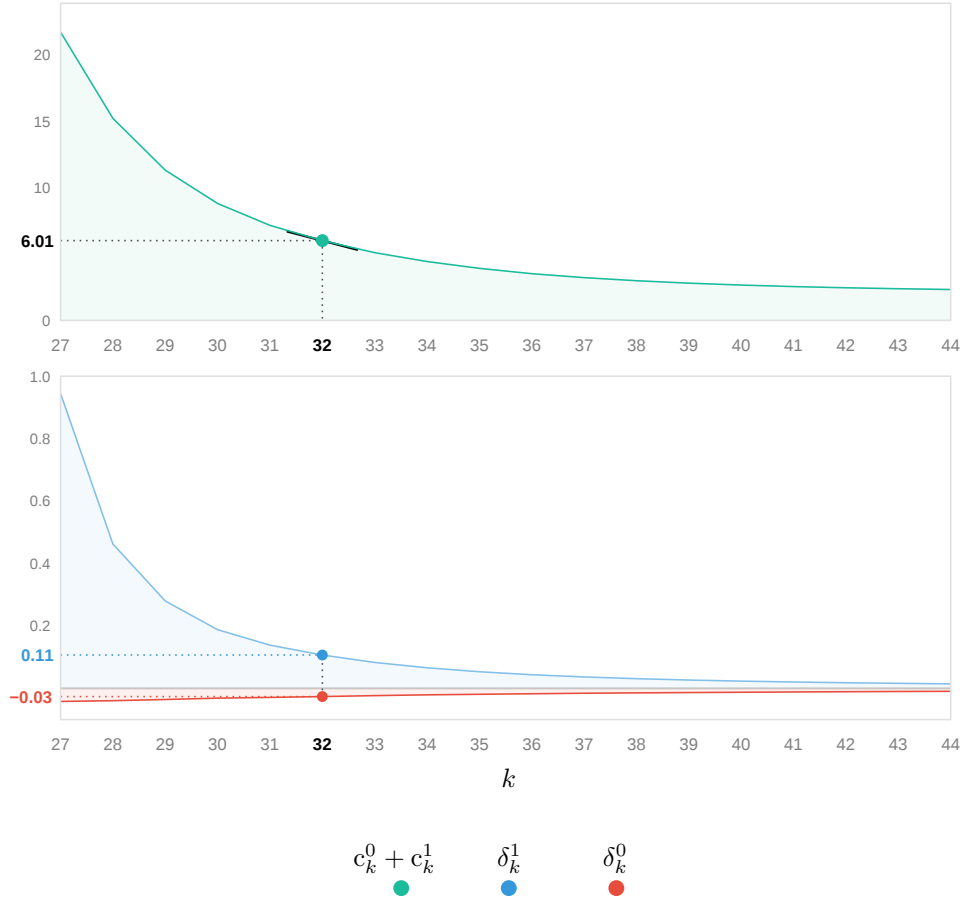
\begin{figure}[ht]
    \centering
{
    \begin{tikzpicture}[scale=1]
        \node[anchor=north, scale=0.7] at (0,0){%
        \pgfimage{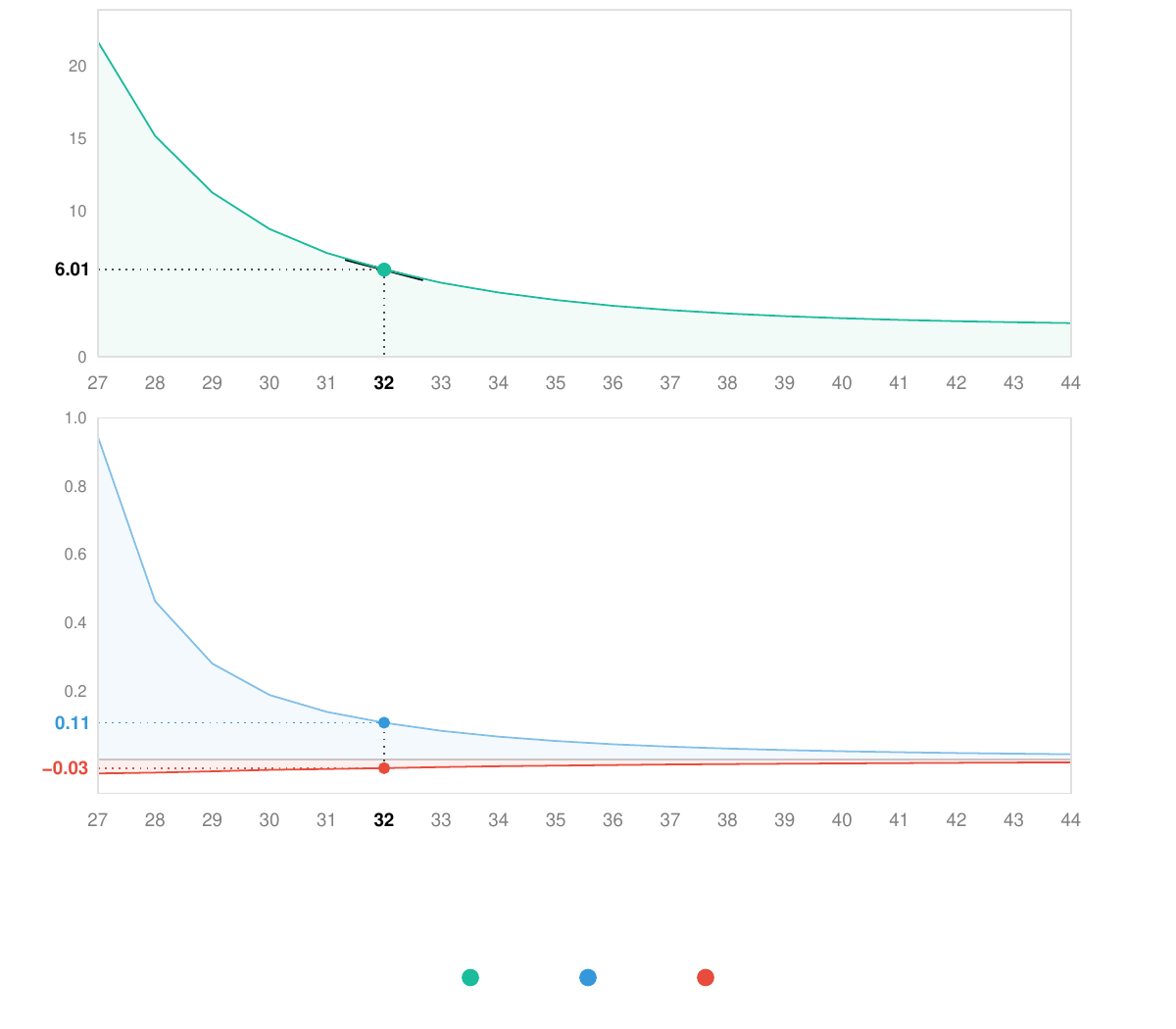}
        };    
       \node[anchor=north] at (0,-10.2){\footnotesize $k$};
        \node[anchor=north west] at (-2.2,-11.2){\footnotesize $\text{c}_k^0 + \text{c}_k^1$};
        \node[anchor=north west] at (-0.3,-11.2){\footnotesize $ \delta_{k}^1$}; 
        \node[anchor=north west] at (1.1,-11.2){\footnotesize $ \delta_{k}^0$}; 
\end{tikzpicture}
}
\caption{Extrapolation curves for case study 2. We have $K_{\text{ex}}=27$. The elbow criterion, reflecting a conservative worst-case scenario, yields $K_{\text{exe}} = 32$.}
\label{fig: extrapolation_curves_mrna_vaccine}
\end{figure}

    \clearpage
    \section{Discussion}
Our work establishes conditions under which MSM-hazard and WSC procedures target causal effects. We also discuss the practical consequences of failing to meet these conditions in different settings. 
For example, in MSM-hazard procedures, omitting $E$ from the conditioning set $V$ is often incompatible with an investigator's working assumptions. In target trial emulation, extrapolation assumptions are required in staggered-entry studies.
To support the use of these procedures in practice, we proposed extrapolation curves as sensitivity measures and illustrated their use in practice. We believe these results will help applied researchers understand and structure their causal analyses. We also believe that with the proliferation of analyses based on observational data (e.g., health records), our results will be increasingly relevant in practice, where the role of $E$ often needs to be considered. 

\section{Acknowledgments}
This work was supported by the Swiss National Science Foundation (Grant number: 207436). Lorenzo dedicates this work to his cousin Francesco.

\subsection*{Software and AI Disclosure}
The code used in Section 5 builds on that provided by \citet{seaman_simulating_2024} and \citet{dumas_re-evaluating_2025}. The authors acknowledge the use of Large Language Models (LLMs) (ChatGPT and Gemini) during the preparation of this manuscript. In the main text, LLMs were used, in part, to assist in reviewing the literature, formulating the substantive justifications for the clinically motivated scenarios in Section 5.3, improving the language, and generating custom symbols and tables. In the Web appendices, LLMs were used to assist in refining the presentation of the proofs, to support the development of the proofs of Propositions S5–S7 (Web appendices D.3 and D.4), to draft the first two paragraphs on page 2, to improve the language, and to develop and polish code. The use of LLMs did not contribute to the study design, methodology, statistical analyses, interpretation of the results, or scientific conclusions. All AI-assisted content was critically reviewed, verified, and revised by the authors. The authors take full responsibility for the content of the manuscript.

    \clearpage
    \renewcommand{\refname}{References} 
    \renewcommand{\bibname}{References}
    
    \bibliographystyle{biom}
    \putbib[bibliography_main] 
\end{bibunit}

\clearpage

\begin{bibunit}
    \begin{appendices}
\setcounter{page}{1}
\setcounter{table}{0}
\renewcommand{\thetable}{S\arabic{table}}
\setcounter{figure}{0}
\renewcommand{\thefigure}{S\arabic{figure}}
\setcounter{equation}{0}
\renewcommand{\theequation}{S\arabic{equation}}
\setcounter{definition}{0}
\renewcommand{\thedefinition}{S\arabic{definition}}
\setcounter{proposition}{0}
\renewcommand{\theproposition}{S\arabic{proposition}}
\setcounter{corollary}{0}
\renewcommand{\thecorollary}{S\arabic{corollary}}
\setcounter{lemma}{0}
\renewcommand{\thelemma}{S\arabic{lemma}}
\setcounter{theorem}{0}
\renewcommand{\thetheorem}{S\arabic{theorem}}

\begin{center}

    {\Huge \bfseries Web Appendices} \\[1.5ex]
    {\large \color{gray} of \textit{Causal inference with staggered entries and effects that change over calendar-time}}
\end{center}
\vspace{3em}

\titlelabel{Web Appendix \thetitle.\quad}
\begingroup
  \hypersetup{linkcolor=black, linktoc=all} 
  
  \linespread{1.5}\selectfont 
  
  \makeatletter
  \renewcommand{\@tocrmarg}{3.5em} 
  \renewcommand{\@pnumwidth}{2em} 
  \makeatother
  
  \DoToC 
\endgroup
\clearpage

\addtocontents{toc}{\protect\setcounter{tocdepth}{1}}
\begin{center}
    \textbf{Guide to the Web Appendices}
\end{center}
\label{sec_app: guide_appendices}
Appendix \ref{sec_app: double_bind_censoring_assumptions_literature_iptw}  formalizes calendar-time changes, clarifies why standard censoring assumptions may be ill-posed, and introduces compatible alternatives together with IPTW results. Appendix \ref{sec_app: consequences_of_ctc} develops the main technical results for WSC and MSM-hazard procedures, including identification, compatibility, and bounds. Appendix \ref{sec_app: extrapolation_assumptions_and_c_values} focuses on interpretation and practice: it contains the extrapolation assumptions underlying MSM-hazard procedures and introduces c-values and extrapolation curves as sensitivity tools. Appendix \ref{sec_app: case_studies_specifics} presents the case studies, and elaborates on the data-generating mechanisms, estimators, and numerical results.

Notation is adapted to each section, and any changes are stated explicitly. Assumptions, in contrast to the main text, are not highlighted in blue. 

Equalities involving random variables are understood almost surely. In proofs, we use the following conventions:
\begin{center}
\begin{minipage}[c]{0.25\textwidth}
\centering
\begin{align*}
    X=& Y \\
    =& W & \Leftarrow Z
\end{align*}
\end{minipage}
\hfill $\equiv$ \hfill
\begin{minipage}[c]{0.25\textwidth}
\centering
\begin{align*}
    &X \\
    =& Y \\
    =& W & \Leftarrow Z
\end{align*}
\end{minipage}
\hfill $\equiv$ \hfill
\begin{minipage}[c]{0.25\textwidth}
\centering
\begin{align*}
    \left[X=Y \text{ and } Z\right] \Rightarrow Y=W,
\end{align*}
\end{minipage}
\end{center}

\begin{center}
\begin{minipage}[c]{0.25\textwidth}
\centering
\begin{align*}
    &X \\
    \Rightarrow& Y & \Leftarrow{Z}
\end{align*}
\end{minipage}
\hfill $\equiv$ \hfill
\begin{minipage}[c]{0.25\textwidth}
\centering
\begin{align*}
     X\Rightarrow& Y & \Leftarrow{Z}
\end{align*}
\end{minipage}
\hfill $\equiv$ \hfill
\begin{minipage}[c]{0.25\textwidth}
\centering
\begin{align*}
    X \text{ and } Z \Rightarrow Y.
\end{align*}
\end{minipage}
\end{center}\clearpage
\clearpage
\section{List of symbols}
\label{sec_app: guide_appendices}

\begin{longtable}{
    >{\raggedright\arraybackslash}p{0.20\textwidth}
    >{\raggedright\arraybackslash}p{0.57\textwidth}
    >{\centering\arraybackslash}p{0.15\textwidth}
}
\caption{Nomenclature}
\label{tab: nomenclature}
\\
\toprule
\textbf{Symbol}
&
\textbf{Meaning}
&
\textbf{First introduced in Section}
\\
\midrule
\endfirsthead

\multicolumn{3}{c}
{\tablename\ \thetable\ -- \textit{continued from previous page}}
\\
\toprule
\textbf{Symbol}
&
\textbf{Meaning}
&
\textbf{First introduced in Section}
\\
\midrule
\endhead

\midrule
\multicolumn{3}{r}{\textit{Continued on next page}}
\\
\endfoot

\bottomrule
\endlastfoot

$E$
&
Calendar time of study entry.
&
\ref{sec: introduction}
\\

$e$
&
A realization of the calendar time of study entry $E$.
&
\ref{sec: censoring_assumptions_are_illposed}
\\

$k$
&
Discrete follow-up time measured from study entry.
&
\ref{sec: censoring_assumptions_are_illposed}
\\

$m$
&
Generic follow-up-time index, generally ranging from $0$ to $k$.
&
\ref{sec: censoring_assumptions_are_illposed}
\\

$j$
&
Generic index, typically used to index follow-up times inside products or sums.
&
\ref{sec: consequences_and_solutions_in_wsc_and_msm_hazard}
\\

$K$
&
Maximum follow-up time.
&
\ref{sec: censoring_assumptions_are_illposed}
\\

$\mathcal{T}$
&
Calendar time of the administrative end of the study.
&
\ref{sec: censoring_assumptions_are_illposed}
\\

$\tau$
&
Calendar time corresponding to study entry $e$ and follow-up $k$, namely
$\tau=e+k$.
&
\ref{sec: censoring_assumptions_are_illposed}
\\

$L_k$
&
Baseline covariates when $k=0$ and time-varying covariates when $k>0$.
&
\ref{sec: censoring_assumptions_are_illposed}
\\

$l_k$
&
A realization of the covariate vector $L_k$.
&
\ref{sec: censoring_assumptions_are_illposed}
\\

$C_{k-1}$
&
Right-censoring indicator before the outcome at follow-up $k$:
$C_{k-1}=1$ denotes right censoring and $C_{k-1}=0$ denotes remaining
uncensored.
&
\ref{sec: censoring_assumptions_are_illposed}
\\

$A_k$
&
Treatment indicator at follow-up $k$:
$A_k=1$ denotes treatment and $A_k=0$ denotes control.
&
\ref{sec: censoring_assumptions_are_illposed}
\\

$a_k$
&
Treatment value assigned at follow-up $k$ under a treatment strategy.
&
\ref{sec: illustrative_example_censoring_assumptions_ill_posed}
\\

$A$
&
Point treatment assigned at baseline, used in the special case without
time-varying treatment.
&
\ref{sec: censoring_assumptions_are_illposed}
\\

$a$
&
A realization of the point treatment $A$, with $a\in\{0,1\}$.
&
\ref{sec: case_study_mrna_vaccine}
\\

$Y_k$
&
Event indicator at follow-up $k$:
$Y_k=1$ denotes death or failure and $Y_k=0$ denotes survival.
&
\ref{sec: censoring_assumptions_are_illposed}
\\

$\overline{A}_k$
&
Treatment history through follow-up $k$,
$\overline{A}_k=\{A_m:m\in\{0,\dots,k\}\}$.
&
\ref{sec: censoring_assumptions_are_illposed}
\\

$\overline{a}_k$
&
Treatment strategy through follow-up $k$.
&
\ref{sec: illustrative_example_censoring_assumptions_ill_posed}
\\

$\overline{L}_k$
&
Covariate history through follow-up $k$.
&
\ref{sec: censoring_assumptions_are_illposed}
\\

$\overline{Y}_k$
&
Event history through follow-up $k$.
&
\ref{sec: censoring_assumptions_are_illposed}
\\

$\overline{C}_{k-1}$
&
Censoring history through follow-up $k-1$.
&
\ref{sec: censoring_assumptions_are_illposed}
\\

$\overline{0}$
&
A vector of $k+1$ zeros; its dimension is determined by context.
&
\ref{sec: censoring_assumptions_are_illposed}
\\

$\overline{1}$
&
A vector of $k+1$ ones; its dimension is determined by context.
&
\ref{sec: censoring_assumptions_are_illposed}
\\

$\emptyset$
&
An absent or undefined variable, including variables indexed outside
$\{0,\dots,K\}$.
&
\ref{sec: censoring_assumptions_are_illposed}
\\

$\mathbb{1}_{\{\cdot\}}$
&
Indicator function of the event appearing between braces.
&
\ref{sec: illustrative_example_censoring_assumptions_ill_posed}
\\

$\operatorname{supp}(P(X))$
&
Support of the distribution of a random variable $X$.
&
\ref{sec: illustrative_example_censoring_assumptions_ill_posed}
\\

$Y_k^{\overline{a}_k,c_{k-1}=0}$
&
Potential outcome at follow-up $k$ under treatment strategy
$\overline{a}_k$ and an intervention abolishing censoring through
follow-up $k-1$.
&
\ref{sec: illustrative_example_censoring_assumptions_ill_posed}
\\

$Y_k^{a,c_{k-1}=0}$
&
Potential outcome at follow-up $k$ under point treatment $a$ and an
intervention abolishing censoring through follow-up $k-1$.
&
\ref{sec: case_study_mrna_vaccine}
\\

$S_k^{\overline{a}_k}$
&
Survival probability at follow-up $k$ under treatment strategy
$\overline{a}_k$ and abolished censoring:
\[
S_k^{\overline{a}_k}
=
P\left\{
Y_k^{\overline{a}_k,c_{k-1}=0}=0
\right\}.
\]
&
\ref{sec: illustrative_example_censoring_assumptions_ill_posed}
\\

&
Survival probability at follow-up $k$ under point treatment $a$ and
abolished censoring:
\[
S_k^a
=
P\left\{
Y_k^{a,c_{k-1}=0}=0
\right\}.
\]
&
\ref{sec: case_study_mrna_vaccine}
\\

\ass{ex}$(e,k)$
&
Conditional sequential exchangeability assumption for treatment and
censoring at study-entry time $e$ and follow-up $k$.
&
\ref{sec: illustrative_example_censoring_assumptions_ill_posed}
\\

\ass{con}$(e,k)$
&
Consistency assumption at study-entry time $e$ and follow-up $k$.
&
\ref{sec: illustrative_example_censoring_assumptions_ill_posed}
\\

\ass{pos}$(e,k)$
&
Positivity assumption at study-entry time $e$ and follow-up $k$.
&
\ref{sec: illustrative_example_censoring_assumptions_ill_posed}
\\

\ass{\CTC}
&
Calendar-time changes: dependencies between study entry and the outcome
of interest that may arise during the post-study period.
&
\ref{sec: the_problem_conditioning_on_e}
\\

$\ass{\FAID}$
&
Intersection of the full exchangeability, consistency, and positivity
conditions, required during both the study and post-study periods.
&
\ref{sec: consequences_and_solutions_in_wsc_and_msm_hazard}
\\

$\ass{\RAID}$
&
Restricted version of $\ass{\FAID}$ in which exchangeability,
consistency, and positivity are required only during the study period.
&
\ref{sec: consequences_and_solutions_in_wsc_and_msm_hazard}
\\

$P$
&
Underlying population law.
&
\ref{sec: censoring_assumptions_are_illposed}
\\

$P^{\mathrm{o}}$
&
Observed-data law induced by the population law $P$.
&
\ref{sec: wsc_procedures}
\\

$P^*$
&
Pseudopopulation law generated from $P$ by weighting.
&
\ref{sec: wsc_procedures}
\\

${\{W_j\}}_{j=0}^m$
&
Weight used to generate the pseudopopulation $P^*$ at follow-up $m$.
&
\ref{sec: wsc_procedures}
\\

$h_m^{*,\overline{a}_m}$
&
Hazard at follow-up $m$ among individuals following treatment strategy
$\overline{a}_m$ and remaining uncensored in the pseudopopulation
$P^*$.
&
\ref{sec: wsc_procedures}
\\

$h_m^{\overline{a}_m}$
&
Marginal counterfactual hazard at follow-up $m$ under treatment strategy
$\overline{a}_m$ and abolished censoring.
&
\ref{sec: wsc_procedures}
\\

$h_{\prstudyp,m}^{\overline{a}_m}$
&
Counterfactual hazard at follow-up $m$ restricted to individuals whose calendar time $E+m$ lies within the study period.
&
\ref{sec: wsc_procedures}
\\

$h_m^{\overline{a}_m}(e)$
&
Counterfactual hazard at follow-up $m$, conditional on study entry
$E=e$.
&
\ref{sec: wsc_procedures}
\\

$S_k^{*,\mathrm{km},\overline{a}_k}$
&
Survival parameter targeted by a weighted Kaplan--Meier or weighted
survival-curve procedure:
\[
S_k^{*,\mathrm{km},\overline{a}_k}
=
\prod_{m=0}^k
\left(1-h_m^{*,\overline{a}_m}\right).
\]
&
\ref{sec: wsc_procedures}
\\

$J_k$
&
Set of follow-up times at which the study-period hazard differs from the
target-population hazard.
&
\ref{sec: wsc_procedures}
\\

$V$
&
Vector of baseline covariates used to define conditional counterfactual
or pseudopopulation hazards; it may include $E$.
&
\ref{sec: msm_hazard_procedures}
\\

$v$
&
A realization of the baseline-covariate vector $V$.
&
\ref{sec: msm_hazard_procedures}
\\

$h_k^{\overline{a}_k}(v)$
&
Counterfactual hazard at follow-up $k$ under strategy
$\overline{a}_k$, conditional on $V=v$.
&
\ref{sec: msm_hazard_procedures}
\\

$h_k^{a_k}(v)$
&
Counterfactual hazard at follow-up $k$ conditional on $V=v$ when,
under \ass{HI}, the hazard depends on the current treatment $a_k$ but
not on past treatment.
&
\ref{sec: msm_hazard_procedures}
\\

\ass{HI}
&
Hazard-independence restriction requiring
$h_k^{\overline{a}_k}(v)$ not to depend on the past treatment history
$\overline{a}_{k-1}$.
&
\ref{sec: msm_hazard_procedures}
\\

$\mathcal{M}(\mathbb{B})$
&
A regression model represented as a collection of functions
$p(\,\cdot\,;\beta)$.
&
\ref{sec: msm_hazard_procedures}
\\

$\mathbb{B}$
&
Parameter space indexing the regression model $\mathcal{M}$.
&
\ref{sec: msm_hazard_procedures}
\\

$p(v,a_k,k;\beta)$
&
Model-implied hazard as a function of baseline covariates $v$, treatment
$a_k$, follow-up $k$, and parameter $\beta$.
&
\ref{sec: msm_hazard_procedures}
\\

$\beta$
&
Generic parameter indexing a regression model.
&
\ref{sec: msm_hazard_procedures}
\\

$\beta_{\prstudyp\prgapp}$
&
Parameter governing the correctly specified counterfactual-hazard model
over both the study and post-study periods.
&
\ref{sec: msm_hazard_procedures}
\\

$\beta^*$
&
Statistical parameter indexing the hazard model in the pseudopopulation
$P^*$.
&
\ref{sec: msm_hazard_procedures}
\\

$\ass{{\mathcal{M}(\mathbb{B})}_{\mid V}^{\mathrm{cs}}}$
&
Correct-specification condition for the counterfactual hazards
conditional on $V$.
&
\ref{sec: msm_hazard_procedures}
\\

$h_k^{*,a_k}(v)$
&
Hazard at follow-up $k$ in the pseudopopulation $P^*$, conditional on
$A_k=a_k$, remaining alive and uncensored, and $V=v$.
&
\ref{sec: msm_hazard_procedures}
\\

$\ass{{\mathcal{M}(\mathbb{B})}_{\mid V}^{*\textnormal{-cs}}}$
&
Correct-specification condition for the pseudopopulation hazards
conditional on $V$.
&
\ref{sec: msm_hazard_procedures}
\\

$S_k^{*,\mathrm{msm},\overline{a}_k}$
&
Survival parameter targeted by an MSM-hazard procedure:
\[
S_k^{*,\mathrm{msm},\overline{a}_k}
=
\mathbb{E}
\left(
\prod_{m=0}^k
(1-p(V,a_m,m;\beta^*))
\right).
\]
&
\ref{sec: msm_hazard_procedures}
\\

$h_{\prstudyp,k}^{a_k}(v)$
&
Counterfactual hazard at follow-up $k$, conditional on $V=v$, among
individuals whose entry time permits observation at follow-up $k$ during
the study period.
&
\ref{sec: extrapolation_assumptions_in_msm_hazard}
\\

$R_k^{\rightsquigarrow,a}$
&
Procedure-implied cumulative risk by follow-up $k$ under treatment
$a$, defined as
$R_k^{\rightsquigarrow,a}
=1-S_k^{*,\mathrm{msm},\overline{a}_k}$.
&
\ref{sec: c_values_and_extrapolation_curves}
\\

$RR_k^{\rightsquigarrow}$
&
Procedure-implied cumulative risk ratio at follow-up $k$:
\[
RR_k^{\rightsquigarrow}
=
\frac{R_k^{\rightsquigarrow,1}}
     {R_k^{\rightsquigarrow,0}}.
\]
&
\ref{sec: c_values_and_extrapolation_curves}
\\

$R_{\prstudyp,k}^{a}$
&
Proportion of individuals who first experience failure during the study
period under treatment strategy $a$.
&
\ref{sec: c_values_and_extrapolation_curves}
\\

$R_{\prgapp,k}^{\rightsquigarrow,a}$
&
Procedure-implied proportion of individuals who first experience
failure during the post-study period under treatment strategy $a$.
&
\ref{sec: c_values_and_extrapolation_curves}
\\

$\mathrm{c}_k^a$
&
C-value at $k$ under the strategy $a$.
&
\ref{sec: c_values_and_extrapolation_curves}
\\

$R_k^{\mathrm{c},a}(\mathrm{c}_k^a)$
&
Shifted cumulative risk under treatment strategy $a$:
\[
R_k^{\mathrm{c},a}(\mathrm{c}_k^a)
=
R_{\prstudyp,k}^{a}
+
\mathrm{c}_k^a \cdot
R_{\prgapp,k}^{\rightsquigarrow,a}.
\]
&
\ref{sec: c_values_and_extrapolation_curves}
\\

$(\mathrm{c}_k^0,\mathrm{c}_k^1)$
&
Pair of c-values minimizing the penalized departure from symmetry while
equalizing the shifted risks under treatment and control.
&
\ref{sec: c_values_and_extrapolation_curves}
\\

$d_k$
&
Quantity entering the closed-form expression for symmetric c-values:
\[
d_k
=
\frac{
R_{\prstudyp,k}^{0}-R_{\prstudyp,k}^{1}
}{
R_{\prgapp,k}^{\rightsquigarrow,1}
}.
\]
&
\ref{sec: c_values_and_extrapolation_curves}
\\

$m_k$
&
Ratio of post-study risks entering the closed-form expression for
symmetric c-values:
\[
m_k
=
\frac{
R_{\prgapp,k}^{\rightsquigarrow,0}
}{
R_{\prgapp,k}^{\rightsquigarrow,1}
}.
\]
&
\ref{sec: c_values_and_extrapolation_curves}
\\

$\delta_k^a$
&
Additive shift in the post-study hazards under treatment strategy $a$
that induces the same shifted cumulative risk as the corresponding
c-value.
&
\ref{sec: c_values_and_extrapolation_curves}
\\

$p_{\delta_k^a}(v,a_m,m;\beta^*)$
&
Hazard obtained by adding $\delta_k^a$ to the model-implied hazard during
the post-study period while leaving study-period hazards unchanged.
&
\ref{sec: c_values_and_extrapolation_curves}
\\

$K_{\mathrm{ex}}$
&
Minimum follow-up time at which c-values exist.
&
\ref{sec: c_values_and_extrapolation_curves}
\\

$\widehat{S}_k^{*,\mathrm{msm}}$
&
Estimated survival curve obtained from the MSM-hazard procedure that
includes $E$ among the baseline covariates.
&
\ref{sec: case_study_hernan_00_sim}
\\

$\widehat{S}_k^{*,\mathrm{msm-we}}$
&
Estimated survival curve obtained from an MSM-hazard procedure that
excludes $E$ from the baseline covariates.
&
\ref{sec: case_study_hernan_00_sim}
\\

$\widehat{S}_k^{*,\mathrm{km}}$
&
Estimated survival curve obtained from the weighted Kaplan--Meier
procedure.
&
\ref{sec: case_study_hernan_00_sim}
\\

$\mathcal{M}_{\circ}$
&
Working hazard model used in Scenario A of the first case study.
&
\ref{sec: case_study_hernan_00_sim}
\\

$\mathcal{M}_{\circ\circ}$
&
Hazard model used in Scenarios B and C, obtained by augmenting
$\mathcal{M}_{\circ}$ with a post-study interaction.
&
\ref{sec: case_study_hernan_00_sim}
\\

$\mathcal{M}_{\mathrm{sat}}$
&
Saturated working model excluding study-entry time $E$.
&
\ref{sec: case_study_hernan_00_sim}
\\

$\prgapp_{\mathrm{int}}(e,k)$
&
Post-study interaction term allowing the treatment effect to change
when $e+k>\mathcal{T}$.
&
\ref{sec: case_study_hernan_00_sim}
\\

$S_k(e)$
&
Survival at follow-up $k$ under the specified treatment strategy,
conditional on study entry $E=e$.
&
\ref{sec: case_study_hernan_00_sim}
\\

$\widehat{RR}_k^{\rightsquigarrow}$
&
Estimated procedure-implied cumulative risk ratio at follow-up $k$.
&
\ref{sec: c_values_case_study_1}
\\

$K_{\mathrm{exe}}$
&
Maximum follow-up time retained for reporting estimates on the basis of
a prespecified sensitivity threshold or an elbow criterion.
&
\ref{sec: c_values_case_study_2}
\\

\end{longtable}\clearpage
\section{Censoring assumptions and IPTW under calendar-time changes}
\label{sec_app: double_bind_censoring_assumptions_literature_iptw}

This Appendix concerns censoring assumptions under calendar-time changes and provides a foundation for establishing the results for WSC and MSM-hazard procedures. The aim of this Appendix is fourfold.

First, we introduce our notion of calendar-time changes and establish preliminary results within a discrete-time causal inference framework. We also formulate an analogous notion for continuous-time survival analysis frameworks.

Second, we define censoring assumptions within causal inference and survival analysis frameworks, show that these assumptions are ill-posed, and formulate alternative assumptions that are not automatically violated.

Third, we review the literature on calendar-time changes and censoring assumptions. In particular, we elaborate on how our notion of calendar-time changes differs from existing ones by capturing features that are relevant for decision-making. We then elaborate on the advantages of an estimand-based methodology in these settings.

Fourth, we restate and prove both existing and new results in IPTW theory\footnote{In the literature, the acronym IPCW is used when the weights are constructed from the conditional distribution of the right-censoring indicator \citep{koul_regression_1981, robins_correcting_2000}, whereas the term IPTW applies when the weights are based on the conditional distribution of the actual treatment assigned \citep{robins_marginal_2000,hernan_marginal_2000}. If the censoring indicator is conceptualized as a treatment, the term IPTW can also be used when weights are constructed from the conditional distributions of both the actual treatment and the right-censoring indicator. We adopt this terminology hereafter.} that will be instrumental in establishing the results in Section \ref{sec: consequences_and_solutions_in_wsc_and_msm_hazard} of the main text.

\subsection{Calendar-time changes as dependencies between $E$ and the target outcome of interest}
\label{sec_app: notion_double_bind}
\subsubsection{Discrete-time counterfactual framework}
\label{sec_app: notion_ctc_discrete_time}
We introduce two distinct notions of dependence between the counterfactual outcome and the calendar time of study-entry. 

Let $L_0^- \coloneqq L_0 \setminus \{E\}$ denote the baseline covariates (not including $E$), and let $B \subseteq L_0^-$ be a given subset. We say that calendar-time changes are present during the study period and post-study period within the subpopulation defined by $B$ at follow-up $k$ and for the treatment strategy $\overline{a}_k$ if, respectively
\begin{itemize}[leftmargin=2.5cm, labelsep=0.5cm]
\item[$\prstudyp_k^{\overline{a}_k} - \textnormal{ctc}(B)$]: 
\begin{align}\label{dep: prstudyp}
    Y_k^{\overline{a}_k, c_{k-1}=0} \not\independent E \mid E + k \leq \mathcal{T}, Y_{k-1}^{\overline{a}_{k-1}, c_{k-2}=0}=0, B \textnormal{ and }
\end{align}
\item[$\prgapp_k^{\overline{a}_k} - \textnormal{ctc}(B)$]: 
\begin{align}\label{dep: prgapp}
    Y_k^{\overline{a}_k, c_{k-1}=0} \not\independent E \mid E + k > \mathcal{T}, Y_{k-1}^{\overline{a}_{k-1}, c_{k-2}=0}=0, B
\end{align}
\end{itemize}
hold. While not strictly required, to simplify our exposition, we will assume throughout that calendar-time changes exhibit the following regularity property:
\begin{itemize}[leftmargin=3.5cm, labelsep=0.5cm]
    \item[ctc-homogeneity]: if there exists at least one subset $B \subseteq L_0^-$ such that $\prstudyp_k^{\overline{a}_k} - \textnormal{ctc}(B)$ holds (respectively, $\prgapp_k^{\overline{a}_k} - \textnormal{ctc}(B)$), then the dependence holds for every possible subset $B' \subseteq L_0^-$. 
\end{itemize}

Assumption ctc-homogeneity implies that the dependence between $E$ and the counterfactual hazard cannot be nullified by any combination of measured baseline characteristics. Under this assumption, we drop the $(B)$ notation and simply refer to these conditions as $\prstudyp_k^{\overline{a}_k} - \textnormal{ctc}$ and $\prgapp_k^{\overline{a}_k} - \textnormal{ctc}$.

It is precisely the (deemed) presence of $\prgapp_k^{\overline{a}_k} - \textnormal{ctc}$ that prompts the investigator to articulate censoring assumptions conditional on $E$, which in turn leads to the positivity violations discussed in the main text. Thus, we will say that \textit{calendar-time changes} are present at follow-up $k$ for strategy $\overline{a}_k$ if dependencies between the outcome of interest and the time of study-entry arise during the post-study period:
\begin{align}\label{expression: ctc}
\prgapp_k^{\overline{a}_k} - \textnormal{ctc}. 
\end{align}

We will simply say that calendar-time changes (\CTC{}) are present if $\prgapp_k^{\overline{a}_k}-\textnormal{ctc}$ holds for at least one $k \in \{0,\dots,K\}$ under the strategy under consideration (or, when considering a contrast, to at least one treatment strategy). To lighten notation, we henceforth omit subscripts and superscripts when referring to these conditions: $\prstudyp-\textnormal{ctc}$ denotes that $\prstudyp_k^{\overline{a}_k}-\textnormal{ctc}$ holds for at least one $k$, while $\neg(\prstudyp-\textnormal{ctc})$ denotes that it holds for no $k \in \{0,\dots,K\}$.

Assumption $\prgapp - \textnormal{ctc}$ can be decomposed into two mutually exclusive cases
\begin{align}\label{expression: ctc_decomposition}
\prgapp - \textnormal{ctc} \Leftrightarrow \textnormal{ either } \prstudyp\prgapp - \textnormal{ctc} \textnormal{ or } (\neg\prstudyp)\prgapp - \textnormal{ctc}, 
\end{align}
where:
\begin{itemize}[leftmargin=3cm, labelsep=0.5cm]
\item[$\prstudyp\prgapp - \textnormal{ctc}$]: $\prstudyp - \textnormal{ctc}$ and $\prgapp - \textnormal{ctc}$,\label{ass: ct_full}
\item[$(\neg\prstudyp)\prgapp - \textnormal{ctc}$]: $\neg(\prstudyp - \textnormal{ctc})$ and $\prgapp - \textnormal{ctc}$. \label{ass: ct_post}
\end{itemize}

Condition $\prstudyp\prgapp - \textnormal{ctc}$ asserts that the dependence between the outcome of interest and the time of study-entry arises \textit{both} during the study period and in the post-study period. Conversely, condition $(\neg\prstudyp)\prgapp - \textnormal{ctc}$ asserts that this dependence arises \textit{only} in the post-study period. In Appendix \ref{sec_app: literature_ctc}, we illustrate how this decomposition allows us to distinguish our notion of calendar-time changes from existing notions.

Finally, analogous considerations extend to other data structures, provided that the definitions of the primitives $\prstudyp - \textnormal{ctc}$ and $\prgapp - \textnormal{ctc}$ are adapted accordingly. For example, in the special case of a point treatment $A$, it suffices to replace $Y_k^{\overline{a}_k, c_{k-1}}$ with $Y_k^{a, c_{k-1}}$ in \eqref{dep: prstudyp} and \eqref{dep: prgapp}. We defer the formulation for the continuous-time case to Appendix \ref{sec_app: continuous_time_ct}.

\subsubsection{Calendar-time changes and functional form of counterfactual hazards}

\CTC{} carries the following implications for the functional form of counterfactual hazards.

\begin{proposition}[Hazard equality under calendar-time changes]\label{proposition: equivalence_hazard_under_ct}
Suppose \CTC{} holds.
Let $E \notin V$. Fix $k \in \{1, \dots, K\}$, $\overline{a}_k \in \{\overline{0},\overline{1}\}$. For every $m \in \{1, \dots, k\}$ and every $v$ at which the quantities below are well-defined:
\begin{enumerate}
    \item under $\ctFull$:
    {\scriptsize
            \begin{align}
        &h_{\prstudyp, m}^{\overline{a}_m}(v) = h_{m}^{\overline{a}_m}(v) \notag \\
            \Leftrightarrow& 
         \frac{\expectation_{E \mid V=v}\left(\mathbb{1}_{\{E+m \leq \mathcal{T}\}} \cdot \prod_{l=0}^m (1 - h_l^{\overline{a}_l}(v,E))\right)}{\expectation_{E \mid V=v}\left(\mathbb{1}_{\{E+m \leq \mathcal{T}\}} \cdot \prod_{l=0}^{m-1} (1 - h_l^{\overline{a}_l}(v,E))\right)} = \frac{\expectation_{E \mid V=v}\left(\mathbb{1}_{\{E+m > \mathcal{T}\}} \cdot \prod_{l=0}^m (1 - h_l^{\overline{a}_l}(v,E))\right)}{\expectation_{E \mid V=v}\left(\mathbb{1}_{\{E+m > \mathcal{T}\}} \cdot \prod_{l=0}^{m-1} (1 - h_l^{\overline{a}_l}(v,E))\right)}.\label{eq: hazard_variation_dependence_ct}
        \end{align}}
    \item under $\ctPost$: the above equivalence holds with the hazards $h_l^{\overline{a}_l}(v,E)$ in the numerator and denominator of the left-hand fraction replaced by $h_{\prstudyp,l}^{\overline{a}_l}(v)$.
\end{enumerate}

\begin{proof}
    Suppose $\ctFull$ holds. Fix $k \in \{1, \dots, K\}$, $\overline{a}_k \in \{\overline{0},\overline{1}\}$. Take $m \in \{1, \dots, k\}$ and let 
    \begin{align*}
        S_{m}^{\overline{a}_m}(v,e) \coloneqq \prod_{j=0}^m \left(1 - h_j^{\overline{a}_j}(v,e)\right);
    \end{align*}
    to lighten notation, define
    {\footnotesize
    \begin{align*}
        X_{m}(v) \coloneqq& \expectation_{E \mid V=v}\left(\mathbb{1}_{\{E+m \leq \mathcal{T}\}} \cdot S_{m}^{\overline{a}_m}(v,E)\right), X_{m}^{-}(v) \coloneqq& \expectation_{E \mid V=v} \left(\mathbb{1}_{\{E+m \leq \mathcal{T}\}} \cdot S_{m-1}^{\overline{a}_{m-1}}(v,E)\right),\\
        Z_m(v) \coloneqq & \expectation_{E \mid V=v}\left(\mathbb{1}_{\{E+m > \mathcal{T}\}} \cdot S_{m}^{\overline{a}_m}(v,E) \right), Z_m^-(v) \coloneqq& \expectation_{E \mid V=v}\left(\mathbb{1}_{\{E+m > \mathcal{T}\}} \cdot S_{m-1}^{\overline{a}_{m-1}}(v,E)\right).
    \end{align*}
    }
It is not difficult to show that $h_m^{\overline{a}_m}(v) = 1 - \frac{X_{m}(v) + Z_m(v)}{X_{m}^-(v)+Z_{m}^-(v)}$ and $h_{\prstudyp,m}^{\overline{a}_m}(v) = 1 - \frac{X_{m}(v)}{X_{m}^-(v)}$. Thus,
{\footnotesize
\begin{align*}
  h_{\prstudyp,m}^{\overline{a}_m}(v) = h_{m}^{\overline{a}_m}(v) & \Leftrightarrow& \frac{X_{m}(v) + Z_m(v)}{X_{m}^-(v) + Z_{m}^-(v)} = \frac{X_m(v)}{X_{m}^-(v)} \\
        &\Leftrightarrow& X_{m}^-(v) X_{m}(v) + X_{m}^-(v) Z_{m}(v) = X_{m}(v) X_{m}^-(v) + X_{m}(v) Z_{m}^-(v) \\
        &\Leftrightarrow& \frac{X_{m}(v)}{X_{m}^-(v)} = \frac{Z_m(v)}{Z_{m}^-(v)}.
    \end{align*}
}
 If $\ctPost$ holds, then  for any fixed $v$ and $\overline{a}_k$, $h_l^{\overline{a}_l}(v,e)$ is constant for every $l + e \leq \mathcal{T}$; thus  $h_l^{\overline{a}_l}(v,e)= h_{\prstudyp,l}^{\overline{a}_l}(v)$.
\end{proof}
\end{proposition}
\begin{corollary}\label{corollary: equivalence_hazard_under_ct_special}
Suppose \textnormal{\CTC{}} holds.
Let $E \notin V$. Fix $k \in \{1, \dots, K\}$, $\overline{a}_k \in \{\overline{0},\overline{1}\}$. For every $m \in \{1, \dots, k\}$ for which the parameters below are well-defined:
\begin{enumerate}
    \item under $\ctFull$:
    {\footnotesize
        \begin{align}\label{eq: hazard_variation_dependence_ct_special}
&h_{\prstudyp, m}^{\overline{a}_m}= h_{m}^{\overline{a}_m} \notag \\
    \Leftrightarrow& 
 \frac{\expectation\left(\mathbb{1}_{\{E+m \leq \mathcal{T}\} }\cdot \prod_{l=0}^m (1 - h_l^{\overline{a}_l}(E))\right)}{\expectation\left(\mathbb{1}_{\{E+m \leq \mathcal{T}\} } \cdot \prod_{l=0}^{m-1} (1 - h_l^{\overline{a}_l}(E))\right)} = \frac{\expectation\left(\mathbb{1}_{\{E+m > \mathcal{T}\} } \cdot \prod_{l=0}^m (1 - h_l^{\overline{a}_l}(E))\right)}{\expectation\left(\mathbb{1}_{\{E+m > \mathcal{T}\} } \cdot \prod_{l=0}^{m-1} (1 - h_l^{\overline{a}_l}(E))\right)}.
 \end{align}
 }
    \item under $\ctPost$: the above equivalence holds with the hazards $h_l^{\overline{a}_l}(E)$ in the numerator and denominator of the left-hand fraction replaced by $h_{\prstudyp,l}^{\overline{a}_l}$.
\end{enumerate}
 \begin{proof}
     Follows from Proposition \ref{proposition: equivalence_hazard_under_ct} by choosing $V =\emptyset$. 
 \end{proof}
\end{corollary}

The following lemma characterizes the primitive conditions \(\prstudyp_k^{\bar a_k}-\mathrm{ctc}(B)\) and \(\prgapp_k^{\bar a_k}-\mathrm{ctc}(B)\) in terms of non-constancy of the counterfactual hazard as a function of $E$.

\begin{lemma}[Equivalence of \CTC{} and non-constant hazards]\label{lemma: ctc_hazard_equivalence}
Let $B \subseteq L_0^-$. Then,
\begin{align*}
    &\prstudyp_k^{\overline{a}_k} - \textnormal{ctc}(B) \\
    \Leftrightarrow& P(Y_k^{\overline{a}_k, c_{k-1}=0} = 1 \mid Y_{k-1}^{\overline{a}_{k-1}, c_{k-2}=0} = 0, B, E=e)
\end{align*}
is not constant as a function of $e$ over the conditional support of $E$ among individuals satisfying $e+k \leq \mathcal{T}$. Analogously
\begin{align*}
    &\prgapp_k^{\overline{a}_k} - \textnormal{ctc}(B) \\
    \Leftrightarrow& P(Y_k^{\overline{a}_k, c_{k-1}=0} = 1 \mid Y_{k-1}^{\overline{a}_{k-1}, c_{k-2}=0} = 0, B, E=e)
\end{align*}
is not constant as a function of $e$ over the conditional support of $E$ among individuals satisfying $e+k > \mathcal{T}$.
\end{lemma}
\begin{proof}
    Follows from the tower property and by noting that:
    \begin{align*}
        &P(Y_k^{\overline{a}_k, c_{k-1}=0} = 1 \mid Y_{k-1}^{\overline{a}_{k-1}, c_{k-2}=0} = 0, B, \mathbb{1}_{\{E+k \leq \mathcal{T}\}} =1, E=e) \\
        =& P(Y_k^{\overline{a}_k, c_{k-1}=0} = 1 \mid Y_{k-1}^{\overline{a}_{k-1}, c_{k-2}=0} = 0, B, E=e) \textnormal{ for $e + k \leq \mathcal{T}$}, \\
        &P(Y_k^{\overline{a}_k, c_{k-1}=0} = 1 \mid Y_{k-1}^{\overline{a}_{k-1}, c_{k-2}=0} = 0, B, \mathbb{1}_{\{E+k > \mathcal{T}\}} =1, E=e) \\
        =& P(Y_k^{\overline{a}_k, c_{k-1}=0} = 1 \mid Y_{k-1}^{\overline{a}_{k-1}, c_{k-2}=0} = 0, B, E=e) \textnormal{ for $e + k > \mathcal{T}$}.
    \end{align*}
\end{proof}

\subsubsection{Continuous-time}
\label{sec_app: continuous_time_ct}
\textbf{Notation and data structure}. Let $[0,t]$ denote the time window within which we want to evaluate survival ($t \leq \mathrm{T})$, with $\mathrm{T}$ denoting the calendar time of the administrative end of the study (which here also denotes the maximum follow-up time) and, with slight abuse of notation, 
let $E$ denote the discrete time of study-entry.\footnote{To remain notationally consistent we will write $E \in \{0, \dots, \mathcal{T}\}$ by which, in continuous time settings, we mean $ E \in \{0, \Delta, 2\Delta,\dots, \mathcal{T} \cdot \Delta \}$ with $(\mathcal{T}+1)\cdot \Delta = \mathrm{T}$; contrary to the discrete-time settings, here $e$ is equal to the actual calendar time but it is denoted as the index of the $e$-th study-entry time. That is, $e$ denotes both the index and the calendar time. Also, we will assume that $E$ takes finitely many values, albeit our results can be generalized to continuous $E$.} 

Let $T$ denote the continuous event time under the strategy of interest, whose indexing is suppressed throughout this continuous-time formulation, and let $C_{\textnormal{loss}}$ denote the continuous loss-to-follow-up time. We assume $T$ and $C_{\textnormal{loss}}$ admit densities with respect to the Lebesgue measure. Let $C = \min(C_{\textnormal{loss}}, t_E)$ denote the total right-censoring time, where $t_E \coloneqq \mathrm{T} - E$. Conditional on $E=e$, $C$ has an absolutely continuous component on $[0,t_e)$ and, whenever $P(C_{\textnormal{loss}} \geq t_e \mid E=e)>0$, a discrete point mass at $t_e$. Finally, let $\Tilde{T} = \min(T, C)$ denote the observed event time and consider throughout that observations are iid draws from a single law $P^{\textnormal{c}}$. 

Letting $\alpha_e(s)$ denote the hazard of the (uncensored) failure times for $\{E=e\}$ at time $s$, \textbf{a continuous-time formulation of \CTC{}}, as well as its mutually exclusive decomposition ($\prstudyp\prgapp - \textnormal{ctc}$ and $(\neg\prstudyp)\prgapp - \textnormal{ctc}$), are defined analogously to the discrete-time case upon adopting the following primitives:
\begin{itemize}[leftmargin=2.5cm, labelsep=0.5cm]
\item[${\prstudyp}_s - \textnormal{ctc}$]: 
\begin{align}\label{eq: prstudyp_ct_cont}
    e \mapsto \alpha_e(s) \textnormal{ is not constant for }e : e+s \leq \mathrm{T};
\end{align}
\item[${\prgapp}_s - \textnormal{ctc}$]: 
\begin{align}\label{eq: prgapp_ct_cont}
    e \mapsto \alpha_e(s) \textnormal{ is not constant for }e : e+s > \mathrm{T}.
\end{align}
\end{itemize}
Here, $E$ constitutes the only baseline covariate and therefore there is no need to implicitly assume ctc-homogeneity. 

\subsection{Censoring assumptions in staggered-entry settings: preliminaries}
\label{sec_app: censoring_assumptions_meaning_violation}
In both the classical survival and causal inference literature, censoring assumptions can be interpreted as non-parametric conditions allowing the identification of a given parameter (or estimand) of interest. Here, we argue that certain non-parametric assumptions are ill-posed in their canonical form, and formulate their restricted form.

\textbf{Assumptions ill-posed, assumptions violated.} 

We consider any proposition that relies on a conditional probability to be \textbf{ill-posed} at $x$ if $x \notin \textnormal{supp}(P(X))$.

In that case, there exists an open ball $B(x, \delta)$ with probability zero, i.e., $P(X \in B(x, \delta)) = 0$, and therefore the observed law does not determine $P(Y\mid X=x)$ (its value can be specified arbitrarily since different specifications at $x$ are compatible with the same observed law.). This differs from conditioning at points $x\in\operatorname{supp}(P(X))$ that have probability zero, where the conditioning value is nevertheless approached by observations with positive probability in every neighborhood.

See the work of \citep{gill_causal_2001} for a formal mathematical treatment.

We say that a \textbf{condition is violated} in two cases: either it evaluates to false, or it enters the positivity assumption and is ill-posed (in which case we additionally say that positivity is violated).
    
For instance, in our illustrative example of Section \ref{sec: the_problem_conditioning_on_e} of the main text, where $K=\mathcal{T}=2$, consider the condition
\begin{align}\label{cond: pos_ill_vio}
    P(C_{1} =0 \mid Y_{1} = 0, \overline{A}_{1}=\overline{a}_1, C_{0}=0, \overline{L}_1=\overline{l}_1, E = 2) > 0
\end{align}
involved in the articulation of the positivity assumption. Because an individual entering at $E=2$ cannot possibly reach follow-up $k=1$ without exceeding the maximum study time $\mathcal{T}=2$, the condition is ill-posed because $(0,\bar a_1,0,\bar l_1,2)
\notin \operatorname{supp}(Y_1,\bar A_1,C_0,\bar L_1,E)$. We therefore say that \eqref{cond: pos_ill_vio} is violated (and correspondingly that $\textnormal{ex}(e=2, k=2)$ is ill-posed).

\subsection{Censoring assumptions: discrete-time counterfactual framework}
\label{sec_app: censoring_assumptions_discrete_time}
In Section \ref{sec: censoring_assumptions_are_illposed} of the main text, we showed that the censoring assumptions are ill-posed. Here, we establish the same violation under an alternative definition of positivity. We then formulate alternative censoring assumptions that are compatible with our working assumptions, focusing in particular on the point treatment setting presented in Section \ref{sec: censoring_assumptions_are_illposed} of the main text and used in the case study of Section \ref{sec: case_study_mrna_vaccine}. 

\subsubsection{Alternative formulation of positivity}
Suppose $E \in L_0$. Following \citep{gill_causal_2001}, consider the following definition of positivity for our illustrative example of Section \ref{sec: illustrative_example_censoring_assumptions_ill_posed}:
{\footnotesize
\begin{multline}
\overline{a}_k = \overline{a}_{k}^{\dagger} \textnormal{ and } \left(e,\overline{l}_{k}, {y}_{k-1}=0, c_{k-2}=0, \overline{a}_{k-1}\right) \in \textnormal{supp}(E, \overline{L}_{k}, {Y}_{k-1}, {C}_{k-2}, \overline{A}_{k-1})  \\ 
\Longrightarrow \left(e, \overline{a}_{k}, 
c_{k-1}=0,
\overline{l}_{k}, {y}_{k-1}=0\right) \in \textnormal{supp}(E, \overline{A}_{k},
{C}_{k-1}, \overline{L}_{k}, {Y}_{k-1}), \\
\textnormal{for every $k \in \{0, \dots, K\}$,} 
\end{multline}
}
where we write $E=e$, even though $E \in L_0$, in the conditioning to emphasize the dependence of the condition on $E$.
We now prove that this condition is not met, that is, the alternative formulation of positivity is violated.

\begin{proof}
Suppose, for simplicity, that $A_k$ and $L_k$ are discrete. Fix $k^* \in \{1, \dots, K\}$ and $\overline{a}_{k^*}^{\dagger}$. Take  $(e, \overline{l}_{k^*}, \overline{a}_{k^*}^{\dagger})$ such that:
\begin{align}
    & (e, \overline{l}_{k^*}, \overline{a}_{k^*-1}, 0,0) \in \textnormal{supp}(E, \overline{L}_{k^*}, \overline{A}_{k^*-1}, {C}_{k^*-2}, {Y}_{k^*-1}), \label{c_proof: pos_vio_1}\\
    & e + k^* - 1 = \mathcal{T}. \label{c_proof: pos_vio_2}
\end{align}
Then,
{\footnotesize
\begin{align*}
\eqref{c_proof: pos_vio_1}& \Rightarrow P \left( E=e, \overline{L}_{k^*}=\overline{l}_{k^*}, \overline{A}_{k^*-1}=\overline{a}_{k^*-1}^{\dagger}, {C}_{k^*-2}=0, {Y}_{k^*-1}=0 \right) > 0; \\
    \eqref{c_proof: pos_vio_1} \textnormal{ and }\eqref{c_proof: pos_vio_2}  &\Rightarrow (e, \overline{l}_{k^*}, \overline{a}_{k^*}^{\dagger}, 0,0) \notin \textnormal{supp}(E, \overline{L}_{k^*}, \overline{A}_{k^*}, {C}_{k^*-1}, {Y}_{k^*-1}) \\
    &\Rightarrow P \left( E=e, \overline{L}_{k^*}=\overline{l}_{k^*}, \overline{A}_{k^*}=\overline{a}_{k^*}^{\dagger}, {C}_{k^*-1}=0, {Y}_{k^*-1}=0 \right) = 0,
\end{align*}}
which concludes the proof.
\end{proof}

\subsubsection{Restricted formulation} 
What prompts the researcher to formulate censoring assumptions conditional on $E$ is the belief that
\begin{align}\label{ass: e-of}
    Y_k^{\overline{a}_k, c_{k-1}=0} \not\independent \mathbb{1}_{\{A_{k}=a_{k}, C_{k-1}=0\}} \mid \overline{L}_{k} \setminus \{E\}, Y_{k-1}=0, \overline{A}_{k-1} = \overline{a}_{k-1}, C_{k-2}=0 \tag{E-conf}
\end{align}
does not hold for some $k \in \{0, \dots, K\}$.  While it can be shown that there exist causal models for which $\RAID$ is compatible with \eqref{ass: e-of}, \CTC{}, and \eqref{eq: secc}, here we illustrate how $\RAID$ can be formulated under different configurations (see Table \ref{table: settings_and_data_structures}). To this end, it suffices to adapt the definitions of ex($e,k$), pos($e,k$), and con($e,k$) to the desired configuration. 

\begin{table}[h!] 
\centering
\footnotesize
\begin{tabular}{>{\bfseries}m{8cm} c c}
\toprule
Setting and data structure & tx & tv covariates \\
\midrule
I: \label{table: conf_1} point tx; no intervention; base conf & $C_{k-1}$ & $\emptyset$ (only $L_0$) \\
II: \label{table: conf_2} point tx; no intervention; tv conf & $C_{k-1}$ & $L_k$ \\
III: \label{table: conf_3} point tx; intervention; base conf & $(C_{k-1}, A)$ & $\emptyset$ (only $L_0$) \\
IV: \label{table: conf_4} point tx; intervention; tv conf & $(C_{k-1}, A)$ & $L_k$ \\
\addlinespace
V: \label{table: conf_5} repeated tx(s); no intervention; base conf & $C_{k-1}$ & $\emptyset$ (only $L_0$) \\ 
VI: \label{table: conf_6} repeated tx(s); no intervention; tv conf & $C_{k-1}$ & $L_k$ \\ 
VII: \label{table: conf_7} repeated tx(s); intervention(s); base conf & $C_{k-1}, A_{k}$ & $\emptyset$ (only $L_0$) \\
VIII: \label{table: conf_8} repeated tx(s); intervention(s); tv conf & $C_{k-1}, A_{k}$ & $L_k$ \\
\bottomrule
\end{tabular}
\caption{Configurations. We will always assume an intervention on the right-censoring indicator. By point tx, we refer to settings with a baseline point treatment; by repeated tx, to settings with repeated treatments over time. By base conf, we refer to settings wherein only adjustment for baseline covariates is needed; by tv conf we refer to settings where adjustment for time-varying confounders is needed.}
\label{table: settings_and_data_structures}
\end{table}

For example, for configuration \hyperref[table: conf_3]{\textbf{III}} we have:
\begin{itemize}
    \item[ex$(e,k)$]:  
    {\tiny
    \begin{align*}
        \begin{cases}
       (Y_K^{a, c_{K-1}=0}, \dots, Y_k^{a, c_{k-1}=0}) \independent C_{k-1} \mid  {Y}_{k-1} = 0,{C}_{k-2} = 0, A=a, L_0, E=e \textnormal{ for every $a\in\{0,1\}$} & k > 0 \\
       (Y_K^{a, c_{K-1}=0}, \dots, Y_k^{a, c_{k-1}=0}) \independent \{ C_{-1}, A \} \mid L_0, E=e & k=0
   \end{cases}
    \end{align*}
    }
\item[con$(e,k)$]:   for every $a \in \{0,1\}$, $E = e, A = a, {C}_{k-1}={0}  \Longrightarrow Y_k^{a, c_{k-1}=0} = Y_k \quad \textnormal{a.s.}$, and
\item[pos$(e,k)$]: {\footnotesize
    \begin{align*}
        \textnormal{for every $a \in \{0,1\}$} \\
        \begin{cases}
       P\left( C_{k-1}=0 \mid  {Y}_{k-1} = 0,{C}_{k-2} = 0, A=a, L_0, E=e\right) > 0  & k > 0 \\
       P\left( C_{-1}=0, A = a \mid L_0, E=e\right) > 0 & k=0
   \end{cases},
   \end{align*}}
\end{itemize} 
where, to emphasize the role of $E$ in the following conditions, we wrote $E=e$ in the conditioning sets even if $E \subseteq L_0$. The formulations presented above can be relaxed. For example, the variable $E$, and thus $E=e$, in the conditioning sets of  ex($e,k$), pos($e,k$), and con($e,k$), can alternatively be replaced by $\mathbb{1}_{\{E \leq \mathcal{T}-k\}} = \mathbb{1}_{\{e \leq \mathcal{T}-k\}}$; in such a case, every $e \in \{0, \dots, \mathcal{T}-k\}$ carries the same information, i.e., whether the event at follow-up $k$ for individuals $E=e$ occurred during the study period ($\mathbb{1}_{\{E \leq \mathcal{T}-k\}}=1$) or during the post-study period ($\mathbb{1}_{\{E \leq \mathcal{T}-k\}}=0$).

\subsection{Censoring assumptions: continuous-time survival analysis framework}
\label{sec_app: continuous_time_secc}
\subsubsection{A continuous-time formulation of \eqref{eq: secc}}
A continuous-time formulation of \eqref{eq: secc} can be articulated as follows:
\begin{align}
    \textnormal{supp} \left(C\mid E=e\right) \subseteq [0, t_e], \quad \textnormal{ for all $e$.}
    \label{cond: secc_continuous_time}
\end{align}
This condition asserts that administrative censoring strictly occurs within the interval $[0, t_e]$. Because the observed time is defined as $\tilde{T} = \min(T, C)$, it follows that the observed event (failure or censoring) must occur within this interval, meaning $\textnormal{supp}(\tilde{T} \mid E=e) \subseteq [0, t_e]$ for all $e$.

Yet, in virtually all settings $t$ is arguably implicitly chosen so a non-zero proportion of individuals experience the event of interest after $t$, that is,
\begin{align}\label{ass: cp}
    [0, t] \subset \textnormal{supp}\left({T} \mid E=e\right), \quad \textnormal{ for all $e$.} \tag{\textbf{C.P}}
\end{align}
Together with \eqref{cond: secc_continuous_time} this condition implies that for all $e$ such that $e+t > \mathrm{T}$, i.e., $t > t_e$, during $(t_e, t] = [0, t] \setminus [0, t_e]$ we cannot observe \textit{any} event of interest for alive individuals at time $t_e$: if an individual is alive at $t_e$ it will be censored right after; if she is not, she has either been censored or has experienced the event before. Thus for all $e$ such that $e+t > \mathrm{T}$, i.e., $t > t_e$:
\begin{align}\label{cond: secc_ramifications_continuous_time}
    \textnormal{supp}\left(\Tilde{T} \mid E=e\right) \subseteq [0, t_e] \subset \textnormal{supp}\left({T} \mid E=e\right).
\end{align}
By construction, the set of $e$ satisfying \eqref{cond: secc_ramifications_continuous_time} is not empty.

Consider now two censoring assumptions, both formulated conditional on $E$:

\begin{enumerate}
    \item Following \citep{andersen_statistical_1993}[pag. 51], 
    \begin{multline}\label{ass: andersen_censoring_ass}
        P^{\textnormal{c}}(\Tilde{T} \in [s, s+ds), \mathbb{1}_{\{T = \Tilde{T}\}} = 1 | \Tilde{T} \geq s, E=e) \\
        = \mathbb{1}_{\{\Tilde{T} \geq s\}} \alpha_e(s) \textnormal{d}s, \\
        s \in [0, t], \textnormal{ for every } e.
    \end{multline}
     In words, at every calendar time of study-entry $e$, for individuals $\{E=e\}$, the right-censoring process does not alter the intensities of the event failure; such a condition permits the identification of the true intensities from censored observations;
    \item Following \citet{robins_correcting_2000}[Eq. 2]:
    \begin{multline}\label{ass: robins_censoring_ass}
    P^{\textnormal{c}}\left(T \in [s, s+ds) \mid T \geq s, C \geq s, E = e\right)  \\
    = P^{\textnormal{c}}\left(T \in [s, s+ds) \mid T \geq s, E = e\right) \\
    s \in [0, t], \textnormal{ for every } e.
    \end{multline}
    In words, at every calendar time of study-entry $e$, for individuals $\{E=e\}$ experiencing the event of interest at time $s$ among individuals who have not experienced it yet is the same as that resulting from considering only uncensored individuals.
\end{enumerate}

For any fixed $e$, condition \ref{cond: secc_ramifications_continuous_time} implies that the conditions in \eqref{ass: andersen_censoring_ass} and \eqref{ass: robins_censoring_ass}, in accordance with our convention, are ill-posed at any $e$ and $s$ such that $s \in (t_e, t]$.

\subsubsection{Restricted formulation}
To establish the results in Appendix \ref{sec_app: continuous_time_wsc}, and in keeping with the discrete-time formulation, we will consider a restricted reformulation of assumption \ref{ass: andersen_censoring_ass} which requires the condition to hold only during the study period:
\begin{multline}\label{ass: cra1}
    P^{\textnormal{c}}(\Tilde{T} \in [s, s+ds), \mathbb{1}_{\{T = \Tilde{T}\}} = 1 | \Tilde{T} \geq s, E=e) \\
    =  \alpha_e(s) \textnormal{d}s, \\
    s \in [0, t_e), \textnormal{ for every } e, \tag{C.R.A1}
\end{multline}
where we wrote $s \in [0, t_e)$ in place of $s \in [0, t]$ for every $e$.

For the assumptions to be well-defined, we will assume:
\begin{align}\label{ass: cra2}
P^c(\tilde{T} \ge s \mid E=e) > 0, \quad \textnormal{ for every } e \textnormal{ and every } s \leq t_e. \tag{C.R.A2}
\end{align}

Since $C = \min(C_{\textnormal{loss}}, t_E)$, we have $P^c(\tilde{T} = t_e \mid E=e) > 0$ whenever individuals have a non-zero probability of surviving both the event process and the loss-to-follow-up process until the study end. Assumption \eqref{ass: cra2} guarantees this positivity, ensuring that the risk set does not deterministically vanish and rendering the asymptotic evaluation of the Nelson–Aalen and Kaplan–Meier estimators meaningful over the closed interval $[0, t_e]$.

To establish our results in Appendix \ref{sec_app: continuous_time_wsc}, we will consider:
\begin{itemize}[leftmargin=3cm, labelsep=0.5cm]
    \item[($\prstudyp-\textnormal{np-cont}$)] \hyperref[ass: cp]{(C.P)} and \hyperref[ass: cra1]{(C.R.A1)} and \hyperref[ass: cra2]{(C.R.A2)}. \label{ass: r.np.cont}
\end{itemize}

\subsection{Literature review}
\label{sec_app: literature_review}
\subsubsection{Existing literature on calendar-time changes}
\label{sec_app: literature_ctc}
Previous work has acknowledged the positivity problems arising from administrative censoring, see, e.g., \citet{robins_estimation_1992, robins_recovery_1992}. The causal inference literature has broadly discussed two main ways to resolve the positivity violation we described in Section \ref{sec: the_problem_conditioning_on_e} of the main text: adopt procedures that extrapolate beyond the support of the observed data, or modify the estimand by either restricting the target population or changing the intervention of interest \citep{moore_causal_2012, petersen_diagnosing_2012, jensen_identification_2024}. 

In epidemiology and causal inference, several terms have been proposed to describe changes in risk over calendar time. In epidemiology, these changes are commonly referred to as secular trends or period effects, which capture systematic variation in disease incidence associated with calendar time. Classically, this is formalized through age–period–cohort models, where period effects isolate temporal changes that simultaneously impact multiple cohorts \citep{clayton_models_1987}.  More broadly, studies where individuals initiate treatment at different calendar times must account for underlying shifts in risk, i.e., must account for calendar-time changes (see, for example, \citep{hernan_observational_2008}). Yet, despite this broad consensus across applied disciplines, the formal mathematical articulation of calendar-time changes has received little attention in the causal inference and classical survival analysis literature.

An exception is the work of \citet{hansen_estimating_2017}, who proposed a non-parametric definition of these calendar-time changes. Motivated by time trends in disease incidence, their study brought attention to the complications arising when conditioning on $E$ is necessary for censoring assumptions to hold. In their framework, the presence of calendar-time trends, i.e., whether $E$ is associated with the outcome $T$, i.e., $T\not\independent E$, is characterized by whether the independent censoring assumption for the administrative part of censoring holds conditional on $E$. This naturally follows because the administrative part of right-censoring depends on $E$ once the calendar time at which the study ends, $\mathrm{T}$, is fixed.

\citet{hansen_estimating_2017} posit that the independence of the administrative part of censoring is equivalent to the absence of calendar-time trends, an assumption they describe as testable.\footnote{\textit{"independence of the administrative part of the censoring is equivalent to the absence of calendar-time trends, which is in fact testable."}} However, their condition $(T\not\independent E)$ is only assessable during the study period. More precisely, using the continuous-time notation introduced in Section \ref{sec_app: continuous_time_secc} (where $t_e = \mathrm{T} - e$), only whether the condition $P(T\leq t \mid  E=e)$ is constant for every $e\in \{0, \dots, \mathcal{T}\}$ and for every $t \leq \textcolor{teal}{t_e}$ can be tested from the observed data. Equivalently, only the variation of the hazard $\alpha_e(s)$ with $e$ such that $s\leq t_e$ can be assessed from the observed data. Thus, only $\prstudyp-\textnormal{ctc}$ is testable from the observed data: $\prgapp-\textnormal{ctc}$ is not. In particular, the hazard may be constant with respect to $e$ throughout the observable study period but vary with $e$ in the post-study period. Hence, failure to detect dependence between $T$ and $E$ during the study period does not establish $T\independent E$:
\begin{align*}
    \neg(\prstudyp-\textnormal{ctc}) \not\Rightarrow \neg(\prgapp-\textnormal{ctc}).
\end{align*}
Hence for their notion of calendar-time trends to be testable from the observed data, the following extrapolation is needed: $\neg (\prstudyp - \textnormal{ctc}) \Rightarrow \neg (\prgapp - \textnormal{ctc})$. 
Thus, in their framework, a testable presence of calendar-time changes is compatible only if $\ctFull$ holds. It is not possible if $\ctPost$ holds. Our notion of \CTC{} accommodates both configurations and provides a more general definition that (i) can be adapted to different frameworks, thereby situating our contributions at the intersection of classical continuous-time and counterfactual frameworks for causal inference, and (ii) captures features relevant for decision making, as illustrated in our re-analysis of a recent influential study (Case Study 2, Section \ref{sec: case_study_mrna_vaccine}).

\subsubsection{The advantages of an estimand-based approach}
This section presupposes that the reader is already familiar with Section \ref{sec: consequences_and_solutions_in_wsc_and_msm_hazard} of the main text and Appendix \ref{sec_app: wsc_discrete}. 

In the presence of calendar-time changes, \citet{hansen_estimating_2017} propose the following strategy. Assume a proportional hazards if deemed plausible; if not, report survival by strata defined by $E$ (in discrete-time, report $S_{k}^{*,\textnormal{km}}(e)$ as in Section \ref{sec: wsc_procedures} for $k \leq \mathcal{T}-e$, $P^*=P^{\textnormal{o}}$). They also argued that a weighted average of $E$-stratified Kaplan-Meier estimates \citep[Eq. 5, p. 8]{hansen_estimating_2017} is a useful summary measure of risk within the study period (e.g., in discrete time, $S_k^{*,\textnormal{km-strat}, \overline{a}_k}$ as in \eqref{eq: wscp_y1} with $P^*=P^{\textnormal{o}}$; and in continuous time, $S^-(t)$ as in \eqref{eq: truncated_cont_survival}).

Consider now the first solution proposed by \citep{hansen_estimating_2017} to counter the presence of calendar-time changes: the use of Proportional Hazards model. The Proportional Hazards model can be thought of as a special case of MSM-hazard procedures. As our characterization (Proposition \ref{proposition: proof_characterization}), which applies to a wide range of settings, makes clear, valid estimation does not require such a model. Thus, relying exclusively on proportional hazards may be overly restrictive in observational settings where such an assumption is implausible. Furthermore, our approach provides an extensive discussion of the assessment and practical implications of extrapolation assumptions (see Sections \ref{sec: wsc_procedures} and \ref{sec: extrapolation_assumptions_in_msm_hazard}), offering a guide for applied researchers.
When the proportional hazards assumption is violated, the descriptive solution across strata proposed by \citet{hansen_estimating_2017} targets a parameter equivalent to an $E$-stratified Kaplan-Meier estimator (e.g., $S_k^{*,\textnormal{km-strat}, \overline{a}_k}$). However, adopting this as a target parameter presents certain limitations in broader contexts. While the authors correctly note that strata comparisons require similar years of entry and follow-up, the $E$-stratified Kaplan-Meier can exhibit undesirable properties in settings where treatment effects manifest over longer follow-up times.  Specifically, as follow-up $k$ increases, the term involving the restricted survival time increasingly dominates the estimator, potentially inflating the inferred treatment effect. See expression \eqref{eq: bias_wscp_y1} of Appendix \ref{sec_app: supplementary_wsc_deferred_main_text}. Thus, while considering an $E$-stratified Kaplan-Meier estimator may be appropriate for the scenarios considered in their work, it is more difficult to justify its use in more general settings. 

\subsection*{Relation to this work}
The estimand-based methodology addresses these limitations by first clearly defining the target estimand and then systematically studying how the statistical parameter targeted by a given procedure (such as WSC procedures) deviates from it. This approach has some advantages.

First, it eliminates the need to rely on purely descriptive analyses when standard modeling assumptions fail. In our work, we characterize the deviation from the targeted estimand in WSC procedures, articulate the conditions required for valid identification in MSM-hazard procedures, and formulate bounds for cases where these conditions are unmet.

Second, our results extend the investigation of consequences of calendar-time changes to a other data structures, including point versus multiple treatments and baseline versus time-varying confounders. Because adjusting for $E$ is frequently necessary in observational studies to satisfy independent censoring assumptions, our results apply to prevailing methodologies, such as target trial emulation. 

Finally, when the assumptions are not met, we propose to use c-values and extrapolation curves. We illustrate how they can refine and assess decisions based on the target estimand(s) of interest.

\subsection{Inverse Probability of Treatment theory}
\label{sec_app: iptw_technical}
Here we establish results for Inverse Probability of Treatment Weighting (IPTW) that are compatible with the positivity violation described in Section \ref{sec: illustrative_example_censoring_assumptions_ill_posed} of the main text. We present these results in discrete time, as the foundational formal theory for IPTW (and IPCW) has traditionally been developed in discrete-time settings; however, the same results can be extended to continuous-time settings.

\textbf{Notation and abbreviations.} We will assume that any probability law $P(X)$ for a vector $X$ admits a density $f_X(x)$ with respect to a dominating measure $\mu$. 
When $(X,Y) = (0,0)$ we will just denote $(X,Y)=0$.  For any two disjoint vectors $X_1$, $X_2$, by marginalizing over $X_1$ at $x_2$ we mean $\int \mathbb{1}_{\{x_2^{\prime} = x_2 \}} f_{(X_1, X_2)}(x_1^{\prime}, x_2^{\prime}) d\mu(x_1^{\prime})$ and write ``marg. over $X_1$''. If a variable is indexed by a negative number in a proposition or function (with the only exception of $C_{-1}$, and unless otherwise specified), it is considered undefined and should be omitted from the expression. For example, $f_{A_0\mid \overline{A}_{-1}, \overline{Y}_{-1}, V } (a_0 \mid \overline{a}_{-1}, 0, v) \equiv f_{A_0 \mid V}(a_0 \mid v)$ and $f_{A_0 \mid \overline{L}_{0}, \overline{A}_{-1}, \overline{Y}_{-1} } (a_0,\mid {l}_{0}, \overline{a}_{-1}, 0) \equiv f_{A_0 \mid L_0}(a_0 \mid l_0)$. Finally, let
\begin{itemize}[leftmargin=1.8cm, labelsep=0.5cm]
    \item[$\RExch$] : ex$(e,k)$ for every $(e,k)$ such that $ \tau = (e+k) \in  \{0,\dots, \mathcal{T}\}$; \label{ass: r.e}
    \item[$\RCons$] : con$(e,k)$ for every $(e,k)$ such that $ \tau = (e+k) \in  \{0,\dots, \mathcal{T}\}$;\label{ass: r.c}
    \item[$\RPos$] : pos$(e,k)$ for every $(e,k)$ such that $ \tau = (e+k) \in  \{0,\dots, \mathcal{T}\}$.\label{ass: r.p}
\end{itemize}

\subsubsection{Meaning of law $P^*$ generated by $P$ by means of weights $\{W_m\}_{m=0}^k$}
\label{sec_app: meaning_of_law_generated_by_weights}
Let $P^{\textnormal{o}} \equiv P^{\textnormal{o}}(P)$ denote the observed data law.
Let $k \in \{0, \dots, K\}$. We say that the weights $\{W_m\}_{m=0}^k$, where
\begin{multline}\label{eq: general_weights}
        W_m(\overline{a}_m, \overline{l}_m) \coloneqq \prod_{j=0}^m \frac{\phi_j\left(\overline{a}_j,v\right)}{f_{A_j, C_{j-1} \mid \overline{L}_{j}, \overline{A}_{j-1},
    {C}_{j-2},
    {Y}_{j-1} } (a_j, 0\mid \overline{l}_{j}, \overline{a}_{j-1}, 0)} \\
    \textnormal{for } P^{\textnormal{o}} - \textnormal{almost all } \overline{a}_m, \overline{l}_m \textnormal{ and for every } m \in \{0, \dots, k\},
\end{multline}
where $\phi_j$ is a function that depends on $\overline{a}_j,v$, $V \subseteq L_0$, generates a law $P^*$ from $P$ through follow-up $k$ if $P^*$ differs from $P^{\textnormal{o}}$ only in its conditional distributions
{\footnotesize
\begin{multline}\label{def: conditionals_of_ps}
P^*(A_m= a, C_{m-1} = c_{m-1} \mid {\overline{L}_m=\overline{l}_m, Y_{m-1}=0, \overline{A}_{m-1} = \overline{a}_{m-1}, C_{m-2}=0})  \\
    \coloneqq\int \mathbb{1}_{\{a^{\prime} = a\}}\mathbb{1}_{\{c_{m-1}^{\prime} = c_{m-1}\}} \left(\frac{\mathbb{1}_{\{c_{m-1} = 0\}} \cdot \phi_m\left(a, \overline{a}_{m-1},v\right) + \mathbb{1}_{\{c_{m-1} = 1\}} \cdot \phi_m^{\textnormal{n}}\left(a, \overline{a}_{m-1},v\right)}{f_{A_{m}, C_{m-1} \mid \overline{L}_{m}, \overline{A}_{m-1},
    \overline{C}_{m-2},
    \overline{Y}_{m-1} } (a, c_{m-1}\mid \overline{l}_{m}, \overline{a}_{m-1}, 0)}\right)  \\
    \textnormal{d}P(A_m=a^{\prime}, C_{m-1} = c_{m-1}^{\prime} \mid {\overline{L}_m=\overline{l}_m, Y_{m-1}=0, \overline{A}_{m-1} = \overline{a}_{m-1}, C_{m-2}=0}), \\
    a,c_{m-1} \in \{0,1\}, \textnormal{for } P^{\textnormal{o}} - \textnormal{almost all } \overline{a}_m, \overline{l}_m \textnormal{ and for every } m \in \{0, \dots k\},
\end{multline}
where $\phi_m^{\textnormal{n}}\left(a, \overline{a}_{m-1},v\right)$ is defined so that 
\begin{align*}
   P^*(A_m= a, C_{m-1} = c_{m-1} \mid {\overline{L}_m=\overline{l}_m, Y_{m-1}=0, \overline{A}_{m-1} = \overline{a}_{m-1}, C_{m-2}=0}) 
\end{align*}
is a probability law for $P^{\textnormal{o}} - \textnormal{almost all } \overline{l}_m, \overline{a}_{m-1} \textnormal{ and for every } m \in \{0, \dots k\}$.
}

Weighting iid observations drawn from $P^{\textnormal{o}}$ is equivalent to drawing iid observations from $P^*$.

In our elaborations, we will consider the so-called stabilized weights \citep{hernan_marginal_2000, hernan_observational_2008, dickerman_avoidable_2019}, 
\begin{align}\label{eq: stabilized_weights}
\phi_m(\overline{a}_m, v) \coloneqq f_{A_m, C_{m-1}\mid \overline{A}_{m-1}, 
    \overline{C}_{m-2},
    \overline{Y}_{m-1}, V } (a_m, 0 \mid \overline{a}_{m-1},0, 0, v), 
\end{align}
and so 
\begin{align*}
\phi_m^{\textnormal{n}}(\overline{a}_m, v) \coloneqq f_{A_m, C_{m-1}\mid \overline{A}_{m-1}, 
    \overline{C}_{m-2},
    \overline{Y}_{m-1}, V } (a_m, 1 \mid \overline{a}_{m-1},0, 0, v).
\end{align*}

When $k$ is not specified, we consider $k=K$. When the weights are not mentioned explicitly, we will simply say that $P^*$ is generated by $P$. \footnote{By convention, we define $W_{-1}$=1 a.s..}

\subsubsection{Technical results}
\label{sec_app: iptw_preliminaries}
We establish results for the setting and data structure presented in the main text (\hyperref[table: conf_8]{\textbf{VIII}} of Table \ref{table: settings_and_data_structures}), which constitutes the most mathematically involved setting; extensions to simpler data structures are straightforward. Throughout, we will consider weights defined as in \eqref{eq: general_weights}.

\begin{lemma}[Property \eqref{eq: secc} is preserved in $P^*$]
\label{lemma: secc_is_preserved_after_iptw}
Suppose \eqref{eq: secc} holds. Consider a law $P^*$ generated from $P$ via some weights $\{W_k\}_{k=0}^K$. Then
\begin{multline}\label{eq: secc_ps}
    \mathbb{1}_{\{E + k \geq \mathcal{T} + 1\}} = 1, Y_{k-1}=0, C_{k-2}=0 \Longrightarrow C_{k-1} = 1, \textnormal{ $P^*-$\textnormal{a.s.}} \\
    \textnormal{ for every } k \in \{0, \dots, K\}.
\end{multline} 
\begin{proof}
    The result follows from observing that all conditional distributions of $P^*$ equal those of $P$, with the exception of those in \eqref{def: conditionals_of_ps}. The result follows by noting that a law is absolutely continuous with respect to itself and that the conditional laws of $P^*$ involved in \eqref{def: conditionals_of_ps} are absolutely continuous with respect to those in $P$. 
\end{proof}
\end{lemma}

Figure \ref{fig: iptw_preserves_secc} offers an intuitive illustration of why the result of Lemma \ref{lemma: secc_is_preserved_after_iptw} holds true.

The following lemma reestablishes well-known identification results.

\begin{lemma}
\label{lemma: kiss}
Suppose $\RAID$ holds. Let $P^*$ be a law generated from $P$ by some weights ${\{W_k\}}_{k=0}^K$ with $\phi_k(\overline{a}_k,\cdot)$ depending at most on $V \subseteq L_0$ for every $k \in \{0, \dots, K\}$. Suppose $E \in V$. Then, for every $k \in \{0, \dots, K\}$ the following expressions coincide:
\begin{align}
    &h_k^{\overline{a}_k^+}(v), \label{eq: p_star_id_1}\\
    &P^*\left( Y_{k} = 1 \mid Y_{k-1} = 0, \overline{A}_{k} = \overline{a}_k^{+}, C_{k-1} = 0, V=v \right), \label{eq: p_star_id_3}
\end{align}
for all $\overline{a}_k^{+}$ and $v$ such that $e(v) + k \le \mathcal{T}$ (ensuring the above probabilities are well-defined).\footnote{We explicitly introduce the notation $\stackrel{\prstudyp}{=}$ to indicate such an equality; see Appendix \ref{sec_app: notion_correct_specification}.}
\end{lemma}

\begin{proof}
Let $k \in \{0, \dots, K\}$. From Lemma \ref{lemma: secc_is_preserved_after_iptw} it follows that for any $v$ such that $e(v) + k > \mathcal{T}$ \eqref{eq: p_star_id_3} is not well-defined. Thus, consider $k$ and $v$ such that $e(v) + k \le \mathcal{T}$, and let $\overline{a}_k \equiv \overline{a}_k^{+}$. 

By the definition of conditional probability in the pseudo-population $P^*$, 
\begin{align}
    &P^*\left( Y_{k} = 1 \mid Y_{k-1} = 0, \overline{A}_{k} = \overline{a}_k, C_{k-1} = 0, V=v \right) \notag \\
    =& \frac{\expectation \left( W_k \mathbb{1}_{\{C_{k-1}=0, \overline{A}_k = \overline{a}_k, Y_{k-1} = 0, Y_k = 1\}} \mid V=v\right)}{\expectation \left( W_k \mathbb{1}_{\{C_{k-1}=0, \overline{A}_k = \overline{a}_k, Y_{k-1} = 0\}} \mid V=v\right)}. \label{eq: ratio_pstar}
\end{align}

Take $m \in \{0, \dots, k\}$, let $H_m \equiv (\overline{L}_m, \overline{A}_{m-1}, \overline{C}_{m-2}, \overline{Y}_{m-1})$, let 
\begin{align*}
  g(\overline{a}_m, v) \equiv W_{m-1} \phi_m(\overline{a}_m, v) \mathbb{1}_{\{C_{m-2}=0, \overline{A}_{m-1} = \overline{a}_{m-1}, Y_{m-1}^{\overline{a}_{m-1}, c_{m-2}=0} = 0\}},  
\end{align*}
and label the following identity by uc($m\mid v$):
{\scriptsize
\begin{align*}
    & \expectation \left( W_m \mathbb{1}_{\{C_{m-1}=0, \overline{A}_m = \overline{a}_m\}} \mathbb{1}_{\{Y_{m-1}^{\overline{a}_{m-1}, c_{m-2}=0} = 0, Y_k^{\overline{a}_k, {c}_{k-1}=0} = 1\}} \mid V=v \right) & \\
    =& \expectation \left( g(\overline{a}_m, v) \expectation \left[ \frac{\mathbb{1}_{\{C_{m-1}=0, A_m = a_m\}} \mathbb{1}_{\{Y_k^{\overline{a}_k, c_{k-1}=0} = 1\}}}{f_{A_m, C_{m-1} \mid H_m)}(a_m, 0 \mid \dots)} \mathrel{\Bigg|} H_m \right] \mathrel{\Bigg|} V=v \right) & \Leftarrow{(i)} \\
    =& \expectation \left( g(\overline{a}_m, v) \expectation \left[ \frac{\mathbb{1}_{\{C_{m-1}=0, A_m = a_m\}}}{f_{A_m, C_{m-1} \mid H_m}(a_m, 0 \mid \dots)} \mathrel{\Bigg|} H_m \right] \expectation \left[ \mathbb{1}_{\{Y_k^{\overline{a}_k, c_{k-1}=0} = 1\}} \mathrel{\Bigg|} H_m \right] \mathrel{\Bigg|} V=v \right) & \Leftarrow{(ii)} \\
    =& \expectation \left( g(\overline{a}_m, v) \cdot 1 \cdot \expectation \left[ \mathbb{1}_{\{Y_k^{\overline{a}_k, c_{k-1}=0} = 1\}} \mathrel{\Bigg|} H_m \right] \mathrel{\Bigg|} V=v \right)  \\
    =& \expectation \left( W_{m-1} \phi_m(\overline{a}_m, v) \mathbb{1}_{\{C_{m-2}=0, \overline{A}_{m-1} = \overline{a}_{m-1}, Y_{m-1}^{\overline{a}_{m-1}, c_{m-2}=0} = 0, Y_k^{\overline{a}_k, c_{k-1}=0} = 1\}} \mid V=v \right) & \Leftarrow{(iii)},
\end{align*}
}
where:
\begin{align*}
    (i)\equiv& \textnormal{def. of } W_m \textnormal{ and } \RPos, \\
    (ii)\equiv& \RExch, \\
    (iii) \equiv& \textnormal{tower property}.
\end{align*}

Iterating from $m=k$ down to $m=0$, the numerator evaluates to:

{
\tiny
\begin{align*}
    &\expectation \left( W_k \mathbb{1}_{\{C_{k-1}=0, \overline{A}_k = \overline{a}_k, Y_{k-1} = 0, Y_k = 1\}} \mid V=v \right) \\
    =& \expectation \left( W_k \mathbb{1}_{\{C_{k-1}=0, \overline{A}_k = \overline{a}_k\}} \mathbb{1}_{\{Y_{k-1}^{\overline{a}_{k-1}, c_{k-2}=0} = 0, Y_k^{\overline{a}_k, c_{k-1}=0} = 1\}} \mid V=v \right) & \Leftarrow{(i)} \\
    =& \expectation \left( W_{k-1} \phi_k(\overline{a}_k, v) \mathbb{1}_{\{C_{k-2}=0, \overline{A}_{k-1} = \overline{a}_{k-1}\}} \mathbb{1}_{\{Y_{k-1}^{\overline{a}_{k-1}, c_{k-2}=0} = 0, Y_k^{\overline{a}_k, c_{k-1}=0} = 1\}} \mathrel{\Bigg|} V=v \right) & \Leftarrow{(ii)} \\
    =& \expectation \left( W_{k-2} \phi_{k-1}(\overline{a}_{k-1}, v) \phi_k(\overline{a}_k, v) \mathbb{1}_{\{C_{k-3}=0, \overline{A}_{k-2} = \overline{a}_{k-2}\}} \mathbb{1}_{\{Y_{k-1}^{\overline{a}_{k-1}, c_{k-2}=0} = 0, Y_k^{\overline{a}_k, c_{k-1}=0} = 1\}} \mathrel{\Bigg|} V=v \right) & \Leftarrow{(iii)} \\
    &\vdots \\
    =& \expectation \left( 1 \cdot \left( \prod_{m=0}^k \phi_m(\overline{a}_m, v) \right) \cdot 1 \cdot \mathbb{1}_{\{Y_{k-1}^{\overline{a}_{k-1}, c_{k-2}=0} = 0, Y_k^{\overline{a}_k, c_{k-1}=0} = 1\}} \mathrel{\Bigg|} V=v \right) &\Leftarrow{(iv)}\\
    =& \left( \prod_{m=0}^k \phi_m(\overline{a}_m, v) \right) P\left(Y_k^{\overline{a}_k, c_{k-1}=0} = 1, Y_{k-1}^{\overline{a}_{k-1}, c_{k-2}=0} = 0 \mid V=v \right) & \Leftarrow{(v)}.
\end{align*}
}
where:
\begin{align*}
(i) \equiv& \RCons, \\
(ii) \equiv& \textnormal{uc($k\mid v$)}, \\
(iii) \equiv& \textnormal{uc($k-1\mid v$)}, \\
(iv) \equiv& \textnormal{uc($0 \mid v$)}, \\
(v) \equiv& \phi_m(\overline{a}_m,v) \textnormal{ const. given } V=v.
\end{align*}
By similar arguments, we have
\begin{align*}
    &\expectation \left( W_k \mathbb{1}_{\{C_{k-1}=0, \overline{A}_k = \overline{a}_k, Y_{k-1} = 0\}} \mid V=v \right) \\
    =& \left( \prod_{m=0}^k \phi_m(\overline{a}_m, v) \right) P\left(Y_{k-1}^{\overline{a}_{k-1}, c_{k-2}=0} = 0 \mid V=v \right)
\end{align*}

Thus,
{
\footnotesize
\begin{align*}
    &P^*\left( Y_{k} = 1 \mid Y_{k-1} = 0, \overline{A}_{k} = \overline{a}_k^+, C_{k-1} = 0, V=v \right) &\Leftarrow{\eqref{eq: ratio_pstar}}\\
    =& \frac{\left( \prod_{m=0}^k \phi_m(\overline{a}_m^+, v) \right) P\left(Y_k^{\overline{a}_k^+, c_{k-1}=0} = 1, Y_{k-1}^{\overline{a}_{k-1}^+, c_{k-2}=0} = 0 \mid V=v \right)}{\left( \prod_{m=0}^k \phi_m(\overline{a}_m^+, v) \right) P\left(Y_{k-1}^{\overline{a}_{k-1}^+, c_{k-2}=0} = 0 \mid V=v \right)} \\
    =& \frac{P\left(Y_k^{\overline{a}_k^+, c_{k-1}=0} = 1, Y_{k-1}^{\overline{a}_{k-1}^+, c_{k-2}=0} = 0 \mid V=v \right)}{P\left(Y_{k-1}^{\overline{a}_{k-1}^+, c_{k-2}=0} = 0 \mid V=v \right)}  \\
    =& P\left(Y_k^{\overline{a}_k^+, c_{k-1}=0} = 1 \mathrel{\Bigg|} Y_{k-1}^{\overline{a}_{k-1}^+, c_{k-2}=0} = 0, V=v \right) \\
    =& h_k^{\overline{a}_k^+}(v) & \Leftarrow{\textnormal{def. of } h_k^{\overline{a}_k}(v)}.
\end{align*}
}
\end{proof}

We now additionally provide an alternative proof by drawing from the marvelous causal proof proposed by \citet{robins_marginal_2000}. Therein, the result has already been argued to hold, but no explicit proof was given.
Following \citet{robins_marginal_2000}, consider a strengthening of ex$(e,k)$, 
{\footnotesize
\begin{align*}
   \textnormal{exs$(e,k)$: } {\scriptsize\left(Y_K^{\mathcal{A}}, \dots, Y_k^{\mathcal{A}} \right)
{\independent}
\{A_k, C_{k-1}\}
\mid
\overline{L}_k, {Y}_{k-1} = 0, \overline{A}_{k-1}, \overline{C}_{k-2} = 0, E=e},
\end{align*}
}
whereby, for $k \in \{0, \dots, K\}$, $Y_k^{\mathcal{A}} \coloneqq \{Y_k^{\overline{a}_k = \overline{0}, c_{k-1} = 0}, Y_k^{\overline{a}_k = \overline{1}, c_{k-1} = 0}\}$.  In Lemmas \ref{lemma: factorization_pz_pstarz} and \ref{lemma: kiss_alternative} we say that $\RAID$ holds if exs$(e,k)$ in place of ex$(e,k)$ holds for every $(e,k)$ such that $\tau = (e+k) \in  \{0,\dots, \mathcal{T}\}$.

For every $k \in \{0, \dots, K\}$, let
\begin{align*}
Z \equiv Z_k \coloneqq (L_0, C_{-1}, A_0, Y_0, \dots, L_{k}, C_{k-1}, A_k, Y_k).
\end{align*}
Let $f_Z$ denote the density of $P(Z)$ and $f_Z^*$ the density of the pseudopopulation $P^*$ generated by $P(Z)$ via weights $\{W_m\}_{m=0}^k$.

\begin{lemma}[Factorization of $f_Z$ and $f_Z^{*}$]
\label{lemma: factorization_pz_pstarz}
Suppose $\RAID$ holds. Let $z$ be such that $c_{k-1}(z)=y_{k-1}(z)=0$. Let $\overline{a}_k \equiv \overline{a}_k(z)$. Then, the density $f_{Z}(z)$ can be factorized as 
\begin{equation}
  \begin{aligned}
    &f_Z(z) \\
    =& f_{L_0, {Y}_{k-1}^{\overline{a}_{k-1}}, {Y}_{k}^{\overline{a}_k}}(l_0, 0, y) \\
    \cdot& \prod_{m=1}^{k}f_{L_m \mid \overline{L}_{m-1}, \overline{A}_{m-1}, {C}_{m-2},{Y}_{k-1}^{\overline{a}_{k-1}}, Y_k^{\overline{a}_k}} (l_m \mid \overline{l}_{m-1}, \overline{a}_{m-1}, 0, y)\\
    \cdot & \prod_{m=0}^{k}f_{A_m, C_{m-1} \mid \overline{L}_{m}, \overline{A}_{m-1}, {C}_{m-2}, {Y}_{m-1} } (a_m,0\mid \overline{l}_{m}, \overline{a}_{m-1}, 0),
\end{aligned}  
\label{eq: f_z}
\end{equation}
and the density $f_Z^{*}(z)$ of $P^*$, with ${\{W_m\}}_{m=0}^k$ can be factorized as
\begin{equation}
    \begin{aligned}
        &f_Z^{*}(z) \\
        =& f_{L_0, {Y}_{k-1}^{\overline{a}_{k-1}}, {Y}_{k}^{\overline{a}_k}}(l_0, 0, y) \\
        \cdot& \prod_{m=1}^{k}f_{L_m \mid \overline{L}_{m-1}, \overline{A}_{m-1}, {C}_{m-2}, {Y}_{k-1}^{\overline{a}_{k-1}}, Y_k^{\overline{a}_k}} (l_m \mid \overline{l}_{m-1}, \overline{a}_{m-1}, 0, y)\\
        \cdot & 
        \prod_{m=0}^{k} \phi_{m}(\overline{a}_m, v).
    \end{aligned}
\label{eq: f_star_z}
\end{equation}

\begin{proof}
Take $k \in \{0, \dots, K\}$. Let $z$ be such that $c_{k-1}(z)=y_{k-1}(z)=0$ (which implies, from \eqref{eq: secc} that $e(z) + k \leq \mathcal{T}$) and denote $\overline{a}_k \equiv \overline{a}_k(z)$. Consider
{\scriptsize
\begin{align*}
    f_Z (z) =& f_{Y_k, {Y}_{k-1}, \overline{L}_k, \overline{A}_k, {C}_{k-1}}(y, 0, \overline{l}_k, \overline{a}_k, 0) & \Leftarrow{\textnormal{def.} \textnormal{ and } \RPos}\\
    =& f_{Y_k^{\overline{a}_k}, {Y}_{k-1}^{\overline{a}_{k-1}}, \overline{L}_k, \overline{A}_k, {C}_{k-1}}(y, {0}, \overline{l}_k, \overline{a}_k, {0}) & \Leftarrow{\RCons}\\
    =& f_{Y_k^{\overline{a}_k}, {Y}_{k-1}^{\overline{a}_{k-1}}, L_0}(y, {0}, l_0) \\
    \cdot & \prod_{m=1}^{k}f_{L_m \mid \overline{L}_{m-1}, \overline{A}_{m-1}, {C}_{m-2},{Y}_{k-1}^{\overline{a}_{k-1}}, Y_k^{\overline{a}_k}} (l_m \mid \overline{l}_{m-1}, \overline{a}_{m-1}, 0, y)  \\
    \cdot & \prod_{m=0}^{k}f_{A_m, C_{m-1}\mid \overline{L}_{m}, \overline{A}_{m-1}, {C}_{m-2}, {Y}_{k-1}^{\overline{a}_{k-1}}, Y_k^{\overline{a}_k} } (a_m,0 \mid \overline{l}_{m}, \overline{a}_{m-1}, 0,y) & \Leftarrow{\RCons}.
\end{align*}
}
For $m \in \{0, \dots, k\}$:
{\footnotesize
\begin{align*}
    &f_{A_m, C_{m-1}\mid \overline{L}_{m}, \overline{A}_{m-1}, {C}_{m-2}, {Y}_{k-1}^{\overline{a}_{k-1}}, Y_k^{\overline{a}_k} } (a_m,0 \mid \overline{l}_{m}, \overline{a}_{m-1}, {0},y) \\
    \equiv& f_{A_m, C_{m-1}\mid \overline{L}_{m}, \overline{A}_{m-1}, {C}_{m-2}, {Y}_{m-1}^{\overline{a}_{m-1}}, {Y}_{m}^{\overline{a}_{m}},  \dots, Y_k^{\overline{a}_k} } (a_m,0 \mid \overline{l}_{m}, \overline{a}_{m-1}, 0,y) \\
    =& f_{A_m, C_{m-1}\mid \overline{L}_{m}, \overline{A}_{m-1}, {C}_{m-2}, {Y}_{m-1}, {Y}_{m}^{\overline{a}_{m}},  \dots, Y_k^{\overline{a}_k} } (a_m,0 \mid \overline{l}_{m}, \overline{a}_{m-1}, 0,y) & \Leftarrow{\RCons} \\
    =& f_{A_m, C_{m-1}\mid \overline{L}_{m}, \overline{A}_{m-1}, {C}_{m-2}, {Y}_{m-1}} (a_m,0 \mid \overline{l}_{m}, \overline{a}_{m-1}, 0) & \Leftarrow{\RExch}.
\end{align*}
}
This establishes \eqref{eq: f_z}. 

Since $P^*$ differs from $P^{\textnormal{o}}$ only in its conditional distributions, i.e.,
{\footnotesize
\begin{align*}
        & \frac{\textnormal{d}P^*(A_m = a, C_{m-1} = 0\mid \overline{L}_m=\overline{l}_m, Y_{m-1}=0, \overline{A}_{m-1} = \overline{a}_{m-1}, C_{m-2}=0)}{\textnormal{d}P^{\textnormal{o}}(A_m = a, C_{m-1} = 0 \mid \overline{L}_m=\overline{l}_m, Y_{m-1}=0, \overline{A}_{m-1} = \overline{a}_{m-1}, C_{m-2}=0)}  & \\
    =&  \frac{\phi\left(a, \overline{a}_{m-1},v\right)}{f_{A_{m}, C_{m-1} \mid \overline{L}_{m}, \overline{A}_{m-1},
        {C}_{m-2},
        {Y}_{m-1} } (a, 0\mid \overline{l}_{m}, \overline{a}_{m-1}, 0)} & \Leftarrow{\eqref{eq: general_weights}}\\
        \Rightarrow& f_{A_{m}, C_{m-1} \mid \overline{L}_{m}, \overline{A}_{m-1},
        {C}_{m-2},
        {Y}_{m-1} }^* (a, 0\mid \overline{l}_{m}, \overline{a}_{m-1}, 0) \\
        =& \phi\left(a, \overline{a}_{m-1},v\right), \textnormal{ for every } m \in \{0, \dots, k\}.
\end{align*}
}
Thus, \eqref{eq: f_star_z} follows from the definition of $f^*(z)$.
\end{proof}
\end{lemma}

Finally,
\begin{lemma}
\label{lemma: kiss_alternative}
Suppose $\RAID$ holds. Let $P^*$ be a law generated from some weights ${\{W_k\}}_{k=0}^K$ with $\phi_k(\overline{a}_k, \cdot)$ depending at most on $V \subseteq L_0$ for every $k \in \{0, \dots, K\}$. Suppose $E \in V$. Then, for every $k \in \{0, \dots, K\}$ the following expressions coincide:
\begin{align*}
    &h_k^{\overline{a}_k^+}(v), \\
    &P^*\left( Y_{k} = 1 \mid Y_{k-1} = 0, \overline{A}_{k} = \overline{a}_k^{+}, C_{k-1} = 0, V=v \right),
\end{align*}
for all $\overline{a}_k^{+},v$ for which the above probabilities are well-defined.
\end{lemma}
\begin{proof}
Fix $k \in \{0, \dots, K\}$ and let $\overline{a}_k^{+} \in \{\overline{0},\overline{1}\}$. Let $z$ be such that $c_{k-1}(z)=y_{k-1}(z)=0$, $y(z)=1$, and $\overline{a}_k \equiv \overline{a}_k(z) = \overline{a}_k^+$  Then, 
{\scriptsize
\begin{align}
         & f_Z^*(z) \notag \\
         =&f_{L_0, {Y}_{k-1}^{\overline{a}_{k-1}}, {Y}_{k}^{\overline{a}_k}}(l_0, 0, y) \\
        \cdot& \prod_{m=1}^{k}f_{L_m \mid \overline{L}_{m-1}, \overline{A}_{m-1}, \overline{C}_{m-2}, {Y}_{k-1}^{\overline{a}_{k-1}}, Y_k^{\overline{a}_k}} (l_m \mid \overline{l}_{m-1}, \overline{a}_{m-1}, {0}, y)\\
        \cdot & 
        \prod_{m=0}^{k} \phi_{m}(\overline{a}_m, v) & \Leftarrow{\textnormal{Lemma }\ref{lemma: factorization_pz_pstarz}, \eqref{eq: f_star_z}}\notag\\
        \Rightarrow& f^*(z\setminus(l_1, \dots, l_k)) = f_{L_0, {Y}_{k-1}^{\overline{a}_{k-1}}, {Y}_{k}^{\overline{a}_k}}(l_0, 0, y)  \cdot\prod_{m=0}^{k} \phi_{m}(\overline{a}_m, v) &\Leftarrow{\textnormal{ marg. over $\overline{L}_k \setminus L_0$}} \notag\\
        \Rightarrow& f^*(z\setminus(l_0 \setminus v, l_1, \dots, l_k)) = f_{V, {Y}_{k-1}^{\overline{a}_{k-1}}, {Y}_{k}^{\overline{a}_k}}(v, 0, y)  \cdot\prod_{m=0}^{k} \phi_{m}(\overline{a}_m, v) &\Leftarrow{\textnormal{ marg. over ${L}_0 \setminus V$}} \notag \\
        =& f_V(v) \cdot h_k^{\overline{a}_k}(v) \prod_{m=0}^{k-1}\left(1 - h_m^{\overline{a}_m}(v)\right)  \cdot \prod_{m=0}^{k} \phi_{m}(\overline{a}_m, v)   &\Leftarrow{\textnormal{factorization}}\notag \\
        =& f_V(v) \cdot h_k^{\overline{a}_k^+}(v) \prod_{m=0}^{k-1}\left(1 - h_m^{\overline{a}_m^+}(v)\right)  \cdot \prod_{m=0}^{k} \phi_{m}(\overline{a}_m, v)  & \Leftarrow{\overline{a}_k(z) = \overline{a}_k^+} \label{eq: hazard_id_star}. 
\end{align}
}
Thus
\begin{align*}
&P^*\left( Y_{k} = 1 \mid \overline{A}_{k} = \overline{a}_k^{+}, C_{k-1} = 0, Y_{k-1} = 0, V=v \right) \\
=& \frac{f_{Y_k, \overline{A}_k, {C}_{k-1}, {Y}_{k-1}, V}^*(y, \overline{a}_k^+, {0}, 0, v)}{f_{\overline{A}_k, {C}_{k-1}, {Y}_{k-1}, V}^*(\overline{a}_k^+, {0}, 0, v)} & \Leftarrow{\textnormal{def. of $f^*(z)$}}\\
    =& h_k^{\overline{a}_k^+}(v) & \Leftarrow{\eqref{eq: hazard_id_star}}.
\end{align*}
\end{proof}

The following two results follow from the previous lemmas.

\begin{corollary}\label{corollary: kiss_in_msm_hazard_models}
Let $k \in \{0, \dots, K\}$. Consider the assumptions and result of Lemma \ref{lemma: kiss}. If \textnormal{\hi{}} holds, then
\begin{align}
    &h_k^{a_k^+}(v), \label{eq: p_star_msm_hazard_id_1}\\
    &P^*\left( Y_{k} = 1 \mid Y_{k-1} = 0, A_{k} = a_k^{+}, C_{k-1} = 0, V=v \right) \label{eq: p_star_msm_hazard_id_3}
\end{align}
coincide for all $a_k^{+}$ and $v$ for which the above probabilities are well-defined.
\begin{proof}
    Follows immediately from Lemma \ref{lemma: kiss} as, for every $k \in\{0, \dots, K\}$ and for all $\overline{a}_{k}^+$, under \hi{}, $h_k^{a_k^+}(v) = h_k^{\overline{a}_k^+}(v)$ for $P(V)$ almost all $v$.
\end{proof}
\end{corollary}
\begin{corollary}[Identification when $E \notin V$]\label{corollary: kiss_when_e_not_in_v}
Let $k \in \{0, \dots, K\}$. Suppose the assumptions of Lemma \ref{lemma: kiss} hold, except that $E \notin V$. Then
\begin{align}\label{eq: p_star_cond_v_not_e_1}
    P^*\left( Y_{k} = 1 \mid Y_{k-1} = 0, \overline{A}_{k} = \overline{a}_k^{+}, C_{k-1} = 0, V=v \right) = h_{\prstudyp, k}^{\overline{a}_k^+}(v),
\end{align}
for all $\overline{a}_k^{+}, v$ for which the probability is well-defined. Furthermore, if \hi{} holds for $(V,E)$ with $h_k^{\overline{a}_k}(v,e) = h_k^{{a}_k}(v,e)$, this reduces to:
\begin{align}\label{eq: p_star_cond_v_not_e_2}
    P^*\left( Y_{k} = 1 \mid Y_{k-1} = 0, A_{k} = a_k^{+}, C_{k-1} = 0, V=v \right) = h_{\prstudyp, k}^{a_k^+}(v),
\end{align}
for all $a_k^{+}, v$ for which the probability is well-defined. 
\end{corollary}

\begin{proof}
Let us start by establishing \eqref{eq: p_star_cond_v_not_e_1}. Fix $k \in \{0, \dots, K\}$ and let $\mathcal{S}_k \coloneqq \{Y_{k-1} = 0, \overline{A}_{k} = \overline{a}_k^+, C_{k-1} = 0, V=v\}$ denote the conditioning event. Then, 
{\footnotesize
\begin{align*}
    P^*\left( Y_{k} = 1 \mid \mathcal{S}_k \right) 
    =& \frac{P^*\left( Y_k = 1, \mathcal{S}_k \right)}{P^*\left( \mathcal{S}_k \right)} & \\
    =& \frac{P^*\left( Y_k = 1, \mathcal{S}_k, E+k \le \mathcal{T} \right)}{P^*\left( \mathcal{S}_k, E+k \le \mathcal{T} \right)} &\Leftarrow{(i)} \\
    =& P^*\left( Y_{k} = 1 \mid \mathcal{S}_k, E+k \leq \mathcal{T} \right).
\end{align*}
}
where:
\begin{align*}
(i) \equiv& \textnormal{ from Lemma } \ref{lemma: secc_is_preserved_after_iptw}, \textnormal{the administrative censoring mechanism} \eqref{eq: secc}  \\
& \textnormal{is preserved in } P^*; \textnormal{ by contrapositive we have }
    \mathbb{1}_{\mathcal{S}_k} = \mathbb{1}_{\mathcal{S}_k} \cdot \mathbb{1}_{\{E+k \leq \mathcal{T}\}}. 
\end{align*}
Applying the steps analogous to those in the proof of Lemma \ref{lemma: kiss}, we have
{\footnotesize
\begin{align*}
    &P^*\left( Y_{k} = 1 \mid Y_{k-1} = 0, \overline{A}_{k} = \overline{a}_k, C_{k-1} = 0, V=v, E+k \leq \mathcal{T} \right) \\
    =& \frac{\expectation \left( W_k \mathbb{1}_{\{C_{k-1}=0, \overline{A}_k = \overline{a}_k, Y_{k-1} = 0, Y_k = 1\}} \mid V=v, E+k \leq \mathcal{T} \right)}{\expectation \left( W_k \mathbb{1}_{\{C_{k-1}=0, \overline{A}_k = \overline{a}_k, Y_{k-1} = 0\}} \mid V=v, E+k \leq \mathcal{T} \right)} \\
    =& \frac{\left( \prod_{m=0}^k \phi_m(\overline{a}_m, v) \right) P\left(Y_k^{\overline{a}_k, {c}_{k-1}=0} = 1, Y_{k-1}^{\overline{a}_{k-1}, {c}_{k-2}=0} = 0 \mid V=v, E+k \leq \mathcal{T} \right)}{\left( \prod_{m=0}^k \phi_m(\overline{a}_m, v) \right) P\left(Y_{k-1}^{\overline{a}_{k-1}, {c}_{k-2}=0} = 0 \mid V=v, E+k \leq \mathcal{T} \right)} & \Leftarrow{(i)} \\
    =& P\left(Y_k^{\overline{a}_k, {c}_{k-1}=0} = 1 \mathrel{\Bigg|} Y_{k-1}^{\overline{a}_{k-1}, {c}_{k-2}=0} = 0, V=v, E+k \leq \mathcal{T} \right) \\
    =& h_{\prstudyp, k}^{\overline{a}_k}(v) & \Leftarrow{(ii)},
\end{align*}
}
where:
{\footnotesize
\begin{align*}
    (i) \equiv& \textnormal{following the iterative steps of Lemma \ref{lemma: kiss} under } \RAID. \textnormal{ For } m \le k, \\
               & \textnormal{we have } H_m \ni L_0 \ni E. \textnormal{ Thus, we condition on } \\
               & (V=v, E=e), \textnormal{ apply the same rationale, and then} \\
               & \textnormal{invoke the tower property to marginalize over the support of } P(E \mid V, E+k \le \mathcal{T}) \\
    (ii) \equiv& \textnormal{def. } h_{\prstudyp, k}^{\overline{a}_k}(v).
\end{align*}
}
Alternatively, using Lemma \ref{lemma: kiss_alternative} the result follows from \eqref{lemma: secc_is_preserved_after_iptw} and by observing for any $k \in \{0, \dots, K\}$ the marginalization over $\overline{L}_k \setminus V$ involves a marginalization over $E$; as we are marginalizing over observed observations, i.e., $C_{k-1}=0$, \eqref{eq: secc_ps} implies that $e$ is integrated (summed) over $\{0, \dots, \mathcal{T}-k\}$, that is $E \in \{0, \dots, \mathcal{T}-k\}$.

Finally, to establish \eqref{eq: p_star_cond_v_not_e_2} it suffices to notice that $h_{\prstudyp,k}^{a_k}(v)$ can be expressed via hazards $h_k^{\overline{a}_k}(v,e)$ (see \eqref{eq: hazard_variation_dependence_ct}) for which \hi{} holds; thus it must depend only on $a_k$; hence, the result follows from \eqref{eq: p_star_cond_v_not_e_1}.
\end{proof}

\begin{figure}[h]
    \centering
    \subcaptionbox[0.6\linewidth]{For $\{E=2\}$, at $k=0$, there are alive and uncensored individuals; weights $W_{s_1}, W_{s_2}, \dots, W_{s_p}$ redefine the proportions of strata $s_1$, $s_2$, $\dots$, $s_p$. \label{fig:iptw_preserves_secc_1}}[0.95\linewidth]{
        \includegraphics[width=0.6\textwidth]{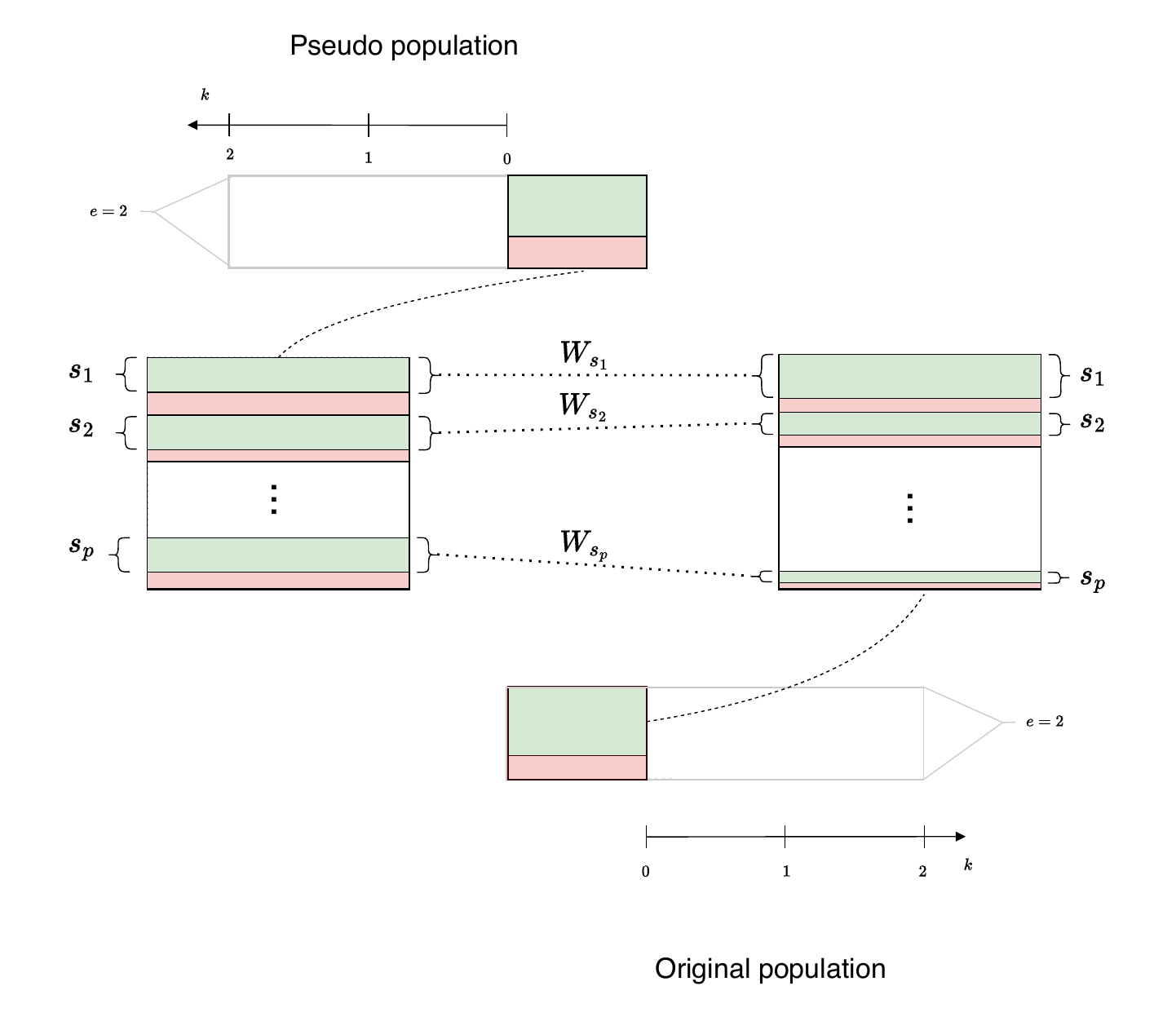}
    }

    \vspace{0.2cm}

    \subcaptionbox[0.6\linewidth]{For $\{E=2\}$, at $k=1$ there are no individuals alive and uncensored (dashed-violet horizontal line); weights $W$ are not even well-defined.
        \label{fig:y}}[0.9\linewidth]{
            \includegraphics[width=0.6\textwidth]{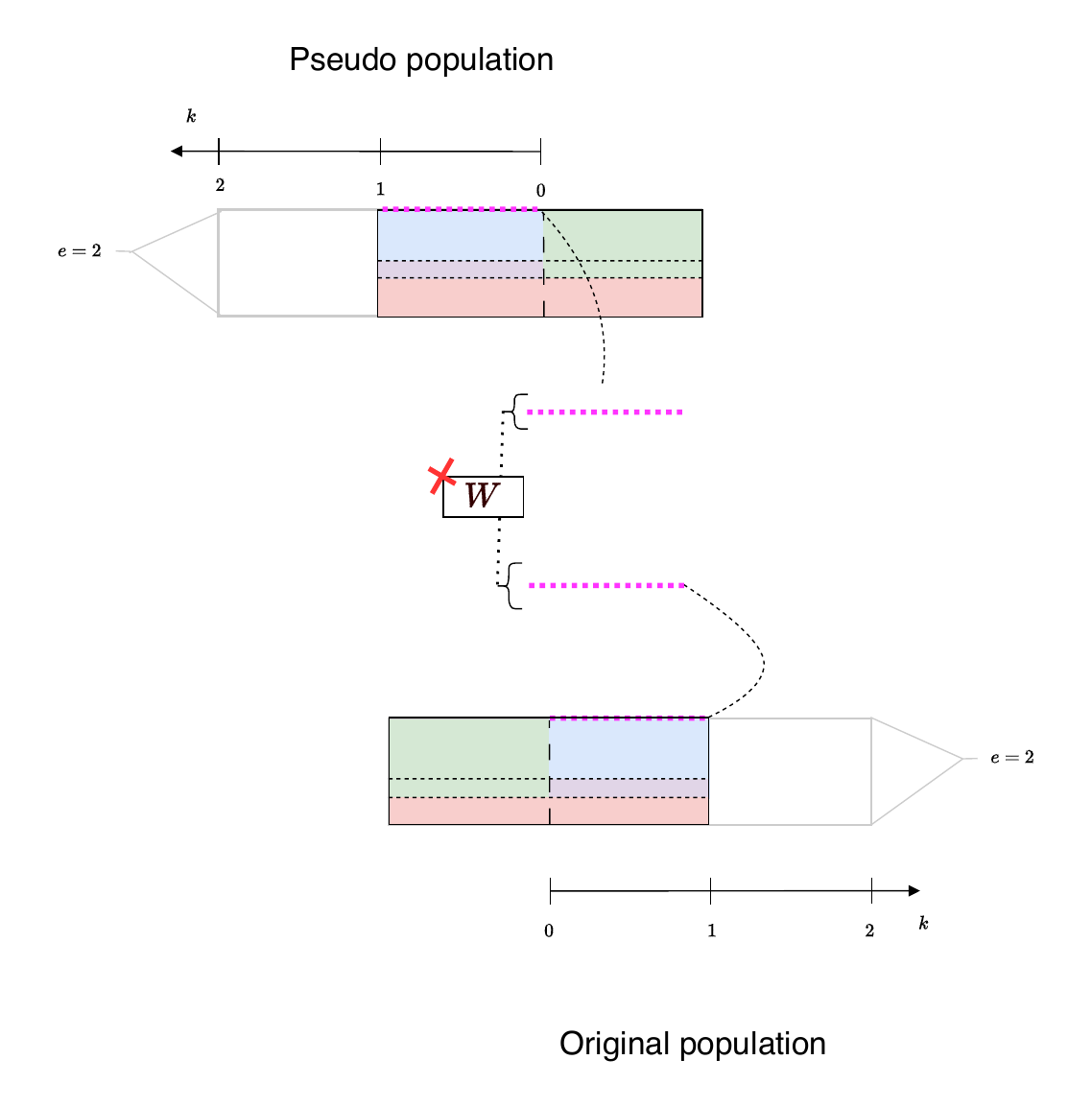}
    }
    \caption{Weights are defined and modified only for those conditional distributions that can be observed from the observed data law $P$. Here, we illustrate it for individuals $\{E=2\}$; the subscript $s$ in $W_{s}$ does not denote the follow-up time $s$, but the value taken by individuals in the stratum $\{S=s\}$.
    }
    \label{fig: iptw_preserves_secc}
\end{figure}\clearpage
\section{Weighted survival curves and MSM-hazard procedures: technical results}
\label{sec_app: consequences_of_ctc}
This appendix presents technical results for WSC and MSM-hazard procedures. It is organized into two main sections. 

Section \ref{sec_app: weighted_survival_curves} concerns WSC procedures and is organized into three subsections. In Section \ref{sec_app: supplementary_wsc_deferred_main_text}, we provide supplementary considerations for the procedure presented in the main text and introduce an additional strategy based on stratification by $E$. Sections \ref{sec_app: wsc_discrete} and \ref{sec_app: continuous_time_wsc} establish formal results for the discrete-time and continuous-time settings. 

Section \ref{sec_app: msm_hazard} presents results for the MSM-hazard procedures introduced in Section \ref{sec: msm_hazard_procedures} of the main text, laying the theoretical groundwork for Section \ref{sec: extrapolation_assumptions_in_msm_hazard} of the main text. Section \ref{sec_app: msm_hazard} is organized into four subsections. In Section \ref{sec_app: formalization_person_time_law}, we begin by describing MSM-hazard procedures and establishing a set of preliminary lemmas. In Section \ref{sec_app: pt_def_and_correspondence}, we offer a definition of the person-time law $P^{\textnormal{pt}}$. We use this to demonstrate that $\beta^*$ can be equivalently conceptualized as the solution to a population score equation defined with respect to a person-time law $P^{\textnormal{pt},*}$, and that $h_k^{*,a_k}(v)$ corresponds to a regression function under this same law. This conceptualization places the notion of correct specification for \textit{``pooled models''} in MSM-hazard procedures on a sound statistical footing, and allows us to formally reason about, and expose the otherwise-overlooked strength of, the conditions required in Proposition \ref{proposition: msm_hazard_characterization} to circumvent staggered-entry-specific complications due to \CTC{}. In Section \ref{sec_app: notion_correct_specification}, we elaborate on the properties of correct specification. In Section \ref{sec_app: staggered_entry_characterization}, we demonstrate the incompatibility of ${\mbps}_{\mid V}^{\textnormal{cs}}$ and ${\mbps}_{\mid V}^{*\textnormal{-cs}}$ when $E \notin V$. We then conclude by providing a formal proof of Proposition \ref{proposition: msm_hazard_characterization} and establishing sharp bounds for scenarios in which its conditions are not satisfied.

\subsection{Weighted survival curves procedures}
\label{sec_app: weighted_survival_curves}

\subsubsection{Considerations deferred from the main text and bias in stratification analyses}
\label{sec_app: supplementary_wsc_deferred_main_text}

Analogous arguments to those presented in Section \ref{sec: wsc_procedures} of the main text apply to settings with a point treatment $A$ \citep{kaplan_nonparametric_1958, cupples_ageadjusted_1995, cole_adjusted_2004, xie_adjusted_2005, bleicher_time_2016, kishan_radical_2018, pouwels_estimating_2020}, upon an appropriate reformulation of $\RAID$ and of the corresponding weights ${\{W_m\}}_{m=0}^k$, both in discrete and continuous settings. For example, studies that include $E$ in the propensity models for weights and that therefore implicitly operate under \CTC{}, e.g., \citet{westreich_time_2010, olarte_parra_trial_2022}, will target parameters isomorphic to \eqref{eq: wscp_y1_km_biased}.

A commonly-used companion strategy to weighting is stratification \citep{cole_adjusted_2004}. Procedures adopting stratification and implicitly operating under censoring assumptions deemed to hold conditional on $E$ target statistical estimands defined as an expectation over $E$ of $E$-stratum-specific Kaplan-Meier curves $S_{k}^{*,\textnormal{km},\overline{a}_k}(E)$ (as in \eqref{eq: wscp_y1_km}, but conditional on $E$):
\begin{align}
\label{eq: wscp_y1}
S_k^{*,\textnormal{km-strat}, \overline{a}_k} \coloneqq \expectation S_{k}^{*,\textnormal{km}, \overline{a}_k}(E).
\end{align} 
We show that under \CTC{}, if $\RAID$ holds, then, for $k> 0$, $S_k^{*, \textnormal{km-strat}, \overline{a}_k}$ targets 
$\mathcal{T}-E$-restricted survival
\begin{align}\label{eq: restricted_survival}
    P\{Y_{k_E}^{\overline{a}_{k_E}, {c}_{k_E-1}= 0}=0\}, \quad k_E \coloneqq \min(\mathcal{T}-E, k)
\end{align}
which differs from  $S_k^{\overline{a}_k}$:
\begin{align}\label{eq: bias_wscp_y1}
      \expectation\left( \mathbb{1}_{\{E+k > \mathcal{T}\}} \cdot (S_{\mathcal{T}-E}^{\overline{a}_{\mathcal{T}-E}}(E)- S_{k}^{\overline{a}_k}(E))\right)= S_k^{*,\textnormal{km-strat}, \overline{a}_k} - S_k^{\overline{a}_k} \geq 0,
\end{align}
where 
\begin{align}\label{eq: survival_conditional}
    S_{k}^{\overline{a}_k}(e) \coloneqq P(Y_k^{\overline{a}_k, {c}_{k-1}= 0}=0\mid E=e).
\end{align}
We establish this result in the following section within a discrete-time counterfactual framework (see Proposition \ref{proposition: bias_discrete_wsc}). We also establish this result within a counting-process continuous-time framework in Proposition \ref{proposition: staggered_entry_mortal_time_bias} of Section \ref{sec_app: continuous_time_wsc} in this appendix.

Expression \eqref{eq: bias_wscp_y1} can be interpreted as the expectation over $E$ of the difference between $S_{k}^{*,\textnormal{km}, \overline{a}_k}(E)$ and $S_{k}^{\overline{a}_k}(E)$. This difference will be positive when one naively extrapolates beyond the last observable follow-up $ \mathcal{T}-E$, to all subsequent follow-up $k$ ($S_{k}^{*,\textnormal{km}, \overline{a}_k}(E) = S_{\mathcal{T}-E}^{*,\textnormal{km}, \overline{a}_k}(E)$ for $k \in \{ \mathcal{T}-E + 1, \dots, \mathcal{T}\}$). That is, $S_k^{*, \textnormal{km-strat}, \overline{a}_k}$ exceeds $S_k^{\overline{a}_k}$ unless 
{\small
\begin{multline} \label{cond: immort_time_period}
\textnormal{on } \{E=e\}, \quad
Y_{\mathcal{T}-e}^{\overline{a}_{\mathcal{T}-e}, \overline{c}_{\mathcal{T}-e-1}= 0} =0 \Rightarrow Y_{k}^{\overline{a}_{k}, \overline{c}_{k-1}= 0} = 0, \quad \textnormal{a.s. } \textnormal{for every } k > \mathcal{T} - e, \\
    \textnormal{for all $e\in \{0, \dots, \mathcal{T}\}$}, \tag{\textnormal{IMT}}
\end{multline}
}
i.e., unless, for every $e$, event-free individuals $\{E=e\}$ at time $k=\mathcal{T}-e$ remain so at subsequent follow-up times. Such a requirement is arguably incompatible with any plausible scientific inquiry. In our running example, for $k=K=2$ and say, $\overline{a}_2=1$, it would require that, in a world where every individual received zidovudine and remained uncensored throughout follow-up time $k=2$, both the proportion of individuals who entered at calendar time $\tau=1$, alive at $k=1$ and able to experience failure at follow-up $k=2$, and the proportion of individuals who entered at calendar time $\tau=2$, alive at $k=0$
and who could experience failure at follow-up $k=2$ must be zero.
That is, for \eqref{eq: bias_wscp_y1} to be zero, the individual-specific time since study-entry $E$ must determine the time interval within which an individual can experience failure; for this reason, we term \eqref{eq: bias_wscp_y1} \textit{study-entry-determined immortal time bias}. In practice, bias in the form of \eqref{eq: bias_wscp_y1} is to be expected if studies that correctly only reported survival curves for different calendar-time periods through follow-up $k$ such that $k \leq \mathcal{T}-E$ \citep{van_not_improving_2024, harrysson_temporal_2025, lian_survival_2025}, instead reported an average of $E$-stratum-specific-$(\mathcal{T}-E)$-truncated Kaplan-Meier curves. It will also manifest in procedures adopting the estimators proposed by \citep{hansen_estimating_2017}[Eq. 5, p. 8], which can be easily shown to target
\eqref{eq: wscp_y1} (see Corollary \ref{corollary: hansen_one_of_us} in Appendix \ref{sec_app: continuous_time_wsc}), if the estimand of interest considered is survival (e.g., $S_k^{\overline{a}_k}$), and not $\mathcal{T}-E$- restricted survival (as in \eqref{eq: restricted_survival}). A trivial instance of bias similar to \textit{study-entry-determined immortal time bias} in a real-world analysis can be found in Appendix \ref{sec_app: case_study_2}.

\subsubsection{Technical results - discrete-time counterfactual framework}
\label{sec_app: wsc_discrete}
We establish results for the setting and data structure presented in the main text (\hyperref[table: conf_8]{\textbf{VIII}} in Table \ref{table: settings_and_data_structures} in Appendix \ref{sec_app: censoring_assumptions_discrete_time}). Our proofs make clear that generalizations to other settings (e.g., point treatment settings) can be established analogously and straightforwardly. Throughout, we consider weights as defined in \eqref{eq: general_weights} in Appendix \ref{sec_app: meaning_of_law_generated_by_weights}.

Consider an i.i.d. sample of $n$ individuals drawn from a common observed law $P$. Let 
\begin{align}
    \hat{S}_k^{*, \textnormal{km}, \overline{a}_k} \coloneqq& \prod_{m=0}^{k} \frac{{P}_n \left(W_m \cdot \mathbb{1}_{\{C_{m-1}=0, \overline{A}_m = \overline{a}_m, Y_{m} = 0\}}\right)}{{P}_n \left(W_m \cdot \mathbb{1}_{\{Y_{m-1} = 0, C_{m-1}=0, \overline{A}_m = \overline{a}_m\}}\right)}, \label{eq: km_estimator}\\
    \hat{S}_k^{*, \textnormal{km-strat}, \overline{a}_k} \coloneqq& {P}_n (\hat{S}_{k}^{*, \textnormal{km}, \overline{a}_k}(E))\label{eq: km_strat_estimator}
\end{align}
where $P_n$ denotes the empirical measure,\footnote{${P}_n(O) \equiv n^{-1} \cdot \sum_{i=1}^n O_i$ with $O_i$ iid draws from $P(O)$} $\hat{S}_{k}^{*, \textnormal{km}, \overline{a}_k}(e)$ the Kaplan-Meier estimator \eqref{eq: km_estimator} restricted to individuals $i$ such that $E_i = e$.

\begin{proposition}[WSC bias: discrete-time]
\label{proposition: bias_discrete_wsc}
Suppose $\RAID$ holds. Suppose that
\begin{enumerate}
    \item for $\hat{S}_k^{*, \textnormal{km}, \overline{a}_k}$, the weights ${\{W_k\}}_{k=0}^K$ are chosen so that $\phi_k(\overline{a}_k, \cdot)$ depends only on $\overline{a}_k$ ($V=\emptyset$), and
    \item for $\hat{S}_k^{*, \textnormal{km-strat}, \overline{a}_k}$, the weights ${\{W_k\}}_{k=0}^K$ are chosen so that $\phi_k(\overline{a}_k, \cdot)$ depends on $\overline{a}_k$ at most additionally on $e$ ($V\in E$).
\end{enumerate}
 Then for every $k \in \{0, \dots, K\}$, 
\begin{align}
    \norminfdiscrete{\hat{S}_m^{*, \textnormal{km}, \overline{a}_m} - S_m^{\overline{a}_m}} &\overset{P}{\longrightarrow} \norminfdiscrete{\prod_{j=0}^m(1 - h_{\prstudyp,j}^{\overline{a}_j} ) - \prod_{j=0}^m(1 - h_{j}^{\overline{a}_j} )}, \label{eq_app: bias_wscp_km}\\
    \norminfdiscrete{\hat{S}_m^{*, \textnormal{km-strat}, \overline{a}_m} - S_m^{\overline{a}_m}} &\overset{P}{\longrightarrow} \expectation \left(\mathbb{1}_{\{E+k > \mathcal{T}\}} \cdot (S_{\mathcal{T}-E}^{\overline{a}_{\mathcal{T}-E}}(E)- S_{k}^{\overline{a}_k}(E))\right).\label{eq_app: bias_wscp_km_strat}
\end{align}
\end{proposition}
\begin{proof}
Let $k \in \{0, \dots, K\}$ and fix $\overline{a}_k \in \{\overline{0},\overline{1}\}$. 

For $(\hat{S}_k^{*,\overline{a}_k},S_k^{*,\overline{a}_k}) \in \{(\hat{S}_k^{*, \textnormal{km}, \overline{a}_k}, {S}_k^{*,\textnormal{km}, \overline{a}_k}), (\hat{S}_k^{*, \textnormal{km-strat}, \overline{a}_k}, {S}_k^{*,\textnormal{km-strat},\overline{a}_k})\}$, from the LLN and continuous mapping theorem \footnote{For $\hat{S}_k^{*,\textnormal{km-strat},\overline{a}_k}$, the sample of size $n$ (with $n \to \infty$) is split into two halves: one used to estimate $\hat{S}_{k}^{*, \textnormal{km}, \overline{a}_k}(e)$ and the other to evaluate ${P}_n (\hat{S}_{k}^{*, \textnormal{km}, \overline{a}_k}(E))$.} it follows that 
\begin{align*}
    \hat{S}_k^{*,\overline{a}_k}&\overset{P}{\longrightarrow} S_k^{*,\overline{a}_k}, \textnormal{ and } \\
    \norminfdiscrete{\hat{S}_m^{*,\overline{a}_m} - S_m^{\overline{a}_m}} &\overset{P}{\longrightarrow} \norminfdiscrete{S_m^{*,\overline{a}_m} - S_m^{\overline{a}_m}}.
\end{align*}
Thus, it suffices to study ${S}_k^{*,\textnormal{km}, \overline{a}_k} - S_k^{\overline{a}_k}$ and ${S}_k^{*,\textnormal{km-strat},\overline{a}_k} - S_k^{\overline{a}_k}$.

To establish \eqref{eq_app: bias_wscp_km},
\begin{align*}
    {S}_k^{*,\textnormal{km}, \overline{a}_k} =& \prod_{m=0}^{k} \frac{\expectation \left(W_m \cdot \mathbb{1}_{\{C_{m-1}=0, \overline{A}_m = \overline{a}_m, Y_{m} = 0\}}\right)}{\expectation \left(W_m \cdot \mathbb{1}_{\{Y_{m-1} = 0, C_{m-1}=0, \overline{A}_m = \overline{a}_m\}}\right)} \\
    =& \prod_{m=0}^k (1 - h_m^{*, \overline{a}_m})&\Leftarrow{(i)}\\
    =& \prod_{m=0}^k(1 - h_{\prstudyp,m}^{\overline{a}_m} ) &\Leftarrow{(ii)}.
\end{align*}
where:
\begin{align*}
     (i) \equiv& \textnormal{def. }h_m^{*, \overline{a}_m} \coloneqq P^*\left( Y_{m} = 1 \mid Y_{m-1} = 0, \overline{A}_{m} = \overline{a}_m, C_{m-1} = 0, V=v \right), \\
     (ii) \equiv& \textnormal{Corollary } \eqref{corollary: kiss_when_e_not_in_v}.
\end{align*}
Next, consider:
{\footnotesize
\begin{align}
    S_{k}^{*,\textnormal{km}, \overline{a}_k}(e) =&\prod_{m=0}^{\min(k, \mathcal{T}-e)} \frac{\expectation \left(W_m \cdot \mathbb{1}_{\{C_{m-1}=0, \overline{A}_m = \overline{a}_m, Y_{m} = 0\}}\mid E=e\right)}{\expectation \left(W_m \cdot \mathbb{1}_{\{Y_{m-1} = 0, C_{m-1}=0, \overline{A}_m = \overline{a}_m\}}\mid E=e\right)} & \Leftarrow{\textnormal{def. and }\textnormal{Lemma }\ref{lemma: secc_is_preserved_after_iptw}}\notag\\
    =&\prod_{m=0}^{\min(k, \mathcal{T}-e)} (1 - h_m^{*, \overline{a}_m}(e))  \notag\\
    =& 1 - \sum_{m=0}^{\textnormal{min}(k, \mathcal{T}-e)} h_m^{*, \overline{a}_m}(e) \prod_{j=0}^{m-1}(1 - h_j^{*, \overline{a}_j}(e)) &\Leftarrow{(iii)}\label{id_proof_1}
\end{align}
}
and similarly, 
\begin{align}
    S_{k}^{\overline{a}_k}(e) 
    =&\prod_{m=0}^{k}(1 - h_m^{\overline{a}_m}(e)) &\Leftarrow{\textnormal{algebra}}\notag\\
    =&1 - \sum_{m=0}^{k} h_m^{\overline{a}_m}(e) \prod_{j=0}^{m-1}(1 - h_j^{\overline{a}_j}(e)) &\Leftarrow{(iii)} \label{id_proof_2}
\end{align}
where $(iii) \equiv [\textnormal{telescoping }\prod_{m=0}^p (1 - x_m) = 1 - \sum_{m=0}^p x_m \prod_{j=0}^{m-1} (1 - x_j) \textnormal{ for any } x_j \in [0,1] ], p > 0$ and where $\prod_{m=0}^{-1}(\cdot)=1$.
Then,
{\footnotesize
\begin{align*}
    &S_{k}^{*,\textnormal{km-strat}, \overline{a}_k} - S_{k}^{\overline{a}_k} \\
    =& \expectation(S_{k}^{*,\textnormal{km}, \overline{a}_k}(E) - S_{k}^{\overline{a}_k}(E))\\
    =& \expectation \left[\sum_{m=0}^{k} h_m^{\overline{a}_m}(E) \prod_{j=0}^{m-1}(1 - h_j^{\overline{a}_j}(E)) -  \sum_{m=0}^{\textnormal{min}(k, \mathcal{T}-E)} h_m^{*,\overline{a}_m}(E) \prod_{j=0}^{m-1}(1 - h_j^{*,\overline{a}_j}(E))\right] & \Leftarrow{(iv)} \\
    =& \expectation \left[\sum_{m=0}^{k} h_m^{\overline{a}_m}(E) \prod_{j=0}^{m-1}(1 - h_j^{\overline{a}_j}(E)) -  \sum_{m=0}^{\textnormal{min}(k, \mathcal{T}-E)} h_m^{\overline{a}_m}(E) \prod_{j=0}^{m-1}(1 - h_j^{\overline{a}_j}(E)) \right]& \Leftarrow{(v)} \\
    =&\expectation \left(\mathbb{1}_{\{E+k > \mathcal{T}\}}\left(S_{\mathcal{T}-E}^{\overline{a}_{\mathcal{T}-E}}(E) - S_{k}^{\overline{a}_{k}}(E)\right)\right) & \Leftarrow{(vi)},
\end{align*}
}
where:
\begin{align*}
    (iv) \equiv& \eqref{id_proof_1}, \eqref{id_proof_2}, \\
    (v) \equiv& \textnormal{Lemma }\ref{lemma: kiss}, \\
    (vi) \equiv& \eqref{id_proof_2}, \textnormal{ and } \pm 1.
\end{align*}
Finally, \eqref{eq_app: bias_wscp_km_strat} follows from noting that $k \mapsto \mathbb{1}_{\{E+k > \mathcal{T}\}}\left(S_{\mathcal{T}-E}^{\overline{a}_{\mathcal{T}-E}}(E) - S_{k}^{\overline{a}_{k}}(E)\right)$ is increasing in $k$ a.s.
\end{proof}

\subsubsection{Technical results - continuous-time counting process framework}
\label{sec_app: continuous_time_wsc}
We preliminarily point the reader toward the notation presented in Appendix \ref{sec_app: continuous_time_ct} and the assumptions formulated in Appendix \ref{sec_app: continuous_time_secc}. Contrary to the discrete-time notation and in keeping with standard conventions, round brackets following a survival symbol enclose the time to event rather than the subpopulation.

Here, we establish our results under a law $P^c$ for which $\RNPcont$ holds. Accordingly, our results apply to any law $P^c \equiv P^{*,c}(P^{c,i})$ derived from an initial law $P^{c,i}$ obtained by weighting, since, analogously to the discrete-time case, \eqref{cond: secc_continuous_time} is preserved under weighting. In the interest of space, we present the continuous-time analogue of \eqref{eq_app: bias_wscp_km_strat}, noting that the analogue of \eqref{eq_app: bias_wscp_km} can also be established. 

In the following, we work toward establishing the following proposition.

\begin{proposition}[Study-entry-determined immortal time bias: continuous-time]
Let  $t \in [0, \mathrm{T}]$. Suppose $\RNPcont$ holds. Then,  under mild regularity conditions
\begin{align*}
    \norminf{\hat{S}(s)- S(s)} \overset{P^c}{\longrightarrow} \expectation\left(\mathbb{1}_{\{t > t_{E}\}} \cdot \left(S_E(t_E) - S_E(t)\right)\right).
\end{align*}
\label{proposition: staggered_entry_mortal_time_bias}
\end{proposition}

To prove Proposition \ref{proposition: staggered_entry_mortal_time_bias}, we preliminarily introduce the following objects. For individuals with $\{E=e\}$, let $S_e(t)$, $\hat{S}_e(t), S_e^{-}(t), S_e^{(n)-}(t)$ denote, respectively, survival, its Kaplan-Meier estimator, the survival truncated at $t_e$, and the survival truncated at the maximum observed time $\Tilde{T}_{e, (n)} \coloneqq \max_{i \in \{1, \dots, n\}, E_i=e}\{\Tilde{T}_{i}\}$:
{\scriptsize
\begin{align}
    S_e(t) &\coloneqq P(T>t \mid E=e), \\
    N_e(t) &\coloneqq \#\{ i: \Tilde{T}_i \leq t \textnormal{, not censored } E_i = e\} & \left(\textnormal{Observed counting process for $\{E_i=e\}$}\right), \nonumber \\
    Y_e^{(n)}(t) &\coloneqq \# \{i : \Tilde{T}_i \geq t \textnormal{ and } E_i=e\} & \left(\textnormal{Num. of individuals at risk at time $t$ for $\{E_i = e\}$}\right), \nonumber  \\
    J_e^{(n)}(t) &\coloneqq \mathbb{1}_{\{\Tilde{T}_{e,(n)} \geq t \}} & \left(\textnormal{Max. obs.time for ind. with $E = e$ exceeds $t$}\right),\\
    \hat{A}_e^{\circ}(t) &\coloneqq \int_{\{0\leq s\leq t\}}  \frac{J_e^{(n)}(s)}{Y_e^{(n)}(s)}\cdot \textnormal{d}N_e(s), \\
    \hat{S}_e(t) &\coloneqq \Prodi\limits_{0 \leq s \leq t}\left(1 - \textnormal{d}\hat{A}_e^{\circ}(s)\right), \nonumber \\
    S_e^{-}(t) &\coloneqq \exp\left(-\int_{0\leq s \leq \min(t,t_e)}\alpha_e(s)\textnormal{d}s\right), \label{eq: e_truncated_cont_survival}\\
    S_e^{(n)-}(t) &\coloneqq \exp\left(-\int_{0\leq s \leq \min(t,\Tilde{T}_{e, (n)})}\alpha_e(s)\textnormal{d}s\right),
\end{align}
}
and let $S(t)$, $\hat{S}(t)$, $S^{-}(t)$, $S^{(n)-}(t)$ denote their marginalization with respect to $P(E)$:
\begin{align}
    S(t) &\coloneqq \expectation S_E(t), \notag\\
    \hat{S}^{\circ}(t) &\coloneqq \expectation \hat{S}_E(t), \notag\\
    {S}^{-}(t) &\coloneqq \expectation{S}_E^{-}(t), \label{eq: truncated_cont_survival}\\
    {S}^{(n)-}(t) &\coloneqq \expectation{S}_E^{(n)-}(t).\notag
\end{align}
Finally, let
\begin{align}
    \hat{S}(t) \coloneqq n^{-1} \cdot \sum_{i \in \{1, \dots, n\}} \hat{S}_{E_i}(t)
\end{align}
denote the adjusted Kaplan-Meier estimator, with $P^c(E=e)$ estimated nonparametrically as $\frac{\#\{E_i = e\}}{n}$.

\begin{lemma}\label{lemma: conv_c1}
Suppose $\RNPcont$ and mild regularity conditions hold. Then,
    \begin{align}
        \norminf{\hat{S}(s) - \hat{S}^{\circ}(s)} \overset{P^c}{\longrightarrow} 0, \label{conv: c1_1}\\
        \norminf{\hat{S}^{\circ}(s) - S^{(n)-}(s)} \overset{P^c}{\longrightarrow} 0, \label{conv: c1_2}\\
        \norminf{{S}^{(n)-}(s) - S^{-}(s)} \overset{P^c}{\longrightarrow} 0. \label{conv: c1_3}
    \end{align}
    \begin{proof}
    Let $\epsilon_1, \epsilon_2 > 0$.
       We start by proving \eqref{conv: c1_1}.
       {\scriptsize
       \begin{align*}
           &P^c\left\{\norminf{\hat{S}(s) - \hat{S}^{\circ}(s)} > \epsilon_2\right\} \\
           =& P^c\left\{\norminf{\sum_{e \in \{0, \dots, \mathcal{T}\}} \left(P^c(E=e) - \frac{\#\{E_i = e\}}{n}\right) \cdot \hat{S}_e(s)} > \epsilon_2 \right\} & \Leftarrow{ \textnormal{Def. of $\hat{S}(s)$ and $\hat{S}^{\circ}(s)$}}\\
           \leq& P^c\left\{\sum_{e \in \{0, \dots, \mathcal{T}\}} \norminf{\left(P^c(E=e) - \frac{\#\{E_i = e\}}{n}\right) \cdot \hat{S}_e(s)} > \epsilon_2 \right\} & \Leftarrow{ \norminf{\expectation X} \leq \expectation\norminf{X}}\\
           \leq& P^c\left\{\sum_{e \in \{0, \dots, \mathcal{T}\}} \left|P^c(E=e) - \frac{\#\{E_i = e\}}{n}\right| > \epsilon_2\right\} & \Leftarrow{ S_e(t) \leq 1} \\
           \leq& P^c\bigcup_{e \in \{0, \dots, \mathcal{T}\}} \left\{\left| P^c(E=e) - \frac{\#\{E_i = e\}}{n}\right| > \frac{\epsilon_2}{|\{0, \dots, \mathcal{T}\}|} \right\} \\
           \leq& \sum_{e \in \{0, \dots, \mathcal{T}\}} P^c\left\{\left| P^c(E=e) - \frac{\#\{E_i = e\}}{n}\right| > \frac{\epsilon_2}{|\{0, \dots, \mathcal{T}\}|}\right\}& \Leftarrow{ \textnormal{Union bound}}.
       \end{align*}
       }
       Then, since for every $e$, $\frac{\#\{E_i = e\}}{n} \overset{P^c}{\longrightarrow}P^c(E=e)$ by the LLN, for every $e$ and for every $\delta_1, \delta_2 > 0$ there exists a $n_e^{\delta_1,\delta_2}$ such that for $n \geq n_e^{\delta_1,\delta_2}$, $P^c\left\{\left|P^c(E=e) - \frac{\#\{E_i = e\}}{n}\right| > \delta_2\right\} < \delta_1$.
       
       Define $n^{\vee} = \max \left( n_0^{\left(\frac{\epsilon_1}{\left|\{0, \dots, \mathcal{T}\}\right|},\frac{\epsilon_2}{\left|\{0, \dots, \mathcal{T}\}\right|}\right)}, \dots, n_{\mathcal{T}}^{\left(\frac{\epsilon_1}{\left|\{0, \dots, \mathcal{T}\}\right|},\frac{\epsilon_2}{\left|\{0, \dots, \mathcal{T}\}\right|}\right)}\right)$; then if $n \geq n^{\vee}$
    \begin{multline*}
        \sum_{e \in \{0, \dots, \mathcal{T}\}} P^c\left\{\left| P^c(E=e) - \frac{\#\{E_i = e\}}{n}\right| > \frac{\epsilon_2}{|\{0, \dots, \mathcal{T}\}|}\right\} \\
        < \sum_{e \in \{0, \dots, \mathcal{T}\}} \frac{\epsilon_1}{|\{0, \dots, \mathcal{T}\}|} \\
        = \epsilon_1,
    \end{multline*}
    which proves \eqref{conv: c1_1}.
    
    To establish \eqref{conv: c1_2} and \eqref{conv: c1_3},  consider the following.
    We have
       {\scriptsize
        \begin{align*}
            \norminf{\hat{S}^{\circ}(s) - S^{(n)-}(s)} =&\norminf{\expectation\left(\hat{S}_E(s) - S_E^{(n)-}(s)\right)} & \Leftarrow{ \textnormal{Def. $\hat{S}^{\circ}(s)$ and $S^{(n)-}(s)$}}\\
            \leq&\expectation\norminf{\hat{S}_E(s) - S_E^{(n)-}(s)}& \Leftarrow{ \norminf{\expectation X} \leq \expectation\norminf{X}} \\
            =&\expectation\left(\norminf{\left(\frac{\hat{S}_E(s)}{S_E^{(n)-}(s)} -1 \right) \cdot S_E^{(n)-}(s)}\right) \\
            \leq&\expectation\left(\norminf{\frac{\hat{S}_E(s)}{S_E^{(n)-}(s)} -1}\right) & \Leftarrow{ S_E^{(n)-}(s) \leq 1}
        \end{align*}
        }
        and, using analogous arguments, we have
        \begin{align*}
             \norminf{{S}^{(n)-}(s) - S^{-}(s)} \leq \expectation\left(\norminf{\frac{S_E^{-}(s)}{S_E^{(n)-}(s)} - 1}\right). 
        \end{align*}
        Because for every $e$ and every $n$ we have that $t_e \geq \Tilde{T}_{e,(n)}$, if $s > t_e$, then $\frac{\hat{S}_e(s)}{S_e^{(n)-}(s)}$ and  $\frac{S_e^{-}(s)}{S_e^{(n)-}(s)}$
        are equal to $\frac{\hat{S}_e(t_e)}{S_e^{(n)-}(t_e)}$ and  $\frac{S_e^{-}(t_e)}{S_e^{(n)-}(t_e)}$, respectively. Thus, we have
        \begin{align*}
        \norminf{\frac{\hat{S}_e(s)}{S_e^{(n)-}(s)} - 1}&= \sup\limits_{s \in \left[0, t_e\right]}\left|\frac{\hat{S}_e(s)}{S_e^{(n)-}(s)} -1\right|, \\
            \norminf{\frac{S_e^{-}(s)}{S_e^{(n)-}(s)} - 1}&= \sup\limits_{s \in \left[0, t_e\right]}\left|\frac{S_e^{-}(s)}{S_e^{(n)-}(s)} - 1\right|
        \end{align*}
        for every $e$. 
        Under mild regularity conditions, $\RNPcont$ implies that $\inf_{s \in \left[0, t_e\right]} Y_e^{(n)}(s) \overset{P^c}{\longrightarrow} \infty$ which in turn implies the following conditions (see p. 190 \citep{andersen_statistical_1993}) 
\begin{align*}
    \int_{\{0 \leq s \leq t_e\}} \frac{J_e^{(n)}(s)}{Y_e^{(n)}(s)} \cdot \alpha_e(s) \cdot \textnormal{d}s \overset{P^c}{\longrightarrow} 0 \textnormal{ and }
    \\
    \int_{\{0 \leq s \leq t_e\}} \left(1 - J_e^{(n)}(s)\right) \cdot \alpha_e(s) \cdot \textnormal{d}s \overset{P^c}{\longrightarrow} 0 
\end{align*}
hold. \footnote{In words, the first condition ensures that as the sample size $n \to \infty$, the number of individuals at risk, $Y_e^{(n)}(s)$, grows to infinity, driving the asymptotic variance of the estimator to zero. The second condition ensures that the maximum observed time $\tilde{T}_{e,(n)}$ approaches its upper bound $t_e$ as $n \to \infty$. Together, these conditions guarantee that the observation window expands to fully cover $[0, t_e]$ and that we have infinitely many observations remaining at risk up to $t_e$. See \citep{andersen_statistical_1993} for more details.} Thus, we can directly apply the results of page 262 in \citep{andersen_statistical_1993}. For any $e \in \{0, \dots, \mathcal{T}\}$, result 4.3.12 implies 
\begin{align*}
    \norminf{\frac{\hat{S}_e(s)}{S_e^{(n)-}(s)} -1} \overset{P^c}{\longrightarrow} 0,
\end{align*}
and result 4.3.13 implies
\begin{align*}
    \norminf{\frac{S_e^{-}(s)}{S_e^{(n)-}(s)} - 1} \overset{P^c}{\longrightarrow} 0.
\end{align*}

Thus, since the $E=e$ specific conditional supremum norms converge to zero in probability, their finite weighted sum (i.e., its expectation over $P(E)$) also converges to zero in probability. Because this sum upper-bounds our target quantities, it establishes \eqref{conv: c1_2} and \eqref{conv: c1_3}.
\end{proof}

\end{lemma}
We are now ready to prove Proposition \ref{proposition: staggered_entry_mortal_time_bias}.
\begin{proposition*}[Re-statement of Proposition \ref{proposition: staggered_entry_mortal_time_bias}]
Let  $t \in [0, \mathrm{T}]$. Suppose $\RNPcont$ holds. Then, under mild regularity conditions
\begin{align*}
    \norminf{\hat{S}(s)- S(s)} \overset{P^c}{\longrightarrow} \expectation\left(\mathbb{1}_{\{t > t_{E}\}} \cdot \left(S_E(t_E) - S_E(t)\right)\right).
\end{align*}
\begin{proof}  
Fix $t \in [0, \mathrm{T}]$, $s \in \left[0, t\right]$ and let 
\begin{align*}
    \circleddash_{1}^{(n)}(s) &\coloneqq \left(\hat{S}(s) - \hat{S}^{\circ}(s)\right) + \left(\hat{S}^{\circ}(s) - S^{(n)-}(s)\right) + \left(S^{(n)-}(s) - S^{-}(s)\right), \\
    \circleddash_{2}(s) &\coloneqq S^{-}(s) - S(s).
\end{align*}
Since $\circleddash_{1}^{(n)}(s) + \circleddash_{2}(s) = \hat{S}(s)- S(s)$, to prove our result, it suffices to show that for every $\epsilon_1, \epsilon_2 > 0$ there exists a $n^{(\epsilon_1, \epsilon_2)} \in \mathbb{N}_{+}$ such that
\begin{align*}
    n \geq n^{(\epsilon_1, \epsilon_2)} \Rightarrow P^c\left\{ \left| \norminf{\circleddash_{1}^{(n)}(s) +  \circleddash_{2}(s)} - \circleddash_{2}(t) \right| > \epsilon_2\right\} < \epsilon_1.
\end{align*}
Let us preliminarily study $\circleddash_{2}(s)$:
\begin{align*}
    \circleddash_{2}(s)
    =& \expectation \left(S_E^{-}(s) - S_E(s)\right)
    &\Leftarrow{ \textnormal{def. $S^{-}(s),S(s)$}}\\
    =& \expectation \left(
    \exp\left\{-\int_{0 \leq u \leq \min(s,t_E)}
    \alpha_E(u)\textnormal{d}u\right\}
    -
    \exp\left\{-\int_{0 \leq u \leq s}
    \alpha_E(u)\textnormal{d}u\right\}
    \right)
    &\Leftarrow{ \textnormal{def. $S_e^{-}(s),S_e(s)$}}\\
    =& \expectation \left(
    \mathbb{1}_{\{s \leq t_E\}} \cdot 0
    +
    \mathbb{1}_{\{s > t_E\}}
    \left(S_E(t_E) - S_E(s)\right)
    \right)\\
    =& \expectation \left(
    \mathbb{1}_{\{s > t_E\}}
    \left(S_E(t_E) - S_E(s)\right)
    \right).
\end{align*}
For every $e \in \{0, \dots, \mathcal{T}\}$ the function
\begin{align*}
    s \mapsto \mathbb{1}_{\{s > t_e\}} \left(S_e(t_e) - S_e(s)\right)
\end{align*}
is non-negative and non-decreasing in $\left[0, t\right]$. It follows that its supremum and so that of $ \left|\circleddash_{2}(s)\right|$ which is equal to $\circleddash_{2}(s)$, is attained at $t$.

Now, 
{\scriptsize
\begin{align*}
    &\left\{\left|\norminf{\hat{S}(s)- S(s)} - \expectation\left(\mathbb{1}_{\{t > t_{E}\}} \cdot \left(S_E(t_E) - S_E(t)\right)\right)\right| > \epsilon_2\right\} \\
    \equiv &\left\{ \left| \norminf{\circleddash_{1}^{(n)}(s) +  \circleddash_{2}(s)} - \circleddash_{2}(t) \right| > \epsilon_2 \right\}\\
    = &\left\{ \left| \norminf{\circleddash_{1}^{(n)}(s) +  \circleddash_{2}(s)} - \norminf{\circleddash_{2}(s)} \right| > \epsilon_2 \right\} \quad \Leftarrow{\norminf{\circleddash_{2}(s)} = \circleddash_{2}(t)} \\
    \subseteq & \left\{\norminf{ \circleddash_{1}^{(n)}(s)} > \epsilon_2\right\} \quad \Leftarrow{\textnormal{reverse triangle ineq.: } \left| \norminf{X+Y} - \norminf{Y} \right| \leq \norminf{X}} \\
    \subseteq & \left\{\underbrace{\norminf{\hat{S}(s) - \hat{S}^{\circ}(s)}}_{\equiv \circleddash_{1}^{\textsc{I}}} >\frac{\epsilon_2}{3}\right\} \bigcup\left\{\underbrace{\norminf{\hat{S}^{\circ}(s) - S^{(n)-}(s)}}_{\equiv \circleddash_{1}^{\textsc{II}}}>\frac{\epsilon_2}{3}\right\}\bigcup\left\{\underbrace{\norminf{S^{(n)-}(s) - S^{-}(s)}}_{\equiv \circleddash_{1}^{\textsc{III}}}>\frac{\epsilon_2}{3} \right\}. 
\end{align*}
}
Thus,
\begin{multline*}
    P^c\left\{ \left| \sup\limits_{s \in \left[0, t\right]} \left| \circleddash_{1}^{(n)}(s) +  \circleddash_{2}(s)\right| - \circleddash_{2}(t) \right| > \epsilon_2\right\} \\
    \leq P^c\left\{ \circleddash_{1}^{\textsc{I}} >\frac{\epsilon_2}{3}\right\}+ P^c\left\{\circleddash_{1}^{\textsc{II}}>\frac{\epsilon_2}{3}\right\} + P^c\left\{\circleddash_{1}^{\textsc{III}}>\frac{\epsilon_2}{3} \right\}.
\end{multline*}
From Lemma \ref{lemma: conv_c1} it follows that for every 
$\circleddash \in \{\circleddash_{1}^{\textsc{I}}, \circleddash_{1}^{\textsc{II}}, \circleddash_{1}^{\textsc{III}}\}$ there exists a $n_{\circleddash}^{\left(\frac{\epsilon_1}{3},\frac{\epsilon_2}{3}\right)}$ such that $n \geq n_{\circleddash}^{\left(\frac{\epsilon_1}{3},\frac{\epsilon_2}{3}\right)} \Rightarrow P^c\left\{ \circleddash >\frac{\epsilon_2}{3}\right\} < \frac{\epsilon_1}{3}$; take $n^{\epsilon_1, \epsilon_2} \coloneqq \max\left(n_{\circleddash_{1}^{\textsc{I}}}^{\left(\frac{\epsilon_1}{3},\frac{\epsilon_2}{3}\right)},n_{\circleddash_{1}^{\textsc{II}}}^{\left(\frac{\epsilon_1}{3},\frac{\epsilon_2}{3}\right)}, n_{\circleddash_{1}^{\textsc{III}}}^{\left(\frac{\epsilon_1}{3},\frac{\epsilon_2}{3}\right)} \right)$ to conclude.
\end{proof}

\end{proposition*}

\begin{corollary}\label{corollary: hansen_one_of_us}
   Under the assumptions of Proposition 
   \ref{proposition: staggered_entry_mortal_time_bias}, the estimator \citep{hansen_estimating_2017}[Eq. 5, p. 8] targets $S^-(t)$.
   \begin{proof}
       First, $\hat{S}_w(t)$ Eq. 5, p. 8 \citep{hansen_estimating_2017} is $\hat{S}(t)$: their $B$ is our $E$, and therefore $\frac{n_i}{n}$ is our 
       \begin{align*}
           \frac{\#\{E_i = i-1 \equiv i^{\textnormal{th}} \textnormal{ (in the increasing order) value in } \{0, \dots, \mathcal{T}\}\}}{n},
       \end{align*}
       their $\hat{S}_i(t)$ is our $\hat{S}_{E=(i-1)}(t)$, and their $k$ is our $\mathcal{T}+1$. The result then follows immediately using the steps in the proof of Proposition 
   \ref{proposition: staggered_entry_mortal_time_bias}.
   \end{proof}
\end{corollary}

\textbf{Remarks on the form and origin of the bias.}
Reasoning about $\hat{S}_{e}(t)$ permits pinpointing the origins of the bias term. Fix $t > 0$. For all those $e$ such that $t > t_e$ the contribution $\alpha_e(s)\textnormal{d}s$ in $S_e^{-}(s)$ is absent because the risk set is empty: from \eqref{cond: secc_ramifications_continuous_time} it follows that no individual is at risk after $t_e$. Since $\hat{S}_e(t)$ targets $S_{e}^{-}$ and $S_e(t) < S_{e}^{-}(t)$ the terms $S_e^{-}(t) - S_e(t)$ are not accounted for in the expectation for all $e$ such that $t > t_e$. Authoritative guidelines have warned about this issue ``which hopefully hardly ever happens if nonparametric estimation of $A(t)$ \textit{(cumulative hazard and thus survival)} is to be meaningful.'' \citep{andersen_statistical_1993}. \footnote{Italics have been added by the present writer.}     

The results of Proposition \ref{proposition: staggered_entry_mortal_time_bias} can be interpreted as follows. In words, the supremum of the absolute difference between the Kaplan-Meier estimates $\hat{S}(t)$ and survival $S(t)$, as the sample size $n$ tends to infinity, comes arbitrarily close to the sum over all $e$ (such that $e+t$ exceeds the calendar time of administrative end of the study $\mathcal{T}$) of the difference between survival at time $\mathcal{T}-e$ and that at time $t$ for the individuals $\{E=e\}$ weighted by the corresponding population proportions. The bias term
\begin{align*}
    \expectation\left(\mathbb{1}_{\{t > t_{E}\}} \cdot \left(S_E(t_E) - S_E(t)\right)\right)
\end{align*}
is non-decreasing in $t$. For a given $t$, it is equal to zero if and only if
\begin{align} \label{eq: staggered_entry structural_failure_condition}
    S_e(t_e)=S_e(t)
    \quad \textnormal{for $P(E)$-a.s. $e$ such that $t>t_e$},
\end{align}
that is, if there is no additional failure probability over $(t_e,t]$ for the relevant $E=e$ strata. A sufficient, but considerably stronger, condition is
\begin{align*}
    [0,t_e]=\textnormal{supp}(T\mid E=e)
    \quad\textnormal{for every $e$}.
\end{align*}

\newpage
\subsection{MSM-hazard procedures}
\label{sec_app: msm_hazard}
Throughout, we denote by $P^*$ the law generated from $P$ by an arbitrary choice of weights compatible with the definition given in Section \ref{sec_app: meaning_of_law_generated_by_weights}. For simplicity, we assume that the weights are known, although the results extend to cases in which they are unknown. We also assume that \hi{} holds throughout. This simplifies the exposition and allows us to introduce the person-time law, which depends only on a single, i.e., point treatment, in the context of MSM-hazard procedures. We note, however, that the results of Proposition \ref{proposition: msm_hazard_characterization} can be straightforwardly extended without assuming \hi{}.

\textbf{Observed data structure}.  Let $K \in \mathbb{N}^{+}$ and let $P^{\textnormal{long}}$ denote the observed data law for
\begin{align}
    O_k = (L_0, C_{-1}, A_0, Y_0, \dots, L_k, C_{k-1}, A_{k}, Y_{k}), \quad A_m, Y_m, C_{m-1} \in \{0,1\},
    \label{def: observation_O}
\end{align}
where $C_{-1}=Y_{-1}=0$. $P^{\textnormal{long}}$ is nothing but $P$ where we added the superscript ``long'' to emphasize that observations are longitudinal. Throughout the remainder of this section, we will appeal to the following defining properties of the observed data structure: 
\begin{align}
    &C_{k-1}=1 \Rightarrow C_{k}=1; Y_k=1 \Rightarrow Y_{k+1}=1, \quad k \in \{0, \dots, K-1\} \label{def: monotonicity_cy}\\
    &C_{k}=1 \Rightarrow (L_{k+1}, A_{k+1}, Y_{k+1}) \textnormal{ cannot be observed},\quad k \in \{0, \dots, K-1\}; \\ 
    &C_{k-1}=Y_{k-1}=0 \textnormal{ and } Y_{k}=1 \Rightarrow C_{k^{'}}=0 \textnormal{ for } k^{'} \geq k \label{def: impossibility_11}.
\end{align}
We will also consider the stopping time
\begin{align}
    &K_{\textnormal{st}} \coloneqq \begin{cases}
        K & C_{K-1}=Y_{K-1}=0 \\
        \textnormal{min}\{l \in \{0, \dots, K-1\} : C_{l} + Y_l > 0 \} & \textnormal{otherwise},
    \end{cases}
    \label{eq: stopping_time_K_star}
\end{align}
which equals the number of person-time observations minus 1. 

Consider a regression model $\mbps$ indexed by $\mathbb{B}^*$. Let $V \subseteq L_0$. \textbf{We say that MSM-hazard procedures target the parameter} $\beta^* \in \mathbb{B}^*$
if, for a given choice of the person-time score function, $s_{\textnormal{pool}}(v,a,k,y;\beta)$, the parameter $\beta^*$ satisfies the population score-equation
\begin{multline} \label{eq: population_pt_score_eq}
0 = {P^{*}\left(K_{\textnormal{st}} + 1\right)}^{-1} \cdot  P^{*} \left(\sum_{k \in \{0, \dots, K\}} \mathbb{1}_{\{K_{\textnormal{st}} \geq k\}} \cdot s_{\textnormal{pool}}(V, A_{k}, k, Y_{k}; \beta) \right). 
\end{multline}
We will consider score functions of the form:
\begin{multline}\label{eq: s_pool}
 s_{\textnormal{pool}}(v,a,k,y;\beta) \coloneqq\left(\frac{y}{p(v,a,k;\beta)} - \frac{1-y}{1-p(v,a,k;\beta)} \right) \nabla_{\beta}p(v,a,k;\beta), \\
p(v,a,k;\beta) \in \mbps
\end{multline}
which, for the popular regression models used in the applied literature
{\small
\begin{align}\label{eq: popular_working_model}
\mbps =& \{ p(v,a,k;\beta) \coloneqq \mu^{-1}(\lambda(v,a,k;\beta)) \\
:& \beta \in \mathbb{B}^*, \mu \textnormal{ link function},  \lambda \textnormal{ linear predictor}\},
\end{align}
}
for $\mu = \text{logit}$ simplifies as
\begin{align}\label{eq: score_popular}
    s_{\textnormal{pool}}(v,a,k,y;\beta) \coloneqq& (y - p(v,a,k; \beta)) \cdot \nabla_{\beta}\lambda(v,a,k; \beta). 
\end{align}

For example, the choice of $\mbps$ resulting in $s_{\textnormal{pool}}(v,a,k,y;\cdot):\beta \mapsto  (y - p(v,a,k; \beta)) \cdot \left[ 1, v, a, k\right]^{\textnormal{T}}$ constitutes the prototypical MSM-logistic models introduced by \citet{robins_marginal_2000}, used in the majority of analyses based on the target trial emulation framework, see e.g., \citep{hernan_observational_2008, dickerman_avoidable_2019}.

\clearpage

\subsubsection{Person-time law in the context of MSM-hazard procedures}
\label{sec_app: formalization_person_time_law}

We begin by establishing a set of preliminary lemmas. We then offer a definition of the person-time law $P^{\textnormal{pt}}$ and show how it can be used to characterize the parameter $\beta^*$ targeted by MSM-hazard procedures via a population score equation. 

For notational convenience throughout Section \ref{sec_app: formalization_person_time_law}, we use probability symbols both to denote probabilities of events and expectations with respect to the corresponding probability law. Probabilities of events are written using curly brackets (except for conditional probabilities, where round brackets are used), while expectations are written using round brackets. In particular, $P\{\cdot\} = P(\mathbb{1}_{\{\cdot\}})$.

\textbf{The following lemmas will be used in the subsequent subsections.}

\begin{lemma}[Stopping time: characterization 1]
\label{lemma: mutual_exclusive}
Let $k \in \{0, \dots, K\}$. Then,
{\footnotesize
\begin{align}\label{eq: characterization_1}
        &\{ K_{\textnormal{st}} = k \} \\
        =& \begin{cases}
            \{C_{k-1} = 0, Y_{k-1} = 0\} &k=K \\
            \{C_{k-1} = 0, Y_{k-1} = 0\}\bigcap \left(\{C_k = 1, Y_k = 0\} \bigcup \{C_k = 0, Y_k = 1\}\right) & \textnormal{otherwise}.
        \end{cases} 
\end{align}
}
\begin{proof}
For $k=K$ the result follows directly from the definition. Consider $k \in \{0, \dots, K-1\}$. To establish the result we note that 
    \begin{align*}
        &\{C_k = 1, Y_k = 0\} \bigcup \{C_k = 0, Y_k = 1\} \\
        =&\{C_k = 1, Y_k = 0\} \oplus \{C_k = 0, Y_k = 1\} \\
        \equiv&\textnormal{ either } \{C_k = 1, Y_k = 0\} \textnormal{ or } \{C_k = 0, Y_k = 1\},
    \end{align*}
    since the events are mutually exclusive and we consider the corresponding propositional formulation
    {\footnotesize
    \begin{align*}
            K_{\textnormal{st}} = k \Leftrightarrow \left[C_{k-1} = 0, Y_{k-1} = 0\right] \textnormal{ and either } \left[C_k = 1, Y_k = 0\right] \textnormal{ or } \left[C_k = 0, Y_k = 1\right]. 
    \end{align*}
    }
        For $``\Leftarrow''$. The result follows directly from \eqref{eq: stopping_time_K_star}.\\
        For $``\Rightarrow''$. Let $k \in \{0, \dots, K-1\}$. From \eqref{eq: stopping_time_K_star}, we have
        \begin{align*}
            K_{\textnormal{st}} = k &\Rightarrow \left[ Y_k = 0  \Rightarrow C_k = 1 \right] &, \\
            K_{\textnormal{st}} = k &\Rightarrow \left[{C}_{k-1} = 0 \textnormal{ and } Y_{k-1} = 0\right] & \\
            &\Rightarrow \left[ Y_k = 1  \Rightarrow C_k = 0 \right]&.
        \end{align*}
        Thus 
        {\small
        \begin{align*}
            K_{\textnormal{st}} = k \Rightarrow& \left[{C}_{k-1} = 0 \textnormal{ and } Y_{k-1} = 0\right] \textnormal{ and } \left[ \left[ Y_k = 0  \Rightarrow C_k = 1 \right] \textnormal{ and } \left[ Y_k = 1  \Rightarrow C_k = 0 \right] \right]\\
            \Leftrightarrow& \left[C_{k-1} = 0, Y_{k-1} = 0\right] \textnormal{ and either } \left[C_k = 1, Y_k = 0\right] \textnormal{ or } \left[ C_k = 0, Y_k = 1\right], 
        \end{align*}}
        which concludes the proof.
    \end{proof}
\end{lemma}
\begin{lemma}[Stopping time: characterization 2]
\label{lemma: stopping_time}
Let $k \in \{0, \dots, K\}$. Then,
\begin{align}
    & \{ K_{\textnormal{st}} \geq  k \} \label{equi_1: stopping_time} \\
    \Leftrightarrow& \{C_{k-1} = 0, Y_{k-1} = 0\}. \label{equi_2: stopping_time}
\end{align}
    \begin{proof}
        To prove \eqref{equi_2: stopping_time}, let us preliminarily denote $S_{(c,y)}^{k} \equiv \{C_{k} = c, Y_{k} = y\}$, $c,y \in \{0,1\}$. We argue by backward induction from $k=K$ to $k=0$. The base case $k=K$ follows from the definition of $K_{\textnormal{st}}$ since $\{K_{\textnormal{st}} \geq K\} = \{ K_{\textnormal{st}} = K \}$.
        Now, let $k \in \{0, \dots, K-1\}$.
        \begin{align*}
            &\{K_{\textnormal{st}} \geq k\} \\
            =& \{K_{\textnormal{st}} = k \} \bigcup\{K_{\textnormal{st}} \geq k+1\} &\Leftarrow{\textnormal{union decomposition}}\\
            =& \left(S_{(0,0)}^{k-1} \bigcap \left(S_{(0,1)}^{k} \bigcup S_{(1,0)}^{k} \right)\right) \bigcup\{K_{\textnormal{st}} \geq k+1\} &\Leftarrow{\textnormal{Lemma \ref{lemma: mutual_exclusive}}, \eqref{eq: characterization_1}}\\
            =& \left(S_{(0,0)}^{k-1} \bigcap \left(S_{(0,1)}^{k} \bigcup S_{(1,0)}^{k} \right)\right) \bigcup S_{(0,0)}^{k} &\Leftarrow{\textnormal{inductive hypothesis}}\\
            =& \left(S_{(0,0)}^{k-1} \bigcup S_{(0,0)}^{k}\right) \bigcap \left(S_{(0,0)}^{k} \bigcup S_{(0,1)}^{k} \bigcup S_{(1,0)}^{k}\right) & \Leftarrow \textnormal{Distributive law}\\
            =& S_{(0,0)}^{k-1} \bigcap \left(S_{(0,0)}^{k} \bigcup S_{(0,1)}^{k} \bigcup S_{(1,0)}^{k}\right) & \Leftarrow \left[S_{(0,0)}^{k-1}  \supseteq S_{(0,0)}^{k} \Leftarrow \eqref{def: monotonicity_cy}\right] \\
            =& S_{(0,0)}^{k-1} & \Leftarrow \left[\emptyset = S_{(1,1)}^{k}  \Leftarrow \eqref{def: impossibility_11}\right].
        \end{align*}
        This concludes the proof.
    \end{proof}
\end{lemma}

\begin{lemma}[Expectation identity]
\label{lemma: expectation_identity}
To lighten notation, let $P \equiv P^{\textnormal{long}}$.
\begin{align}
    \sum_{k \in \{0, \dots, K\}}P\{ K_{\textnormal{st}} \geq k\} = P(K_{\textnormal{st}} + 1)
\end{align}
\begin{proof}
Let us preliminarily denote $S_{(c,y)}^{k} \equiv \{C_{k} = c, Y_{k} = y\}$, $c,y \in \{0,1\}$. For $k^{\prime}, k^{\prime\prime} \in \{0, \dots, K\}$, $k^{\prime} < k^{\prime\prime}$, it can be shown that $\{K_{\textnormal{st}} = k^{\prime}\} \bigcap \{K_{\textnormal{st}} = k^{\prime\prime}\} = \emptyset$. Thus,
{\footnotesize    
\begin{align*}
        &\sum_{k \in \{0, \dots, K\}}P\{ K_{\textnormal{st}} \geq k\} \\ 
        =& \sum_{k \in \{0, \dots, K\}} P\bigcup_{l \in \{k, \dots, K\}}\{K_{\textnormal{st}} = l\} &\Leftarrow{(i)}\\
        =& \sum_{k \in \{0, \dots, K\}} P\{K_{\textnormal{st}} = k\} +  \dots + P\{K_{\textnormal{st}} = K\} &\Leftarrow{(ii)}\\
        =& \phantom{+}P\{K_{\textnormal{st}} = 0\} + P\{K_{\textnormal{st}} = 1\} + P\{K_{\textnormal{st}} = 2\} + \dots + P\{K_{\textnormal{st}} = K\} \\
        \phantom{=}&+ \phantom{P\{K_{\textnormal{st}} = 0\} + }P\{K_{\textnormal{st}} = 1\} + P\{K_{\textnormal{st}} = 2\} + \dots + P\{K_{\textnormal{st}} = K\} \\
        \phantom{=}&+ \phantom{P\{K_{\textnormal{st}} = 0\} + P\{K_{\textnormal{st}} = 1\} + }P\{K_{\textnormal{st}} = 2\} + \dots + P\{K_{\textnormal{st}} = K\} \\
        \phantom{=}&+ \phantom{P\{K_{\textnormal{st}} = 0\} + P\{K_{\textnormal{st}} = 1\} + P\{K_{\textnormal{st}} = 2\} +} \dots \phantom{+ P\{K_{\textnormal{st}} = K\}}\\
        \phantom{=}&+ \phantom{P\{K_{\textnormal{st}} = 0\} + P\{K_{\textnormal{st}} = 1\} + P\{K_{\textnormal{st}} = 2\} + \dots + }P\{K_{\textnormal{st}} = K\} \\
        =& \sum_{k\in\{0, \dots, K\}} (k+1) \cdot P\{K_{\textnormal{st}} = k\} & \Leftarrow{(iii)}\\
        =& P(K_{\textnormal{st}}+1) &\Leftarrow{(iv)},
\end{align*}
}
where:
\begin{align*}
    (i) \equiv& \textnormal{union decomposition}, \\
    (ii) \equiv& \{K_{\textnormal{st}} = k^{\prime}\} \bigcap \{K_{\textnormal{st}} = k^{\prime\prime}\} = \emptyset, \\ 
    (iii) \equiv& \textnormal{rearranging terms}, \\
    (iv) \equiv&\textnormal{def. of expectation w.r.t. to $P(K_{\textnormal{st}})$}. 
\end{align*}
This concludes the proof.
\end{proof}
\end{lemma}

\subsubsection{Person-time law definition and correspondence}
\label{sec_app: pt_def_and_correspondence}

To estimate the parameter $\beta^*$, in many analyses (see, e.g., \citep{hernan_observational_2008, dickerman_avoidable_2019}), users of MSM-hazard procedures fit a model using datasets formed by the so-called person-time observations generated from longitudinal individual observations
\citep{singer_its_1993}. To be explicit, if an observation 
$(V=v, A_0=0, Y_0=0, C_0=0, A_1=1, Y_1=0, C_1=1, A_2 = 1, Y_2 = 0)$ is drawn from the observed law $P^{\textnormal{o}}(P)$, then the person-time observations generated from that individual longitudinal observation are $\{(V=v, A=0, K_{\textnormal{pt}}= 0, Y=0), (V=v, A=1, K_{\textnormal{pt}}= 1, Y=0)\}$, where $Y$ and $A$ indicate the outcome and treatment for an individual with $V=v$ at follow-up time $K_{\textnormal{pt}}=k$.  The marginal distribution of person-time observations can be represented by a law $P^{\textnormal{pt}} \equiv P^{\textnormal{pt}}(P)$ generated from $P$, which we term \textit{person-time law}. The law $P^{\textnormal{pt},*} \equiv P^{\textnormal{pt},*}(P^*)$ appearing in the statement of Proposition \ref{proposition: msm_hazard_characterization} in Section \ref{sec_app: staggered_entry_characterization} is the person-time
law generated from the pseudopopulation $P^*$. 

In the following, we define the person-time law $P^{\textnormal{pt}}$ induced by the longitudinal law $P^{\textnormal{long}}$ as the probability law obtained by normalizing
the expected occupation measure of the observable person-time records
generated by longitudinal observations. In words, $P^{\textnormal{pt}}$ represents the marginal distribution of a record drawn from the pooled population of observable person-time records. Person-time records contributed by the same individual need not be independent.

In order to aid intuition and simplify notation, we present our results in discrete-time settings and for binary treatment, but our results can be extended to continuous covariates and treatments.
\begin{definition}[Person-time law]
The person-time law $P^{\textnormal{pt}} \equiv P^{\textnormal{pt}}\left(P\right)$ generated by the longitudinal law $P$ is defined as:
\begin{align}
    &P^{\textnormal{pt}}\left( \{l\}\times\{ a\}\times  \{y \} \times B_{\textnormal{fw}}\right) \nonumber\\
    \coloneqq&  P(K_{\textnormal{st}}+1)^{-1} \cdot \sum_{k \in B_{\textnormal{fw}}}P
     (L_0 = l, A_{k} = a, Y_{k} = y, \mathbb{1}_{\{K_{\textnormal{st}} \geq k\}} = 1),\label{eq: discrete_person_time_def} \\
     &\textnormal{for every $a,y,l$ in their respective supports and every $B_{\textnormal{fw}} \subseteq \{0, \dots, K\}$.} \notag
\end{align}
For a given $(l,a,y,k)$ we will denote by $L, A, Y, K_{\textnormal{pt}}$ the corresponding coordinate projections.
\end{definition}
The above display defines a probability law. 

Indeed, $P^{\textnormal{pt}}\emptyset = 0$. Furthermore,
{\footnotesize
\begin{align*}
    &P(K_{\textnormal{st}}+1) \cdot P^{\textnormal{pt}}\left( \mathcal{L}\times\{ 0,1\} \times \{0,1 \} \times \{0, \dots, K\}\right) \\
    =& \sum_{k \in \{0, \dots, K\}}
    P\{L_0 \in \mathcal{L}, A_{k} \in \{0,1\}, Y_{k} \in \{0,1\}, \mathbb{1}_{\{ K_{\textnormal{st}} \geq k\}} = 1\} & \Leftarrow{ (i)} \\
    =&\sum_{k \in \{0, \dots, K\}}P\{ K_{\textnormal{st}} \geq k\} & \Leftarrow (ii)\\
    =&P(K_{\textnormal{st}} + 1)& \Leftarrow (iii) \\
    \Rightarrow& P^{\textnormal{pt}}\left((L, A, Y, K_{\textnormal{pt}}) \in \mathcal{L}\times\{ 0,1\} \times \{0,1 \} \times \{0, \dots, K\}\right) = 1,
\end{align*}
}
where
\begin{align*}
    (i) \equiv& \textnormal{def.}  \eqref{eq: discrete_person_time_def}, \\
    (ii) \equiv& \textnormal{marg. w.r.t. $L_0, A_{k}, Y_{k}$}, \\
    (iii) \equiv&\textnormal{Lemma \ref{lemma: expectation_identity}}.
\end{align*}

We now relate conditional probabilities formulated under $P^{\textnormal{pt}}$ and $P$. 

\begin{lemma}[Conditioning: person-time and longitudinal correspondence]
\label{lemma: factorization_pt_long}
Let $k \in \{0, \dots, K\}$ and $X \subseteq L_0$. Let $P^{\textnormal{pt}} \equiv P^{\textnormal{pt}}(P)$ denote the person-time law generated from the law $P$, as defined in Definition \ref{eq: discrete_person_time_def}. Then, 
{\footnotesize
    \begin{align}
    &P^{\textnormal{pt}} (X=x, K_{\textnormal{pt}}=k) \nonumber  \\
    =& P(K_{\textnormal{st}}+1)^{-1} \cdot P\left(X =x, C_{k-1} = 0, Y_{k-1} = 0\right), \label{eq: distintegration_marginal} \\
    &P^{\textnormal{pt}} (Y=y, A=a \mid X=x, K_{\textnormal{pt}} = k) \nonumber \\
    =& P(Y_{k} = y,  A_{k}=a \mid X = x, C_{k-1} = 0, Y_{k-1} = 0), \label{eq: distintegration_conditional_1}  \\
    &P^{\textnormal{pt}} (Y = y \mid X=x, A=a, K_{\textnormal{pt}} = k) \nonumber \\
    =& P(Y_{k} = y \mid X = x, A_{k} = a, C_{k-1} = 0, Y_{k-1} = 0). \label{eq: distintegration_conditional_2} 
\end{align}
}
for every $x,a,y$ in their respective supports.
\begin{proof}
    For \eqref{eq: distintegration_marginal}. Let $k \in \{0, \dots, K\}$. Then, for every $l$:
    {\footnotesize
    \begin{align*}
        &P^{\textnormal{pt}}(L=l, K_{\textnormal{pt}} = k) \\
        =&  P^{\textnormal{pt}}(\{l\} \times \{0,1\} \times \{0,1\}  \times \{ k\}) & \Leftarrow (i) \\
        =&\sum_{(a,y) \in \{0,1\}^{2}}P^{\textnormal{pt}}\left(L=l, A=a, Y=y, K_{\textnormal{pt}} = k\right) \\
        =& \sum_{(a,y) \in \{0,1\}^{2}}P(K_{\textnormal{st}}+1)^{-1} \cdot P(L_0=l, A_{k}=a, Y_{k}=y, \mathbb{1}_{\{K_{\textnormal{st}} \geq k\}} = 1) & \Leftarrow{(ii)} \\
        =& P(K_{\textnormal{st}}+1)^{-1} \cdot \sum_{(a,y) \in \{0,1\}^{2}}  P(L_0=l, C_{k-1} = 0, Y_{k-1} = 0, A_{k}=a, Y_{k}=y) & \Leftarrow{(iii)}\\
        =& P(K_{\textnormal{st}}+1)^{-1} \cdot   
        P(L_0=l, C_{k-1}=0, Y_{k-1}=0) \\
        \Rightarrow& P^{\textnormal{pt}}(X=x, K_{\textnormal{pt}}=k) \\
        =& \sum_{l\setminus x} P(K_{\textnormal{st}}+1)^{-1} \cdot   
        P(L_0=l, C_{k-1}=0, Y_{k-1}=0) \\
        =& P(K_{\textnormal{st}}+1)^{-1} \cdot   
        P(X=x, C_{k-1}=0, Y_{k-1}=0),
    \end{align*}
    }
    where:
    \begin{align*}
        (i) \equiv& \textnormal{def. of joint law},\\
        (ii) \equiv& \textnormal{def. }  \eqref{eq: discrete_person_time_def},\\
        (iii) \equiv& \textnormal{Lemma } \ref{lemma: stopping_time}.
    \end{align*}
    For \eqref{eq: distintegration_conditional_1}:
    {\footnotesize
    \begin{align*}
        &P^{\textnormal{pt}} (Y=y, A=a \mid X=x, K_{\textnormal{pt}} = k) \\
        =&\frac{P^{\textnormal{pt}}(X=x, A=a, Y=y, K_{\textnormal{pt}} = k)}{P^{\textnormal{pt}}(X=x, K_{\textnormal{pt}} =k)} &\\ 
        =&\frac{P(K_{\textnormal{st}}+1)^{-1} \cdot P(X = x,  A_{k} = a, Y_{k} = y,  \{K_{\textnormal{st}} \geq k\} = 1)}{P(K_{\textnormal{st}}+1)^{-1} \cdot   
        P(X=x, \{K_{\textnormal{st}} \geq k\} = 1)} &  \Leftarrow{\textnormal{def. }\eqref{eq: discrete_person_time_def}} \\
        =&\frac{P(K_{\textnormal{st}}+1)^{-1} \cdot P(X = x,  A_{k} = a, Y_{k} = y,  Y_{k-1} = 0, C_{k-1}=0)}{P(K_{\textnormal{st}}+1)^{-1} \cdot   
        P(X=x, C_{k-1}=0, Y_{k-1}=0)} &\Leftarrow{\textnormal{Lemma } \ref{lemma: stopping_time}} \\
        =&P(Y_{k} = y,  A_{k}=a \mid X = x, C_{k-1} = 0, Y_{k-1} = 0).  &
    \end{align*}
    }
    Analogous arguments can be used to establish \eqref{eq: distintegration_conditional_2}.
\end{proof}
\end{lemma}

\begin{landscape}
\begin{figure}[H]
\includegraphics[width=16cm]{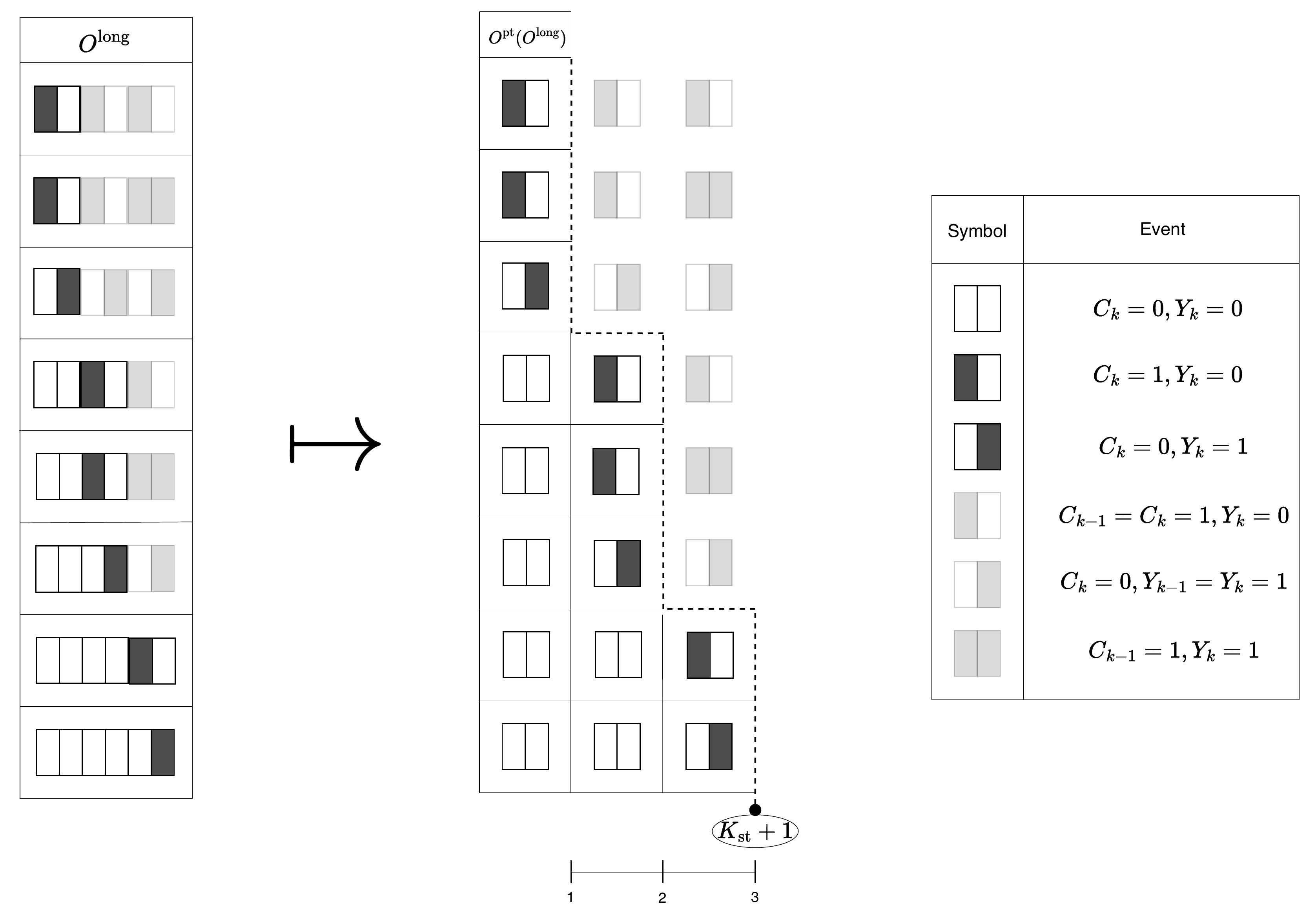}
\caption{Let $O^{\textnormal{long}}$ denote an observation drawn from $P^{\textnormal{o}}(P)$. From $O^{\textbf{long}}$ practitioners create $K_{\textnormal{st}}(O^{\textbf{long}}) + 1$ person-time observations, that is, as many observations as the number of follow-ups that can be observed before the individual either experiences failure or is right-censored. In keeping with standard practice, person-time units corresponding to any $k > K_{\textnormal{st}}$, depicted with opaque symbols, are not considered.
}
\label{fig: person_time_illustration}
\end{figure}
\end{landscape} 

We are now ready to characterize $\beta^*$ as the solution of a score equation formulated under an appropriate law.

\begin{proposition}[Person-time equivalent formulation of population person-time score equation]
\label{proposition: correspondence_score_equation}
Let $P^{\textnormal{pt},*}(P^*)$ denote the person-time law generated from the pseudopopulation $P^*$ as defined in Definition \ref{eq: discrete_person_time_def}. Suppose $\mbps^{\textnormal{*-\textnormal{cs}}}$ holds. With a slight abuse of notation, let $V$ denote the coordinate projection of $L$ in the entries of the vector $V$ in $L_0$. Assume the parameter space $\mathbb{B}^{*}$ is open, $p(v,a,k;\beta) \in (0,1)$ is continuously differentiable with respect to $\beta$, and the norm of the gradient of the log-likelihood is bounded by a $P^{\textnormal{pt},*}$-integrable function. Then,
\begin{enumerate}
    \item $\beta^* \textnormal{ satisfies } \eqref{eq: population_pt_score_eq} \Leftrightarrow  \beta^* \textnormal{ satisfies } P^{\textnormal{pt},*}\left(s_{\textnormal{pool}}(V,A,K_{\textnormal{pt}},Y;\beta)\right)=0$
    \item $P^{\textnormal{pt},*}\left(s_{\textnormal{pool}}(V,A,Y,K_{\textnormal{pt}};\beta)\right)=0$ represents the population score equation of a model for the conditional distribution of $P^{\textnormal{pt},*}(Y \mid V=v, A=a, K_{\textnormal{pt}}=k)$ (for $(v,a,k) \in \textnormal{supp}(P^{\textnormal{pt},*}(V,A,K_{\textnormal{pt}}))$) assumed to be Bernoulli with probability of success $p(v,a,k;\beta)$. 
\end{enumerate}

\begin{proof}
First,
{\scriptsize
    \begin{align*}
    & P^*\left(\sum_{k \in \{0, \dots, K \}} \mathbb{1}_{\{K_{\textnormal{st}} \geq k\}} \cdot s_{\textnormal{pool}}(V, A_{k}, Y_{k}, k; \beta) \right) & \Leftarrow{\eqref{eq: population_pt_score_eq}}\\
        =& \sum_{k \in \{0, \dots, K \}} P^*\left(\mathbb{1}_{\{K_{\textnormal{st}} \geq k\}} \cdot s_{\textnormal{pool}}(V, A_{k}, Y_{k}, k; \beta)\right) & \Leftarrow \textnormal{linearity}\\
        =& \sum_{k \in \{0, \dots, K \}}
        \sum_{(v,a,y)} s_{\textnormal{pool}}(v,a,y,k; \beta)\cdot P^*\{V = v, A_k=a, Y_k=y, \mathbb{1}_\{K_{\textnormal{st}} \geq k\} = 1 \} & \Leftarrow \textnormal{def.}\\
        =& P^*(K_{\textnormal{st}}+1) \sum_{k,v,a,y} s_{\textnormal{pool}}(v,a,y,k; \beta) \cdot P^{\textnormal{pt},*}(V=v, A=a, Y=y,K_{\textnormal{pt}}=k) &\Leftarrow \times \frac{P^*(K_{\textnormal{st}}+1) }{P^*(K_{\textnormal{st}}+1) }\\
        =& P^*(K_{\textnormal{st}}+1) \cdot P^{\textnormal{pt},*}\left( s_{\textnormal{pool}}(V,A,Y,K_{\textnormal{pt}}; \beta) \right) & \Leftarrow \eqref{eq: discrete_person_time_def}
    \end{align*}
    }
    Second, 
    {\footnotesize
    \begin{align*}
        &P^{\textnormal{pt}, *} \left( s_{\textnormal{pool}}(V,A,Y,K_{\textnormal{pt}}; \beta) \right)\\
        =& P^{\textnormal{pt}, *} \left( \left( \frac{Y}{p(V,A,K_{\textnormal{pt}}; \beta)} - \frac{1-Y}{1-p(V,A,K_{\textnormal{pt}}; \beta)} \right) \nabla_{\beta} p(V,A,K_{\textnormal{pt}}; \beta) \right) 
        &\Leftarrow{\textnormal{def.} \eqref{eq: s_pool}} \\
        =& P^{\textnormal{pt}, *} \left( \nabla_{\beta} \left(Y \log p(V,A,K_{\textnormal{pt}}; \beta) + (1-Y) \log (1-p(V,A,K_{\textnormal{pt}}; \beta))\right) \right) 
        &\Leftarrow{\textnormal{algebra}}  \\
        =& \nabla_{\beta}P^{\textnormal{pt}, *} \left( Y \log p(V,A,K_{\textnormal{pt}}; \beta) + (1-Y) \log (1-p(V,A,K_{\textnormal{pt}}; \beta))\right) 
        &\Leftarrow{\textnormal{reg. cond.}} 
    \end{align*}
    }
    Thus, since $\mbps^{\textnormal{*-\textnormal{cs}}}$ holds, $P^{\textnormal{pt},*}\left(s_{\textnormal{pool}}(V,A,K_{\textnormal{pt}},Y;\beta)\right)=0$ represents the population score equation of a model for the conditional distribution of $P^{\textnormal{pt},*}(Y \mid V=v, A=a, K_{\textnormal{pt}}=k)$ (for $(v,a,k) \in \textnormal{supp}(P^{\textnormal{pt},*}(V,A,K_{\textnormal{pt}}))$) assumed to be Bernoulli with probability of success $p(v,a,k;\beta)$.
\end{proof}
\end{proposition}

\subsubsection{Correct specification in MSM-hazard procedures}
\label{sec_app: notion_correct_specification}
Correct specification of a regression model $\mathcal{M}(\mathbb{B})$ indexed by a set $\mathbb{B}$ for a given regression parameter, say $h_k^{a_k}(v)$, requires the existence of a unique $\beta \in \mathbb{B}$ such that $h_k^{a_k}(v) = p(v,a_k,k;\beta) \in \mathcal{M}(\mathbb{B})$. We emphasize that the domain in which the equality between $p(v,a_k,k;\beta)$ and the regression parameter $h_k^{a_k}(v)$ is required to hold is a distinctive feature in our elaborations. In
$\mb^{\textnormal{cs}}$ the equality should be understood for every $v \in \textnormal{supp}( P(V))$ for every $k \in \{0, \dots, K\}$ and every $a_k \in \{0,1\}$, while in $\mb^{*\textnormal{-cs}}$ it should be understood for every $v \in \textnormal{supp}(P^{\textnormal{pt},*}(V\mid A=a_k, K_{\textnormal{pt}}=k))$ for every $k \in \{0, \dots, K\}$ and every $a_k \in \{0,1\}$. A feature of $P^{\textnormal{pt},*}$ is that for any two functions, $f_1$ and $f_2$, the equality
$f_1(v,a_{k},k)=f_2(v,a_{k},k)$
is weaker if it holds for every $v \in \textnormal{supp}(P^{\textnormal{pt},*}(V \mid A=a_k, K_{\textnormal{pt}}=k))$ than if it holds for every $v \in \textnormal{supp}(P(V))$ whenever  $k \in \{1,\dots, K\}$ and $a_k \in \{0,1\}$. Conceptually, this follows because $P^{\textnormal{pt},*}$ involves only uncensored observations; if $K_{\textnormal{pt}}=k$, that is, if $Y_{k-1}=C_{k-1}=0$, then since $\eqref{eq: secc}$ is preserved upon weighting (see the arguments in the proof of Corollary \ref{corollary: kiss_when_e_not_in_v} in Appendix \ref{sec_app: iptw_technical}), the equality is required to hold only for $v$ such that $e(v) \equiv e \in \{0, \dots, \mathcal{T}-k\}$. To be explicit, let $V=E$ and consider, as in our running example in Section \ref{sec: illustrative_example_censoring_assumptions_ill_posed} of the main text, $K=\mathcal{T}=2$; a person-time observation, say, $(E=2, A=a_k, K_{\textnormal{pt}} = 2, Y=y)$, for any $a_k,y \in \{0,1\}$, cannot possibly be realized under $P^{\textnormal{pt},*}$ because it would imply that an observation $(E=2, A_0=a_0, Y_0=0, C_0=0, A_1=a_1, Y_1=0, C_1=0, A_2 = a_2, Y_2 = y)$ is realizable under $P^*$, which, since $C_{k\textcolor{gray}{\equiv 2} -1} = 0, Y_{k\textcolor{gray}{\equiv 2} - 1} = 0 \Rightarrow E\textcolor{gray}{\equiv 2} \textnormal{ + } \textcolor{gray}{2\equiv}k < \mathcal{T}\textcolor{gray}{\equiv 2} + 1$, $P^*$-\textnormal{a.s.}, contradicts \eqref{eq: secc}. See Lemma \ref{lemma: cp_full_implies_observed_pt}.

For any two functions, $f_1(v,a_{k},k)$ and $f_2(v,a_{k},k)$, we write $f_1(v,a_{k},k)\stackrel{\prstudyp}{=}f_2(v,a_{k},k)$ if for every $a_k \in \{0,1\}$, every $k \in \{0, \dots, K\}$, and every $v \in \textnormal{supp}(P^{\textnormal{pt},*}(V \mid A=a_k, K_{\textnormal{pt}}=k))$ it holds that
$f_1(v,a_{k},k)=f_2(v,a_{k},k)$. We write $f_1(v,a_{k},k)\stackrel{\prstudyp\prgapp}{=}f_2(v,a_{k},k)$ if the same equality holds for every $a_k \in \{0,1\}$, every $k \in \{0, \dots, K\}$, and for every $v \in \textnormal{supp}(P(V))$.

\begin{lemma}[Relations between correct specifications]\label{lemma: cp_full_implies_observed_pt}
   Consider two functions $f_1(v,a_{k},k)$ and $f_2(v,a_{k},k)$ well-defined for every $k \in \{0, \dots, K\}$ and $a_k \in \{0,1\}$, for every $v \in \textnormal{supp}(P(V))$. Then,
    \begin{align}
    f_1(v,a_{k},k)\stackrel{\prstudyp\prgapp}{=}f_2(v,a_{k},k) \Rightarrow f_1(v,a_{k},k)\stackrel{\prstudyp}{=}f_2(v,a_{k},k).
    \end{align}
    \begin{proof}
        Fix $k \in \{0, \dots, K\}$ and $a_k \in \{0,1\}$. Then
        {\footnotesize
        \begin{align*}
            &\textnormal{supp}(P^{\textnormal{pt},*}(V \mid A=a_k, K_{\textnormal{pt}}=k)) \\
            =& \textnormal{supp}(P^{\textnormal{pt},*}(V \mid A=a_k, K_{\textnormal{pt}}=k, E \leq \mathcal{T}-k)) &\Leftarrow{\textnormal{Corollary \ref{corollary: kiss_when_e_not_in_v}}}\\ 
            \subseteq& \textnormal{supp}(P^{\textnormal{pt},*}(V \mid E \leq \mathcal{T}-k)) &\Leftarrow{\textnormal{supp}(X_1 \mid X_3) \supseteq \textnormal{supp}(X_1 \mid X_2, X_3)} \\
            \subseteq& \textnormal{supp}(P^{*}(V \mid E \leq \mathcal{T}-k)) & \Leftarrow{\textnormal{def. } \eqref{eq: discrete_person_time_def}}\\
            =& \textnormal{supp}(P(V \mid E \leq \mathcal{T}-k)) & \Leftarrow{\textnormal{def. of $P^*$}} \\
            \subseteq& \textnormal{supp}(P(V)). &\Leftarrow{\textnormal{supp}(X_1) \supseteq \textnormal{supp}(X_1 \mid X_3)}
        \end{align*}
        }
        By definition, $f_1 \stackrel{\prstudyp\prgapp}{=} f_2$ implies that the equality holds for every $v \in \textnormal{supp}(P(V))$. Because we established that $\textnormal{supp}(P^{\textnormal{pt},*}(V \mid A=a_k, K_{\textnormal{pt}}=k)) \subseteq \textnormal{supp}(P(V))$, the equality necessarily holds for every $v \in \textnormal{supp}(P^{\textnormal{pt},*}(V \mid A=a_k, K_{\textnormal{pt}}=k))$ as well. This directly satisfies the definition of $\stackrel{\prstudyp}{=}$, concluding the proof.
    \end{proof}
\end{lemma}

A counterexample showing that $f_1(v,a_k,k)\stackrel{\prstudyp}{=}f_2(v,a_{k},k)$ does not imply $f_1(v,a_{k},k)\stackrel{\prstudyp\prgapp}{=}f_2(v,a_{k},k)$ is illustrated in
Figure \ref{fig: not.i.3}.
\begin{landscape}
   \begin{figure}[H]
\includegraphics[width=20cm]{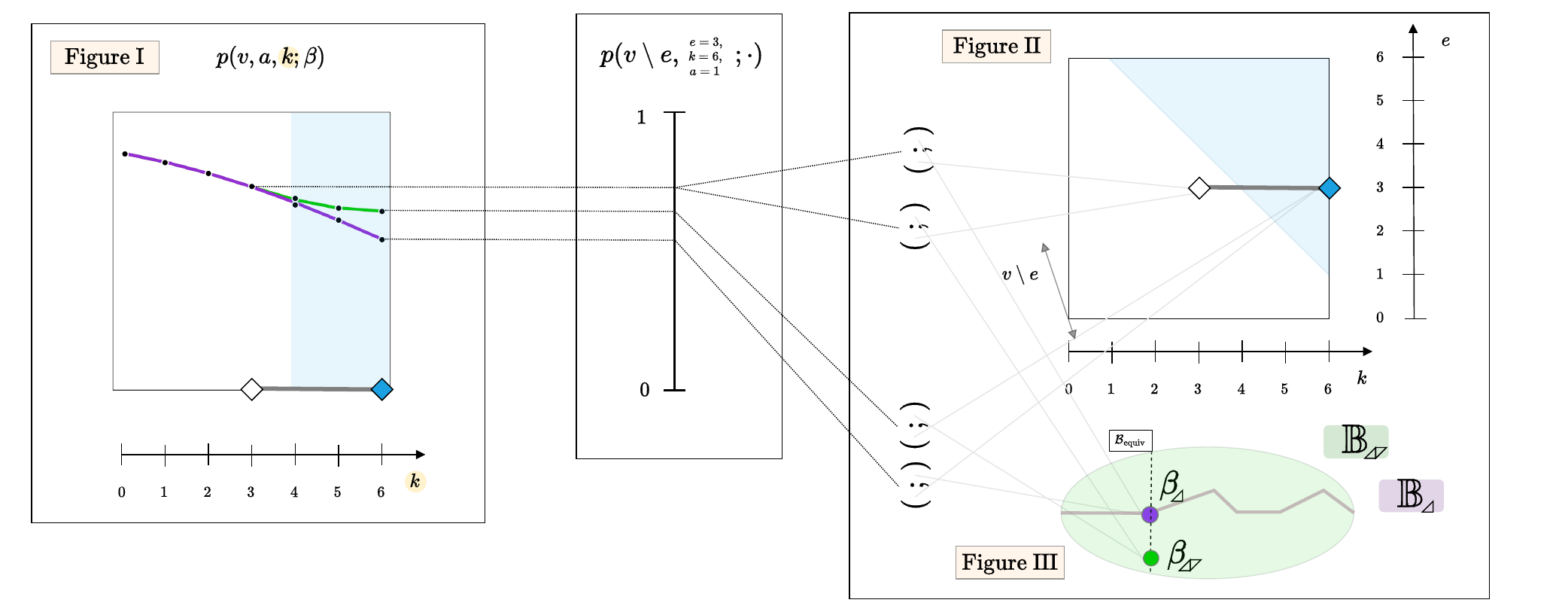}
\caption{Suppose $K=\mathcal{T}=6$ and consider, without loss of generality, a fixed $e \equiv e(v) = 3$, $v\setminus e$, and, say, $a=1$. Suppose ${\protect\mbs}^{*-\textnormal{cs}}$ and ${\protect\mbf}^{\textnormal{cs}}$ hold for two different models $\protect\mbs$ and $\protect\mbf$; let $\beta_{\protect\prstudyp} \in \mathbb{B}_{\protect\prstudyp}$ and $\beta_{\protect\prstudyp\protect\prgapp} \in \mathbb{B}_{\protect\prstudyp\protect\prgapp}$ denote the unique values for which ${\protect\mbs}^{*-\textnormal{cs}}$ and ${\protect\mbf}^{\textnormal{cs}}$ respectively hold.
Let $p(v \setminus e, e=3, a=1,k; \beta_{\protect\prstudyp\protect\prgapp}) = \textnormal{expit}\left(0.2 -0.3 \cdot k - 0.05 \cdot \mathbb{1}_{\{e+k > \mathcal{T}\}} \cdot (e+k - \mathcal{T})^{2}\right)$, $(\beta_{K}, \beta_{\mathcal{T}>}) = (-0.3, - 0.05)$.  The functions $k \mapsto p(v \setminus e, e=3, a=1,k;\textcolor{violet}{\beta_{\protect\prstudyp}}) = \textnormal{expit}\left(0.2 -0.3 \cdot k \right)$ and $k \mapsto p(v \setminus e, e=3, a=1,k;\textcolor{Green}{\beta_{\protect\prstudyp\protect\prgapp}}) = \textnormal{expit}\left(0.2 -0.3 \cdot k - 0.05 \cdot \mathbb{1}_{\{e+k > \mathcal{T}\}} \cdot (e+k - \mathcal{T})^{2}\right)$ in \textbf{Figure I}, on the left-hand side, correspond respectively to the violet and green curves for $e=3$, $a=1$, and the given $v\setminus e$. In keeping with condition ${\protect\mbs}^{*-\textnormal{cs}}$, $p(v,a,k;\beta_{\protect\prstudyp}) = p(v,a,k;\beta_{\protect\prstudyp\protect\prgapp})$ for $(v\setminus e,e=3,a=1,k)$ such that $k \in \{0,1,2,3\}$; however, $p(v,a,k;\beta_{\protect\prstudyp}) \neq p(v,a,k;\beta_{\protect\prstudyp\protect\prgapp})$ for $(v\setminus e,e=3,a=1,k)$ such that $k \in \{4,5,6\}$. Hence, ${\protect\mbs}^{\textnormal{cs}}$ does not hold.  \textbf{Figure II}, on the right-hand side, clarifies that the two curves differ when evaluated in the blue region, corresponding to the set of points that cannot be possibly realized from $P^{\textnormal{pt},*}$. \textbf{Figure III} illustrates the relation between $\protect\mbs$ and $\protect\mbf$.
}

\label{fig: not.i.3}
\end{figure} 
\end{landscape}

\subsubsection{Staggered-entry characterization for MSM-hazard procedures}
\label{sec_app: staggered_entry_characterization}
Throughout, it is assumed that \CTC{} are present (i.e., \CTC{}, as an assumption, holds). We will also consider the following regularity condition:
\begin{align}\label{ass: regularity_support}
\textnormal{supp}(P^*(V)) = \textnormal{supp}(P^{\textnormal{pt},*}(V \mid K_{\textnormal{pt}} =k, A=a_k)), \notag \\
\textnormal{ for every } k \in \{0, \dots, K\}, \, a_k \in \{0,1\}. \tag{RSUP}
\end{align}
Since $\textnormal{supp}(P^*(V))=\textnormal{supp}(P(V))$, it follows from \eqref{eq: secc} that \ref{ass: regularity_support} is automatically violated if $E \in V$. If instead $E \notin V$, it can be shown that for any two functions $f_1$ and $f_2$,  $f_1(v,a_k,k) \stackrel{\prstudyp}{=} f_2(v,a_k,k) \Leftrightarrow f_1(v,a_k,k) \stackrel{\prstudyp\prgapp}{=} f_2(v,a_k,k)$. \footnote{For any two functions, $f_1(v,a_{k},k)$ and $f_2(v,a_{k},k)$, we write $f_1(v,a_{k},k)\stackrel{\prstudyp}{=}f_2(v,a_{k},k)$ if for every $a_k \in \{0,1\}$, every $k \in \{0, \dots, K\}$, and every $v \in \textnormal{supp}(P^{\textnormal{pt},*}(V \mid A=a_k, K_{\textnormal{pt}}=k))$ it holds that
$f_1(v,a_{k},k)=f_2(v,a_{k},k)$. We write $f_1(v,a_{k},k)\stackrel{\prstudyp\prgapp}{=}f_2(v,a_{k},k)$ if the same equality holds for every $a_k \in \{0,1\}$, every $k \in \{0, \dots, K\}$, and for every $v \in \textnormal{supp}(P(V))$.}

\begin{lemma}[Omission of $E$ in $V$]\label{lemma: omission_e_is_problematic}
Suppose $\RAID$ and \ref{ass: regularity_support} hold. 

\begin{enumerate}
    \item If ${\mbps}_{\mid V}^{*-\textnormal{cs}}$ holds with $E \notin V$, then $\mbps \ni p(v,a_k,k;\beta^*) \stackrel{\prstudyp\prgapp}{=} h_{\prstudyp,k}^{a_k}(v)$.
    \item Suppose that there are no perfect cancellations, i.e., \eqref{eq: hazard_variation_dependence_ct} does not hold. If $E \notin V$, then ${\mbps}_{\mid V}^{\textnormal{cs}}$ and ${\mbps}_{\mid V}^{*\textnormal{-cs}}$ cannot simultaneously hold for the same parameter values indexing a given working regression model $\mbps$, i.e., $\beta^*=\beta_{\prstudyp\prgapp}$ cannot hold for $\beta^*,\beta_{\prstudyp\prgapp} \in \mathbb{B}^*$.
\end{enumerate}

\end{lemma}
\begin{proof}
Suppose ${\mbps}_{\mid V}^{*\textnormal{-cs}}$ holds with $E \notin V$. Then, for some $\beta^* \in \mathbb{B}^*$ 
\begin{align*}
    {\mbps}_{\mid V}^{*\textnormal{-cs}} \Rightarrow \mbps \ni p(v,a_k,k;\beta^*) \stackrel{\prstudyp}{=}& h_k^{*,a_k}(v) \\
    \stackrel{\prstudyp}{=}&h_{\prstudyp,k}^{a_k}(v) &\Leftarrow{\textnormal{Corollary }\ref{corollary: kiss_when_e_not_in_v}, \RAID}.\\
    \stackrel{\prstudyp\prgapp}{=}&h_{\prstudyp,k}^{a_k}(v) &\Leftarrow{\ref{ass: regularity_support}}.
\end{align*}
This proves the first point. 

Next, suppose ${\mbps}_{\mid V}^{\textnormal{cs}}$ holds with $E \notin V$. Let $\beta^*, \beta_{\prstudyp\prgapp} \in \mathbb{B}^*$. If $\beta^* = \beta_{\prstudyp\prgapp}$, then

\begin{align*}
    h_{\prstudyp,k}^{a_k}(v) \stackrel{\prstudyp\prgapp}{=} p(v,a_k,k;\beta^*) 
    \stackrel{\prstudyp\prgapp}{=}p(v,a_k,k;\beta_{\prstudyp\prgapp}) \stackrel{\prstudyp\prgapp}{=} h_k^{a_k}(v), 
\end{align*}
i.e., $h_{\prstudyp,k}^{a_k}(v)\stackrel{\prstudyp\prgapp}{=}h_k^{a_k}(v)$, which is absurd given that we assumed \eqref{eq: hazard_variation_dependence_ct} does not hold. This establishes the second point.
\end{proof}

\clearpage
\begin{proposition*}[Restatement of MSM-hazard identification]
\label{proposition: proof_characterization}
Suppose $\RAID$. If ${\mbps}^{*\textnormal{-cs}}$ holds and $\beta^* \in \mathbb{B}^*$ is identifiable from $P^{\textnormal{pt},*} \equiv P^{\textnormal{pt},*}(P^*)$, then, under the following regularity conditions: 
\begin{enumerate}
    \item \ref{ass: regularity_support} holds for $E \notin V$;
    \item $p(v,a_k,k;\beta^*) < 1$ for every $k\in\{0, \dots, K\}$ and $a_k \in \{0,1\}$ and for every $v \in \textnormal{supp}(P(V))$;
    \item no perfect cancellations, i.e., if neither
    {\scriptsize
    \begin{align*}
        (\textnormal{pc}_i) & \equiv \expectation\left(\prod_{m=0}^k (1 - h_{m}^{a_m}(V))\right) = \expectation\left(\prod_{m=0}^k (1 - h_{\prstudyp, m}^{a_m}(V))\right), \textit{ nor} \\
        (\textnormal{pc}_{ii}) &\equiv  \expectation \left( \mathbb{1}_{\{k > \mathcal{T} - E \}} \prod_{m=0}^{k}(1 - p(V,a_m,m;\beta^*)) \cdot \left( \prod_{m=\mathcal{T}-E+1}^{k}\frac{1 - h_m^{a_m}(V)}{1 - p(V,a_m,m;\beta^*)}- 1\right)\right) = 0
    \end{align*}
    }
holds for every $k \in \{0, \dots, K\}$;
\end{enumerate}
we have
\begin{align*}
{\mbps}_{\mid V}^{\textnormal{cs}} \textnormal{ with }
    \textnormal{$E \in V$ } \quad \Leftrightarrow \quad S_k^{\overline{a}_k} = S_k^{*, \textnormal{msm}, \overline{a}_k}\\ \small{\textnormal{for every } k \in \{0, \dots, K\} \textnormal{ and }  \overline{a}_k \in \{\overline{0}, \overline{1}\}}.
\end{align*}
\begin{proof}
For $``\Rightarrow''$. Let $k \in \{0, \dots, K\}$, $\overline{a}_k \in \{\overline{0}, \overline{1}\}$. %
We have
\begin{align*}
    S_k^{*,\overline{a}_k} =& \expectation\left(\prod_{m=0}^k (1 - p(V,a_m,m;\beta^*))\right) &\Leftarrow{\textnormal{def. } \eqref{eq: survival_star}}\\
    =&\expectation\left(\prod_{m=0}^k (1 - p(V,a_m,m;\beta_{\prstudyp\prgapp}))\right) &\Leftarrow{(i)} \\
    =&\expectation\left(\prod_{m=0}^k (1 - h_m^{a_m}(V))\right)  &\Leftarrow{{\mb}^{\textnormal{cs}}} \\
    =& S_k^{\overline{a}_k}.
\end{align*}
where $(i)$ denotes the following implication $(i) \equiv [$
{\footnotesize
\begin{align*}
\textnormal{for some } \beta^*, \beta_{\prstudyp\prgapp} \in \mathbb{B}^*,&\\
    \mbps^{*-\textnormal{cs}} \Rightarrow p(v,a_k,k;\beta^*) \stackrel{\prstudyp}{=}&h_k^{*,a_k}(v) \\
    \stackrel{\prstudyp}{=}&h_k^{a_k}(v) & \Leftarrow{\textnormal{Corollary \ref{corollary: kiss_in_msm_hazard_models}}, \RAID}\\
    \stackrel{\prstudyp}{=}&p(v,a_k,k;\beta_{\prstudyp\prgapp}) & \Leftarrow{\mbps^{\textnormal{cs}}} \\
\Rightarrow& \beta^*= \beta_{\prstudyp\prgapp}].
\end{align*}}

For $``\Leftarrow''$. Consider the contrapositive statement:
\begin{align*}
    \neg \left[ {\mbps}_{\mid V}^{\textnormal{cs}} \textnormal{ with }
    \textnormal{$E \in V$ } \right] \Rightarrow S_k^{\overline{a}_k} \neq S_k^{*,\overline{a}_k}
\end{align*}
for some $k \in \{0, \dots, K\}$ and $\overline{a}_k \in \{\overline{0}, \overline{1}\}$. We consider two cases: either $E \notin V$, or
$E \in V$ and ${\mbps}_{\mid V}^{\textnormal{cs}}$ does not hold.

\begin{enumerate}
    \item Suppose $E \notin V$. By
    ${\mbps}_{\mid V}^{*\textnormal{-cs}}$, Lemma
    \ref{lemma: omission_e_is_problematic}, and
    \ref{ass: regularity_support},
    \begin{align*}
        p(v,a_k,k;\beta^*)
        \stackrel{\prstudyp\prgapp}{=}
        h_{\prstudyp,k}^{a_k}(v).
    \end{align*}
    Hence, for every $k \in \{0,\dots,K\}$,
    \begin{align*}
    S_k^{*,\overline{a}_k}
    =&
    \expectation\left(
        \prod_{m=0}^k
        (1-p(V,a_m,m;\beta^*))
    \right)
    &&\Leftarrow{\textnormal{def. } \eqref{eq: survival_star}}
    \\
    =&
    \expectation\left(
        \prod_{m=0}^k
        (1-h_{\prstudyp,m}^{a_m}(V))
    \right)
    &&\Leftarrow{
        {\mbps}_{\mid V}^{*\textnormal{-cs}},
        \textnormal{ Lemma }\ref{lemma: omission_e_is_problematic}
    }.
    \end{align*}
    Therefore, since $(\textnormal{pc}_i)$ does not hold for every
    $k \in \{0,\dots,K\}$, there exists some $k$ such that
    \begin{align*}
        S_k^{*,\overline{a}_k}
        \neq
        \expectation\left(
            \prod_{m=0}^k
            (1-h_m^{a_m}(V))
        \right)
        =
        S_k^{\overline{a}_k}.
    \end{align*}

    \item Suppose $E \in V$ and
    ${\mbps}_{\mid V}^{\textnormal{cs}}$ does not hold.
    Let $k \in \{1,\dots,K\}$. Then,
    {\footnotesize
    \begin{align*}
        S_k^{\overline{a}_k}
        =&
        \expectation\left(
            \prod_{m=0}^k
            (1-h_m^{a_m}(V))
        \right)
        \\
        =&
        \expectation\Bigg(
        \mathbb{1}_{\{k > \mathcal{T}-E\}}
        \prod_{m=0}^{\mathcal{T}-E}
        (1-h_m^{a_m}(V))
        \prod_{m=\mathcal{T}-E+1}^{k}
        (1-h_m^{a_m}(V))
        \\
        &\qquad\qquad+
        \mathbb{1}_{\{k \leq \mathcal{T}-E\}}
        \prod_{m=0}^{k}
        (1-h_m^{a_m}(V))
        \Bigg)
        &&\Leftarrow{(ii)}
        \\
        =&
        \expectation\Bigg(
        \mathbb{1}_{\{k > \mathcal{T}-E\}}
        \prod_{m=0}^{\mathcal{T}-E}
        (1-p(V,a_m,m;\beta^*))
        \prod_{m=\mathcal{T}-E+1}^{k}
        (1-h_m^{a_m}(V))
        \\
        &\qquad\qquad+
        \mathbb{1}_{\{k \leq \mathcal{T}-E\}}
        \prod_{m=0}^{k}
        (1-p(V,a_m,m;\beta^*))
        \Bigg)
        &&\Leftarrow{(iii)},
    \end{align*}
    }
    where
    {\footnotesize
    \begin{align*}
        (ii) \equiv&
        \ 1 =
        \mathbb{1}_{\{k > \mathcal{T}-E\}}
        +
        \mathbb{1}_{\{k \leq \mathcal{T}-E\}},
        \\
        (iii) \equiv&
        \left[
        p(v,a_k,k;\beta^*)
        \stackrel{\prstudyp}{=}
        h_k^{*,a_k}(v)
        \stackrel{\prstudyp}{=}
        h_k^{a_k}(v)
        \textnormal{ for every }a_k\in\{0,1\},
        \ k\in\{0,\dots,K\}
        \right]
        \\
        &\hspace{5cm}
        \Leftarrow
        \RAID
        \textnormal{ and }
        \mbps^{*-\textnormal{cs}}.
    \end{align*}
    }
    Thus,
    {\scriptsize
    \begin{align}
    \label{eq: case_3_characterization}
        S_k^{\overline{a}_k}
        -
        S_k^{\overline{a}_k,*}
        =&
        \expectation\Bigg(
        \mathbb{1}_{\{k > \mathcal{T}-E\}}
        \prod_{m=0}^{\mathcal{T}-E}
        (1-p(V,a_m,m;\beta^*))
        \prod_{m=\mathcal{T}-E+1}^{k}
        (1-h_m^{a_m}(V))
        \notag\\
        &\qquad-
        \mathbb{1}_{\{k > \mathcal{T}-E\}}
        \prod_{m=0}^{k}
        (1-p(V,a_m,m;\beta^*))
        \Bigg)
        \notag\\
        =&
        \expectation\left(
        \mathbb{1}_{\{k > \mathcal{T}-E\}}
        \prod_{m=0}^{k}
        (1-p(V,a_m,m;\beta^*))
        \left[
        \prod_{m=\mathcal{T}-E+1}^{k}
        \frac{1-h_m^{a_m}(V)}
             {1-p(V,a_m,m;\beta^*)}
        -1
        \right]
        \right)
        \\
        \neq&\ 0
        \qquad\Leftarrow{\neg(\textnormal{pc}_{ii})}.
        \notag
    \end{align}
    }
\end{enumerate}

Thus, whenever
$\neg[{\mbps}_{\mid V}^{\textnormal{cs}}\textnormal{ with }E\in V]$
holds, either $E\notin V$, in which case
$\neg(\textnormal{pc}_i)$ yields the result, or $E\in V$ and
${\mbps}_{\mid V}^{\textnormal{cs}}$ fails, in which case
$\neg(\textnormal{pc}_{ii})$ yields the result. This concludes the proof.
\end{proof}
\end{proposition*}

The same arguments apply to settings with a point treatment $A$. We define the point-treatment setting as the case where $A \coloneqq A_0$, i.e., the treatment received at baseline, and $A_k = A_0$ for all $k \in \{1, \dots, K\}$.

\textbf{Remarks on the conditions involving perfect cancellations:}
\begin{itemize}
    \item $(\textnormal{pc}_{i})$ will generally not hold under \CTC{}, because equality of the two marginal survival expressions requires a particular cancellation of the differences between the full and study-period hazards. See for example eq. \eqref{eq: hazard_variation_dependence_wsc_main} in Section \ref{sec: wsc_procedures} of the main text.
    \item  \eqref{eq: case_3_characterization} is defined as the expectation of the product of two functions: $\mathbb{1}_{\{k > \mathcal{T} - E \}} \prod_{m=0}^{k}(1 - p(V,a_m,m;\beta^*))$ and $\left( \prod_{m=\mathcal{T}-E+1}^{k}\frac{1 - h_m^{a_m}(V)}{1 - p(V,a_m,m;\beta^*)}- 1\right)$. Since, by hypothesis, $p(v,a_m,m;\beta^*) < 1$ and (when ${\mbps}_{\mid V}^{\textnormal{cs}}$ with $E \in V$ does not hold) $p(v,a_m,m;\beta^*) = h_m^{a_m}(v)$ do not hold for every $v \in \textnormal{supp}(P(V))$, then the integrand will generally be non-zero on part of the post-study region. Its expectation can nevertheless equal zero through cancellation but condition $(\textnormal{pc}_{ii})$ explicitly rules out this possibility.
\end{itemize}

\textbf{To establish bounds}, consider the case where $E \in V$. Let $\overline{f}_k = \{f_0(v), f_1(v), \dots, f_k(v)\}$ be a vector of real-valued functions, and define the functional:
{\scriptsize
    \begin{align*}
    \Gamma(\overline{f}_k)&=\expectation \left( \mathbb{1}_{\{k > \mathcal{T} - E \}} \prod_{m=0}^{k}(1 - p(V,a_m,m;\beta^*)) \cdot \left( \prod_{m=\mathcal{T}-E+1}^{k}\frac{1 - f_m(V) \cdot p(V,a_m,m;\beta^*)}{1 - p(V,a_m,m;\beta^*)}- 1\right)\right).
    \end{align*}
    }
\begin{corollary}[Bounds]
\label{corollary: bounds_ass_characterization_not_met}
Consider the assumptions in Proposition
\ref{proposition: proof_characterization}. Suppose $E \in V$ and ${\mbps}^{\textnormal{cs}}$ does not hold.
Let $k \in \{1, \dots, K\}$. Suppose two vectors of functions,
$\overline{l}_k^{\Gamma}$ and $\overline{u}_k^{\Gamma}$, are such that
\begin{align*}
h_{m}^{a_m}(v)
\in
[\textnormal{l}_m^{\Gamma}(v) \cdot p(v,a_m,m;\beta^*),
 \textnormal{u}_m^{\Gamma}(v) \cdot p(v,a_m,m;\beta^*)]
\end{align*}
for every $(v,a_m,m)$ such that $e+m>\mathcal{T}$ and every
$m\in\{1,\dots,k\}$, with
\begin{align*}
0 \leq
\textnormal{l}_m^{\Gamma}(v) \cdot p(v,a_m,m;\beta^*)
\leq
\textnormal{u}_m^{\Gamma}(v) \cdot p(v,a_m,m;\beta^*)
\leq 1.
\end{align*}
Then,
\begin{align*}
    \Gamma(\overline{u}_k^{\Gamma})
    \leq
    S_k^{\overline{a}_k} - S_k^{\overline{a}_k,*}
    \leq
    \Gamma(\overline{l}_k^{\Gamma}).
\end{align*}

\begin{proof}
Let $k \in \{1, \dots, K\}$. In the proof of Proposition
\ref{proposition: proof_characterization}, we established the identity
$S_k^{\overline{a}_k} - S_k^{\overline{a}_k,*} =$ 
\eqref{eq: case_3_characterization}.
Because all factors $(1-p(\cdot;\beta^*))$ are non-negative, the
expectation is a monotonically decreasing function of the true hazards
$h_m^{a_m}(V)$. Substituting $u_m^{\Gamma}(V)$ for $f_m(V)$ in
$\Gamma(\cdot)$ minimizes
\begin{align*}
    \prod_{m=\mathcal{T}-E+1}^{k}
    \frac{1-f_m(V)\cdot p(V,a_m,m;\beta^*)}
         {1-p(V,a_m,m;\beta^*)}.
\end{align*}
Because all other terms in the expectation are functions of
$\mathbb{1}_{\{k>\mathcal{T}-E\}}$ and $p(V,a_m,m;\beta^*)$, this
minimization gives the lower bound
$\Gamma(\overline{u}_k^{\Gamma})$. Substituting
$l_m^{\Gamma}(V)$ for $f_m(V)$ maximizes the same product, giving the
upper bound $\Gamma(\overline{l}_k^{\Gamma})$.
\end{proof}
\end{corollary}\clearpage
\section{Extrapolation assumptions, c-values and extrapolation curves}
\label{sec_app: extrapolation_assumptions_and_c_values}
This appendix supplements Sections \ref{sec: extrapolation_assumptions_in_msm_hazard} and \ref{sec: c_values_and_extrapolation_curves} of the main text and is organized into four sections. In Section \ref{sec_app: scrutiny_proof_characterization}, we revisit the conditions in Proposition \ref{proposition: msm_hazard_characterization} that permit identification in MSM-hazard procedures. In Section \ref{sec_app: assessment_extrapolation_ass_msm}, we propose a user-assisted procedure for assessing the plausibility of extrapolation assumptions in MSM-hazard procedures. In Section \ref{sec_app: first_sensitivity_analysis}, we introduce a related sensitivity analysis strategy, which we argue exhibits undesirable properties. To avoid technical complications, the results of Sections \ref{sec_app: assessment_extrapolation_ass_msm} and \ref{sec_app: first_sensitivity_analysis} apply when $V$ is discrete, but can, in principle, be extended to accommodate continuous $V$. Finally, Section \ref{sec_app: c_value} presents a formal treatment of c-values and extrapolation curves, where results deferred from the main text are provided.

Unless stated otherwise, throughout we will consider MSM-hazard procedures for a working model of the type \eqref{eq: popular_working_model}:
\begin{align*}
\mbps =& \{ p(v,a,k;\beta) \coloneqq \mu^{-1}(\lambda(v,a,k;\beta)):  \beta \in \mathbb{B}^*\},
\end{align*}
for some set $\mathbb{B}^*$, link function $\mu$, and linear predictor $\lambda$.

\subsection{Conditions of Proposition \ref{proposition: proof_characterization} as extrapolation assumptions}
\label{sec_app: scrutiny_proof_characterization}
Assume $E \in V$. Consider $\stackrel{\prstudyp}{=}$ and $\stackrel{\prstudyp\prgapp}{=}$ as defined in Appendix \ref{sec_app: notion_correct_specification}. If, for given model $\mb$, condition ${\mb}^{\text{cs}}$ is met, then there must exist a unique $\beta_{\prstudyp\prgapp} \in \mathbb{B}$ such that $h_k^{a_k}(v) \stackrel{\prstudyp\prgapp}{=} p(v,a_k,k; {\beta}_{\prstudyp\prgapp}) \in \mb$. Under $\RAID$, if $\mbps^{*\text{-cs}}$ holds for some $\mbps$, then\footnote{See proof of Proposition \ref{proposition: proof_characterization} in Appendix \ref{sec_app: staggered_entry_characterization} for a justification.}
\begin{align}
     \mbps\ni p(v,a_k,k;\beta^*)\stackrel{\prstudyp}{=} h_k^{*,a_k}(v)\stackrel{\prstudyp}{=} h_k^{a_k}(v).
     \label{id: kiss_prstudyp}
\end{align}
This condition is implied by, but
does not necessarily imply that,
\begin{align}
     \mbps\ni p(v,a_k,k;\beta^*)\stackrel{\prstudyp\prgapp}{=} h_k^{a_k}(v) \text{ and } \beta^*  = \beta_{\prstudyp\prgapp},
     \label{id: kiss_prstudypprgapp}
\end{align}
that is, it does not imply that the model $\mb$ satisfying ${\mb}^{\text{cs}}$ is $\mbps$. Thus, to meet the conditions of Proposition \ref{proposition: msm_hazard_characterization}, researchers must assume the following extrapolation:
\begin{align}\label{eq: cs_extrapolation}
   \mbps\ni p(v,a_k,k;\beta^*) \stackrel{\prstudyp}{=} h_k^{a_k}(v) \Rightarrow  {\mbps}^{\text{cs}} \text{ holds.}
\end{align}

To be explicit, let us return to our running example of Section \ref{sec: illustrative_example_censoring_assumptions_ill_posed} of the main text, where $\mathcal{T}=K=2$ and $V=E$, and consider
\begin{multline}\label{eq: stylized_msm}
        \mbps = \Big\{  
        \text{logit} ( p(e,a_k,k;\beta) ) \stackrel{\prstudyp\prgapp}{=}\beta_0 \
+ \beta_{k=0} \cdot \mathbb{1}_{\{k=0\}} + \beta_{k=1} \cdot \mathbb{1}_{\{k=1\}} \\
        + \beta_{E=0} \cdot \mathbb{1}_{\{e=0\}} +  \beta_{E=1} \cdot \mathbb{1}_{\{e=1\}} +  \beta_A \cdot \mathbb{1}_{\{a_k = 1\}} \\
        \text{ for }\left(\beta_0, \beta_{k=0}, \beta_{k=1}, \beta_{E=0}, \beta_{E=1}, \beta_A \right)^{\text{T}} \equiv \beta \in  \mathbb{B}^* \subseteq \mathbb{R}^{6}
        \Big\},
    \end{multline}  
and suppose that on the basis of, say, goodness of fit tests, it is deemed to satisfy $\mbps^{*\text{-cs}}$ and $\mbps^{\text{cs}}$. Let us turn our attention to Table \ref{tab: calendar_time_running}. While $h_{k}^{a_k=1}(e) = p(e,a_k=1,k;\beta^*) \ni \mbps$ for every $(e,a_k=1,k) \in \mathcal{S}_{\prstudyp}^{\overline{1}}$, $h_{k}^{a_k=1}(e) \neq p(e,a_k=1,k;\beta^*)$ for every $(e,a_k=1,k) \in \mathcal{S}_{\prgapp}^{\overline{1}}$. In contrast, letting $\mbcirc$ be a supermodel of $\mbps$ defined as
        \begin{multline}
         \mbcirc  \coloneqq  \{  \text{logit}(p(e,a_k,k; \beta(\beta^{\circ}))) + \textcolor{gray}{\beta_{\geq 3}^{(e,k)}}(\beta^{\circ}) \cdot \mathbb{1}_{\{e+k \geq 3, a_k=1\}} \\  : 
         \beta^{\circ} \in \mathbb{B}^{\circ} \coloneqq \{ (\beta, \beta_{\geq 3}^{(2,1)}, \beta_{\geq 3}^{(1,2)}, \beta_{\geq 3}^{(2,2)}) : \beta \in \mathbb{B}^*, \textcolor{gray}{\beta_{\geq 3}^{(2,1)}, \beta_{\geq 3}^{(1,2)}, \beta_{\geq 3}^{(2,2)}}\in \mathbb{R}\}, \\
         \text{logit}(p(e,a_k,k;\beta(\beta^{\circ}))) \in \mbps \},
    \end{multline}
with $\beta(\beta^{\circ})$ and $\beta_{\geq}^{(e,k)}(\beta^{\circ})$ denoting, respectively, the coordinate projections of $\beta^{\circ}$ onto $\beta$ and $\beta_{\geq}^{(e,k)}$, we have that 
$\mbcirc \ni p(e,a_k=1,k;\beta^{\circ}=\beta_{\prstudyp\prgapp}) = h_{k}^{a_k=1}(e)$ holds for every $(e,a_k=1,k) \in \mathcal{S}_{\prstudyp}^{\overline{1}} \cup \mathcal{S}_{\prgapp}^{\overline{1}}$. Thus, in our running example, only $\mbcirc^{\text{cs}}$ holds: $\mbps^{*\text{-cs}}$ holds, but $\mbps^{\text{cs}}$ does not.

Clearly, \textit{only} the selection of $\mbps$ designed to satisfy \eqref{id: kiss_prstudyp} can be informed via the observed data law $P$. Returning to our running example, where $V=E$ and $\mathcal{T}=K=2$, for any $a_k \in \{0,1\}$ there are no observations to test whether $h_k^{a_k}(e) = p(e,a_k,k;\beta^*)$ holds for, say, $k=2$ and $e \in \{1,2\}$ because $P^{\text{pt},*}(E=1 \mid A = a_k, K_{\text{pt}}=2)=P^{\text{pt},*}(E=2 \mid A = a_k, K_{\text{pt}}=2)=0$: all $(v, a_k,k)$ such that $\tau = (e(v)+k) \in \{\mathcal{T}+1,\dots,\mathcal{T}+K\}$ cannot be observed as they are located during the post-study period. Such a condition indirectly imposes restrictions on the true functional form of $h_k^{a_k}(v)$ and, consequently, on the effects of $A_k$ on $Y_k$ that can be inferred. Whether the restrictions on the true functional form of $h_k^{a_k}(v)$ implied by $P^{\text{pt},*}$-identifiability of the working model, $\mbps$, are justified depends on the degree to which the regularity structure induced by $p(v,a_k,k;\beta^*) \in \mbps$ on $h_k^{a_k}(v)$ accords with the deemed contributions of $(v,a_k,k)$ on $h_k^{a_k}(v)$ during the post-study period. That is, the contributions $(v,a_k,k)$ on $h_k^{a_k}(v)$ during the post-study period, even under \CTC{}, must perfectly align with what is implied by $p(e,a_k,k;\beta^*) \in \mbps$ under \eqref{eq: cs_extrapolation}, i.e., there are no contributions of the form $ g(v,a_k,k) \cdot \mathbb{1}_{\{e(v)+k > \mathcal{T}\}}$ (for some non-trivially zero function $g$) on $h_k^{a_k}(v)$.

\subsection{Assessment of extrapolation assumptions in MSM-hazard procedures}
\label{sec_app: assessment_extrapolation_ass_msm}

A second feature of $P^{\text{pt},*}$ is that if 
${\mbps}^{*\text{-cs}}$ holds, then 
\begin{align*}
    p(v,a_k,k;\beta^*) 
\stackrel{\prstudyp}{=} P^{\text{pt},*}\left(Y \mid V=v, A=a_k, K_{\text{pt}} = k\right).
\end{align*}
See Proposition \ref{proposition: correspondence_score_equation} of Appendix \ref{sec_app: msm_hazard} for a formal proof. Through such a conceptualization, here, we propose a procedure to aid the assessment of \eqref{eq: cs_extrapolation} that is based on a user-assisted inspection of the behaviour of $h_k^{*,a_k}(v)$ outside the support of $P^{\text{pt},*}(V,A,K_{\text{pt}})$. Our strategy clarifies that extrapolation assumptions will likely not be met whenever characteristic features of the estimand of interest change from the study-period to the post-study period and when such features are captured as a change in the contribution of $(v,a_k,k)$ on $h_k^{a_k}(v)$ during the post-study period from what is otherwise dictated by the working model, $p(v,a_k,k;\beta^*) \in \mbps$.

\subsubsection{The procedure}
\label{sec_app: procedure_assess_extrapolation_assumption_msm_hazard}
\begin{table}[!htpt]
\centering
\scriptsize
\setlength{\tabcolsep}{6pt} 
\resizebox{\textwidth}{!}{%
\begin{tabular}{@{}lccc:cc@{}}
\hline
Calendar time year(s) & 87' & 88--89' & 90--93' & 93--95' & 95--97' \\
\hline
$\tau$ & 0 & 1 & 2 & 3 & 4 \\
Study($\bullet$) - post-study($\circ$) period & $\bullet$ & $\bullet$ & $\bullet$ & $\circ$ & $\circ$ \\
$(e,k)$ & (0,0) & (0,1), (1,0) & (1,1), (2,0), (0,2) & (2,1), (1,2) & (2,2) \\
$(e,a_k=1,k) \in\mathcal{S}_{\prstudyp}^{\overline{1}}$ & $\{(0,1,0)$ & $(0,1,1),(1,1,0)$ & $(1,1,1),(2,1,0),(0,1,2)\}$ &  &  \\ 
$(e,a_k=1,k) \in\mathcal{S}_{\prgapp}^{\overline{1}}$ &  &  &  & $\{(2,1,1),(1,1,2)$ & $\{(2,1,2)\}$ \\
$h_k^{a_k=1}(e)$ & 0.05 & 0.13, 0.05 & 0.11, 0.04, 0.29 & 0.07, 0.24 & 0.20 \\
$p(e,a_k=1,k;\beta^*) \in \mbps $ & 0.05 & 0.13, 0.05 & 0.11, 0.04, 0.29 & 0.09, 0.26 & 0.21 \\
$p(e,a_k=1,k;\beta_{\prstudyp\prgapp}) \in \mbcirc$ & 0.05 & 0.13, 0.05 & 0.11, 0.04, 0.29 & \textcolor{blue}{0.07}, \textcolor{blue}{0.24} & \textcolor{blue}{0.20} \\
\hline\hline
\multicolumn{6}{@{}l@{}}{}\\
\multicolumn{6}{@{}c@{}}{%
  \begin{tabular}{@{}lcccccc:ccc@{}}
    & $\beta_0$ & $\beta_{k=0}$ & $\beta_{k=1}$ & $\beta_{E=0}$ & $\beta_{E=1}$ & $\beta_A$ & $\beta_{\geq 3}^{(2,1)}$ & $\beta_{\geq 3}^{(1,2)}$ & $\beta_{\geq 3}^{(2,2)}$\\
    \hline
    $\mathbb{B}^* \ni \beta^*$ & $-1$ & $-2$ & $-1$ & $0.4$ & $0.25$ & $-0.3$ & $\times$ & $\times$ & $\times$ \\
    $\mathbb{B}^{\circ} \ni \beta^{\circ,*}$ & $-1$ & $-2$ & $-1$ & $0.4$ & $0.25$ & $-0.3$ & $0$ & $0$ & $0$\\
    $\mathbb{B}^{\circ} \ni \beta_{\prstudyp\prgapp}$ 
& $-1$ & $-2$ & $-1$ & $0.4$ & $0.25$ & $-0.3$ & \textcolor{blue}{$-0.287$} & \textcolor{blue}{$-0.103$} & \textcolor{blue}{$-0.086$} \\
  \end{tabular}
}\\[30pt]
\multicolumn{6}{@{}c@{}}{%
  \begin{tabular}{@{}lccccccccc@{}}
    & $\gamma_{(0,0)}^{*,\cdot}$ & $\gamma_{(0,1)}^{*,\cdot}$ & $\gamma_{(1,0)}^{*,\cdot}$ & $\gamma_{(1,1)}^{*,\cdot}$ & $\gamma_{(2,0)}^{*,\cdot}$ & $\gamma_{(0,2)}^{*,\cdot}$ & $\gamma_{e=2}^{*,\cdot}$ & $\gamma_{k=2}^{*,\cdot}$ & $\gamma_{\geq 3}^{*,\cdot}$ \\
    \hline
    $\mathbb{\Gamma}^* \ni \gamma^{*, a_k=0}$ & 0.07 & 0.17 & 0.06 & 0.15 & 0.10 & 0.20 & -0.05 & 0.15 & 0.17 \\
    $\mathbb{\Gamma}^* \ni \gamma^{*, a_k=1}$ & 0.05 & 0.13 & 0.05 & 0.11 & 0.09 & 0.17 & -0.05 & 0.12 & 0.14 \\
  \end{tabular}
}\\[8pt]
\hline
\end{tabular}
}
\caption{Stylized scenario drawn from \citep{hernan_marginal_2000}. The model $\mbps$ and the parameter $\beta^*$ inspired by \citet[Table~2, p.~564]{hernan_marginal_2000}. In the original analysis, time intervals differ in length; here, to facilitate the illustration of our points, we assumed time points are equally-spaced.}
\label{tab: calendar_time_running}
\end{table}

To illustrate this, consider our running example and let us place ourselves from the perspective of a researcher in \citet{hernan_marginal_2000} analyzing data from the MACS cohort under the working model $\mbps$. Our procedure requires the following ordered steps. 
\begin{enumerate}
\item Select a set of terms of the type $\mathbb{1}_{\{e+k>\mathcal{T}\}} \cdot g$ for some function $g$ through which meaningful \textit{deviations}
in the contributions of $v,a_k,k$ on $h_k^{a_k}(v)$ for $(v,a_k,k)$ in $\mathcal{S}_{\prgapp}^{\overline{a}_k}$, relative to what is dictated by $p(e,a_k,k;\beta^*)$, can be made explicit. Call $\mathcal{U}$ the set of such \textit{calendar-time-interaction terms}. 
Returning to our running example, during the time period 93'-95', here $\tau=3$, combined therapy (e.g., zidovudine + ddl) constituted the most prevalent treatment \citep{detels_effectiveness_1998} in the MACS cohort. 
 \citep{hernan_marginal_2000} assumed that patients remain on zidovudine once they start it. Thus, it is possible that a subset of individuals who initiated monotherapy at say, $\tau=2$, received combined therapy, i.e., zidovudine, during or after $\tau=3$. Hence, if interest lies in zidovudine administered as either monotherapy or combined therapy, then (given that combined therapy has been shown to improve survival \citep{darbyshire_delta_1996}) subject-matter experts can inspect whether the hazards 
 implied by the procedure coincide with those deemed to hold given such state of knowledge; that is, subject matter experts can ascertain whether
$p(e,a_k,k;\beta^*)$ actually accord with what $h_k^{a_k=1}(e)$ is deemed to be for $(v,a_k,k) \in \mathcal{S}_{\prgapp}^{\overline{1}}$, through the calendar-time-interaction term $\mathbb{1}_{\{e+k \geq 3\}}\mathbb{1}_{\{a_k=1\}}=\mathbb{1}_{\{e+k \geq 3, a_k=1\}} \in \mathcal{U}$. 

\item Select a set of variables
$\mathcal{B}_{\mathcal{U}} \equiv \mathcal{B}_{\mathcal{U}}(\mathcal{U})$, a user-specified bijective link-function $\Phi$, and a function $\gamma_{\mathcal{U}}$  
such that for any parameter $\beta \in \mathbb{B}^*$ indexing the working model $\mbps$, there exists a unique parameter $\gamma \equiv \gamma_{\mathcal{U}}(\beta)$ in the set $ \mathbb{\Gamma}$ indexing the saturated model with variables $\mathcal{B}_{\mathcal{U}}$,
$\mathcal{M}_{\mathcal{U}}(\mathbb{\Gamma}) \coloneqq \{ \Psi(e,a_k,k;\gamma) \equiv \sum_{v_j(e,a_k,k) \in \mathcal{B}_{\mathcal{U}}}v_j(e,a_k,k) \cdot  \gamma_{j} : (\gamma_{1}, \dots,  \gamma_{d_{\mathcal{U}}})\equiv \gamma \in \mathbb{\Gamma}\}$, satisfying 
\begin{align}\label{eq: identity_saturated_with_working}
\Psi(e,a_k,k;\gamma) \stackrel{\prstudyp\prgapp}{=} \Phi ( p(e,a_k,k;\beta)).
\end{align}
In our running example $p(e,a_k,k;\beta^*)$ is equal to
\begin{multline}\label{eq: reparametrization_new_msm_spec}
\Psi(e,a_k,k;\gamma^*) = \sum_{(e^\prime,k^{\prime}) : e^\prime + k^{\prime} \leq 2 } \mathbb{1}_{\{e=e^{\prime}, k=k^{\prime}\}} \cdot \gamma_{(e^{\prime},k^{\prime})}^{*, a_k} \\
+ \mathbb{1}_{\{e=2\}} \cdot \gamma_{e=2}^{*, a_k}  + \mathbb{1}_{\{k=2\}} \cdot \gamma_{k=2}^{*, a_k} \\
+ \mathbb{1}_{\{e+k \geq 3\}} \cdot \gamma_{\geq 3}^{*, a_k},
\end{multline}
whereby $\Phi = \text{identity}$, $\mathcal{B}_{\mathcal{U}} = \{ \mathbb{1}_{\{e=e^{\prime}, k=k^{\prime}\}}:e^{\prime}+k^{\prime} \leq 2 \} \cup \{ \mathbb{1}_{\{e=2\}}\}\cup \{\mathbb{1}_{\{k=2\}}\}\cup \left(\mathcal{U} \equiv \{\mathbb{1}_{\{e+k \geq 3 \}}\}\right)$, and $\gamma^* \coloneqq (\gamma^{*, a_k=0}, \gamma^{*, a_k=1})$ as tabulated in Table \ref{tab: calendar_time_running}. We emphasize that both the choice of the scale $\Phi$ and that of $\mathcal{U}$, rather than serving computational convenience, are intended to facilitate the assessment of whether the contributions of $(e,a_k,k) \in \mathcal{S}_{\prgapp}^{\overline{1}}$ on $h_k^{a_k}(e)$, as dictated by $p(e,a_k,k;\beta^*) \stackrel{\prstudyp\prgapp}{=} \Psi(e,a_k,k;\gamma^*)$, accord with what is deemed true on subject-matter grounds. 
\item If so, then $\mbps^{\text{cs}}$ is deemed tenable; if not, then $p(e,a_k,k;\beta^*) \stackrel{\prstudyp\prgapp}{=} h_{k}^{a_k}(e)$ is not deemed to hold and, correspondingly, $\mbps^{\text{cs}}$ is not deemed tenable. 
\end{enumerate}

\subsubsection{Existence of saturated model $\mathcal{M}_{\mathcal{U}}(\mathbb{\Gamma})$ formed by user-specified set of variables $\mathcal{U}$}

Throughout Section \ref{sec_app: limitations_first_sensitivity_analysis} of this Appendix, for ease of exposition, denote $\Omega$ the set of all possible strata $(v,a_k,k) \in \text{supp}(P(V)) \times \{0,1\} \times \{0, \dots K\}$, with $|\Omega| = N$. Let $\Omega_{\text{id}} \coloneqq \{ \omega \in \Omega : P^{\text{pt},*}(\omega) > 0 \}$ denote the combinations in the support of $P^{\text{pt},*}(V,A,K_{\text{pt}})$ (as implied by $\RAID$) and $\Omega_{\neg\text{id}} \coloneqq \{ \omega \in \Omega : P^{\text{pt},*}(\omega) = 0 \}$ such that $\Omega = \Omega_{\text{id}} \cup \Omega_{\neg\text{id}}$.

Let $\mathbf{V}_{\mathcal{M}}$ denote the design matrix for a $P^{\text{pt},*}$-identifiable model $\mbps$. 

Then, $P^{\text{pt},*}$-identifiability (which we subsequently make precise) implies that for any hazard vector 
\begin{align*}
    \mathbf{h}^* \coloneqq \text{vec}\left( \{p(v,a_k,k;\beta^*) \in \mbps : (v,a_k,k) \in \Omega\} \right) \in \mathbb{R}^N,
\end{align*}
there exists a unique $\beta^*$ such that 
\begin{align*}
    \Lambda^* \coloneqq \mu(\mathbf{h}^*) = \mathbf{V}_{\mathcal{M}} \beta^*.
\end{align*}

We want to find a set of variables $\mathcal{B}_{\mathcal{U}} \equiv \mathcal{B}_{\mathcal{U}}(\mathcal{U})$ defining the saturated model $\mathcal{M}_{\mathcal{U}}(\mathbb{\Gamma})$ indexed by $\mathbb{\Gamma}$. Let $\mathcal{U} = \{ v_1, \dots, v_m \}$ be the set of user-specified terms, and define 
\begin{align*}
    \mathbf{v}_{j} =& \text{vec}\left( \{v_j(v,a_k,k) : (v,a_k,k) \in \Omega\}\right), \quad j \in \{1, \dots, m\}, \\
    \mathbf{V}_{\mathcal{U}, \text{sub}} =& [\mathbf{v}_{1} \mid \dots \mid \mathbf{v}_{m}].
\end{align*}
Since $\Omega$ is finite, the space of all possible hazard functions is isomorphic to $\mathbb{R}^N$. Thus, on the scale induced by $\Phi$, manufacturing a saturated model containing the user-specified terms in $\mathcal{U}$ reduces to completing the corresponding vectors to a basis in $\mathbb{R}^N$.

Let, with a slight abuse of notation, $\mathbf{V}_{\mathcal{U}}$ indicate both the set of basis vectors composing the design matrix and the design matrix itself. Then,
\begin{enumerate}
    \item if $\text{rank}(\mathbf{V}_{\mathcal{U},\text{sub}}) = m$, then, by the basis extension theorem, this set can be completed to form a full saturated basis $\mathbf{V}_{\mathcal{U}} \coloneqq  [\mathbf{V}_{\mathcal{U},\text{sub}} \mid \mathbf{V}_{\mathcal{U},\text{comp}}] \in \mathbb{R}^{N \times N}$; the choice of $\mathbf{V}_{\mathcal{U},\text{comp}}$ is not unique and can be chosen to accommodate the user's preferences.
    \item if $\text{rank}(\mathbf{V}_{\mathcal{U},\text{sub}}) < m$, the set $\mathcal{U}$ induces a design matrix $\mathbf{V}_{\mathcal{U},\text{sub}}$ with linearly dependent columns and therefore the user-specified vectors cannot all be included as distinct elements of a basis of $\mathbb{R}^N$. Thus, one must redefine $\mathcal{U}$ and go to step 1.
\end{enumerate}
\subsubsection{Non-uniqueness of the parameter for a given user-specified choice of coefficients}
\label{sec_app: non_uniqueness_projection}
Suppose the user found a saturated model $\mathcal{M}_{\mathcal{U}}(\mathbb{\Gamma})$. Suppose the investigator believes the true coefficients expressed in the saturated basis are  $\gamma_{\text{u}}$. 
Let $\mathbf{h}_{\text{u}} \coloneqq \text{vec}\left( \{h_{\text{u},k}^{a_k}(v) : (v,a_k,k) \in \Omega\}\right) \in \mathbb{R}^N$ the vector of hazards specified by the user and let 

\begin{align*}
    \Lambda_{\text{u}} \coloneqq \mu(\mathbf{h}_{\text{u}}) = \mathbf{V}_{\mathcal{U}} \gamma_{\text{u}}.
\end{align*}
Suppose $\gamma_{\text{u}}$ differs from the coefficients $\gamma^*$ satisfying $\Lambda^* \coloneqq \mathbf{V}_{\mathcal{U}} \gamma^*$ (such a $\gamma^*$ always exists as $\mathbf{V}_{\mathcal{U}}$ is full rank), only in those coefficients that cannot be identified from $P^{\text{pt},*}$. 

For instance, in the running example of Section \ref{sec_app: procedure_assess_extrapolation_assumption_msm_hazard}, the coefficient $\gamma_{\geq 3}^{*,a_k=1}$, corresponding to the interaction term $\mathbb{1}_{\{e+k \geq 3,a_k=1\}}$ is not $P^{\text{pt},*}$-identifiable as $P^{\text{pt},*}( E=e, K_{\text{pt}} = k) = 0$ for $k +e \geq 3$. Then, 
in the following Proposition we show that there is no $\beta_{\text{u}} \in \mathbb{B}^*$ satisfying $\mathbf{V}_{\mathcal{U}} \gamma_{\text{u}} = \mathbf{V}_{\mathcal{M}} \beta_{\text{u}}$. 

\begin{proposition}[Projection necessity]\label{proposition: projection_necessity}
Let $\gamma_{\text{u}} = \gamma^* + \Delta_{\text{u}}$ for some $\Delta_{\text{u}} \neq 0$.
If $\Delta_{\text{u}}$ is non-zero only at components corresponding to basis functions in $\mathcal{B}_{\mathcal{U}}$ supported exclusively on $\Omega_{\neg \text{id}}$, then $\Lambda_{\text{u}} \notin \text{span}(\mathbf{V}_{\mathcal{M}})$.  
\end{proposition}
\begin{proof}
Denote $\Lambda_{\Delta} \equiv \mathbf{V}_{\mathcal{U}}\Delta_{\text{u}}$. We have
\begin{align*}
        \Lambda_{\text{u}} = \mathbf{V}_{\mathcal{U}} \gamma_{\text{u}}
= \mathbf{V}_{\mathcal{U}} \left( \gamma^* +  \Delta_{\text{u}}\right)
= \Lambda^* +  \mathbf{V}_{\mathcal{U}}\Delta_{\text{u}}
\equiv\Lambda^* + \Lambda_{\Delta}.
\end{align*}
We want to show that $\Lambda_{\Delta} \notin \text{span}(\mathbf{V}_{\mathcal{M}})$. We establish the proof by contradiction. Let $\mathbf{P}_{\text{id}}$ be the linear operator projecting onto the subspace supported on the observed law $P^{\text{pt},*}$, $\Omega_{\text{id}}$, and  recall
that a model is $P^{\text{pt},*}$-identifiable if and only if 
\begin{align}
    &\lambda(v,a,k;\beta_1) = \lambda(v,a,k;\beta_2) \quad \text{for every } (v,a_k,k) \in \Omega_{\text{id}} \Rightarrow \beta_1 = \beta_2,  \notag\\
    \Leftrightarrow&\mathbf{P}_{\text{id}} \mathbf{V}_{\mathcal{M}} \beta_1 = \mathbf{P}_{\text{id}} \mathbf{V}_{\mathcal{M}} \beta_2 \Rightarrow \beta_1 = \beta_2 \notag\\
    \Leftrightarrow& \text{rank}(\mathbf{P}_{\text{id}} \mathbf{V}_{\mathcal{M}}) = \text{rank}(\mathbf{V}_{\mathcal{M}})  \notag\\
    \Leftrightarrow& \text{ker}(\mathbf{P}_{\text{id}} \mathbf{V}_{\mathcal{M}}) = \{0\}. \label{cond: trivial_kernel}
\end{align}

Letting
{\footnotesize
\begin{align*}
    (i) \equiv& [ \mathbf{P}_{\text{id}} \Lambda_{\Delta} = \mathbf{0} \Leftarrow \Lambda_{\Delta}(v,a_k,k) = 0 \text{ for every } (v,a_k,k) \in \Omega_{\text{id}} \\
    &\text{ since } \Lambda_{\Delta} \text{ is a linear combination of basis functions supported exclusively on } \Omega_{\neg \text{id}} ],
\end{align*}
}
we have
\begin{align*}
    \Lambda_{\Delta} \in \text{span}(\mathbf{V}_{\mathcal{M}}) \Rightarrow& \Lambda_{\Delta} = \mathbf{V}_{\mathcal{M}} c,\quad  \text{for some } c &\\
    \Rightarrow &\mathbf{P}_{\text{id}} \Lambda_{\Delta} = \mathbf{P}_{\text{id}} \mathbf{V}_{\mathcal{M}} c,\quad  \text{for some } c & \\
    \Rightarrow&\mathbf{0}=\mathbf{P}_{\text{id}} \mathbf{V}_{\mathcal{M}} c,\quad  \text{for some } c&\Leftarrow{(i)} \\
        \Rightarrow&c=\mathbf{0}&\Leftarrow{\eqref{cond: trivial_kernel}} \\
    \Rightarrow& \Lambda_{\Delta} = \mathbf{0}& 
\end{align*}
which is a contradiction since $\mathbf{V}_{\mathcal{U}}$ is full rank and thus $\Delta_{\text{u}} \neq \mathbf{0} \Rightarrow \Lambda_{\Delta} \neq \mathbf{0}$.
\end{proof}

Thus, since under \CTC{} the user will likely select $\gamma_{\text{u}}$ as in the conditions of Proposition \ref{proposition: projection_necessity},
${\mbps}_{\mid V}^{\text{cs}}$ will typically fail as no $\beta_{\text{u}} \in \mathbb{B}^*$ satisfies $\mathbf{V}_{\mathcal{U}} \gamma_{\text{u}} = \mathbf{V}_{\mathcal{M}} \beta_{\text{u}}$.

\subsection{A first strategy as a sensitivity analysis for extrapolation assumptions}
\label{sec_app: first_sensitivity_analysis}
An investigator might nonetheless wish to select a $\beta_{\text{u}}$ that aligns, in a sense made precise below, with $\mathbf{h}_{\text{u}}$ implied by a given choice of $\gamma_{\text{u}}$. Returning to our running example (where $E = V$), suppose we have grounds to encode changes due to the adoption of combined therapy from $\tau = 3$ onward as a decrease in the coefficient $\gamma_{\geq 3}^{*,a_k=1}$ from $0.14$ to, say, $0.12 \eqqcolon \gamma_{\text{u}, \geq 3}^{a_k=1}$. This amounts to assuming that changes in the standard of care not captured by statistical procedures based on $P^{\text{o}}$ reduce the linear predictor by $0.14 - 0.12 = 0.02$ for any individual at risk at any calendar time $\tau \in \{3,4\}$. The analogy with $\beta_A^*$ (up to scale) is straightforward: at \textit{any individual-specific follow-up time} $k$, $\beta_A^* = - 0.3$ in \eqref{eq: stylized_msm} implies a reduction in the rate of failure due to receiving treatment, in the logit scale, of 0.3. Letting $\gamma_{\text{u}} \coloneqq (\gamma^* \setminus \gamma_{\geq 3}^{*,a_k=1}, \gamma_{\text{u}, \geq 3}^{a_k=1})$ and $h_{\text{u},k}^{a_k}(e) \coloneqq \Phi^{-1}\Psi(e,a_k,k;\gamma_{\text{u}})$, we can choose a $\beta_{\text{u}}$ that aligns with 
\begin{align*}
    \mathbf{h}_{\text{u}} \coloneqq \textnormal{vec} \left( \{ h_{\text{u},k}^{a_k}(e) : (v \equiv e, a_k, k) \in \Omega\}\right).
\end{align*}

One way of summarizing the discrepancy between $\Lambda_{\mathrm u}$ and the working model is to focus on the treatment contrasts. Here, we illustrate this approach by projecting the stratum-specific treatment contrasts implied by $\Lambda_{\mathrm u}$ onto the space of constant treatment contrasts. We offer a high-level idea of such a strategy in the context of proportional MSM-hazard models, as in \eqref{eq: stylized_msm} or \eqref{eq: model_hernan_00}, although this approach extends to more elaborate models. Let
\begin{align*}
    \Omega^- \coloneqq \{(v,k) : (v,a_k=1,k) \in \Omega\},
\end{align*}
and define the stratum-specific treatment contrasts
\begin{align*}
    D_{\text{u}}(v,k)
    \coloneqq
    \mu(h_{\text{u},k}^{a_k=1}(v))
    -
    \mu(h_{\text{u},k}^{a_k=0}(v)),
    \qquad (v,k)\in\Omega^-.
\end{align*}
In a proportional MSM-hazard model, $\beta_{\text{u},A}$ represents a
constant treatment contrast. It can therefore be selected by projecting
the collection of contrasts
$\{D_{\text{u}}(v,k):(v,k)\in\Omega^-\}$ onto the space of constant treatment contrasts. 

For nonnegative weights $w_{(v,k)}$ satisfying
\begin{align*}
    w_{(v,k)} &\geq 0, \qquad (v,k)\in\Omega^-, \\
    \sum_{(v,k)\in\Omega^-} w_{(v,k)} &= 1,
\end{align*}
consider
\begin{align*}
    \beta_{\text{u},A}
    =
    \arg\min_{b\in\mathbb R}
    \sum_{(v,k)\in\Omega^-}
    w_{(v,k)}
    \cdot {\left( D_{\text{u}}(v,k)-b\right)}^2.
\end{align*}
The solution is
\begin{align*}
    \beta_{\text{u},A}
    =
    \sum_{(v,k)\in\Omega^-}
    w_{(v,k)} \cdot D_{\text{u}}(v,k).
\end{align*}

For example, placing all weight on a $(v,k)$ pair with the smallest
treatment contrast gives
\begin{align*}
    \beta_{\text{u},A}
    =
    \inf_{(v,k)\in\Omega^-}
    \left(
    \mu(h_{\text{u},k}^{a_k=1}(v))
    -
    \mu(h_{\text{u},k}^{a_k=0}(v))
    \right),
\end{align*}
i.e., the largest treatment benefit when negative contrasts are
beneficial. Vice versa, placing all weight on a $(v,k)$ pair with the largest treatment contrast gives
\begin{align*}
    \beta_{\text{u},A}
    =
    \sup_{(v,k)\in\Omega^-}
    \left(
    \mu(h_{\text{u},k}^{a_k=1}(v))
    -
    \mu(h_{\text{u},k}^{a_k=0}(v))
    \right).
\end{align*}
Finally, applying uniform weights,
$w_{(v,k)}=|\Omega^-|^{-1}$, gives
\begin{align*}
    \beta_{\text{u},A}
    =
    \frac{1}{|\Omega^-|}
    \sum_{(v,k)\in\Omega^-}
    \left(
    \mu(h_{\text{u},k}^{a_k=1}(v))
    -
    \mu(h_{\text{u},k}^{a_k=0}(v))
    \right).
\end{align*}
Thus, $\beta_{\text{u},A}$ represents the unweighted average of the
stratum-specific treatment contrasts on the additive link-function
(e.g., logit) scale.

We emphasize that all three choices above have properties that differ
from those implied by the extrapolation assumptions. Under
${\mbps}^{\text{cs}}$,
\begin{align*}
    &\operatorname{logit}\left(p(v,a_k=1,k;\beta^*)\right)
    -
    \operatorname{logit}\left(p(v,a_k=0,k;\beta^*)\right)
    =
    \beta_A^*
\end{align*}
for every stratum $V=v$ and every follow-up time $k$, i.e., $\beta_A^*$ represents a treatment contrast that is constant across $v$ and $k$.

\subsubsection{Limitations of the projection-based strategy}
\label{sec_app: limitations_first_sensitivity_analysis}
The procedure presented in the preceding section requires (i) to find a saturated model $\mathcal{M}_{\mathcal{U}}(\mathbb{\Gamma})$, (ii) quantitatively encode post-study variations in this model, (iii) select an appropriate projection, and (iv) interpret the resulting parameter $\beta_{\text{u},A}$ as a function of the chosen projection.

These requirements present substantial practical and theoretical challenges. Even if a saturated model can be specified, applied researchers often lack the precise quantitative knowledge required to meaningfully encode unobserved post-study variations. Furthermore, the choice of projection determines the interpretation of $\beta_{\text{u}, A}$, i.e., the population for which the effects are defined. It is highly impractical to expect applied researchers to contextually justify a specific projection matrix for every analysis. Furthermore, even if we, by default, resorted to, say, the equal-weight least squares solution projection which leads to $\beta_{\text{u},A} = {\mid \Omega^- \mid}^{-1} \cdot \sum_{(v,k) \in \Omega^-}\left(\mu(h_{\text{u},k}^{a_k=1}(v)) - \mu(h_{\text{u},k}^{a_k=0}(v))\right)$, then a decision based on parameter $\beta_{\text{u},A}$ might be sensitive to the values $h_{\text{u},k}^{a_k}(v)$ assumed near the boundaries, $\{0,1\}$. This carries the implication that it might always be possible to find a parameter $\beta_{\text{u},A}$ whose sign is opposed to that of $\beta_{A}^*$. As an example, consider the following heuristic argument. Because the logit link function $\mu(x) = \log({x} \cdot {(1-x)}^{-1})$ diverges towards $+\infty$ as $x \to 1$ and $-\infty$ as $x \to 0$, positing that even a single stratum $(v', k')$ experiences an unobserved post-study hazard under treatment $h_{\text{u},k'}^{a_k=1}(v')$ extremely close to $1$ (while the untreated hazard $h_{\text{u},k'}^{a_k=0}(v')$ remains moderate) produces an arbitrarily large positive value for $\mu(h_{\text{u},k'}^{a_k=1}(v')) - \mu(h_{\text{u},k'}^{a_k=0}(v'))$. Because $\beta_{\text{u},A}$ is a simple arithmetic mean, this single extreme log-odds difference can easily dominate the entire sum, leading to a positive $\beta_{\text{u},A}$, even if the treatment effect remains strictly negative (with negative logit differences) for all other $\mid \Omega^- \mid -1$ observations in $\Omega^-$.

This also makes it clear that if we were to define a measure of extrapolation based on an automation of the above user-assisted procedure introduced in Section \ref{sec_app: procedure_assess_extrapolation_assumption_msm_hazard}, it would come with undesirable properties.\footnote{For example, one can consider $\gamma_{\text{u}}$ of the form $\gamma_{\text{u}} = \gamma^* + \Delta_{\text{u}}$ where the components of $\Delta_{\text{u}}$ corresponding to $(v,a_k,k) \in \Omega_{\text{id}}$ are zero and those $(v,a_k,k) \in \Omega_{\neg \text{id}}$  are equal to $\Delta_{\neg \text{id}}^{a=0}$ for $a_k=0$ and $\Delta_{\neg \text{id}}^{a=1}$ for $a_k=1$, with $\Delta_{\neg \text{id}}^{a=0}$ and $\Delta_{\neg \text{id}}^{a=1}$ two constants.}

\subsection{C-values}
\label{sec_app: c_value}

\subsubsection{Preliminaries}
\label{sec_app: c_values_preliminaries}
Let $E \in V \subseteq L_0$. For every $k \in \{0, \dots, K\}$ we will consider two strategies, $\overline{a}_k \in \{\overline{0}, \overline{1}\}$. To lighten notation, we will therefore use the symbol $a \in \{0,1\}$ ($\overline{a}_k = \overline{1} \cdot a$) and we let $R_{\prstudyp,k}^{\rightsquigarrow, a} \equiv {\prstudyp}_k^{a}$ and $R_{\prgapp,k}^{\rightsquigarrow, a} \equiv {\prgapp}_k^{a}$, so that $R_k^{\textnormal{c}, a}(\textnormal{c}_k^a) = {\prstudyp}_k^{a} + \textnormal{c}_k^a \cdot {\prgapp}_k^{a}$.

For every $a \in \{0, 1\}$, every $k \in \{0, \dots, K\}$ and every $v \in \text{supp}(P(V))$, consider:
\begin{align}\label{eq: r_k}
    h_k^{\rightsquigarrow,a}(v),& \notag\\
    S_{-1}^{\rightsquigarrow,a}(v)& \coloneqq 1, \notag\\
    S_{k}^{\rightsquigarrow,a}(v)& \coloneqq \prod_{m=0}^k(1 - h_m^{\rightsquigarrow,a}(v)) \notag, \\
    r_k^{\rightsquigarrow,a}(v)& \coloneqq S_{k-1}^{\rightsquigarrow,a}(v) \cdot  h_k^{\rightsquigarrow, a}(v).
\end{align} 
For individuals with $V=v$, $h_k^{\rightsquigarrow,a}(v)$ denotes the hazard under extrapolation assumptions (e.g., in MSM-hazard procedures $h_k^{\rightsquigarrow,a}(v) \coloneqq p(v,a_k=a,k;\beta^*)$), while $r_k^{\rightsquigarrow,a}(v)$ represents the proportion of individuals who first experience failure at follow-up time $k$ under extrapolation assumptions and it is given by the product of the survival up to the previous follow-up, $S_{k-1}^{\rightsquigarrow, a}(v)$, and the hazard at the current follow-up, $h_k^{\rightsquigarrow,a}(v)$.

Further, for every $a \in \{0, 1\}$, every $k \in \{0, \dots, K\}$ and every $v \in \text{supp}(P(V))$ let
\begin{align}
    \prstudyp_{k}^a(v) \coloneqq& \sum_{m=0}^{k} \mathbb{1}_{\{e+m \leq \mathcal{T}\}} \cdot r_m^{\rightsquigarrow,a}(v), \text{ and } \label{eq: pr_studyp_v}\\
    \prgapp_{k}^a(v) \coloneqq& \sum_{m=0}^{k} \mathbb{1}_{\{e+m > \mathcal{T}\}} \cdot r_m^{\rightsquigarrow,a}(v),\label{eq: pr_gapp_v}
\end{align}
denote the proportion of individuals with $V=v$ who experience failure by follow-up $k$ during the study and post-study period, respectively.

The following lemma permits us to express $\prstudyp_{k}^a(v)$ and $\prgapp_{k}^a(v)$ as a function of survival.

\begin{lemma}[Risk decomposition]
Let $k \in \{1, \dots, K\}$. For an individual $v$ with entry time $e \equiv e(v)$:
    \begin{align*}
    \prstudyp_{k}^a(v) =& 1 - S_{\min(k, \mathcal{T}-e)}^{\rightsquigarrow,a}(v), \\
    \prgapp_{k}^a(v) =& \mathbb{1}_{\{e+k > \mathcal{T}\}} \left(S_{\mathcal{T}-e}^{\rightsquigarrow,a}(v) - S_{k}^{\rightsquigarrow,a}(v)\right).
    \end{align*}
Marginalizing over $P(V)$ we obtain the marginal risks:
    \begin{align*}
    \prstudyp_k^a =& \expectation\left( 1 - S_{\min(k, \mathcal{T}-E)}^{\rightsquigarrow, a}(V)\right),\\
    {\prgapp}_k^{a} =& \expectation\left(\mathbb{1}_{\{E+k > \mathcal{T}\}} \cdot (S_{\mathcal{T}-E}^{\rightsquigarrow,a}(V) - S_{k}^{\rightsquigarrow,a}(V))\right).
\end{align*}
\end{lemma}
\begin{proof}
    Let $k \in \{1, \dots, K\}$. We have:
{\footnotesize
\begin{align*}
    \prstudyp_{k}^a(v) =& \sum_{m=0}^{k} \mathbb{1}_{\{e+m \leq \mathcal{T}\}} \cdot r_m^{\rightsquigarrow,a}(v) & \Leftarrow{\eqref{eq: pr_studyp_v}}\\
    =& \sum_{m=0}^{\min(k, \mathcal{T}-e)} r_m^{\rightsquigarrow,a}(v) \\
    =& \sum_{m=0}^{\min(k, \mathcal{T}-e)} \left( S_{m-1}^{\rightsquigarrow,a}(v) - S_{m}^{\rightsquigarrow,a}(v) \right) &\Leftarrow{r_m^{\rightsquigarrow,a}(v) = S_{m-1}^{\rightsquigarrow,a}(v) - S_{m}^{\rightsquigarrow,a}(v)}\\
    =& 1 - S_{\min(k, \mathcal{T}-e)}^{\rightsquigarrow,a}(v)&\Leftarrow{\text{telescoping}}\\
\end{align*}
}
and similarly for the post-study gap:
{\footnotesize
\begin{align*}
        \prgapp_{k}^a(v) =& \sum_{m=0}^{k} \mathbb{1}_{\{e+m > \mathcal{T}\}} \cdot r_m^{\rightsquigarrow,a}(v)& \Leftarrow{\eqref{eq: pr_gapp_v}} \\
    =& \mathbb{1}_{\{e+k> \mathcal{T}\}} \cdot \left(\sum_{m=\mathcal{T}-e+1}^{k} r_m^{\rightsquigarrow,a}(v)\right) \\
        =& \mathbb{1}_{\{e+k> \mathcal{T}\}} \cdot\left(\sum_{m=0}^{k} r_m^{\rightsquigarrow,a}(v) - \sum_{m=0}^{\mathcal{T}-e} r_m^{\rightsquigarrow,a}(v) \right)\\
    =&\mathbb{1}_{\{e+k > \mathcal{T}\}} \cdot \left(S_{\mathcal{T}-e}^{\rightsquigarrow,a}(v) - S_{k}^{\rightsquigarrow,a}(v)\right) & \Leftarrow{r_m^{\rightsquigarrow,a}(v) = S_{m-1}^{\rightsquigarrow,a}(v) - S_{m}^{\rightsquigarrow,a}(v)}.
\end{align*}
}
The remaining results follow trivially.
\end{proof}

\begin{corollary}
Let $a \in \{0,1\}$, $k \in \{1, \dots, K\}$. Suppose ${\prgapp}_k^{a} > 0$. Shifted risks are compatible with risks under extrapolation assumptions if:
\begin{align}\label{eq: bounds_shifted_gapp}
    \textnormal{c}_k^a \cdot {\prgapp}_k^{a} \leq \expectation\left(\mathbb{1}_{\{E+k > \mathcal{T}\}} \cdot S_{\mathcal{T}-E}^{\rightsquigarrow,a}(V)\right),
\end{align} 
i.e., $\textnormal{c}_k^a \cdot {\prgapp}_k^{a} $ cannot exceed the proportion of individuals alive at follow-up $\mathcal{T}-E$ under $\overline{a}_k$ and such that $E+k > \mathcal{T}$. Therefore, we will consider
\begin{align}\label{eq: max_variation_bullet}
\textnormal{c}_k^a \in& \left[0, \textnormal{c}_k^{a,+}\right] \text{ with } \notag \\
 {\textnormal{c}}_k^{a,+} \coloneqq& \frac{\expectation\left(\mathbb{1}_{\{E+k > \mathcal{T}\}} \cdot S_{\mathcal{T}-E}^{\rightsquigarrow,a}(V)\right)}{\prgapp_k^{a}},
\end{align}
guaranteeing that 
\begin{align}\label{eq: shifted_risks_are_bounded}
     {R_k^{\textnormal{c}, a}(\textnormal{c}_k^{a,+}) \leq 1}.
\end{align}
\end{corollary}
\begin{proof}
    \eqref{eq: bounds_shifted_gapp} trivial; \eqref{eq: max_variation_bullet} follows directly from \eqref{eq: bounds_shifted_gapp}. For \eqref{eq: shifted_risks_are_bounded} consider
    {\tiny    
    \begin{align*}
    \prstudyp_k^a + \textnormal{c}_k^{a,+} \cdot \prgapp_k^a =& \expectation\left( 1 - S_{\min(k, \mathcal{T}-E)}^{\rightsquigarrow, a}(V)\right) + \expectation\left(  \mathbb{1}_{\{E+k > \mathcal{T}\}} \cdot S_{\mathcal{T}-E}^{\rightsquigarrow, a}(V) \right) \\
    =& 1 - \left[ \expectation\left( \mathbb{1}_{\{E+k \le \mathcal{T}\}} \cdot S_{k}^{\rightsquigarrow, a}(V)\right) + \expectation\left( \mathbb{1}_{\{E+k > \mathcal{T}\}} \cdot S_{\mathcal{T}-E}^{\rightsquigarrow, a}(V) \right) \right] + \expectation\left( \mathbb{1}_{\{E+k > \mathcal{T}\}} \cdot S_{\mathcal{T}-E}^{\rightsquigarrow, a}(V) \right) & \Leftarrow{(i)}\\
    =& 1 - \expectation\left( \mathbb{1}_{\{E+k \le \mathcal{T}\}} \cdot S_{k}^{\rightsquigarrow, a}(V)\right) \\
    \leq& 1
\end{align*}
}
where 
{\footnotesize
\begin{align*}
     (i) \equiv \expectation\left( S_{\min(k, \mathcal{T}-E)}^{\rightsquigarrow, a}(V)\right) = \expectation\left( \mathbb{1}_{\{E+k \le \mathcal{T}\}} \cdot S_{k}^{\rightsquigarrow, a}(V) \right) + \expectation\left(\mathbb{1}_{\{E+k > \mathcal{T}\}} \cdot S_{\mathcal{T}-E}^{\rightsquigarrow, a}(V) \right).
\end{align*}
}
\end{proof}

\subsubsection{C-values as a minimization problem}
Having bounded the maximum possible variation for each strategy individually, we now aim to find a paired set of variations that nullify the treatment effect inferred under extrapolation assumptions while remaining as symmetric as possible.

Let $k \in \{1, \dots, K\}$. We denote by $\Omega_k$ the region of valid perturbations:
\begin{align*}
    \Omega_k = \left[0, \textnormal{c}_{k}^{0,+}\right] \times \left[0, \textnormal{c}_{k}^{1,+}\right].
\end{align*}
 
We seek to find a pair $(\textnormal{c}_k^{0}, \textnormal{c}_k^{1}) \in \Omega_k$ that satisfies two conditions. First, we want 
$(\textnormal{c}_k^{0}, \textnormal{c}_k^{1})$ to minimize 
the objective function
\begin{align}\label{eq: penalization_c_values}
    J(\textnormal{c}_k^{0}, \textnormal{c}_k^{1}) = \begin{cases}
        \left(\ln(\textnormal{c}_k^{0} \cdot \textnormal{c}_k^{1})\right)^2 & c_k^0 \cdot c_k^1 > 0 \\
        +\infty & c_k^0 \cdot c_k^1 = 0.
    \end{cases}
\end{align}
We select the squared logarithmic penalty because it is smooth, strictly convex in the log-product, and penalizes deviations from the ideal symmetric solution $\textnormal{c}_k^{0} \cdot \textnormal{c}_k^{1} = 1$.

Second, $(\textnormal{c}_k^{0}, \textnormal{c}_k^{1})$ must nullify the effect, i.e., $R_k^{\textnormal{c}, 0}(\textnormal{c}_k^0) = R_k^{\textnormal{c}, 1}(\textnormal{c}_k^1)$ (as in \eqref{eq: shifted_risks_equality}). 

We call c-values the pair $(\textnormal{c}_k^{0}, \textnormal{c}_k^{1}) \in \Omega_k$ minimizing \eqref{eq: penalization_c_values} and satisfying \eqref{eq: shifted_risks_equality}.

\subsubsection{Existence of c-values}
\label{sec_app: existence_c_values}
C-values do not necessarily exist for every $k \in \{1, \dots, K\}$. To fix ideas, suppose $RR_m^{\rightsquigarrow} < 1$ for every $m \in \{1, \dots, K\}$, i.e., $\overline{a}_k=\overline{1}$ constitutes the most beneficial treatment strategy. This occurs when, for a given $k \in \{1, \dots, K\}$, even the most extreme variation, i.e., that resulting in failure by follow-up $k$ of all individuals under $\overline{a}_k=\overline{1}$ ($a=1$) who are alive at the end of the study period, together with that implying no failure by follow-up $k$ of all individuals alive under $\overline{a}_k=\overline{0}$ ($a=0$), is insufficient to make the maximum shifted risk under $\overline{a}_k=\overline{1}$ exceed the minimum shifted risk under $\overline{a}_k=\overline{0}$.

This idea is formalized in the following lemma.

\begin{lemma}[Existence of positive feasible perturbations at follow-up $k$]
\label{lemma: existence_c_values}
Let $k \in \{1,\dots,K\}$ and suppose $\prgapp_k^a>0$ for
$a\in\{0,1\}$. Suppose ${RR}_k^{\rightsquigarrow} \neq 1$. Then,
\begin{align*}
&\exists\,
(\textnormal{c}_k^0,\textnormal{c}_k^1)\in\Omega_k
\textnormal{ with }
\textnormal{c}_k^0>0,\,
\textnormal{c}_k^1>0
\textnormal{ such that }
R_k^{\textnormal{c},0}(\textnormal{c}_k^0)
=
R_k^{\textnormal{c},1}(\textnormal{c}_k^1)
\\
\stackrel{(1)}{\Leftrightarrow}\;&
\begin{cases}
\prstudyp_k^1
+\textnormal{c}_k^{1,+}\cdot\prgapp_k^1
>
\prstudyp_k^0,
& {RR}_k^{\rightsquigarrow}<1,\\[3pt]
\prstudyp_k^0
+\textnormal{c}_k^{0,+}\cdot\prgapp_k^0
>
\prstudyp_k^1,
& {RR}_k^{\rightsquigarrow}>1, 
\end{cases} \\
\stackrel{(2)}{\Leftrightarrow}\;&
\max\{\prstudyp_k^0,\prstudyp_k^1\}
<
\min\Big\{
\prstudyp_k^0+\textnormal{c}_k^{0,+} \cdot \prgapp_k^0,\,
\prstudyp_k^1+\textnormal{c}_k^{1,+} \cdot\prgapp_k^1
\Big\}.
\end{align*}
\end{lemma}

\begin{proof}
$\stackrel{(1)}{\Leftrightarrow}$ is trivial. Equivalence
$\stackrel{(2)}{\Leftrightarrow}$ follows since $\prgapp_k^a>0$,
\begin{align*}
\left\{
R_k^{\textnormal{c},a}(\textnormal{c}_k^a):
\textnormal{c}_k^a\in(0,\textnormal{c}_k^{a,+}]
\right\}
=
\left(
\prstudyp_k^a,\,
\prstudyp_k^a+\textnormal{c}_k^{a,+}\cdot \prgapp_k^a
\right].
\end{align*}
Hence,
\begin{align*}
&\exists\,
(\textnormal{c}_k^0,\textnormal{c}_k^1)\in\Omega_k
\textnormal{ with }
\textnormal{c}_k^0>0,\,
\textnormal{c}_k^1>0
\textnormal{ such that }
R_k^{\textnormal{c},0}(\textnormal{c}_k^0)
=
R_k^{\textnormal{c},1}(\textnormal{c}_k^1)
\\
\Leftrightarrow\;&
\left(
\prstudyp_k^0,\,
\prstudyp_k^0+\textnormal{c}_k^{0,+}\prgapp_k^0
\right]
\cap
\left(
\prstudyp_k^1,\,
\prstudyp_k^1+\textnormal{c}_k^{1,+}\prgapp_k^1
\right]
\neq\emptyset
\\
\Leftrightarrow\;&
\max\{\prstudyp_k^0,\prstudyp_k^1\}
<
\min\Big\{
\prstudyp_k^0+\textnormal{c}_k^{0,+}\prgapp_k^0,\,
\prstudyp_k^1+\textnormal{c}_k^{1,+}\prgapp_k^1
\Big\}.
\end{align*}
This concludes the proof.
\end{proof}

Under mild regularity assumptions, it can be shown that when c-values exist for $k=K_{\textnormal{ex}}$, they exist for the subsequent follow-up times.

Furthermore, when a solution exists, it is unique.

\begin{proposition}[C-values as the solution of an optimization problem]\label{proposition: c_values_derivation}
Let $k \in \{1, \dots, K\}$. Assume that c-values at follow-up $k$ exist, i.e., conditions in Lemma \ref{lemma: existence_c_values} are met. Define $\mathcal{D}_k \coloneqq \{(\textnormal{c}_k^0, \textnormal{c}_k^1) \in \Omega_k : R_k^{\textnormal{c}, 0}(\textnormal{c}_k^0) = R_k^{\textnormal{c}, 1}(\textnormal{c}_k^1)\}$. A solution $(\textnormal{c}_k^{0}, \textnormal{c}_k^{1}) \in \mathcal{D}_k$ that minimizes \eqref{eq: penalization_c_values} exists, is unique, and is given by:
\begin{align*}
(\textnormal{c}_k^{0}, \textnormal{c}_k^{1}) = 
\begin{cases}
    \left(\textnormal{c}_k^{\mid}, \, (\textnormal{c}_k^{\mid})^{-1}\right) & \text{if } \textnormal{c}_k^{0,+,\mathcal{D}} \cdot \textnormal{c}_k^{1,+,\mathcal{D}} \ge 1, \\
    \left(\textnormal{c}_k^{0,+,\mathcal{D}}, \, \textnormal{c}_k^{1,+,\mathcal{D}}\right) & \text{if } \textnormal{c}_k^{0,+,\mathcal{D}} \cdot \textnormal{c}_k^{1,+,\mathcal{D}} < 1,
\end{cases} 
\end{align*}
where the boundary limits $\textnormal{c}_k^{0,+,\mathcal{D}}$, $\textnormal{c}_k^{1,+,\mathcal{D}}$, and the symmetric root $\textnormal{c}_k^{\mid}$ are algebraic constants defined solely by $\{{\prstudyp}_k^{a}, {\prgapp}_k^a, \textnormal{c}_k^{a,+} : a \in \{0, 1\}\}$:
\begin{align*}
m_k \equiv& \frac{\prgapp_k^0}{\prgapp_k^1} > 0, \\ 
d_k \equiv& \frac{\prstudyp_k^0 - \prstudyp_k^1}{\prgapp_k^1} \in \mathbb{R}, \\
\textnormal{c}_k^{\mid} \equiv& \frac{-d_k + \sqrt{d_k^2 + 4m_k}}{2m_k}, \\
\textnormal{c}_k^{0,-} \equiv& \max(0, -\frac{d_k}{m_k}), \\
\textnormal{c}_k^{0,+,\mathcal{D}} \equiv& \min\left(\textnormal{c}_k^{0,+}, \frac{\textnormal{c}_k^{1,+} - d_k}{m_k}\right), \\
\textnormal{c}_k^{1,+,\mathcal{D}} \equiv& m_k \cdot \textnormal{c}_k^{0,+,\mathcal{D}} + d_k.
\end{align*}
\end{proposition}
\begin{proof}
To lighten notation, for every symbol we drop the subscripts ($k$). By Lemma \ref{lemma: existence_c_values}, $\mathcal{D}$ is non-empty. Furthermore, we have
\begin{align*}
(\textnormal{c}^0, \textnormal{c}^1) \in \mathcal{D} \Leftrightarrow& \Big[ (\textnormal{c}^0, \textnormal{c}^1) \in \Omega \Big] \land \Big[ R^{\textnormal{c}, 0}(\textnormal{c}^0) = R^{\textnormal{c}, 1}(\textnormal{c}^1) \Big]\\
\Leftrightarrow& \Big[ \textnormal{c}^0 \in [0, \textnormal{c}^{0,+}]\land \textnormal{c}^1 \in [0, \textnormal{c}^{1,+}] \Big] \land \Big[ \textnormal{c}^1 = m \cdot \textnormal{c}^0 + d \Big] \\
\Leftrightarrow& \left[0 \le \textnormal{c}^0 \le \textnormal{c}^{0,+} \land 0 \le m \cdot \textnormal{c}^0 + d \le \textnormal{c}^{1,+}\right]\land \Big[ \textnormal{c}^1 = m \cdot \textnormal{c}^0 + d \Big] \\
\Leftrightarrow& \left[\textnormal{c}^0 \in [0, \textnormal{c}^{0,+}]  \land \textnormal{c}^0 \in \left[-\frac{d}{m}, \frac{\textnormal{c}^{1,+} - d}{m}\right] \right] \land \Big[ \textnormal{c}^1 = m \cdot \textnormal{c}^0 + d \Big] \\
\Leftrightarrow& \left[\textnormal{c}^0 \in [\textnormal{c}^{0,-}, \textnormal{c}^{0,+, \mathcal{D}}\right] \land \Big[ \textnormal{c}^1 = m \cdot \textnormal{c}^0 + d \Big],
\end{align*}
where $\textnormal{c}^{0,-} \equiv \max(0,- \frac{d}{m})$ and $\textnormal{c}^{0,+, \mathcal{D}} = \min(\textnormal{c}^{0,+}, \frac{\textnormal{c}^{1,+} -d}{m})$.
Thus, the interval $[\textnormal{c}^{0,-}, \textnormal{c}^{0,+, \mathcal{D}}]$ characterizes the domain of $\textnormal{c}^{0}$ for all valid solutions. To find $(\textnormal{c}^{0}, \textnormal{c}^{1}) \in \mathcal{D}$ that minimizes $J(\textnormal{c}_k^{0}, \textnormal{c}_k^{1})$ (eq. \eqref{eq: penalization_c_values}) for $(\textnormal{c}^{0},\textnormal{c}^{1}) \in \mathcal{D}$, we will consider the function $f: [\textnormal{c}^{0,-}, \textnormal{c}^{0,+,\mathcal{D}}] \ni \textnormal{c}^0 \mapsto \textnormal{c}^0 \cdot(m \textnormal{c}^0 + d) \in \mathbb{R}$. Clearly $f$ is non-negative, continuous and furthermore its derivative is strictly positive, i.e., $f^{\prime}(\textnormal{c}^0) = 2m\textnormal{c}^0 + d > 0$, on the interior of $[{\textnormal{c}}^{0,-}, {\textnormal{c}}^{0,+,\mathcal{D}}]$. For $\textnormal{c}^0 \in [\textnormal{c}^{0,-}, \textnormal{c}^{0,+, \mathcal{D}}]$ we have two cases: either $f(\textnormal{c}^{0,+,\mathcal{D}}) < 1$ or $f(\textnormal{c}^{0,+,\mathcal{D}}) \ge 1$:\footnote{This case is non-empty, i.e., there exists a law satisfying Lemma \ref{lemma: existence_c_values} that leads to $f(\textnormal{c}^{0,+,\mathcal{D}}) \ge 1$: $\prstudyp^0 = 0.30$; $\prstudyp^1 = 0.20$; $\prgapp^0 = 0.05$; $\prgapp^1 = 0.04$; $\textnormal{c}^{0,+} = \frac{0.60}{0.05} = 12$; $\textnormal{c}^{1,+} = \frac{0.70}{0.04} = 17.5$. So that $m = \frac{0.05}{0.04} = 1.25$, $d = \frac{0.30 - 0.20}{0.04} = 2.5$. And, $\textnormal{c}^{0,-} = \max(0,\frac{-2.5}{1.25} ) = 0$, $\textnormal{c}^{0,+,\mathcal{D}} = \min(12, \frac{17.5 - 2.5}{1.25}) = \min(12, 12) = 12$. Finally,  $f(12) = 12 \cdot (1.25 \cdot 12 + 2.5) = 12 \cdot 17.5 = 210$.}
\begin{enumerate}
    \item $f(\textnormal{c}^{0,+,\mathcal{D}}) \ge 1$. Since $f(\textnormal{c}^{0,-}) = 0 \le 1 \le f(\textnormal{c}^{0,+,\mathcal{D}})$ and $f$ is continuous, the Intermediate Value Theorem guarantees that there exists a root where $f(\textnormal{c}^0) = 1$; since $f$ is strictly increasing, this root is unique. But, $J(\textnormal{c}^0, m \textnormal{c}^0 + d) = {\ln(f(\textnormal{c}^0))}^{2}=(\ln 1)^2= 0$,
    i.e., $(\textnormal{c}^0, m \textnormal{c}^0 + d)$ is the global minimum and furthermore $(m\textnormal{c}^0 +d ) = (\textnormal{c}^0)^{-1}$. For $\textnormal{c}^0 > 0$, we have $f(\textnormal{c}^0) = 1 \Leftrightarrow 
    m \cdot {(\textnormal{c}^0) }^2 + d \cdot \textnormal{c}^0 - 1 = 0 \Leftrightarrow \textnormal{c}^0  = \frac{-d + \sqrt{d^2 + 4m}}{2m}$.
    \item $f({\textnormal{c}}^{0,+,\mathcal{D}}) < 1$. Because the strict existence condition in Lemma \ref{lemma: existence_c_values} guarantees $\textnormal{c}^{1,+} > d$ and $\textnormal{c}^{0,+} > 0$, the upper bound of the domain is strictly positive: $\textnormal{c}^{0,+,\mathcal{D}} > 0$.\footnote{By definition, $\textnormal{c}^{0,+,\mathcal{D}} = \min\left(\textnormal{c}^{0,+}, \frac{\textnormal{c}^{1,+} - d}{m}\right)$. By assumption, $\textnormal{c}^{0,+} > 0$ and $m > 0$. We need only show $\textnormal{c}^{1,+} > d$. If $d \le 0$, this holds trivially since $\textnormal{c}^{1,+} > 0$. If $d > 0$, Lemma \ref{lemma: existence_c_values} requires $\prstudyp^1 + \textnormal{c}^{1,+} \cdot \prgapp^1 > \prstudyp^0$, which gives $\textnormal{c}^{1,+} > d$.} Thus, the maximum product is strictly positive: $f({\textnormal{c}}^{0,+,\mathcal{D}}) > 0$. Since $f$ is strictly increasing, we have $0 \le f(\textnormal{c}^0) < 1$ for all $\textnormal{c}^0 \in [{\textnormal{c}}^{0,-}, {\textnormal{c}}^{0,+,\mathcal{D}}]$. Because the function $x \mapsto (\ln x)^2$ is strictly decreasing on the interval $(0, 1)$, minimizing $\left(\ln f({\textnormal{c}}^{0})\right)^2$ over $\mathcal{D}$ is equivalent to maximizing $f({\textnormal{c}}^{0})$. Since $f$ is strictly increasing, its maximum is attained at $\textnormal{c}^0 = \textnormal{c}^{0,+,\mathcal{D}}$ and it is unique; thus, $\textnormal{c}^1 = m \cdot \textnormal{c}^{0,+,\mathcal{D}} + d$.
\end{enumerate}
That is, in each case, the solution in the statement of the Proposition is the $({\textnormal{c}}^{0}, {\textnormal{c}}^{1}) \in \mathcal{D}$ minimizing $J(\textnormal{c}^0, \textnormal{c}^1)$.
\end{proof}

\subsubsection{Extrapolation curves: c-values and their equivalent representations of variation through additive hazard shifts}
\label{sec_app: equivalence_c_values_hazards}
Let $k \in \{K_{\text{ex}}, \dots, K\}$, $E \in V$. Suppose  c-values $(\textnormal{c}_k^{0}, \textnormal{c}_k^{1})$ exist. We want to show that for any given c-values $(\textnormal{c}_k^{0}, \textnormal{c}_k^{1})$, it is possible to find a $(\delta_k^{0}, \delta_k^{1})$ that leads to the same shifted risks. To this end, for every $a \in \{0, 1\}$, every $k \in \{0, \dots, K\}$ and every $v \in \text{supp}(P(V))$ consider
\begin{align*}
h_k^{\rightsquigarrow,a}(v;\delta,\epsilon) =& \begin{cases}
   h_k^{\rightsquigarrow,a}(v) & e + k \leq \mathcal{T}  \\
       \min (\max(h_k^{\rightsquigarrow,a}(v) + \delta,\epsilon),1-\epsilon) & \text{otherwise}. 
\end{cases}
\\
    S_k^{\rightsquigarrow,a}(v;\delta, \epsilon)\coloneqq& \prod_{m=0}^k(1-h_m^{\rightsquigarrow,a}(v;\delta,\epsilon)), \\
    \prstudyp_{k}^a(v;\delta,\epsilon) \coloneqq& 1 - S_{\min(k, \mathcal{T}-e)}^{\rightsquigarrow,a}(v;\delta, \epsilon), \\
    \prgapp_{k}^a(v;\delta,\epsilon) \coloneqq& \mathbb{1}_{\{e+k > \mathcal{T}\}} \cdot \left(S_{\mathcal{T}-e}^{\rightsquigarrow,a}(v;\delta, \epsilon) - S_{k}^{\rightsquigarrow,a}(v;\delta, \epsilon)\right), \\
    \prstudyp_{k}^a(\delta,\epsilon) \coloneqq& \expectation\left(\prstudyp_{k}^a(V;\delta,\epsilon)\right), \\
    \prgapp_{k}^a(\delta,\epsilon) \coloneqq& \expectation\left(\prgapp_{k}^a(V;\delta,\epsilon)\right),
\end{align*}
where $h_k^{\rightsquigarrow,a}(v;\delta,\epsilon)$ are the hazards shifted by $\delta$ and clipped at $(\epsilon, 1-\epsilon)$, and $S_k^{\rightsquigarrow,a}(v;\delta, \epsilon)$, $\prstudyp_{k}^a(v;\delta,\epsilon)$, $\prgapp_{k}^a(v;\delta,\epsilon)$, $ \prstudyp_{k}^a(\delta,\epsilon)$, $\prgapp_{k}^a(\delta,\epsilon)$ are the analogues of parameters to those defined in Appendix \ref{sec_app: c_values_preliminaries} resulting from replacing $h_k^{\rightsquigarrow,a}$ with $h_k^{\rightsquigarrow,a}(v;\delta,\epsilon)$.

For a given pair of c-values $(\textnormal{c}_k^0, \textnormal{c}_k^1)$, we say that equivalent
hazards exist on the additive scale if, for every $a \in \{0,1\}$, there exist $\epsilon^a \geq 0$ and $\delta_k^a \in [-1,1]$ such that
\begin{align}\label{eq: kiss_c_values_and_hazards}
    \textnormal{c}_k^a \cdot
    \prgapp_{k}^a(\delta=0;\epsilon^a) = \prgapp_{k}^a(\delta=\delta_k^a;\epsilon^a).
\end{align}
When $0<\textnormal{c}_k^a<\textnormal{c}_k^{a,+}$, we refer to the unique $\delta_k^a$ satisfying \eqref{eq: kiss_c_values_and_hazards} as the equivalent additive hazard shift.

Now we show that this is the case under mild regularity conditions: the post-study hazards must be uniformly bounded away from the deterministic values 0 and 1. We also define the active index set $\mathcal{I}_k \coloneqq \{m \le k \mid P(E+m > \mathcal{T}) > 0\}$, which collects all follow-up times corresponding to the post-study period for a positive proportion of the population.

\begin{proposition}
\label{prop: equivalence_c_values_and_hazards}
Let $a\in\{0,1\}$ and $k \in \{K_{\textnormal{ex}},\dots,K\}$. Suppose that, $\prgapp_k^a>0$ and $\textnormal{c}_k^a\in(0,\textnormal{c}_k^{a,+}]$.
Assume there exists $\epsilon^+\in(0,1/2)$ such that, for every $m\in\mathcal I_k$,
\begin{align*}
    \operatorname{supp}
\left(P(h_m^{\rightsquigarrow,a}(V)\mid E+m>\mathcal T)\right) \subset(\epsilon^+,1-\epsilon^+).
\end{align*}
If
\begin{align*}
    0<\textnormal{c}_k^a<\textnormal{c}_k^{a,+},
\end{align*}
then there exists a sufficiently small $\epsilon^a>0$ for which there exists a unique $\delta_k^a\in(-1,1)$ satisfying
\begin{align*}
    \prgapp_k^a(\delta_k^a,\epsilon^a) = \textnormal{c}_k^a \cdot \prgapp_k^a(0,\epsilon^a).
\end{align*}
If
\begin{align*}
    \textnormal{c}_k^a=\textnormal{c}_k^{a,+},
\end{align*}
a solution exists with $\epsilon^a=0$ (for example,
$\delta_k^a=1$), but the solution need not be unique.
\end{proposition}
\begin{proof}
Let $a \in \{0,1\}$ and $k \in \{K_{\textnormal{ex}},\dots,K\}$
satisfy the conditions of the proposition, and omit the superscript $a$ except for the solution $\delta_k^a$. If $\textnormal{c}_k = \textnormal{c}_k^+$, trivially set $\epsilon = 0$ and $\delta_k^a = 1$, so that post-study hazards are identically equal to $1$. 

If $0 < \textnormal{c}_k < \textnormal{c}_k^+$, fix $\epsilon \in (0, 0.5)$. As
\begin{align*}
\prgapp_k(v;\delta_k, \epsilon) \coloneqq \mathbb{1}_{\{e+k>\mathcal{T}\}} \cdot S_{\mathcal{T}-e}(v) \cdot \left( 1 - \prod_{m=\mathcal{T}-e+1}^k (1 - h_m(v; \delta_k, \epsilon)) \right),
\end{align*}
is defined as a composition of continuous functions in $\delta_k$ and it is bounded by $1$ for all $v \in \textnormal{supp}(P(V))$, the Dominated Convergence Theorem ensures that $\delta_k \mapsto \prgapp_k(\delta_k, \epsilon) = \mathbb{E}\left(\prgapp_k(V; \delta_k, \epsilon)\right)$ is continuous on $[-1,1]$. 

Define the bounds $L_k(\epsilon) \coloneqq \prgapp_k(-1, \epsilon)$ and $U_k(\epsilon) \coloneqq \prgapp_k(1, \epsilon)$, and target $T_k(\epsilon) \coloneqq \textnormal{c}_k \cdot \prgapp_k(0, \epsilon)$. As $\epsilon \downarrow 0$, these limits converge to $0$, $\textnormal{c}_k^+ \cdot \prgapp_k(0,0)$, and $\textnormal{c}_k \cdot \prgapp_k(0,0)$, respectively. Letting $\mathcal{S}_m \equiv \textnormal{supp}(P(h_m(V) \mid E+m >\mathcal{T}))$ for every $m \in \mathcal{I}_k$, we have
\begin{align*}
    &0 < \textnormal{c}_k < \textnormal{c}_k^+ \text{ and } \prgapp_k(0,0) > 0 &\Leftarrow{(i)}\\
    \Rightarrow& 0 < \lim_{\epsilon \downarrow 0} T_k(\epsilon) < \lim_{\epsilon \downarrow 0} U_k(\epsilon) \\
    \Rightarrow& L_k(\epsilon) < T_k(\epsilon) < U_k(\epsilon) \text{ for all } \epsilon \in (0, \epsilon_1) \text{ for some } \epsilon_1 > 0, &\Leftarrow{(ii)} \\
    \Rightarrow& \exists \delta_k^a \in (-1, 1) \text{ such that } \prgapp_k(\delta_k^a, \epsilon) = T_k(\epsilon) &\Leftarrow{(iii)}.
\end{align*}
Here, $(i) \equiv \textnormal{hypotheses}$; $(ii) \equiv$ strict order-preserving property of limits: because the limits of the differences $T_k(\epsilon) - L_k(\epsilon)$ and $U_k(\epsilon) - T_k(\epsilon)$ as $\epsilon \downarrow 0$ are strictly positive constants, there exists a sufficiently small neighborhood $(0, \epsilon_1)$ where both differences remain strictly positive; $(iii) \equiv$  Intermediate Value Theorem and the continuity of $\delta_k \mapsto \prgapp_k(\delta_k, \epsilon)$.

To prove uniqueness, first note that $\delta_k \mapsto \prgapp_k(\delta_k, \epsilon)$ is globally (i.e., in $[-1,1]$) non-decreasing as $x \mapsto \min(\max(x, \epsilon), 1-\epsilon)$ is non-decreasing. 
Thus, it suffices to show that $\delta_k \mapsto \prgapp_k(\delta_k, \epsilon)$ is strictly increasing at any root $\delta_k^a$. Fix $\epsilon \in (0, \min(\epsilon_1, \epsilon^+))$ and the root $\delta_k^a$; define $I(\delta_k^a) \coloneqq (\epsilon - \delta_k^a, 1 - \epsilon - \delta_k^a)$ which has length $|I(\delta_k^a)| =1 - 2\epsilon$. Since
{\small
\begin{align*}
    & P(h_m(V) \in I(\delta_k^a) \mid E+m > \mathcal{T}) = 0 \text{ for all } m \in \mathcal{I}_k \\
    \Rightarrow& \mathcal{S}_m \cap I(\delta_k^a) = \emptyset \\
    \Rightarrow& \mathcal{S}_m \text{ lies entirely left or right of } I(\delta_k^a) & \Leftarrow{(iv)} \\
    \Rightarrow& \mathcal{S}_m \text{ is on the same side of } I(\delta_k^a) \text{ for all } m \in \mathcal{I}_k & \Leftarrow{(v)} \\
    \Rightarrow& \text{either } [\forall m, h_m(v; \delta_k^a, \epsilon) = \epsilon] \text{ or } [\forall m, h_m(v; \delta_k^a, \epsilon) = 1-\epsilon] & \Leftarrow{(vi)} \\
    \Rightarrow& \prgapp_k(\delta_k^a, \epsilon) \in \{L_k(\epsilon), U_k(\epsilon)\} \\
    \Rightarrow& \neg \left[L_k(\epsilon) < T_k(\epsilon) < U_k(\epsilon)\right],
\end{align*}
}
where
\begin{align*}
    (iv) \equiv& \sup(\mathcal{S}_m) - \inf(\mathcal{S}_m) \le 1 - 2\epsilon^+ < |I(\delta_k^a)|,\\
    (v) \equiv& \left[\text{left} \Rightarrow \delta_k^a \le \epsilon-\epsilon^+ < 0; \, \text{right} \Rightarrow \delta_k^a \ge \epsilon^+-\epsilon > 0\right],\\
    (vi) \equiv& \text{for every }m\in\mathcal{I}_k\text{ and every }v \text{ with }e(v)+m>\mathcal{T},
\end{align*}
is absurd, we have that for some $m^{\sim} \in \mathcal{I}_k$, the set
\begin{align*}
    B^{\sim} \coloneqq \{v \in \textnormal{supp}(P(V)) \mid e+m^{\sim} > \mathcal{T} \text{ and } h_{m^{\sim}}(v) + \delta_k^a \in (\epsilon, 1-\epsilon)\}
\end{align*}
satisfies $P(V \in B^{\sim}) > 0$. Now, for any $\delta_k > \delta_k^a$:
{\footnotesize
\begin{align*}
    & h_{m^{\sim}}(v; \delta_k, \epsilon) > h_{m^{\sim}}(v; \delta_k^a, \epsilon), \quad v \in B^{\sim} \\
    \Rightarrow& \prod_{m=\mathcal{T}-e+1}^k (1 - h_m(v; \delta_k, \epsilon)) < \prod_{m=\mathcal{T}-e+1}^k (1 - h_m(v; \delta_k^a, \epsilon)), \quad v \in B^{\sim} &\Leftarrow{(vii)}\\
    \Rightarrow& \prgapp_k(v; \delta_k, \epsilon) > \prgapp_k(v; \delta_k^a, \epsilon), \quad v \in B^{\sim} \\
    \Rightarrow& \expectation\left(\prgapp_k(V; \delta_k, \epsilon)\right) > \expectation\left(\prgapp_k(V; \delta_k^a, \epsilon)\right) &\Leftarrow{(viii)}\\
    \Rightarrow& \prgapp_k(\delta_k, \epsilon) > \prgapp_k(\delta_k^a, \epsilon).
\end{align*}
}
where
{\footnotesize
\begin{align*}
    (vii) \equiv& \text{$1 - h_m(v; \cdot, \epsilon) \ge \epsilon > 0$ for all $m$} \\
    (viii) \equiv& \text{$\prgapp_k(v; \delta_k, \epsilon) \ge \prgapp_k(v; \delta_k^a, \epsilon)$ for every $v \in \textnormal{supp}(P(V))$ and $P(V \in B^{\sim}) > 0$}.
\end{align*}
}
A symmetric argument holds for $\delta_k < \delta_k^a$. This concludes the proof.
\end{proof}

\textbf{Remarks on the choice of }$\epsilon^a$. For every $a \in \{0,1\}$ and every $k$, the exact solution $\delta_k^a$ depends on the choice of $\epsilon^a$. As established, to guarantee that a unique solution exists, $\epsilon^a$ must be chosen sufficiently small. In practice, a value such as $10^{-4}$ will often be sufficiently small, although this should be verified for the hazards under consideration. $\delta_k^a$ is typically found using a standard 1D scalar root-finding algorithm (e.g., Brent's method) bounded on the interval $[-1, 1]$. However, when evaluating this iterative solver at scale is computationally demanding it is possible to bypass the solver entirely by deriving a closed-form, first-order approximation of $\delta_k^a$.
\clearpage
\section{Case studies: specifics and additional results}
\label{sec_app: case_studies_specifics}
The aim of this appendix is twofold. In Section \ref{sec_app: case_study_1}, we provide supplementary details for the simulations of case study 1. In Section \ref{sec_app: case_study_2}, we report estimates for case study 2 and further elaborate on why some of the findings published in \citep{grippin_sars-cov-2_2025} are biased, with the resulting bias having an expression that resembles that of \textit{study-entry-determined immortal time bias}.

\subsection{Case study 1}
\label{sec_app: case_study_1}
\subsubsection{Data structure}
\label{sec_app: case_study_1_data_structure}
\begin{figure}[H]
\centering
\fbox{%
\resizebox{0.6\textwidth}{!}{%
\begin{tikzpicture}[
  edge/.style={line width=0.6pt},
  lab/.style={font=\small},
  every node/.style={inner sep=1pt}
]

\node (L0)   at (0,  -8)   {$L_0$};
\node (Lkp)  at (4.8, 0.2) {$L_k^{\prime}$};
\node (tx)  at (0, 2.2) {For $k \in \{0, \dots, 15=\mathcal{T}\}$:};
\node (tx)  at (0, -1.8) {For $k =0$:};

\node (Ak)   at (0, 0.2) {$A_k$};
\node (Yk)   at (2.4, 0.2) {$Y_k$};
\node[lab, right=2mm of Ak] {$\in\{0,1\}$};
\node[lab, right=2mm of Yk] {$\in\{0,1\}$};

\node (L0m)  at (2.2,  -5.2) {$L_0^{-}$};
\node (L0p0) at (4.8, -10.2) {$L_{0}^{\prime}$};

\draw[edge] (L0) -- (L0m);
\draw[edge] (L0) -- (L0p0);

\node (V) at (4.8,  -3.6) {$V$};
\node (B) at (4.8,  -6.8) {$B$};

\draw[edge] (L0m) -- (V);
\draw[edge] (L0m) -- (B);

\node (V1) at (7.8,  -2.8) {$V_1$};
\node (V2) at (7.8,  -4.4) {$V_2$};

\draw[edge] (V) -- (V1);
\draw[edge] (V) -- (V2);

\node[lab, right=2mm of V1] {$\in0.25 \cdot\{0,\dots,15\}$};
\node[lab, right=2mm of V2] {$\in\{0,1\}$};

\node (B1) at (7.8,  -6.4) {$B_1$};
\node (B2) at (7.8, -8.2) {$B_2$};

\draw[edge] (B) -- (B1);
\draw[edge] (B) -- (B2);

\node[lab, right=2mm of B1] {$\in\{0,1\}$};
\node[lab, right=2mm of B2] {$\in\{0,1\}$};

\node (L10p) at (7.8, -9.0) {$L_{1,0}^{\prime}$};
\node (L20p) at (7.8, -11.2) {$L_{2,0}^{\prime}$};

\draw[edge] (L0p0) -- (L10p);
\draw[edge] (L0p0) -- (L20p);

\node[lab, right=2mm of L10p] {$\in\{0,1\}$};
\node[lab, right=2mm of L20p] {$\in\{0,1\}$};

\node (L1kp) at (7.8, 1.2) {$L_{1,k}^{\prime}$};
\node (L2kp) at (7.8, -0.4) {$L_{2,k}^{\prime}$};

\draw[edge] (Lkp) -- (L1kp);
\draw[edge] (Lkp) -- (L2kp);

\node[lab, right=2mm of L1kp] {$\in\{0,1\}$};
\node[lab, right=2mm of L2kp] {$\in\{0,1\}$};

\end{tikzpicture}}}
\end{figure}
We considered four baseline covariates, $L_0^{-}$, and 2 time-varying confounders, $L_{k}^{\prime}$, so that, according to the data structure of the main text, $L_0 \coloneqq L_0^{-} \cup L_{0}^\prime$, $L_{k} \coloneqq L_{k}^{\prime}$ for $k>0$. 

The variable $V_1$ is just a rescaled version of $E$, i.e., $V_1 = 0.25 \cdot E$. Accordingly, when we write $E \notin V$ in the main text, we mean $V = (V_2)$. 

\subsubsection{Data generating mechanism}
\label{sec_app: scenarios_dgms}
To generate the observed data, we adapted the procedure proposed by \citet{seaman_simulating_2024}. For all scenarios we considered: \footnote{$E \sim  \text{Discrete uniform}\{0, \dots, \mathcal{T}\}$ if $P(E=e) = (\mathcal{T}+1)^{-1}$, $E \sim  \text{Discr. exponential with parameter } \lambda $ if $P(E=e) =  \frac{\text{exp}(-\lambda \cdot e)}{\sum_{j=0}^{\mathcal{T}} \text{exp}(-\lambda \cdot j)} \quad \text{for } e \in \{0, 1, 2, \dots, \mathcal{T}\}$.}
{\scriptsize
\begin{align*}
    V_1 \sim& 0.25 \cdot 
    \begin{cases} 
      \text{Discr. exponential } (\lambda=0.5) & \text{scenarios C) and F)}, \\
      \text{Discr. uniform} & \text{otherwise},
    \end{cases} \\
    V_2 \mid V_1 \sim& \text{Bernoulli}(0.5), \\
    B_1 \mid V_2, V_1 \sim& \text{Normal}(-0.2 + 0.4 \cdot V_2), \\
    B_2 \mid B_1, V_2, V_1 \sim& \text{Normal}(0.2 \cdot V_1), \\
    L_{1,0}^{\prime} \mid L_0^- \sim& \text{Bernoulli}(\text{expit}(-0.3 + 0.6\cdot V_1)), \\
    L_{2,0}^{\prime} \mid L_{1,0}^{\prime}, L_0^- \sim& \text{Bernoulli}(\text{expit}(-0.2 + 0.4\cdot V_2)), \\ 
    L_{1,k}^{\prime} \mid Y_{k-1}=0, \overline{A}_{k-1}, \overline{L}_{k-1} \sim& \text{Bernoulli}(\text{expit}(-0.3 + 0.6\cdot V_1)), \\
    L_{2,k}^{\prime} \mid L_{1,k}^{\prime}, Y_{k-1}=0, \overline{A}_{k-1}, \overline{L}_{k-1} \sim& \text{Bernoulli}(\text{expit}(-0.2 + 0.4\cdot V_2)), \\ 
    A_k \mid \overline{L}_k \sim& \text{Bernoulli}\big(\text{expit}\big(-1 + 0.2 \cdot V_1 + 0.3 \cdot V_2 + 0.2 \cdot B_1 \\
    &\qquad\qquad\qquad\quad\; + 0.6 \cdot L_{1,k}^{\prime} + 0.6 \cdot L_{2,k}^{\prime} + A_{k-1}\big)\big),
\end{align*}
}
and a model for $Y_k \mid \overline{A}_k, \overline{L}_k, Y_{k-1}$ is generated in such a way as to be compatible, under $\RAID$ and $\FAID$, with the scenario-specific MSM-hazard for $h_k^{a_k}(v)$:
{\footnotesize
\begin{align*}
h_k^{a_k}(v) =& \begin{cases}
        \text{expit}(h_0 + \beta_{v_1}(e) + 1.9 \cdot v_2 - 0.4 \cdot a_k), & \text{A) and D)} \\
         \text{expit}(h_0 + \beta_{v_1}(e) + 1.9 \cdot v_2 - 0.4 \cdot a_k - 0.5 \cdot a_k \cdot  \mathbb{1}_{\{e+k > \mathcal{T} \}}) , & \text{B),C),E), and F)},
    \end{cases} \\
        \beta_{v_1}(e)=& -\frac{2}{\mathcal{T}}\cdot  e + 2,
\end{align*}
}
where $\beta_{v_1}(e)$ is decreasing in $e$ to be consistent with the estimates reported by \citet{hernan_marginal_2000}.
\subsubsection{Scenario-specific estimators}
\label{sec_app: procedures_estimators_case_study_1}
The estimators used in our procedures are:
{\footnotesize
\begin{align*}
    &\text{1) MSM-hazard}, \quad \hat{S}_k^{*,\text{msm}} = {P}_n\left(\prod_{m=0}^{k}\left(1-p(V,a_m,m;\hat{\beta}^{*})\right)\right) \quad V=(V_1, V_2),\\
    &\text{2) MSM-hazard } E \notin V, \quad \hat{S}_k^{*,\text{msm-we}}= {P}_n\left(\prod_{m=0}^{k}\left(1-p(V,a_m,m;\hat{\beta}^{*})\right)\right) \quad V=V_2,\\
    &\text{3) WSC-Kaplan-Meier},\quad \hat{S}_k^{*,\text{km}}= \prod_{m=0}^{k} \frac{{P}_n \left(\hat{W_m} \cdot \mathbb{1}_{\{C_{m-1}=0, \overline{A}_m = \overline{a}_m, Y_{m} = 0\}}\right)}{{P}_n \left(\hat{W_m} \cdot \mathbb{1}_{\{Y_{m-1} = 0, C_{m-1}=0, \overline{A}_m = \overline{a}_m\}}\right)}.
\end{align*}
}
The propensity score models in the weights $W_m$ used to generate $P^*$ from $P$ have been estimated under a correctly specified model. In procedures 1) and 2), we have considered 
{\small
\begin{align*}
\mathcal{M}_{\circ} \equiv \mathcal{M}(\mathbb{B}_{\text{A}}^{\text{D}})  =& \{\text{logit}(p(v,a_k,k;\beta)) = \lambda(v,a_k,k;\beta)\}, \\
\mathcal{M}_{\circ\circ} \equiv \mathcal{M}(\mathbb{B}_{\text{BC}}^{\text{EF}}) =& \{\text{logit}(p(v,a_k,k;\beta)) =\lambda(v,a_k,k;\beta) + \beta_{A, \geq \mathcal{T}} \cdot a_k \cdot \mathbb{1}_{\{e+k > \mathcal{T}\}}\}, \\
\mathcal{M}(\mathbb{B}_{\text{sat}}) =& \{\text{logit}(p(v,a_k,k;\beta)) = \lambda_{\text{we}}(v,a_k,k;\beta)\},
\end{align*}
}
where
{\footnotesize
\begin{align}
\lambda(v,a_k,k;\beta) =&
\beta_0
+ \beta_{v_1} v_1
+ \beta_{v_2} v_2
+ \beta_{A} a_k \label{cs_1_msm_hazard_no_interaction},\\
\lambda_{\text{we}}(v_2,a_k,k;\beta) =&
\text{all possible interactions between}(v_2, a_k, \text{natural cubic spline}(k)).
\end{align}
}
Here, $\mathcal{M}(\mathbb{B}_{\text{A}}^{\text{D}})$ constitutes a stylized version of the working model considered in \citep{hernan_marginal_2000}; the choice of $\mathcal{M}(\mathbb{B}_{\text{sat}})$ is meant to guarantee that ${\mathcal{M}(\mathbb{B}_{\text{sat}})}^{*-\text{cs}}$ holds. This, under $\RAID$, guarantees that $h_{\prstudyp,k}^{a_k}$ is correctly specified, i.e., there exists a unique $\beta \in \mathbb{B}_{\text{sat}}$ such that $p(v,a_k,k;\beta) \stackrel{\prstudyp}{=} h_{\prstudyp,k}^{a_k}$. \footnote{Technically, here the model is not saturated, but is chosen to be sufficiently flexible for ${\mathcal{M}(\mathbb{B}_{\text{sat}})}^{*-\text{cs}}$ to be tenable. This suffices to make our points.} Thereby, we can validate the bias resulting from not including $E$ in $V$. Finally, for $\hat{S}_k^{*,\text{km}, \overline{a}_k}$ to gain efficiency, we consider an observation to be consistent with the strategy $\overline{a}_k=\overline{1}$ if, over the four most recent follow-ups up to time $k$, $\{k-3, k-2, k-1, k\}$, we have $A_j = 1$ for all $j \in \{k-3, k-2, k-1, k\}$.

\subsubsection{Comparison estimates vs. ground-truth values and extrapolation curves}
\label{sec_app: plots_case_study_1}
Figure \ref{fig: sim_hernan_00} presents the survival curves for scenarios A), B), and C). Figure \ref{fig: sim_hernan_00_df} presents the survival curves for scenarios D), E), and F). Figure \ref{fig: extrapolation_curves_hernan_A_B} presents the extrapolation curves for scenarios A) and B).

\begin{figure}[H]
    \centering
    \adjustbox{trim={.0 \width} {.01\height} {0\width} {.02\height},clip}%
{
    \begin{tikzpicture}[scale=1]
        \node[anchor=north, scale=0.4] at (0,0){%
            \pgfimage{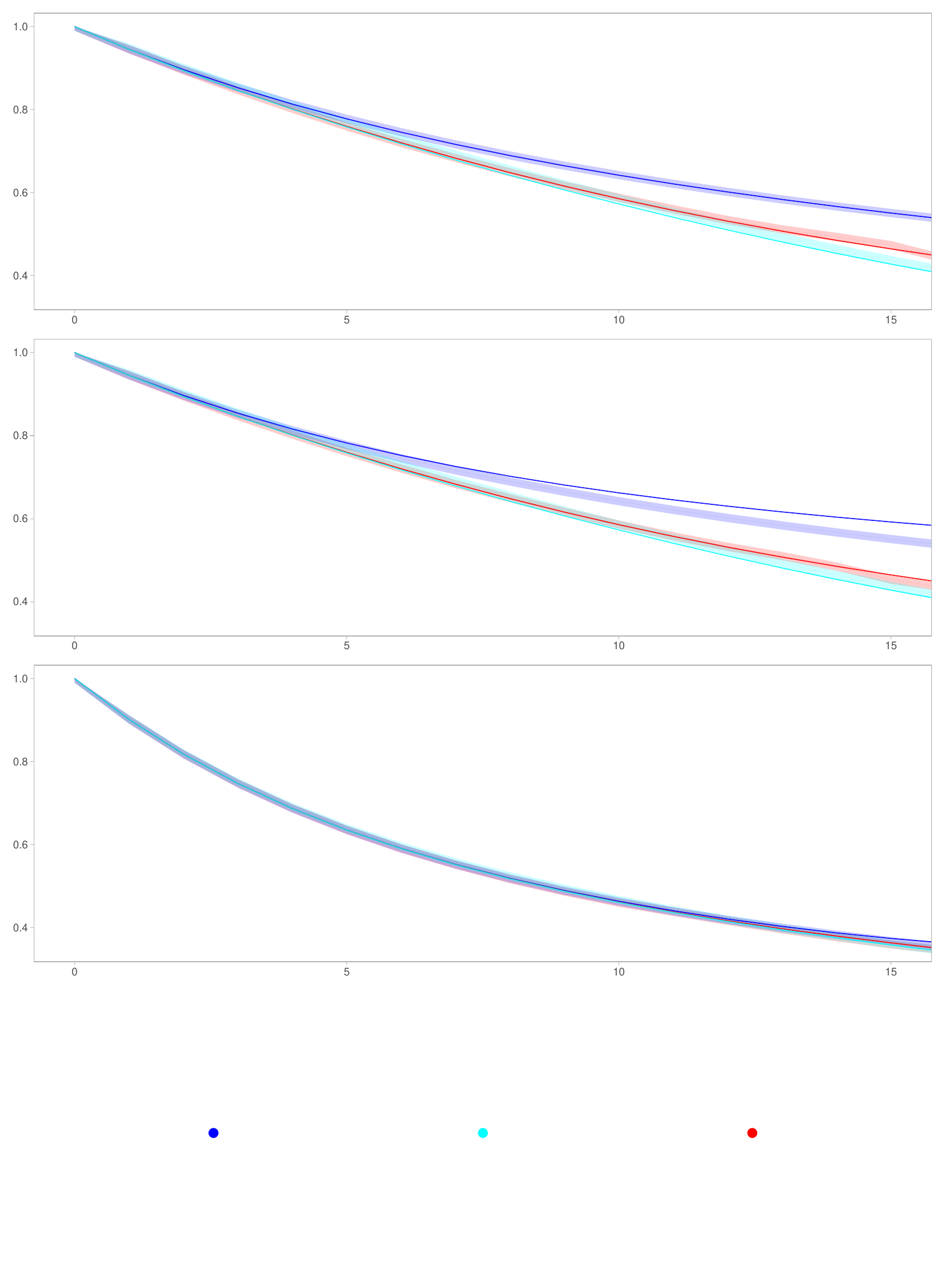}};

        \node[anchor=north west] at (-7,-0.2){\small A)}; 
        \node[anchor=north west] at (-7,-4){\small B)}; 
        \node[anchor=north west] at (-7,-8){\small C)}; 
        
        \node[anchor=north] at (0.5,-11.8){\small $k+1$};
        \node[anchor=north west] at (-3.5,-12.5){\small $\hat{S}_k^{*,\text{msm}}$};
        \node[anchor=north west] at (-0.2,-12.5){\small $\hat{S}_k^{*,\text{msm-we}}$};
        \node[anchor=north west] at (3,-12.5){\small $\hat{S}_k^{*,\text{km}}$}; 
\end{tikzpicture}
}
\caption{Simulation results for Scenarios A), B), and C) for sample size $n = 5 \cdot 10^6$ . The x-axis reports time $k^{\prime} \coloneqq k+1$ (At $k^{\prime}=0$, all survival curves equal 1). Highlighted curves correpond to point estimates; solid lines show the scenario-specific ground truth values (In particular, the solid blue line corresponds to $S_k^{\overline{a}_k=\overline{1}}$).}
\label{fig: sim_hernan_00}
\end{figure}

\begin{figure}
\centering
    \adjustbox{trim={.0 \width} {.01\height} {0\width} {.02\height},clip}%
{
    \begin{tikzpicture}[scale=1]
        \node[anchor=north, scale=0.4] at (0,0){%
            \pgfimage{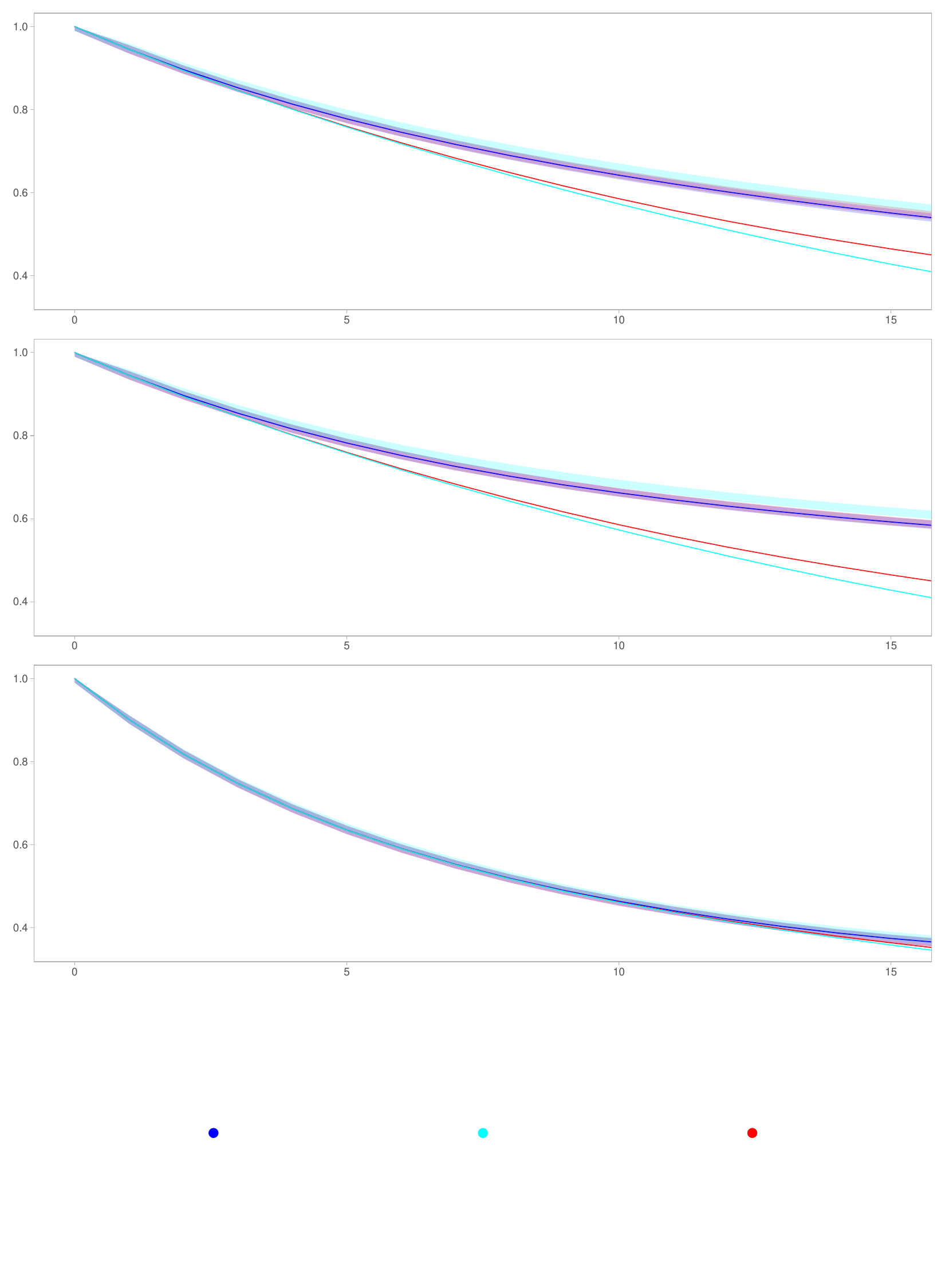}};

        \node[anchor=north west] at (-7,-0.2){\small D)}; 
        \node[anchor=north west] at (-7,-4){\small E)}; 
        \node[anchor=north west] at (-7,-8){\small F)}; 
        
        \node[anchor=north] at (0.5,-11.8){\small $k+1$};
        \node[anchor=north west] at (-3.5,-12.5){\small $\hat{S}_k^{*,\text{msm}}$};
        \node[anchor=north west] at (-0.2,-12.5){\small $\hat{S}_k^{*,\text{msm-we}}$};
        \node[anchor=north west] at (3,-12.5){\small $\hat{S}_k^{*,\text{km}}$}; 
\end{tikzpicture}
}
\caption{Simulation results for Scenarios D), E) and F) for a sample size $n = 50000$. The x-axis reports time $k^{\prime} \coloneqq k+1$ (At $k^{\prime}=0$, all survival curves equal 1). Highlighted curves correpond to point estimates; solid lines show the scenario-specific ground truth values (In particular, the solid blue line corresponds to $S_k^{\overline{a}_k=\overline{1}}$).}
\label{fig: sim_hernan_00_df}
\end{figure}
\begin{figure}[H]
    \centering%
{
    \begin{tikzpicture}[scale=1]
        \node[anchor=north, scale=0.6] at (0,0){\pgfimage{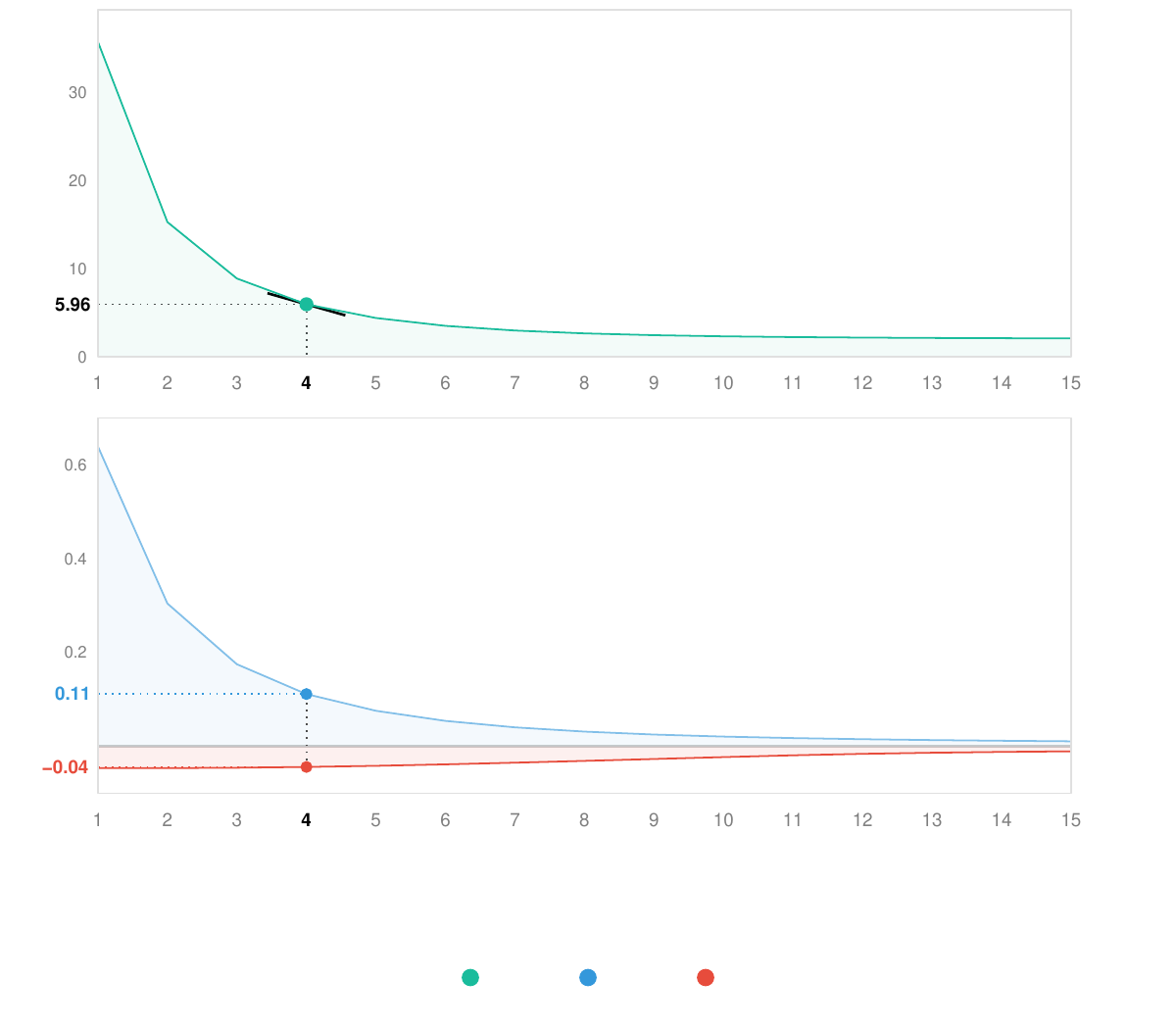}
        };    
        \node[anchor=north] at (0,-9){\footnotesize $k$};
        \node[anchor=north west] at (-2,-9.5){\footnotesize $\text{c}_k^0 + \text{c}_k^1$};
        \node[anchor=north west] at (-0.3,-9.5){\footnotesize $ \delta_{k}^1$}; 
        \node[anchor=north west] at (0.8,-9.5){\footnotesize $ \delta_{k}^0$}; 
\end{tikzpicture}
}
\caption{Extrapolation curves for case study 1, scenarios A) and B).}
\label{fig: extrapolation_curves_hernan_A_B}
\end{figure}

\subsubsection{Comments on scenarios C) and F)}
\label{sec_app: comments_cf_study_case_1}
To explain why the curves are barely distinguishable, consider the following heuristic argument. The survival of interest, $S_k$, can be expressed as  $S_k = \expectation S_{k}(E)$, where $S_{k}(e)$ is defined as in \eqref{eq: survival_conditional} (and alternatively expressible as $S_{k}(e) \coloneqq \expectation \left(\frac{\mathbb{1}_{\{E=e\}}}{P(E=e)} \cdot \prod_{m=0}^{k}\left(1-h_m^{a_m}(V)\right)\right)$) and denotes the survival of interest at follow-up $k$ for individuals $\{E=e\}$. Because approximately 95\% of the population entered the study during the calendar-time period $\{0, \dots, 5\}$, i.e., $P\{E \leq  5\} \simeq 0.95$, $S_k$ is dominated by $S_{k}(e)$ for $e \in \{0,\dots, 5\}$, i.e., $S_k \simeq \expectation( \mathbb{1}_{\{E \leq 5\}} \cdot S_{k}(E))$. Furthermore, since $e \mapsto S_k(e)$ is sufficiently smooth in $\{0, \dots, 5\}$, at every $k \in \{0, \dots, 15=K\}$, $S_k$ can be approximated using two survival quantities, $S_{\prstudyp}$ and $S_{\prgapp}$, as $S_k\simeq S_{\prstudyp} \cdot P\{E \leq 5, E \leq \mathcal{T}-k\} + S_{\prgapp} \cdot P\{E \leq 5, E > \mathcal{T}-k\}$. 

First consider $k \in \{0,\dots, 10\}$. We have $P\{E \leq 5, E > \mathcal{T}-k\} = 0$ and thus $S_k\simeq S_{\prstudyp} \cdot P\{E \leq 5, E \leq \mathcal{T}-k\}$. Hence,
{\scriptsize
\begin{align*}
    h_{k}^{a_k=1} = 1 - \frac{\expectation((\mathbb{1}_{\{E+k \leq \mathcal{T}\}} + \mathbb{1}_{\{E+k > \mathcal{T}\}} ) \cdot S_{k}(E))}{\expectation((\mathbb{1}_{\{E+k \leq \mathcal{T}\}} + \mathbb{1}_{\{E+k > \mathcal{T}\}} ) \cdot S_{k-1}(E))} \simeq 1 - \frac{\expectation(\mathbb{1}_{\{E+k \leq \mathcal{T}\}} \cdot S_{k}(E))}{\expectation(\mathbb{1}_{\{E+k \leq \mathcal{T}\}} \cdot S_{k-1}(E))}=  h_{\prstudyp, k}^{a_k=1}
\end{align*}
}
and therefore $\hat{S}_k^{*, \text{km}} \simeq S_k$. Using analogous arguments, it can be shown that $\hat{S}_k^{*, \text{msm}} \simeq \hat{S}_k^{*, \text{msm}-\text{we}}$. Thus, since $\hat{S}_k^{*, \text{msm}}$ targets a parameter that, similarly to $S_k$, can be approximated via $S_{\prstudyp} \cdot P\{E \leq 5, E \leq \mathcal{T}-k\} + S_{\prgapp}^{\sim} \cdot P\{E \leq 5, E > \mathcal{T}-k\}$, $\hat{S}_k^{*, \text{msm}}$ approximately targets $S_{\prstudyp} \cdot  P\{E \leq 5, E \leq \mathcal{T}-k\} \simeq S_k$. \footnote{$S_{\prstudyp}$ in this last expression is the same as used to decompose $S_k$; this follows from the fact that in scenario C) $\RAID$ holds and $\mathcal{M}(\mathbb{B}_{\text{BC}}^{\text{EF}}) \ni p(v,a_k,k;\beta^*) \stackrel{\prstudyp}{=} h_k^{a_k}(v))$.}

Next consider $k \in \{11, \dots, 15\}$. Now $P\{E \leq 5, E > \mathcal{T}-k\} = P(E \in \{\mathcal{T}-k, \dots, 5\})$. Because $e \mapsto S_{k}(e)$ is sufficiently smooth $e \in \{0,\dots,5\}$, these conditional survivals are approximately equal for $k$ close to $10$, but they increasingly diverge as $k$ grows. Consequently, 
$S_k \simeq S_{\prstudyp} \cdot P(E \in \{ 0, \dots, \mathcal{T}-k\}) + S_{\prgapp} \cdot P( E \in \{\mathcal{T}-k, \dots, 5\}) \simeq S_{\prstudyp} \cdot P(E \in \{ 0, \dots, 5\})$ for $k \simeq 10$ and this approximation becomes increasingly worse as $k$
increases. Thus, $h_k^{a_k}$ increasingly diverges from $h_{\prstudyp,k}^{a_k}$ and using similar arguments, $h_k^{a_k}(v)$ increasingly diverges from
$h_{\prstudyp,k}^{a_k}(v)$. See Table \ref{table: case_study_1_estimates} for the numerical results.

\subsubsection{Numerical results}
\label{sec_app: numerical_results_case_study_1}
\begin{landscape}
\begin{table}[p]
\centering
\scriptsize
\setlength{\tabcolsep}{4pt}
\begin{minipage}[t]{0.32\linewidth}
\centering
\captionof{table}{A)}
\begin{tabular}{rrrr}
\hline
$k^{\prime} \equiv k+1$ & $\hat{S}_k^{*, \text{km}}$ & $\hat{S}_k^{*, \text{msm}}$ & $\hat{S}_k^{*, \text{msm-we}}$ \\
\hline
   0 & 1.000 & 1.000 & 1.000 \\ 
     1 & 0.945 & 0.945 & 0.949 \\ 
     2 & 0.894 & 0.896 & 0.900 \\ 
     3 & 0.847 & 0.853 & 0.854 \\ 
     4 & 0.802 & 0.813 & 0.811 \\ 
     5 & 0.759 & 0.778 & 0.769 \\ 
     6 & 0.719 & 0.746 & 0.730 \\ 
     7 & 0.683 & 0.716 & 0.692 \\ 
     8 & 0.649 & 0.690 & 0.655 \\ 
     9 & 0.617 & 0.665 & 0.620 \\ 
    10 & 0.588 & 0.642 & 0.586 \\ 
    11 & 0.560 & 0.621 & 0.554 \\ 
    12 & 0.534 & 0.602 & 0.523 \\ 
    13 & 0.512 & 0.584 & 0.493 \\ 
    14 & 0.493 & 0.567 & 0.465 \\ 
    15 & 0.474 & 0.551 & 0.438 \\ 
    16 & 0.440 & 0.536 & 0.413 \\ 
\hline
\end{tabular}
\end{minipage}
\hfill
\begin{minipage}[t]{0.32\linewidth}
\centering
\captionof{table}{B)}
\begin{tabular}{rrrr}
\hline
$k^{\prime} \equiv k+1$ & $\hat{S}_k^{*, \text{km}}$ & $\hat{S}_k^{*, \text{msm}}$ & $\hat{S}_k^{*, \text{msm-we}}$ \\
\hline
   0 & 1.000 & 1.000 & 1.000 \\ 
     1 & 0.945 & 0.945 & 0.949 \\ 
     2 & 0.894 & 0.896 & 0.901 \\ 
     3 & 0.847 & 0.852 & 0.855 \\ 
     4 & 0.802 & 0.813 & 0.811 \\ 
     5 & 0.760 & 0.778 & 0.769 \\ 
     6 & 0.721 & 0.745 & 0.729 \\ 
     7 & 0.683 & 0.716 & 0.691 \\ 
     8 & 0.649 & 0.689 & 0.655 \\ 
     9 & 0.617 & 0.665 & 0.620 \\ 
    10 & 0.586 & 0.642 & 0.586 \\ 
    11 & 0.558 & 0.621 & 0.553 \\ 
    12 & 0.532 & 0.602 & 0.522 \\ 
    13 & 0.510 & 0.584 & 0.493 \\ 
    14 & 0.485 & 0.567 & 0.465 \\ 
    15 & 0.453 & 0.551 & 0.438 \\ 
    16 & 0.434 & 0.536 & 0.413 \\ 
\hline
\end{tabular}
\end{minipage}
\hfill
\begin{minipage}[t]{0.32\linewidth}
\centering
\captionof{table}{C)}
\begin{tabular}{rrrr}
\hline
$k^{\prime} \equiv k+1$ & $\hat{S}_k^{*, \text{km}}$ & $\hat{S}_k^{*, \text{msm}}$ & $\hat{S}_k^{*, \text{msm-we}}$ \\
\hline
   0 & 1.000 & 1.000 & 1.000 \\ 
     1 & 0.901 & 0.901 & 0.902 \\ 
     2 & 0.818 & 0.818 & 0.819 \\ 
     3 & 0.747 & 0.747 & 0.749 \\ 
     4 & 0.687 & 0.687 & 0.690 \\ 
     5 & 0.635 & 0.635 & 0.639 \\ 
     6 & 0.590 & 0.591 & 0.596 \\ 
     7 & 0.552 & 0.552 & 0.557 \\ 
     8 & 0.517 & 0.519 & 0.524 \\ 
     9 & 0.487 & 0.489 & 0.494 \\ 
    10 & 0.460 & 0.463 & 0.466 \\ 
    11 & 0.438 & 0.440 & 0.442 \\ 
    12 & 0.417 & 0.419 & 0.419 \\ 
    13 & 0.396 & 0.400 & 0.398 \\ 
    14 & 0.377 & 0.383 & 0.379 \\ 
    15 & 0.360 & 0.368 & 0.361 \\ 
    16 & 0.346 & 0.353 & 0.344 \\ 
\hline
\end{tabular}
\end{minipage}

\par\medskip

\begin{minipage}[t]{0.32\linewidth}
\centering
\captionof{table}{D)}
\begin{tabular}{rrrr}
\hline
$k^{\prime} \equiv k+1$ & $\hat{S}_k^{*, \text{km}}$ & $\hat{S}_k^{*, \text{msm}}$ & $\hat{S}_k^{*, \text{msm-we}}$ \\
  \hline
   0 & 1.000 & 1.000 & 1.000 \\ 
     1 & 0.945 & 0.945 & 0.949 \\ 
     2 & 0.897 & 0.896 & 0.903 \\ 
     3 & 0.853 & 0.852 & 0.861 \\ 
     4 & 0.813 & 0.813 & 0.824 \\ 
     5 & 0.778 & 0.777 & 0.790 \\ 
     6 & 0.745 & 0.745 & 0.759 \\ 
     7 & 0.717 & 0.716 & 0.731 \\ 
     8 & 0.690 & 0.689 & 0.705 \\ 
     9 & 0.666 & 0.665 & 0.682 \\ 
    10 & 0.644 & 0.642 & 0.660 \\ 
    11 & 0.624 & 0.621 & 0.640 \\ 
    12 & 0.605 & 0.602 & 0.621 \\ 
    13 & 0.588 & 0.584 & 0.604 \\ 
    14 & 0.572 & 0.567 & 0.588 \\ 
    15 & 0.556 & 0.551 & 0.572 \\ 
    16 & 0.542 & 0.537 & 0.558 \\ 
   \hline
\end{tabular}
\end{minipage}
\hfill
\begin{minipage}[t]{0.32\linewidth}
\centering
\captionof{table}{E)}
\begin{tabular}{rrrr}
\hline
$k^{\prime} \equiv k+1$ & $\hat{S}_k^{*, \text{km}}$ & $\hat{S}_k^{*, \text{msm}}$ & $\hat{S}_k^{*, \text{msm-we}}$ \\
  \hline
   0 & 1.000 & 1.000 & 1.000 \\ 
     1 & 0.945 & 0.945 & 0.949 \\ 
     2 & 0.897 & 0.897 & 0.903 \\ 
     3 & 0.854 & 0.854 & 0.863 \\ 
     4 & 0.816 & 0.816 & 0.828 \\ 
     5 & 0.783 & 0.783 & 0.796 \\ 
     6 & 0.753 & 0.753 & 0.768 \\ 
     7 & 0.726 & 0.726 & 0.743 \\ 
     8 & 0.703 & 0.703 & 0.721 \\ 
     9 & 0.682 & 0.682 & 0.701 \\ 
    10 & 0.663 & 0.663 & 0.684 \\ 
    11 & 0.647 & 0.646 & 0.668 \\ 
    12 & 0.632 & 0.631 & 0.653 \\ 
    13 & 0.619 & 0.617 & 0.640 \\ 
    14 & 0.606 & 0.605 & 0.628 \\ 
    15 & 0.594 & 0.593 & 0.617 \\ 
    16 & 0.584 & 0.583 & 0.607 \\ 
   \hline
\end{tabular}
\end{minipage}
\hfill
\begin{minipage}[t]{0.32\linewidth}
\centering
\captionof{table}{F)}
\begin{tabular}{rrrr}
\hline
$k^{\prime} \equiv k+1$ & $\hat{S}_k^{*, \text{km}}$ & $\hat{S}_k^{*, \text{msm}}$ & $\hat{S}_k^{*, \text{msm-we}}$ \\
  \hline
   0 & 1.000 & 1.000 & 1.000 \\ 
     1 & 0.901 & 0.901 & 0.901 \\ 
     2 & 0.818 & 0.818 & 0.819 \\ 
     3 & 0.748 & 0.747 & 0.750 \\ 
     4 & 0.688 & 0.687 & 0.691 \\ 
     5 & 0.636 & 0.635 & 0.641 \\ 
     6 & 0.591 & 0.591 & 0.597 \\ 
     7 & 0.552 & 0.552 & 0.558 \\ 
     8 & 0.518 & 0.519 & 0.524 \\ 
     9 & 0.489 & 0.489 & 0.494 \\ 
    10 & 0.463 & 0.463 & 0.468 \\ 
    11 & 0.441 & 0.441 & 0.445 \\ 
    12 & 0.421 & 0.420 & 0.426 \\ 
    13 & 0.402 & 0.403 & 0.408 \\ 
    14 & 0.386 & 0.387 & 0.393 \\ 
    15 & 0.373 & 0.374 & 0.380 \\ 
    16 & 0.362 & 0.362 & 0.368 \\ 
   \hline
\end{tabular}
\end{minipage}
\caption{Survival curves values for case study 1.}
\label{table: case_study_1_estimates}
\end{table}
\end{landscape}

\begin{table}[ht]
\centering
\small

\captionof{table}{scenarios A) and B)}
\scriptsize
\begin{tabular}{ccccccc}
  \toprule
$k$ & $\text{c}_k^0+\text{c}_k^1$ & $\text{c}_k^0$ & $\text{c}_k^1$ & $\delta_k^0$ & $\delta_k^1$ \\ 
  \midrule
  1 & 35.72 & 0.03 & 35.70 & -0.05 & 0.64  \\ 
    2 & 15.28 & 0.07 & 15.22 & -0.05 & 0.30 \\
    3 & 8.88 & 0.11 & 8.77 & -0.05 & 0.18 \\
    4 & 5.96 & 0.17 & 5.79 & -0.04 & 0.11 \\
    5 & 4.41 & 0.24 & 4.17 & -0.04 & 0.08 \\
    6 & 3.52 & 0.31 & 3.21 & -0.04 & 0.05 \\
    7 & 2.99 & 0.38 & 2.60 & -0.03 & 0.04 \\
    8 & 2.66 & 0.45 & 2.21 & -0.03 & 0.03 \\
    9 & 2.46 & 0.51 & 1.95 & -0.03 & 0.03 \\
   10 & 2.33 & 0.57 & 1.76 & -0.02 & 0.02 \\
   11 & 2.25 & 0.61 & 1.63 & -0.02 & 0.02 \\
   12 & 2.19 & 0.65 & 1.54 & -0.02 & 0.02 \\
   13 & 2.15 & 0.68 & 1.46 & -0.01 & 0.01 \\
   14 & 2.12 & 0.71 & 1.41 & -0.01 & 0.01 \\
   15 & 2.10 & 0.73 & 1.36 & -0.01 & 0.01 \\
   \bottomrule
\end{tabular}

\vspace{1em} 

\captionof{table}{scenario C)}
\scriptsize
\begin{tabular}{ccccccc}
  \toprule
$k$ & $\text{c}_k^0+\text{c}_k^1$ & $\text{c}_k^0$ & $\text{c}_k^1$ & $\delta_k^0$ & $\delta_k^1$ \\
  \midrule
   13 & 7.54 & 0.13 & 7.41 & -0.10 & 0.43 \\
   14 & 5.11 & 0.20 & 4.90 & -0.09 & 0.21 \\
   15 & 3.68 & 0.30 & 3.38 & -0.05 & 0.11 \\
   \bottomrule
\end{tabular}
\caption{Extrapolation curves for case study 1}
\end{table}

\begin{figure}[H]
    \centering
{
    \begin{tikzpicture}[scale=0.55]
        \node[anchor=north, scale=0.4] at (0,0){%
        \pgfimage{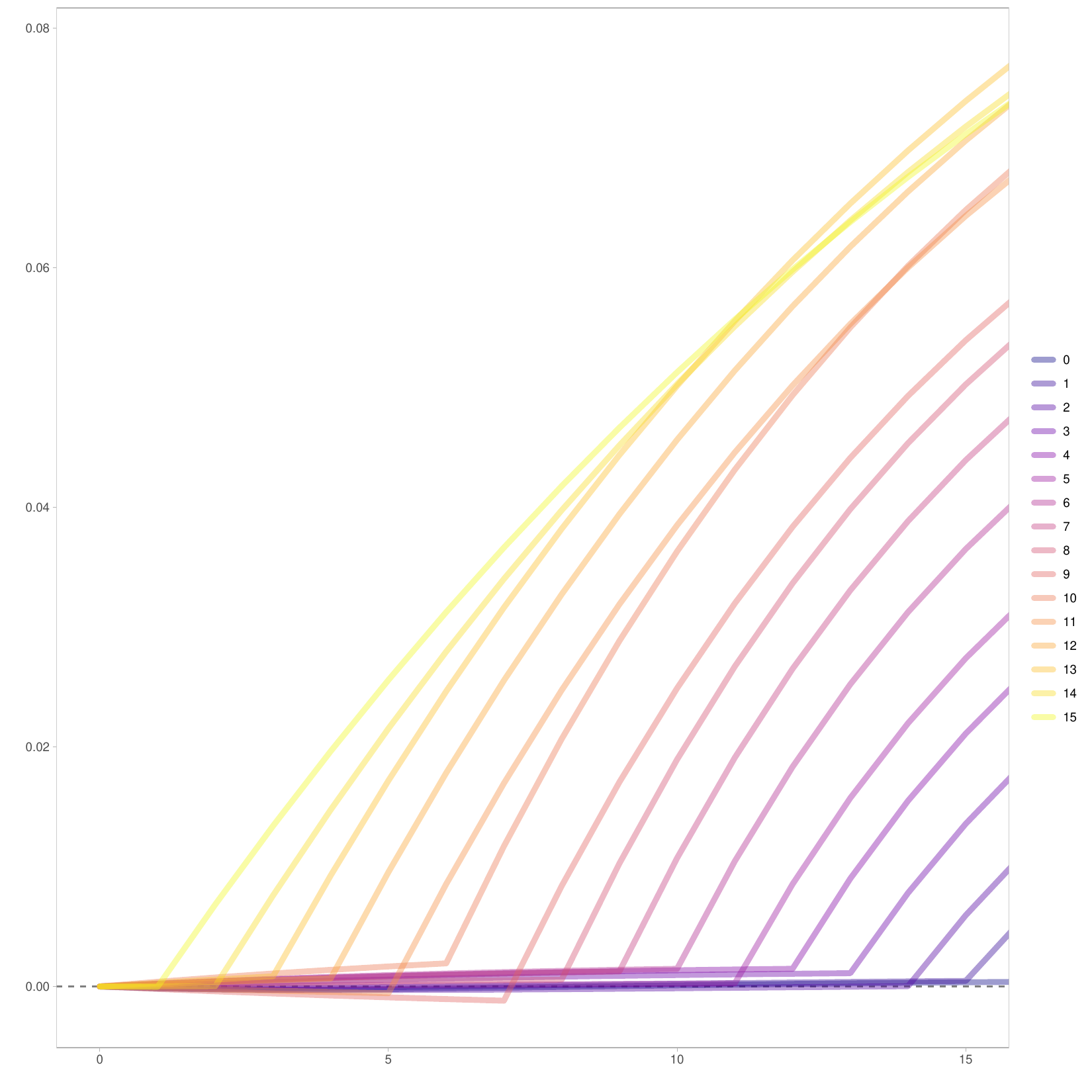}
        };    
        \node[anchor=north] at (0,-21){\small $k$};
        \node[anchor=north] at (9.2,-5.8){\footnotesize $e$};
        \node[anchor=north] at (-11.8,-6){\scriptsize $S_{k}(e) - \hat{S}_{k}^{*,\text{msm}}(e)$};

\end{tikzpicture}
}
\vspace{0.5cm}
\caption{ $S_{k}(e) - \hat{S}_{k}^{*, \text{msm}}(e)$ for $(e,k) \in \{0, \dots, K\} \times \{0, \dots, \mathcal{T}\} = {\{0, \dots, K\}}^2$ for case-study 1, scenario C).}
\label{fig: s_e_diff}
\end{figure}

\subsection{Case study 2}
\label{sec_app: case_study_2}
\subsubsection{Specifics and numerical results}
\label{sec_app: specifics_numerical_case_study_2}
The assumption $\RAID$ mentioned in the main text is articulated as in setting $\hyperref[table: conf_3]{\textbf{III}}$ in Appendix \ref{sec_app: censoring_assumptions_meaning_violation}, with 
\begin{align*}
    L_0 = \{ \text{treatment year}, \text{gender}, \text{ethnicity}, 
\text{age at ici start}, \text{ecog}, \text{liver failure}\}.
\end{align*}
See \citep{grippin_sars-cov-2_2025} for details on the covariates. 
The time since study-entry was available at the year level, not at the month level. Therefore, to compute c-values and extrapolation curves, we assumed individuals were uniformly spread within each year. Survival estimates, along with 95\% confidence intervals, are tabulated in Table \ref{table: survival_case_study_2}. Numerical results for extrapolation curves are tabulated in Table \ref{table: extrapolation_curves_case_study_2}.

\begin{table}[ht]
\centering
\tiny
\begin{tabular}{lcccccc}
  \toprule
$k^{\prime} \equiv k+1$ & Vaccine $(a=1)$ & ${\text{CI}}_{\text{low}}^{a=1}$ & ${\text{CI}}_{\text{high}}^{a=1}$ & No vaccine $(a=0)$ & ${\text{CI}}_{\text{low}}^{a=0}$ & ${\text{CI}}_{\text{high}}^{a=0}$ \\ 
  \midrule
0 & 1.00 & 1.00 & 1.00 & 1.00 & 1.00 & 1.00 \\ 
  1 & 0.99 & 0.97 & 1.00 & 0.99 & 0.97 & 1.00 \\ 
  2 & 0.96 & 0.94 & 0.98 & 0.96 & 0.94 & 0.98 \\ 
  3 & 0.95 & 0.92 & 0.97 & 0.95 & 0.92 & 0.97 \\ 
  4 & 0.93 & 0.90 & 0.96 & 0.93 & 0.90 & 0.96 \\ 
  5 & 0.93 & 0.89 & 0.95 & 0.90 & 0.83 & 0.94 \\ 
  6 & 0.92 & 0.88 & 0.95 & 0.88 & 0.81 & 0.93 \\ 
  7 & 0.90 & 0.86 & 0.94 & 0.85 & 0.77 & 0.91 \\ 
  8 & 0.88 & 0.81 & 0.93 & 0.85 & 0.77 & 0.91 \\ 
  9 & 0.86 & 0.79 & 0.92 & 0.82 & 0.73 & 0.88 \\ 
  10 & 0.82 & 0.75 & 0.89 & 0.77 & 0.68 & 0.84 \\ 
  11 & 0.82 & 0.74 & 0.88 & 0.72 & 0.63 & 0.82 \\ 
  12 & 0.81 & 0.73 & 0.87 & 0.71 & 0.61 & 0.81 \\ 
  13 & 0.77 & 0.68 & 0.84 & 0.71 & 0.61 & 0.81 \\ 
  14 & 0.77 & 0.68 & 0.84 & 0.68 & 0.58 & 0.78 \\ 
  15 & 0.74 & 0.64 & 0.82 & 0.64 & 0.54 & 0.75 \\ 
  16 & 0.71 & 0.61 & 0.79 & 0.62 & 0.53 & 0.73 \\ 
  17 & 0.69 & 0.59 & 0.78 & 0.60 & 0.50 & 0.72 \\ 
  18 & 0.66 & 0.56 & 0.75 & 0.60 & 0.50 & 0.72 \\ 
  19 & 0.66 & 0.56 & 0.75 & 0.58 & 0.48 & 0.71 \\ 
  20 & 0.63 & 0.53 & 0.73 & 0.57 & 0.47 & 0.70 \\ 
  21 & 0.63 & 0.53 & 0.73 & 0.56 & 0.46 & 0.69 \\ 
  22 & 0.62 & 0.53 & 0.72 & 0.54 & 0.44 & 0.66 \\ 
  23 & 0.62 & 0.52 & 0.71 & 0.54 & 0.44 & 0.66 \\ 
  24 & 0.62 & 0.52 & 0.71 & 0.54 & 0.44 & 0.66 \\ 
  25 & 0.62 & 0.52 & 0.71 & 0.52 & 0.42 & 0.65 \\ 
  26 & 0.60 & 0.50 & 0.70 & 0.49 & 0.38 & 0.62 \\ 
  27 & 0.60 & 0.50 & 0.70 & 0.49 & 0.38 & 0.62 \\ 
  28 & 0.60 & 0.50 & 0.70 & 0.49 & 0.38 & 0.62 \\ 
  29 & 0.60 & 0.50 & 0.70 & 0.46 & 0.35 & 0.59 \\ 
  30 & 0.60 & 0.50 & 0.70 & 0.46 & 0.35 & 0.59 \\ 
  31 & 0.59 & 0.49 & 0.68 & 0.46 & 0.35 & 0.59 \\ 
  32 & 0.55 & 0.44 & 0.67 & 0.46 & 0.35 & 0.59 \\ 
  33 & 0.55 & 0.44 & 0.67 & 0.44 & 0.34 & 0.58 \\ 
  34 & 0.54 & 0.43 & 0.66 & 0.43 & 0.32 & 0.56 \\ 
  35 & 0.54 & 0.43 & 0.66 & 0.43 & 0.32 & 0.56 \\ 
  36 & 0.54 & 0.43 & 0.66 & 0.43 & 0.32 & 0.56 \\ 
  37 & 0.54 & 0.43 & 0.66 & 0.43 & 0.32 & 0.56 \\ 
  38 & 0.48 & 0.32 & 0.62 & 0.43 & 0.32 & 0.56 \\ 
  39 & 0.48 & 0.32 & 0.62 & 0.43 & 0.32 & 0.56 \\ 
  40 & 0.48 & 0.32 & 0.62 & 0.40 & 0.29 & 0.54 \\ 
  41 & 0.48 & 0.32 & 0.62 & 0.40 & 0.29 & 0.54 \\ 
  42 & 0.48 & 0.32 & 0.62 & 0.40 & 0.29 & 0.54 \\ 
  43 & 0.48 & 0.32 & 0.62 & 0.40 & 0.29 & 0.54 \\ 
  44 & 0.48 & 0.32 & 0.62 & 0.40 & 0.29 & 0.54 \\ 
  45 & 0.48 & 0.32 & 0.62 & 0.40 & 0.29 & 0.54 \\ 
   \bottomrule
\end{tabular}
\caption{Estimates $\hat{S}_k^{*,\text{km}, a}$ of \eqref{estimand_mrna} for case study 2.} \label{table: survival_case_study_2}
\end{table}

\begin{table}[H]
\centering
\scriptsize
\begin{tabular}{ccccccccccc}
  \toprule
$k$ & $\textnormal{c}_k^0 + \textnormal{c}_k^1$ & $\textnormal{c}_k^0$ & $\textnormal{c}_k^1$ & $\delta_k^0$ & $\delta_k^1$ \\ 
  \midrule
    27 & 21.702 & 0.046 & 21.656 & -0.042 & 0.945 \\
    28 & 15.210 & 0.066 & 15.144 & -0.039 & 0.463 \\
    29 & 11.311 & 0.089 & 11.222 & -0.036 & 0.280 \\
    30 & 8.801 & 0.115 & 8.686 & -0.031 & 0.188  \\
    31 & 7.154 & 0.143 & 7.011 & -0.029 & 0.139 \\
    32 & 6.007 & 0.171 & 5.836 & -0.026 & 0.107 \\
    33 & 5.100 & 0.204 & 4.896 & -0.023 & 0.083 \\
    34 & 4.432 & 0.238 & 4.194 & -0.021 & 0.066 \\
    35 & 3.920 & 0.274 & 3.646 & -0.019 & 0.053 \\
    36 & 3.518 & 0.312 & 3.206 & -0.017 & 0.044 \\
    37 & 3.220 & 0.348 & 2.872 & -0.016 & 0.037 \\
    38 & 2.987 & 0.384 & 2.603 & -0.014 & 0.031 \\
    39 & 2.805 & 0.419 & 2.385 & -0.013 & 0.027 \\
    40 & 2.661 & 0.453 & 2.209 & -0.012 & 0.023 \\
    41 & 2.548 & 0.485 & 2.063 & -0.012 & 0.020 \\
    42 & 2.457 & 0.515 & 1.942 & -0.011 & 0.018 \\
    43 & 2.386 & 0.543 & 1.843 & -0.010 & 0.016 \\
    44 & 2.327 & 0.569 & 1.758 & -0.009 & 0.015 \\
   \bottomrule
\end{tabular}
\caption{Numerical results for extrapolation curves for case study 2.}
\label{table: extrapolation_curves_case_study_2}
\end{table}

\subsubsection{Additional complications in the analysis of \citep{grippin_sars-cov-2_2025}}
\label{sec_app: additional_complications_grippin}
Here, we illustrate in which sense the survival curves reported in \citet{grippin_sars-cov-2_2025}[Figure 1b, page 2] are biased and provide an explanation of how such a bias relates to \textit{study-entry-determined immortal time bias} introduced in Appendix \ref{sec_app: supplementary_wsc_deferred_main_text}.

Consider the same data structure as in Section \ref{sec: case_study_mrna_vaccine} of the main text; restrict the population to non-small cell lung cancer (NSCLC) (n=138) and conceptualize such observations, as implied by the methods used in the original analysis (continuous-time Kaplan-Meier estimator), as drawn from a population of i.i.d. individuals. Taking overall survival as the outcome, the discrete-time counterfactual formulation of the (implicit) estimands of interest is
\begin{align}\label{grippin_estimand_mrna}
    S_{k\mid a}^{a} \coloneqq P(Y_k^{a, \overline{c}_{k-1}=0} = 0\mid A=a), \, a \in \{0,1\} \\ 
    \text{ for } k \in\{0, \dots, K=40\}.
\end{align}
That is, the implicit contrast of interest is 
survival under SARS-CoV-2 mRNA vaccination $(a=1)$, among vaccinated $\{A = 1\}$, vs. survival under no vaccination $(a=0)$, among unvaccinated $\{A = 0\}$, in the absence of censoring.

To faithfully reproduce the results of \cite{grippin_sars-cov-2_2025} we considered WSC-Kaplan-Meier survival curves estimates, $\hat{S}_k^{*,\text{km},a=1}$ and $\hat{S}_k^{*,\text{km},a=0}$, where we did not correct for informative censoring ($P^{\textnormal{o}}=P^*$) and we did not consider a grace period.  Figure \ref{fig: nsclc_stage_iii} presents the two curves.
\begin{figure}[H]
    \centering%
{
    \begin{tikzpicture}[scale=1]
        \node[anchor=north, scale=0.6] at (0,0){%
        \pgfimage{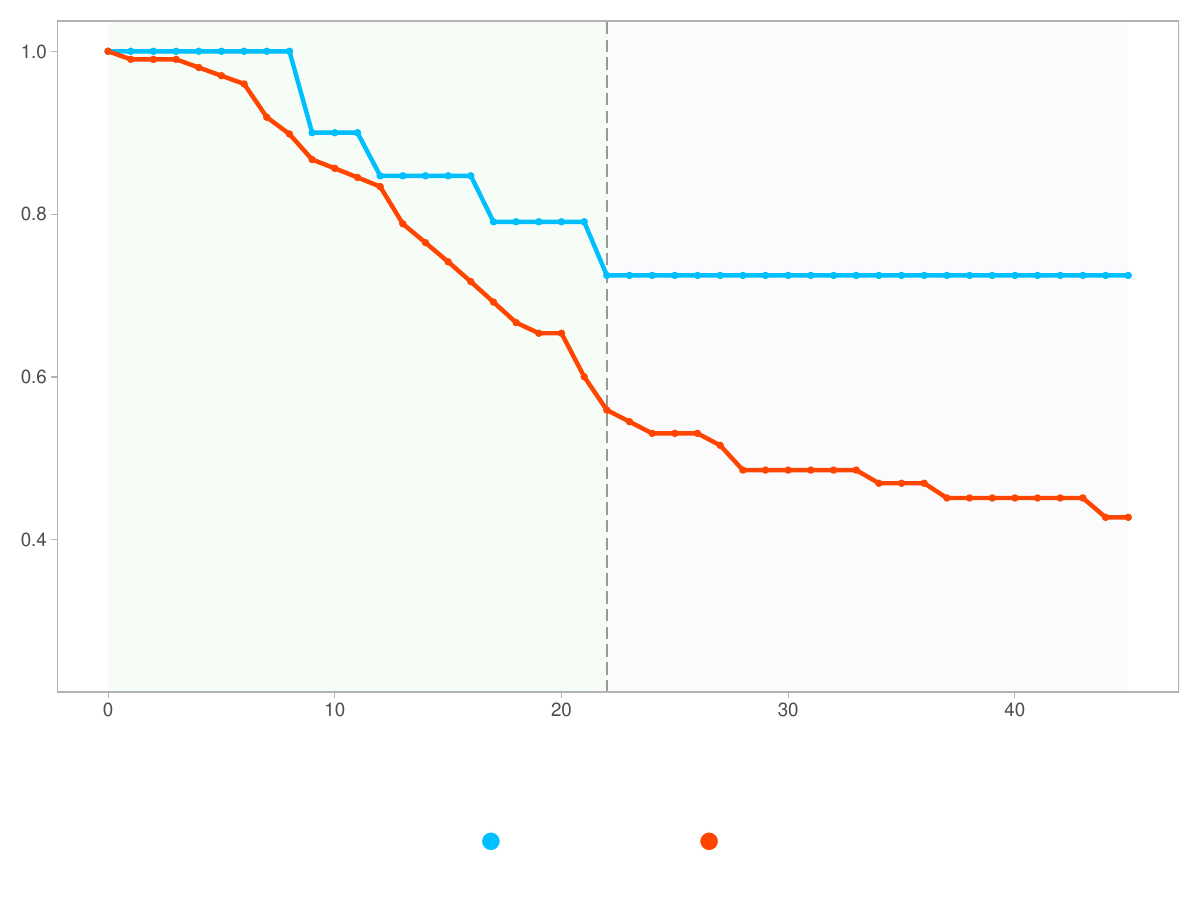}
        };    
        \node[anchor=north] at (0,-7.4){\scriptsize $k+1$};
        \node[anchor=north west] at (-1.4,-7.8){\small $\hat{S}_k^{*,\text{km},a=1}$};
        \node[anchor=north west] at (0.8,-7.8){\small $\hat{S}_k^{*,\text{km},a=0}$}; 
\end{tikzpicture}
}
\caption{$\hat{S}_k^{*,\text{km},a=1}$ and $\hat{S}_k^{*,\text{km},a=0}$ curves. The x-axis reports time $k^{\prime} = k+1$ (At $k^{\prime}=0$, all survival curves equal 1). The dashed-grey vertical line marks the time $k+1=21+1$ before (light-green area) and after (lavender area) which extrapolation assumptions are needed for the parameter targeted by $\hat{S}_k^{*,\text{km},a=1}$ to coincide with $S_{k\mid a=1}^{a=1}$.}
\label{fig: nsclc_stage_iii}
\end{figure}

Focus on $\hat{S}_k^{*,\text{km},a=1}$.  Among vaccinated individuals, no failures are observed after follow-up $k=21$ and therefore $\hat{S}_k^{*,\text{km},a=1}$ is constant for all $k>21$, i.e., the vaccination arm behaves as if follow-up had ended at $\Tilde{\mathcal{T}}=21$.

For the sake of simplicity, suppose that \CTC{} does not hold. Under $\RAID$ (formulated as in Appendix \ref{sec_app: censoring_assumptions_meaning_violation} with $L_0 = \emptyset$), the parameter targeted by $\hat{S}_k^{*,\text{km},a=1}$, $S_k^{*,a=1}$, can be shown to equal $S_{\text{min}(k,21)\mid a=1}^{a=1}$. Thus, $S_k^{*,a=1} - S_{k\mid a=1}^{a=1} = \mathbb{1}_{\{k > 21\}} \cdot \left(S_{21\mid a=1}^{a=1} - S_{k\mid a=1}^{a=1}\right)$, an expression that closely resembles that of \textit{study-entry-determined immortal time bias}, as both arise from extrapolating values after the end of the study. Accordingly, unless for every $k \in \{\Tilde{\mathcal{T}}+1=21+1, \dots, K=40\}$ $Y_{\Tilde{\mathcal{T}}}^{a=1, {c}_{\Tilde{\mathcal{T}}-1}=0} = 0  \Rightarrow Y_{k}^{a=1, {c}_{k-1}=0}=0$, $S_{k\mid a=1}^{a=1}$ undershoots the parameter targeted by $\hat{S}_k^{*,\text{km},a=1}$, $S_k^{*,a=1}$. An important consequence is that the protective effect of the vaccine for the given cohort is going to be inflated. Thus, estimates should be reported up to a maximum follow-up time $K^+ = 21$.

\begin{table}[ht]
\centering
\tiny
\begin{tabular}{lcc}
  \toprule
$k^{\prime} \equiv k+1$ & Vaccine ($a=1$) & No vaccine ($a=0$) \\ 
  \midrule
0 & 1.000 & 1.000 \\ 
  1 & 1.000 & 0.990 \\ 
  2 & 1.000 & 0.990 \\ 
  3 & 1.000 & 0.990 \\ 
  4 & 1.000 & 0.980 \\ 
  5 & 1.000 & 0.970 \\ 
  6 & 1.000 & 0.960 \\ 
  7 & 1.000 & 0.919 \\ 
  8 & 1.000 & 0.898 \\ 
  9 & 0.900 & 0.867 \\ 
  10 & 0.900 & 0.856 \\ 
  11 & 0.900 & 0.845 \\ 
  12 & 0.847 & 0.834 \\ 
  13 & 0.847 & 0.788 \\ 
  14 & 0.847 & 0.765 \\ 
  15 & 0.847 & 0.741 \\ 
  16 & 0.847 & 0.717 \\ 
  17 & 0.791 & 0.692 \\ 
  18 & 0.791 & 0.667 \\ 
  19 & 0.791 & 0.654 \\ 
  20 & 0.791 & 0.654 \\ 
  21 & 0.791 & 0.600 \\ 
  22 & 0.725 & 0.559 \\ 
  23 & 0.725 & 0.545 \\ 
  24 & 0.725 & 0.531 \\ 
  25 & 0.725 & 0.531 \\ 
  26 & 0.725 & 0.531 \\ 
  27 & 0.725 & 0.516 \\ 
  28 & 0.725 & 0.486 \\ 
  29 & 0.725 & 0.486 \\ 
  30 & 0.725 & 0.486 \\ 
  31 & 0.725 & 0.486 \\ 
  32 & 0.725 & 0.486 \\ 
  33 & 0.725 & 0.486 \\ 
  34 & 0.725 & 0.469 \\ 
  35 & 0.725 & 0.469 \\ 
  36 & 0.725 & 0.469 \\ 
  37 & 0.725 & 0.451 \\ 
  38 & 0.725 & 0.451 \\ 
  39 & 0.725 & 0.451 \\ 
  40 & 0.725 & 0.451 \\ 
  41 & 0.725 & 0.451 \\ 
   \bottomrule
\end{tabular}
\caption{Estimates $\hat{S}_k^{*,\text{km}, a}$ of \eqref{grippin_estimand_mrna}.}
\end{table}
\clearpage
\end{appendices}
    
    \renewcommand{\refname}{Appendix References}
    \renewcommand{\bibname}{Appendix References}
    
    \bibliographystyle{biom}
    \putbib[bibliography_appendices] 
\end{bibunit}

\end{document}